\documentclass[12pt,twoside]{bristol_thesis2-1}
\usepackage{latexsym}
\usepackage{cite}

\usepackage{url}
\usepackage{esint}
\usepackage{amsmath,amssymb,amsthm,latexsym}

\usepackage{graphicx}
\usepackage{esint}
\usepackage{float}
\usepackage{units}
\usepackage[all]{xy}

\usepackage{times}
\usepackage[all]{xy}
\usepackage{float}
\usepackage{nicefrac}

\newtheorem{Fact}{Fact}
\newtheorem{lemma}{Lemma}
\newtheorem{theorem}[lemma]{Theorem}

\newtheorem{remark}[lemma]{Remark}

\newtheorem{definition}[lemma]{Definition}
\newtheorem{example}[lemma]{Example}
\numberwithin{equation}{section} \numberwithin{lemma}{section}
\newcommand{\bbZ}{{\mathbb Z}}
\newcommand{\bbR}{{\mathbb R}}

\newcommand{\vx}{\mathbf{x}}
\newcommand\Vcalbar{\bar{\cal V}}
\newcommand\Vbar{{\bar V}}

\newcommand{\half}{{\textstyle \frac{1}{2}}}
\newcommand\Cbar{{\bar P}}
\newcommand\Gbar{{\bar \Gamma}}
\newcommand\Omegabar{{\bar \Omega}}

\begin{document}

\author{Adam Sawicki}

\title{Topology of graph configuration spaces and quantum statistics}

\titlepage

\pagenumbering{roman}

\prefacesection{Abstract}

In this thesis we develop a full characterization of abelian quantum statistics on graphs. We explain how the number of anyon phases is related to connectivity. For 2-connected graphs the independence of quantum statistics with respect to the number of particles is proven. For non-planar 3-connected graphs we identify bosons and fermions as the only possible statistics, whereas for planar 3-connected graphs we show that one anyon phase exists. Our approach also yields an alternative proof of the structure theorem for the first homology group of n-particle graph configuration spaces. Finally, we determine the topological gauge potentials for 2-connected graphs. Moreover we present an alternative application of discrete Morse theory for two-particle graph configuration spaces. In contrast to previous constructions, which are based on discrete Morse vector fields, our approach is through Morse functions, which have a nice physical interpretation as two-body potentials constructed from one-body potentials. We also give a brief introduction to discrete Morse theory. 

\prefacesection{Dedication}

I dedicate this thesis to my family and friends. 

\prefacesection{Acknowlegdements}

Foremost, I would like to express my sincere gratitude to my advisors Prof. Jon P. Keating and Dr Jonathan M. Robbins for discussions we had, their patience and constant encouragement.  

Besides my advisors, I would like to thank my fellow officemates in Mathematics Department: Jo Dwyer, Orestis Georgiou, Jo Hutchinson, Andy Poulton and James Walton (in fact a chemist) for all the fun we have had in the last three years. I hope our friendship will continue! I also wish to thank my friend Rami Band for many stimulating discussions we had on the subject of this thesis and the experience of the 2011 Royal Society Summer Science Exhibition.   

Finally, I would like to thank my family, in particular my sister Ania for being so positively crazy!

\declarationpage

\comment{
\prefacesection{Author's Declaration}
I declare that the work in this dissertation was carried out in accordance with the requirements of the University's Regulations and Code of Practice for Research Degree Programmes and that it has not been submitted for any other academic award. Except where indicated by specific reference in the text, the work is the candidate's own work. Work done in collaboration with, or with the assistance of, others, is indicated as such. Any views expressed in the dissertation are those of the author.
\vspace*{20mm}

\begin{flushright}
\begin{tabular*}{0.75\textwidth}{cccr}
  \hline
  \empty & \empty & \empty & \empty \\
\end{tabular*}
\end{flushright}

\vspace*{-10mm}
\hspace*{65mm}A.N. Other

\vspace*{15mm}

\hspace*{27mm} Date:}

\tableofcontents

\listoffigures


\chapter{Introduction}
\pagenumbering{arabic}
This thesis concerns the characterization of quantum statistics on graphs. Naturally, one should first explain what it means. As with many problems in mathematical physics it is hard to do it in a one sentence. However, in the subsequent sections of the introduction it is done. The subject, as I see it, is inevitably connected to some basic concepts in algebraic topology and graph theory. The main purpose of this, rather short, introduction is to persuade the reader that quantum statistics and the first homology group of an appropriate configuration space are one and the same thing. I knowingly avoid using the full formalism of quantum mechanics on non-simply connected spaces. This can be found in many textbooks and in my opinion is not relevant to understand the problem and the main results of the thesis. Writing this text I tried to minimize the number of irrelevant details so that the key ideas and concepts were clearly visible. Therefore, for example, I do not prove theorems whose proofs do not contribute to the understanding of the main flow of the text. The interested reader is asked to consult the cited references. On the other hand, in order to make the manuscript available to a reader not familiar with homology groups and graph theory I include a basic discussion of the relevant facts. Although one can find it unnecessary, from time to time, I repeat definitions and key properties of some important objects. I believe that it is better to do this rather than to send the reader to a distant page where they were discussed for the first time.

The chapter is organized as follows: In section \ref{WhatIsStatistics} I shortly explain the concept of quantum statistics describing two approaches. The first one is standard and the second topological. Then in section \ref{AB-top-ex} the Aharonov-Bohm effect is discussed as an example of a topological phase. The subsequent five sections contain the discussion of basic properties of graphs, cell complexes and their homotopy and homology groups. Next, in section \ref{topologyAndStaistics} I define the many-particle configuration space and explain that its first homology group encodes the information about quantum statistics. The calculation for the case of particles living in $\mathbb{R}^2$ and $\mathbb{R}^n$, where $n\geq 3$ is included. Then in section \ref{graph-statistics} I generalize the above concept to graphs and introduce the basic mathematical object of this thesis, i.e. the discrete configuration space of $n$-particles, $\mathcal{D}^n(\Gamma)$. This space has the structure of a cell complex and is topologically equivalent to the configuration space of $n$-particles on a graph, $C_n(\Gamma)$. In the last section of this chapter I discuss the tight-binding model of $n$-particles on a graph, define the topological gauge potentials and explain the connection between them, the first homology and quantum statistics. The background material of the introduction is mostly based on \cite{Nakahara} and \cite{tutte01}.

\section{Quantum statistics } \label{WhatIsStatistics}

In this section I describe two approaches to quantum statistics.
The first one introduces it as an additional postulate of quantum
mechanics. The second, which I will follow throughout the thesis,
is topological in its nature.

\subsection{Standard approach to quantum statistics }

In quantum mechanics any quantum system is described by its underlying
Hilbert space. Let us denote by $\mathcal{H}_{1}$ the one-particle Hilbert
space, i.e. the Hilbert space of a single particle. By one of
the postulates of quantum mechanics the Hilbert space of $n$ distinguishable
particles, $\mathcal{H}_{n}$, is the tensor product of the Hilbert
spaces of the constituents, i.e.
\begin{gather*}
\mathcal{H}_{n}=\underbrace{\mathcal{H}_{1}\otimes\ldots\otimes\mathcal{H}_{1}}_{n}.
\end{gather*}
If we want to treat particles as indistinguishable some additional
modifications of $\mathcal{H}_n$ are required. First, the indistiguishability implies
that all observables need to commute with permutations of the particle
labels. Therefore, one decomposes $\mathcal{H}_{n}$ into irreducible
representations of the permutation group $S_{n}$:
\begin{gather*}
\mathcal{H}=\bigoplus_{\lambda}\mathcal{H}_{\lambda},
\end{gather*}
where $\lambda$ labels those representations. The components
$\mathcal{H}_{\lambda}$ represent essentially different permutation
symmetries. Note that {\it a priori} all components $\mathcal{H}_{\lambda}$
are equally good, i.e. none of them is distinguished in any way. The
distinction between them is due to symmetrization postulates of quantum
mechanics, i.e. physically realizable components $\mathcal{H}_{\lambda}$
are only
\begin{gather*}
\mbox{1. symmetric tensors}:\,\,\,\mathcal{H}_{\lambda}=S^{n}\mathcal{H}_{1},\\
\mbox{2. antisymmetric tensors}:\,\,\,\mathcal{H}_{\lambda}=\bigwedge^{n}\mathcal{H}_{1},
\end{gather*}
which are trivial and sign representations of the permutation group
$S_{n}$, respectively. The first one corresponds to bosons and the second to fermions.
Other components or equivalently other representations of $S_{n}$
are physically excluded. In order to decide if the considered
particles obey Bose or Fermi statistics one looks at the spin. The
spin-statistics theorem \cite{P40} says that particles with integer spin
are bosons and with half-integer, fermions. It is worth mentioning 
that at the level of non-relativistic quantum mechanics the spin-statistics
theorem is actually a postulate as it is proved only in the framework of
quantum filed theory. Nevertheless, there were attempts to deduce
it on the level of QM (see for example \cite{BR00}). The antisymmetric property of fermionic states
is also known as the Pauli exclusion principle which says that no two identical
fermions may occupy the same quantum state simultaneously. Finally,
let us mention that symmetrization postulate has an important consequences
if one looks at the energy distribution of many non-interacting particles.
More precisely, assume that we have a collection of non-interacting
indistinguishable particles and ask how they occupy a set of available
discrete energy states. Then the expected number of particles
in the $i$-th energy state is given by:
\begin{gather*}
n_{i}=\frac{g_{i}}{e^{\left(E_{i}-\mu\right)/kT}-1},\,\,\,\mbox{for bosons,}\\
n_{i}=\frac{g_{i}}{e^{\left(E_{i}-\mu\right)/kT}+1},\,\,\,\mbox{for fermions,}
\end{gather*}
where, $T$ is temperature, $k$ is Boltzmann constant and $g_{i}$ is
the degeneracy of the $E_{i}$ energy state.

\subsection{Topological approach to quantum statistics}

After discussing the standard way of introducing quantum statistics
we switch to the topological approach. Interestingly, it is based on
the topological properties of the classical configuration space.

In classical mechanics, particles are considered distinguishable.
Therefore, the $n$-particle configuration space is the Cartesian
product, $M^{\times n}$, where $M$ is the one-particle configuration
space. By contrast, in quantum mechanics elementary particles may
be considered indistinguishable. This conceptual difference in the
description of many-body systems prompted Leinaas and Myrheim \cite{LM77}
(see also \cite{S70,W90}) to study classical configuration spaces
of indistinguishable particles, $C_{n}(M)$ which led to the discovery
of anyon statistics. We first briefly describe the work of Leinaas and Myrheim.

As noted by the authors of \cite{LM77} indistinguishability of classical
particles places constraints on the usual configuration space, $M^{\times n}$.
Configurations that differ by particle exchange must be identified.
One also assumes that two classical particles cannot occupy the same
configuration. Consequently, the classical configuration space of
$n$ indistinguishable particles is the orbit space
\begin{gather*}
C_{n}(M)=(M^{\times n}-\Delta)/S_{n},
\end{gather*}
where $\Delta$ corresponds to the configurations for which at least
two particle are at the same point in $M$, and $S_{n}$ is the permutation
group. Significantly, the space $C_{n}(M)$ may have non-trivial topology.
As permuted configurations are identified in $C_{n}(M)$ any closed
curve in $C_{n}(M)$ corresponds to a process in which particles start
at some configuration and then return to the same configuration modulo
they might have been exchanged. Some of these curves are non-contractible
and therefore the space $C_{n}(M)$ has nontrivial fundamental group
$\pi_{1}(C_{n}(M))$.

\paragraph{Quantum mechanics on non-simply connected configuration spaces }

For many (or just one particle) whose classical configuration space
$\mathcal{C}$ is non-simply connected quantum mechanics allows an
additional freedom stemming from the non-triviality of the fundamental
group $\pi_{1}(\mathcal{C})$. In order to describe this freedom we
assume in the following that all particles are free, i.e. there are no
external fields and on the classical level they do not interact. In
the subsequent section we discuss in details the Aharonov-Bohm effect
which is an example of the general concept we describe here.

Let $A$ be a connection $1$-form of a $d$-dimensional vector bundle
over $\mathcal{C}$ with the structure group $U(d)$ (see \cite{Nakahara} for more details). As we do not
want to affect classical mechanics, we assume that the curvature $2$-form
$F=DA$ vanishes. In the following we will need the notion of the
holonomy group. Let $\gamma:[0,1]\rightarrow\mathcal{C}$ be a closed
curve. As we consider $d$-dimensional vector bundle, over any point
of $\gamma$ there is a $d$-dimensional vector space $V_{d}$. For
any vector $v_{0}$ over the point $\gamma(0)$ we consider the parallel
transport through $\gamma$. The result of this process is vector
$v_{1}$. Notably $v_{0}$ and $v_{1}$ need not to be the same. Therefore
to each loop one can assign a matrix $M_{\gamma}$ which depends only
on the loop and
\begin{gather*}
v_{1}=M_{\gamma}v_{0},\,\,\,\mbox{\ensuremath{\forall}}v_{0}\in V_{d}.
\end{gather*}
The collection of all matrices $M_{\gamma}$ for all loops based at
some fixed point $p\in\mathcal{C}$ is called the holonomy group. Moreover,
when $F=0$, $M_{\gamma}$ depends only on the homotopy type of
the loop. Therefore holonomy group is a $d$-dimensional representation
of the fundamental group (see section \ref{fund-group} for definition of fundamental group). When $d=1$ this representation is abelian
and assigns phase factors to non-contractible loops in $\mathcal{C}$.
When $d>1$ it assigns in general non-commuting unitary matrices to
non-contractible loops in $\mathcal{C}$. Finally, these matrices act
on $d$-component wavefunction.

\paragraph{Classical configuration spaces and quantum statistics}

In 1977 Leinaas and Myrheim \cite{LM77} considered the classical configuration
space of $n$ indistinguishable particles, $C_{n}(M)$ in the above
described context. Their work showed that the representations of the
fundamental group $\pi_{1}(C_{n}(M))$ determine all possible quantum
statistics. In particular they described in details the cases when
$M=\mathbb{R}^{2}$ and $M=\mathbb{R}^{k}$, where $k\geq3$. Notably
for $M=\mathbb{R}^{2}$ they found that the fundamental group is the
braid group which led to the discovery of anyon statistics. Similar
results were obtained by Laidlaw and  DeWitt \cite{LD71} who considered the problem
of quantum statistics using the language of path integrals. As
clearly pointed out by Dowker \cite{D85} when one is interested in the abelian
quantum statistics only, determination of the fundamental group is
not actually necessary. Instead, the first homology group which is the abelianized version
of $\pi_{1}(C_{n}(M))$ plays the major role. In this thesis we determine
it for graph configuration spaces. 

\section{Aharonov-Bohm Effect as an example of topological phase}\label{AB-top-ex}

In this section we discuss the Aharonov-Bohm effect. In particular we
explain the topological nature of the phase gained by the wavefunction
when it goes around the magnetic flux. Our exposition mainly follows
\cite{geoemtric-phase}.

In non-relativistic quantum mechanics the canonical commutation relations
for a free particle living in $n$-dimensional space $M$ are given
by:
\begin{gather}
[x_{i},x_{j}]=0=[p_{i},p_{j}],\quad[x_{i},p_{j}]=i\delta_{ij},\label{eq:commutation-relations}
\end{gather}
where $i,j\in\{1,\ldots,n\}$. The standard representation of position
and momenta operators satisfying (\ref{eq:commutation-relations})
is given by:
\begin{gather}
(x_{i}\Psi)(x)=x_{i}f(x),\quad(p_{i}\Psi)(x)=-i\frac{\partial}{\partial x_{i}}f(x).\label{eq:first-defi-momentum}
\end{gather}
It was perhaps first noticed\footnote{According to authors of \cite{geoemtric-phase}.
} by Dirac \cite{Drirac58}, that operators:
\begin{gather}
p_{i}^{\omega}=-i\frac{\partial}{\partial x_{i}}+\omega_{i}(x),\label{eq:secnod-momentum-defi}
\end{gather}
where
\begin{gather}
\omega=\sum_{i}\omega_{i}dx^{i},\,\,\,\, d\omega=0,
\end{gather}
satisfy the canonical commutation relations, i.e.
\begin{gather}
[x_{i},x_{j}]=0=[p_{i}^{\omega},p_{j}^{\omega}],\quad[x_{i},p_{j}^{\omega}]=\delta_{ij},
\end{gather}
as well. When the configuration space $M$ has the trivial topology, e.g. $M=\mathbb{R}^{n}$
\begin{gather}
d\omega=0\Rightarrow\exists f\,\,\omega=df.\label{eq:implication}
\end{gather}
Therefore, using gauge freedom, i.e. $\Psi^{\prime}(x)=e^{-if(x)}\Psi(x)$
it is possible to remove $\omega$ from $p_{i}$. To this end, note
that
\begin{gather*}
p_{i}^{\omega}\Psi^{\prime}(x)=\left(-i\frac{\partial}{\partial x_{i}}+\omega_{i}(x)\right)e^{-if(x)}\Psi(x)=\\
=-ie^{-if(x)}\frac{\partial}{\partial x_{i}}\Psi(x)-\frac{\partial}{\partial x_{i}}f(x)e^{-if(x)}\Psi(x)+\omega_{i}e^{-if(x)}\Psi(x)=\\
=-ie^{-if(x)}\frac{\partial}{\partial x_{i}}\Psi(x)=e^{-if(x)}p_{i}\Psi(x).
\end{gather*}
On the other hand, when configuration space has a non-trivial topology
the implication given by (\ref{eq:implication}) does not hold and
it is not possible to use the above argument. Before discussing the Aharonov-Bohm
effect which is, in some sense, a manifestation of this phenomenon we
first focus on a more general situation. The operators $p_{i}$ are
generators of translation and when $\omega=0$ one has
\begin{gather}
(e^{\epsilon^{i}p_{i}}\Psi)(x)=\Psi(x+\epsilon).
\end{gather}
It is easy to verify that when transporting the state vector $\Psi$
along curve $C:[0,T]\rightarrow M$ we get
\begin{gather}
(e^{i\int_{C}dx^{i}p_{i}^{\omega}}\Psi)(x)=\Psi(x+\Delta x)e^{\int_{C}\omega}.
\end{gather}
Therefore for a closed loop $C$
\begin{gather}
(e^{i\int_{C}dx^{i}p_{i}}\Psi)(x)=\Psi(x)e^{\oint_{C}\omega}.\label{eq:phase-shift}
\end{gather}
Let us consider two situations when $M=\mathbb{R}^{2}$ and when $M=\mathbb{R}^{2}-D(0,\rho)$, where $D(0,\rho)$ is a disk of radius $\rho$ (see figures \ref{AB111}(a) and \ref{AB111}(b), respectively). For the first case the loop $C$
is contractible and we have
\begin{gather}
\oint_{C}\omega=\int_{S}d\omega=0.
\end{gather}
For figure \ref{AB111}(b), that is, when the disk $D(0,\rho)$
is removed from the domain contained inside the loop $C$, i.e. when
$M=\mathbb{R}^{2}-D(0,\rho)$ we have
\begin{gather}
0=\int_{S}d\omega=\oint_{C}\omega-\oint_{\partial D}\omega\Rightarrow\oint_{C}\omega=\oint_{\partial D}\omega,
\end{gather}
and hence the phase $\phi=\oint_{C}\omega$ in equation (\ref{eq:phase-shift})
might be non-zero. For a general loop $C$ which goes around the disk
$D$ clockwise $n_{+}$ times and anticlockwise $n_{-}$ times one
gets
\begin{gather}
\oint_{C}\omega=(n_{+}-n_{-})\oint_{\partial D}\omega=(n_{+}-n_{-})\phi.
\end{gather}

\begin{figure}
\includegraphics[scale=0.48]{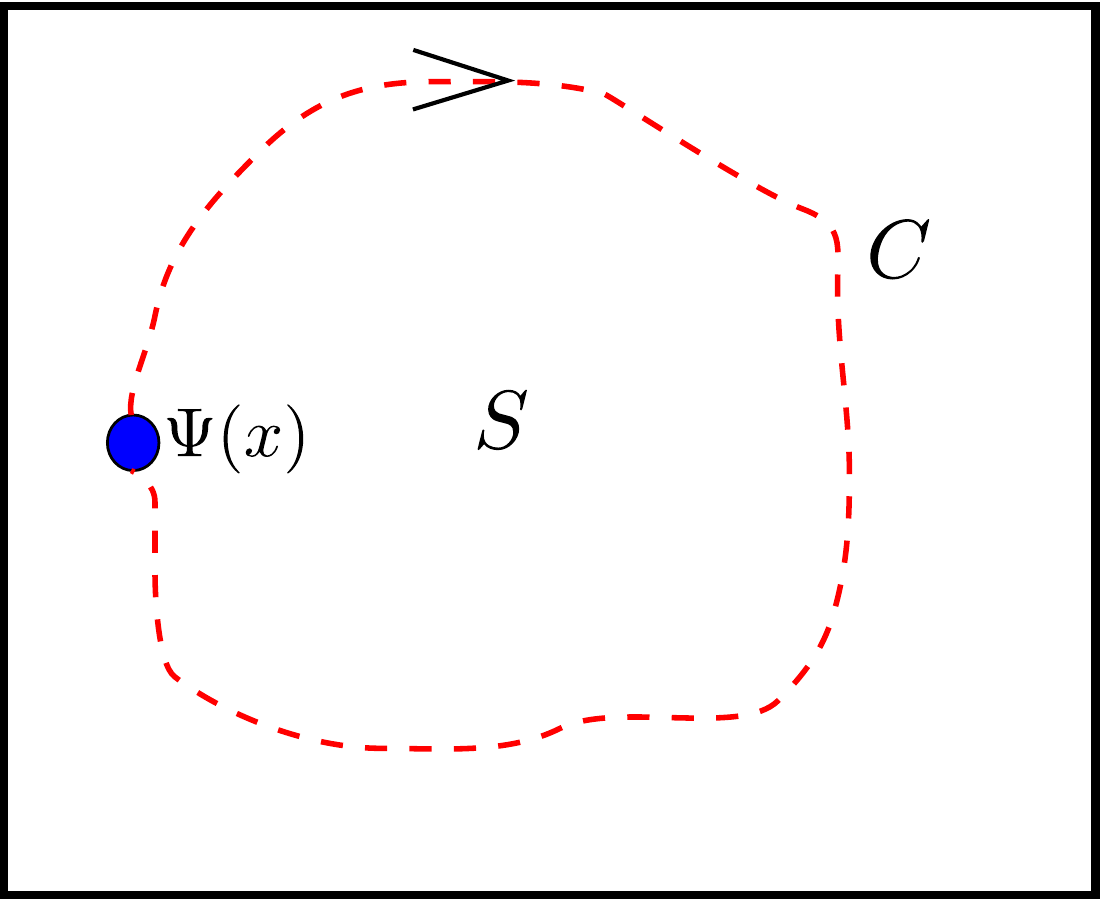}~~~~~~~~~~~~~~~~~\includegraphics[scale=0.48]{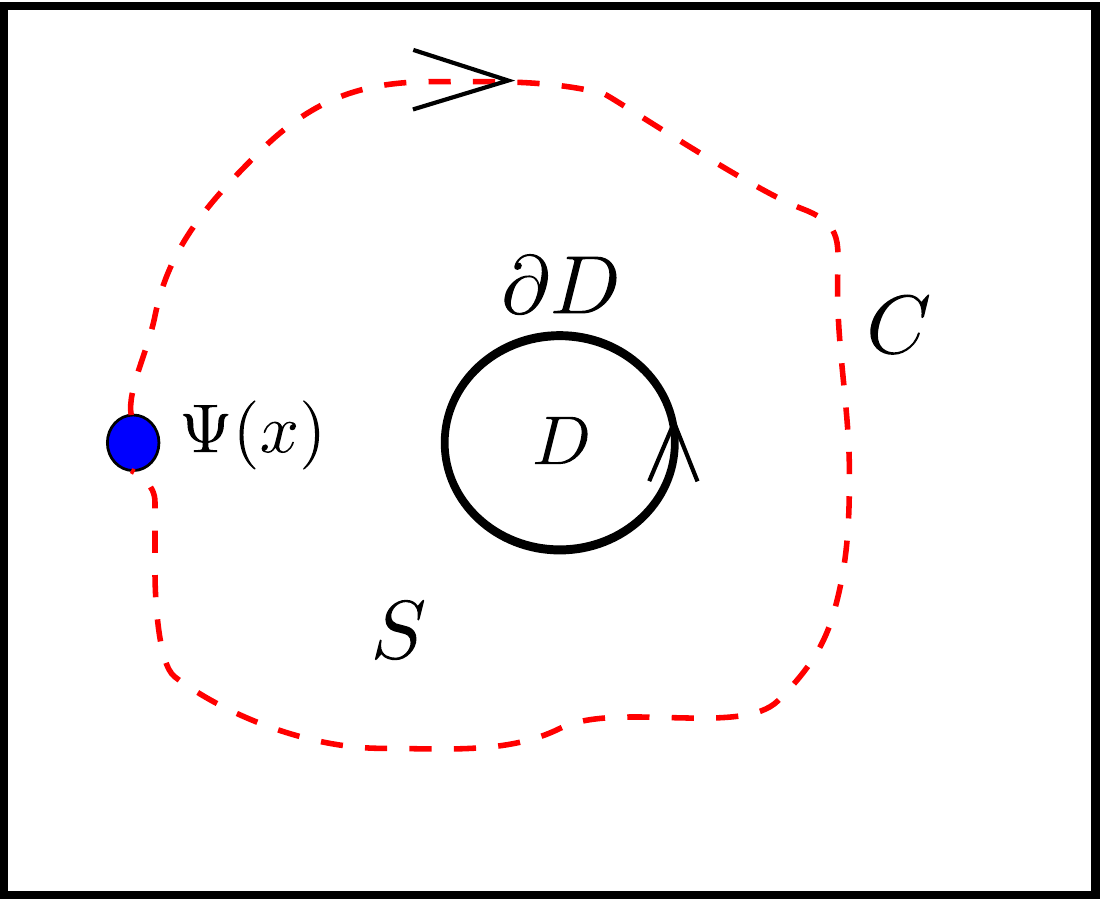}\caption{\label{AB111} (a) The contractible loop in $\mathbb{R}^{2}$ (b) the non-contractible
loop in $\mathbb{R}^{2}-D(0,\rho)$}
\end{figure}
\noindent As a conclusion we see that in certain topologies it does matter which
definition of momentum operators (\ref{eq:first-defi-momentum}) or
(\ref{eq:secnod-momentum-defi}) we use. On the other hand, it is
also clear that the differential $1$-form $\omega$ should be taken into
account only if it has some physical meaning. From a physics perspective
the simplest example of such physical realization is the magnetic
field whose potential $A$ is a connection $1$-form. Recall, that by
the minimal coupling principle, in the presence of a magnetic field
$B=dA$ all derivatives in all equations of physics should be substituted
by covariant derivatives. Thus
\[
p_{i}\rightarrow p_{i}-eA_{i},\quad B=dA.
\]
Therefore the magnetic potential $eA$ plays the role of $\omega$ from
the previous considerations. Let us next consider the situation shown
in figure \ref{AB11}. We assume
\[
dA(x)=0\,\,\mathrm{if}\, x\in M.
\]
\begin{figure}
\begin{center}\includegraphics[scale=0.5]{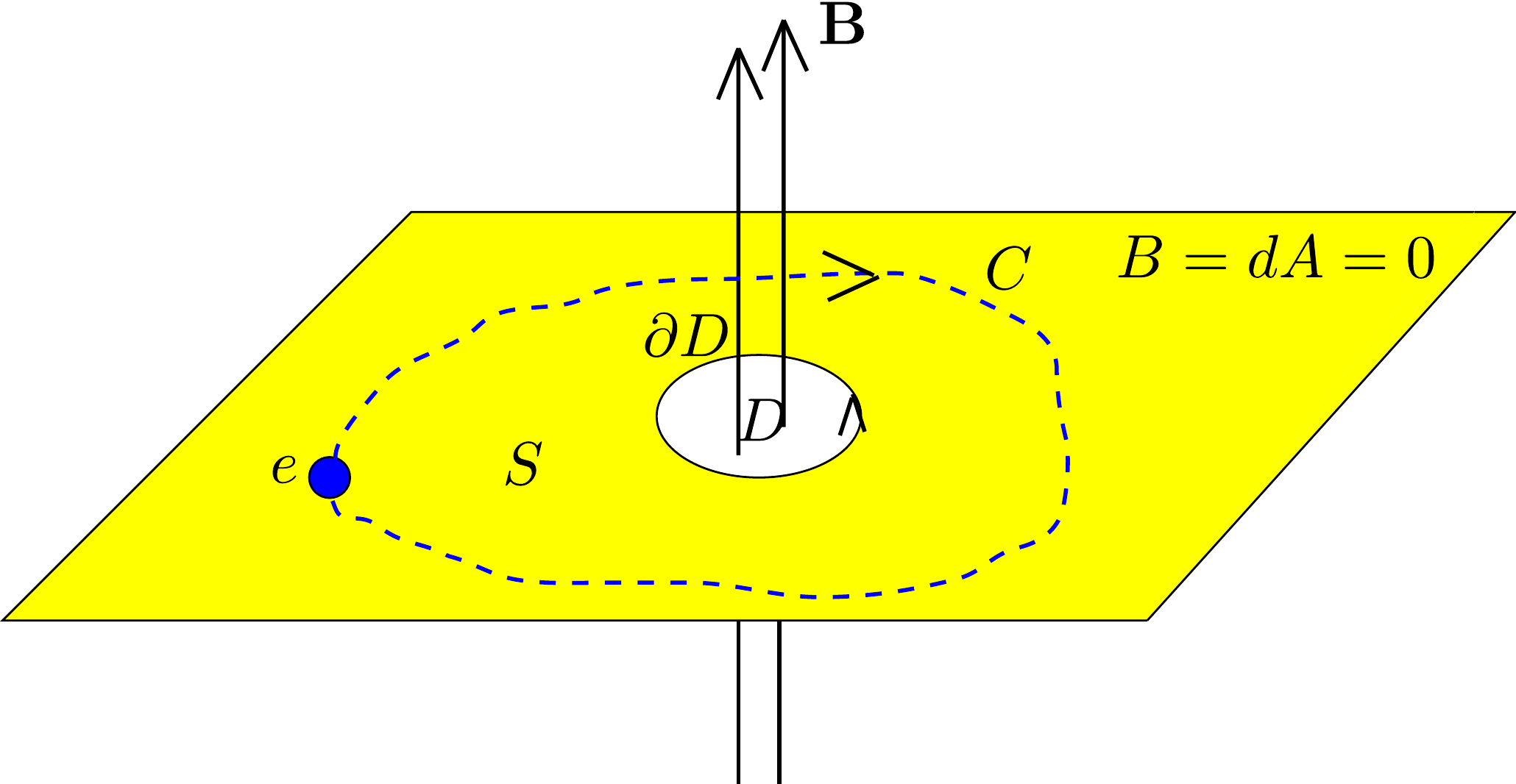}\end{center}
\caption{\label{AB11} The magnetic flux through a disk in $M=\mathbb{R}^{2}-D(0,\rho)$}
\end{figure}
\noindent The phase $\phi=\oint_{C}eA$ is called the Aharonov-Bohm phase and is
given by
\begin{gather*}
\oint_{C}eA=(n_{+}-n_{-})\oint_{\partial D}eA=(n_{+}-n_{-})e\int_{D}dA=(n_{+}-n_{-})e\int_{D}B=\\
=(n_{+}-n_{-})e\phi_{B},
\end{gather*}
i.e. it is proportional to the flux $\phi_{B}=\int_{D}B$ of the magnetic
field $B$ through the disk $D$. The existence of this phase was
experimentally verified in 1960 \cite{Chambers60} by measuring the
shift in the interference pattern of the electrons in the geometry
which is almost identical with the one shown in figure \ref{AB11}. Finally,
let us note that the existence of the Aharonov-Bohm phase is inevitably
related to the non-contractibility of the loop through which electrons
travel. Thus for graphs, we will call the phase gained by a particle
going around the cycle an Aharonov-Bohm phase (see chapter \ref{chQS} for
more details).

\section{Graphs}\label{graphs}
In this section we introduce the notion of graphs and discuss their basic properties. A graph $\Gamma=(E,V)$, where $V$ and $E$ are finite sets,
is a collection of $|V|$ points called vertices and $|E|$ edges which connect some of the vertices. We will write $v_1\sim v_2$ ($v_1\nsim v_2$) if two vertices $v_1,v_2\in V$ are connected (not connected) by an edge, respectively. An undirected edge between $v_1$ and $v_2$ will be denoted by $e=v_1\leftrightarrow v_2$. Similarly a directed edge from $v_1$ to $v_2$ ($v_2$ to $v_1$) will be denoted by $v_1\rightarrow v_2$ ($v_2\rightarrow v_1$). In the following we will consider only simple graphs, i.e. graphs for which any pair of vertices is connected by at most one edge (there are no multiple edges) and each edge is connected to exactly two different vertices (there are no loops). A typical way to encode the information about connections between vertices of $\Gamma$ is by means of the so-called adjacency matrix. The adjacency matrix of $\Gamma$ is a $|V|\times |V|$ matrix such that
\begin{equation}\label{adjacency}
A_{i,j}:=A_{v_i,v_j}=1\,\, \mathrm{if}\, v_i \sim v_j\in E,
\end{equation}
and $A_{ij}=0$ otherwise. If vertices $v_i$ and $v_j$ are connected by an edge, i.e. if $A_{ij}\neq 0$, we say they are adjacent. It is straightforward to see that $|E|=\frac{1}{2}\sum_{jk}A_{jk}$.

\subsection{Subgraphs, paths, trees and cycles}\label{PTC}
Here we assume that $\Gamma=(V, E)$ is a simple connected graph. A subgraph $\Gamma^\prime=(V^\prime,E^\prime)$ of the graph $\Gamma$ is a graph such that $V^\prime \subset V$, $E^\prime \subset E$ and edges from $E^\prime$ connect vertices from $V^\prime$. There are two elementary methods for constructing a subgraph out of the given graph. For $e\in E$ one can consider a graph with $|E|-1$ edges obtained from $\Gamma$ by deletion of the edge $e$. It will be denoted by $\Gamma\setminus e$. Similarly for a vertex $v\in V$ one defines graph $\Gamma\setminus v$ which is a result of deleting vertex $v$ together with all the edges connected to $v$. The generalization of these procedures to many edges or vertices is straightforward.

\noindent We proceed with definitions of other important subgraphs: paths, cycles and trees.
\begin{definition}
A path $P=(v_1,v_2,\ldots,v_k)$ on $\Gamma$ is a subgraph of $\Gamma$ such that $v_i\sim v_{i+1}$.\end{definition}
\noindent We will call the vertices $v_2,\ldots,v_{k-1}$ of path $P=(v_1,v_2,\ldots,v_k)$ the internal vertices. A path $P=(v_1,v_2,\ldots,v_k)$ will be called a {\it simple path} if $v_i\neq v_j$ for any $i,j\in \{1,\ldots,k\}$.
\begin{definition}
 A cycle $C=(v_1,v_2,\ldots,v_k)$ on $\Gamma$ is a subgraph of $\Gamma$ such that $v_i\sim v_{i+1}$ and $v_k=v_1$.
\end{definition}
\begin{definition}
A tree $T$ of $\Gamma$ is a subgraph of $\Gamma$ such that any pair of vertices is connected by exactly one path.
\end{definition}
\noindent Equivalently, $T\subset \Gamma$ is a tree if it contains no cycles. Among all trees of $\Gamma$ we distinguish the so-called {\it spanning trees}. A spanning tree $T=\{V_T,E_T\}$ of $\Gamma$ is a tree such that its set of vertices is exactly the set of vertices of $\Gamma$, i.e. $V_T=V$. Therefore, a spanning tree is a maximal subgraph of $\Gamma$ without cycles.
In order to calculate the number of cycles of a given graph $\Gamma=\{V,E\}$ we note that any spanning tree of $\Gamma$ has $|V|-1$ edges. Therefore, the number of cycles is
\begin{equation}\label{adjacency}
\beta_1(\Gamma)=|E|-(|V|-1)=|E|-|V|+1.
\end{equation}
This number, which is called the first Betti number, will play a major role in the next chapters.

\begin{figure}[h]
\begin{center}
\includegraphics[scale=0.55]{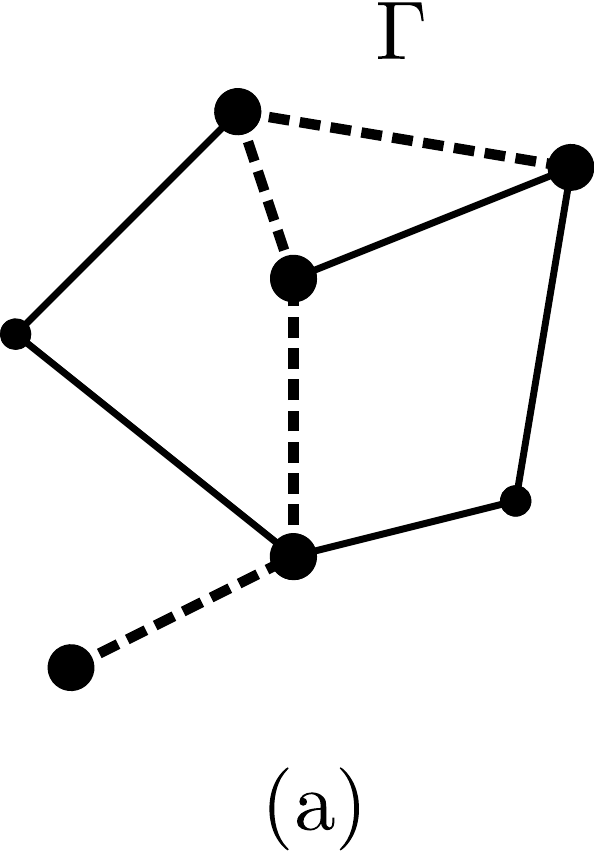}~~~~~~~\includegraphics[scale=0.55]{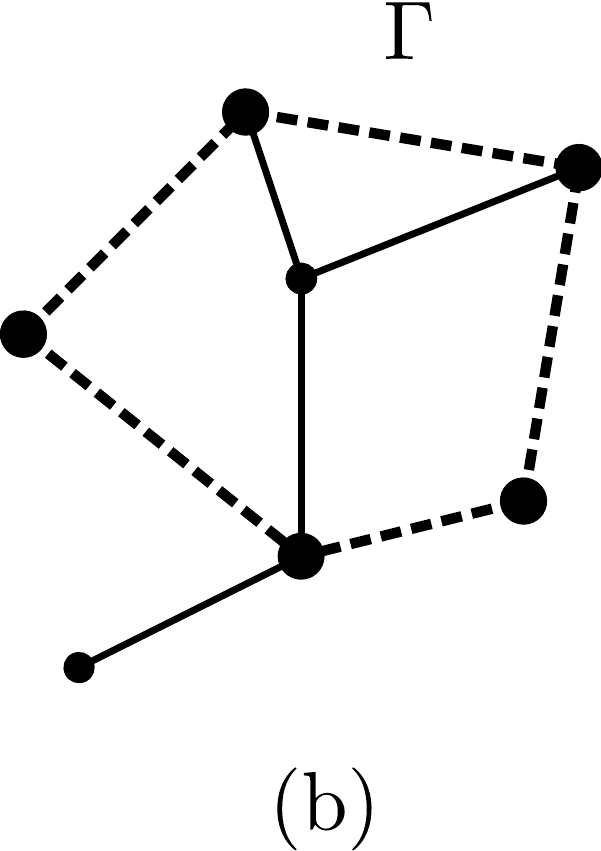}~~~~~~~\includegraphics[scale=0.55]{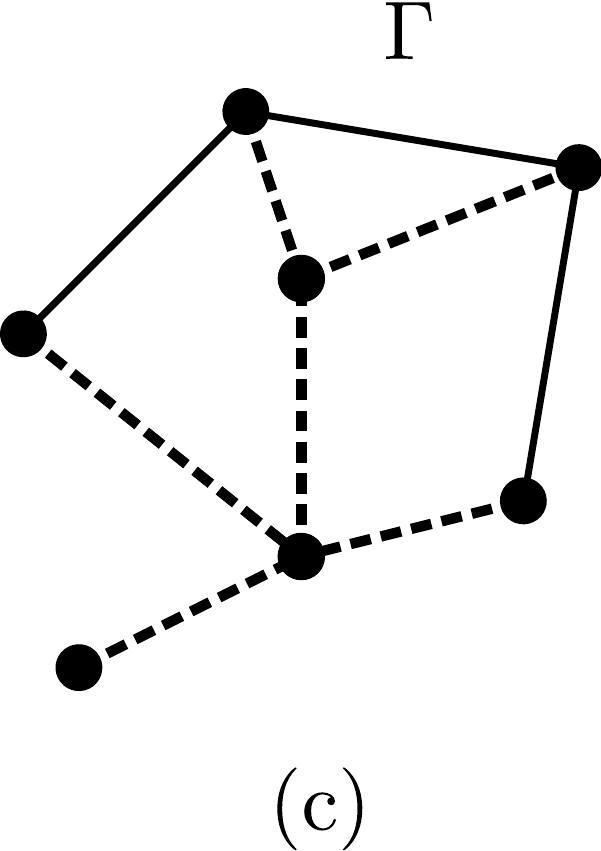}
\end{center}
\caption{The dashed edges represent examples of (a) a path, (b) a cycle, (c) a spanning tree, of graph $\Gamma$.}
\end{figure}

\subsection{Connectivity}\label{connectivity}
In this section we discuss the notion of connectivity of a graph. We start with the definition of a connected graph.
\begin{definition}
A graph is connected, if any pair of its vertices is connected by a path.
\end{definition}
\noindent In the following we will need the notion of $k$-connected graphs. Note at the beginning that after the removal of a vertex (or vertices) from $\Gamma$, the graph $\Gamma$ can split into several disjoint connected components (see figure \ref{comp} (b)). The topological closures of connected components of $\Gamma\setminus X$ will be called topological components or, if it does not cause ambiguity, just components (see figure \ref{comp} (c) for an intuitive definition of the topological closure).
\begin{figure}[h]
\begin{center}
\includegraphics[scale=0.55]{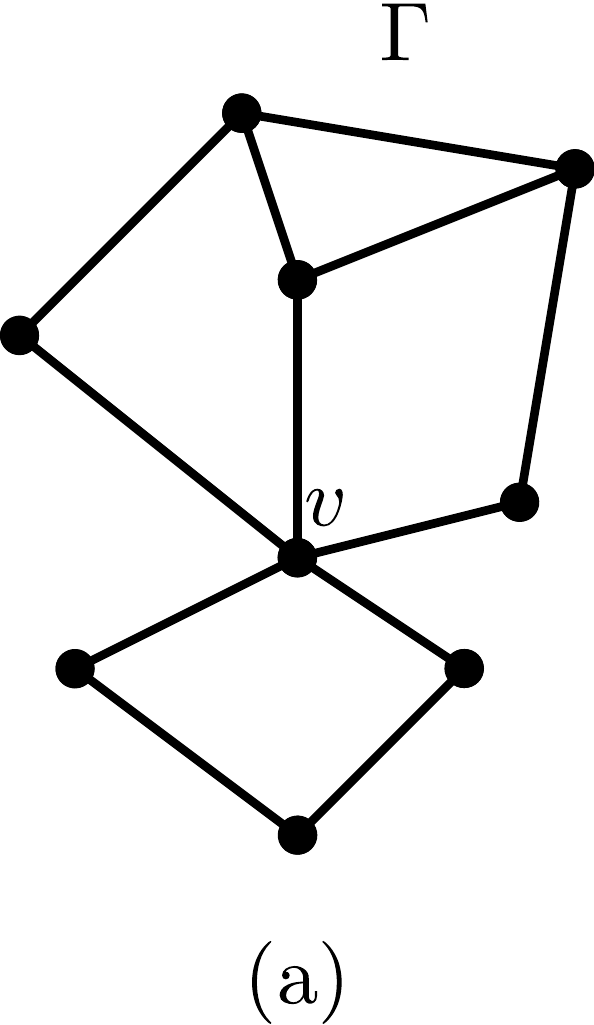}~~~~~~~\includegraphics[scale=0.55]{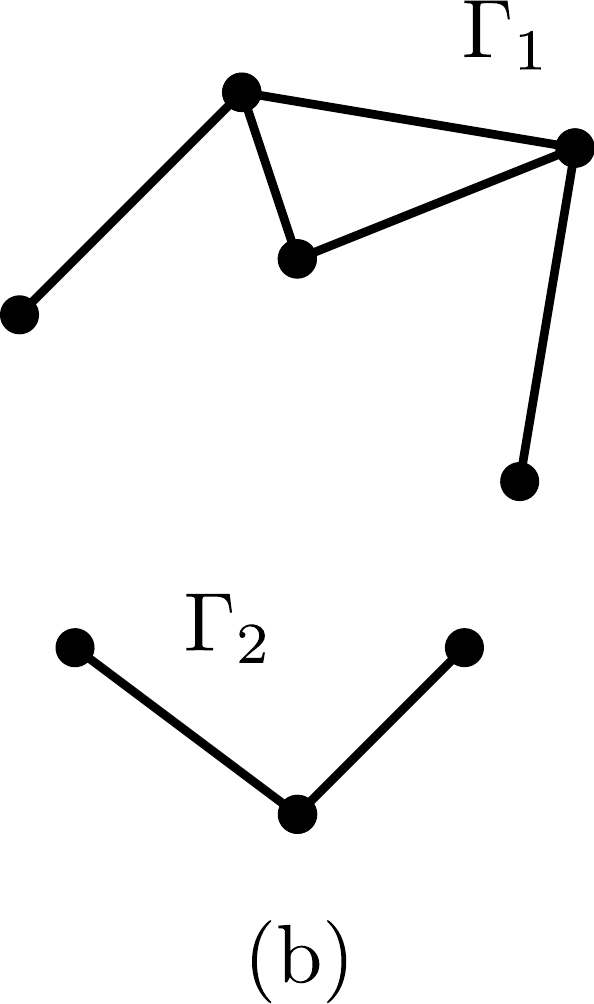}~~~~~~~\includegraphics[scale=0.55]{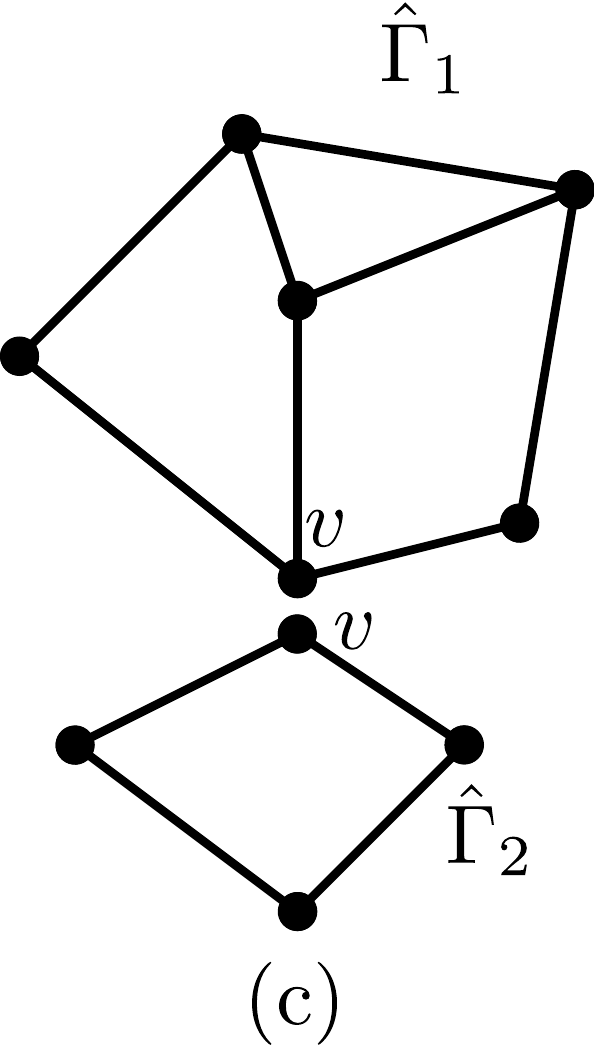}
\end{center}
\caption{\label{comp} (a) A graph $\Gamma$, (b) Components of $\Gamma$, (c) Topological components of $\Gamma$.}
\end{figure}
\noindent The definition of $k$-connected graph is closely related to this notion, i.e.
\begin{definition}\label{k-connected-def}
A graph $\Gamma=\{V,E\}$ is $k$-connected, where $k\in\mathbb{N}$, if $|V|>k$ and $\Gamma \setminus X$ is connected for any set $X\subset V$ with $|X|<k$.
\end{definition}
\noindent In definition \ref{k-connected-def} the graph $\Gamma \setminus X$ should be understood as a graph obtained from $\Gamma$ by removal of vertices $X$ as explained in section \ref{PTC}. Moreover, we assume that every graph is $0$-connected. Note also that by definition \ref{k-connected-def} all connected graphs are $1$-connected. We next define the connectivity of a graph.
\begin{definition}\label{k-connectivity}
The connectivity of $\Gamma$, $\kappa(\Gamma)$, is the greatest integer $k$ such that $\Gamma$ is $k$-connected.
\end{definition}
\noindent Note that the definition of $k$-connected graph is phrased in terms of vertex removals rather than in terms of paths joining pairs of vertices. In order to link it with paths we first define independent paths between pairs of vertices.
\begin{definition}\label{paths}
Two (or more) paths between vertices $v_1$ and $v_2$ of $\Gamma$ are independent if they do not have common inner vertices (see figure \ref{paths}).
\end{definition}
\noindent The following Menger's theorem \cite{tutte01} gives the characterization of $k$-connected graphs in terms of independent paths.
\begin{theorem}
A graph is $k$-connected if and only if there are $k$ independent paths between any two of its vertices.
\end{theorem}

\begin{figure}[h]
\begin{center}
\includegraphics[scale=0.65]{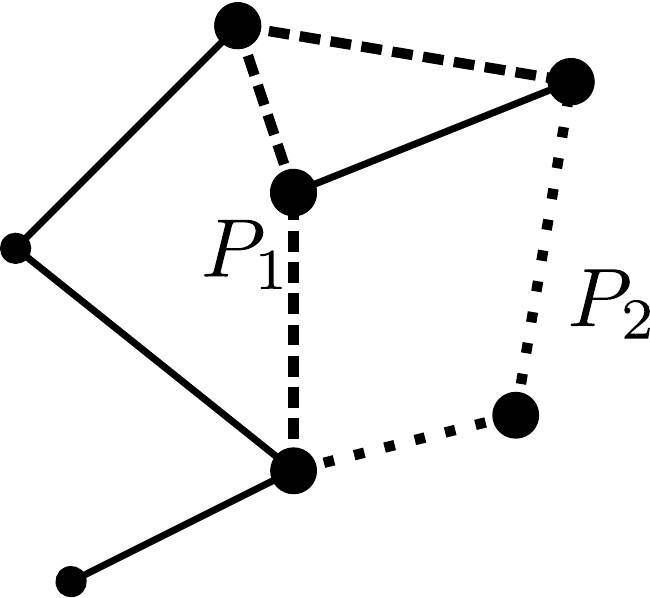}
\end{center}
\caption{\label{paths} Example of two independent paths between vertices $v_1$ and $v_2$, dashed and dotted subgraphs of $\Gamma$.}
\end{figure}

\subsection{Decomposition of a graph into $3$-connected components}
As we will see in the next sections, the characterisation of quantum statistics on graphs requires an understanding of the decomposition of a graph into $3$-connected components. We start with the definition of a cut.
\begin{definition}
A cut $X$ of a graph $\Gamma$ is a set of vertices such that $\Gamma\setminus X$ is disconnected, that is, it consists of at least two components.
\end{definition}
\noindent We will say that $X$ is an $n$-cut if $|X|=n$. We next introduce the notion of a block of the graph.
\begin{definition}\label{block}
A block is a maximal connected subgraph without a $1$-cut.
\end{definition}

\noindent Note that by definition \ref{block} any block is either a maximal $2$-connected subgraph or a single edge. For example, if $\Gamma$ is a tree then its blocks are precisely the edges. Having a $1$-connected simple graph one can consider the set of its $1$-cuts. For each $1$-cut we can further consider its topological components. Next for each component, if possible we apply the remaining $1$-cuts. Repeating this process iteratively we arrive with topological components which are either $2$-connected or given by edges. This way we decompose a $1$-connected graph into the set of $2$-connected components and edges, which are in fact the blocks of the considered graph (see figure \ref{decomp1-2} for an example of this kind of decomposition). It can be shown that the decomposition is unique \cite{tutte01}.
\begin{figure}[h]
\begin{center}
\includegraphics[scale=0.5]{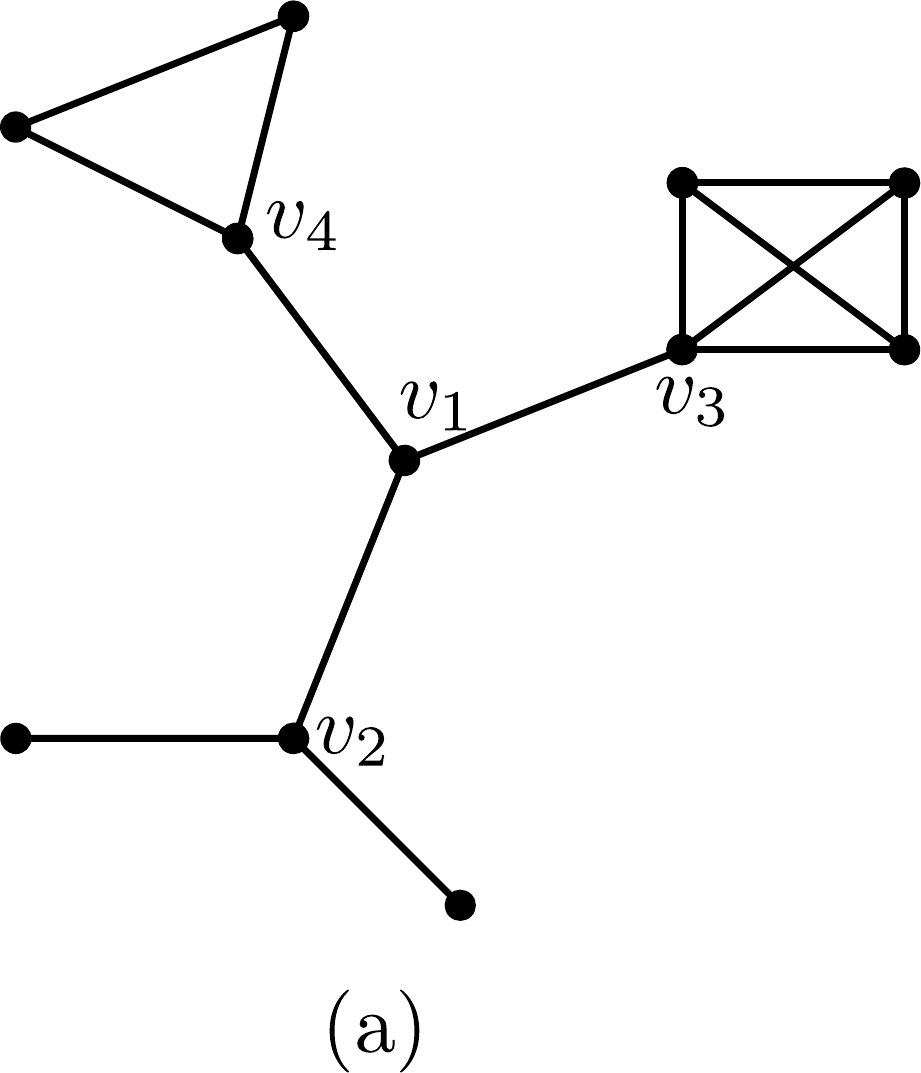}~~~~~~~\includegraphics[scale=0.5]{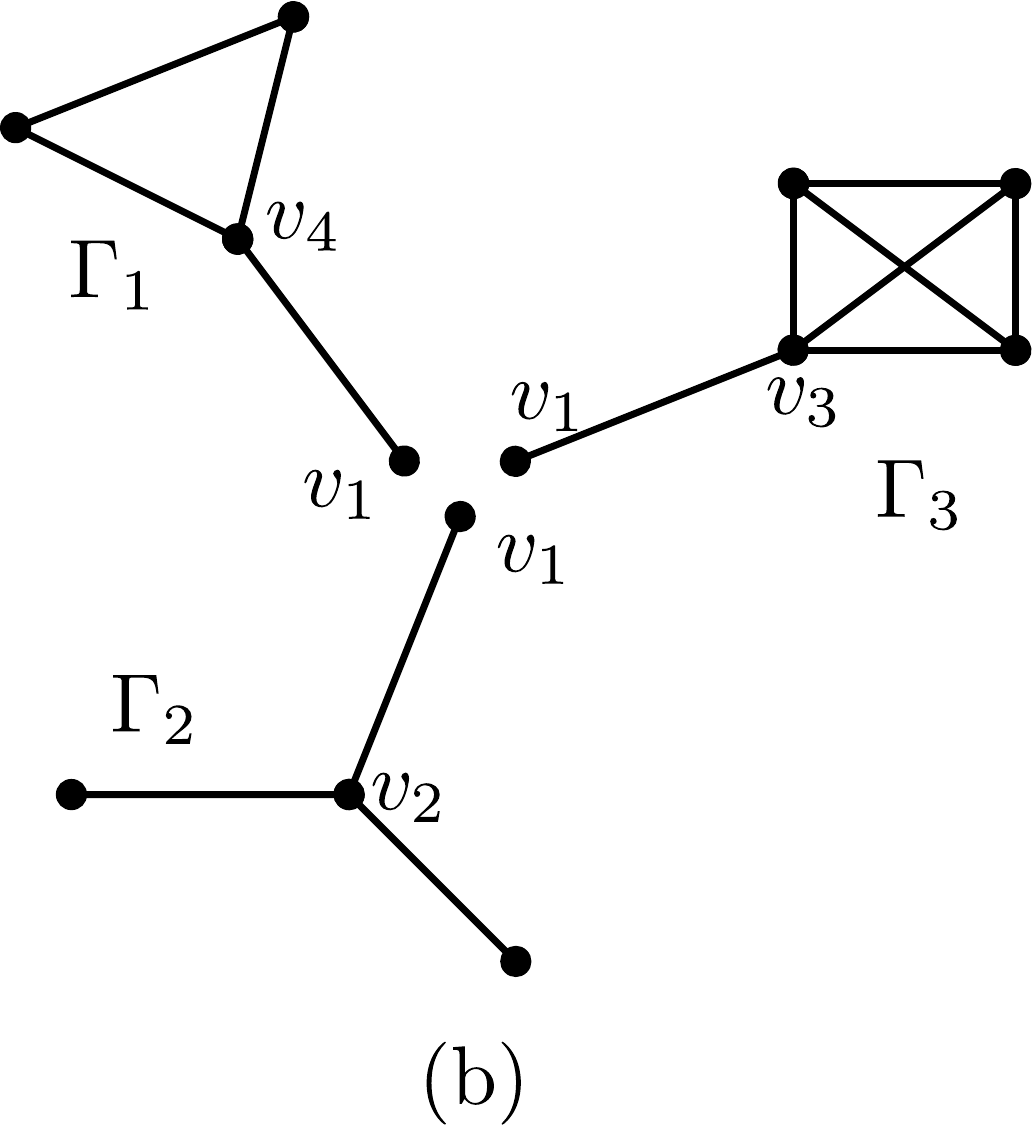}
\end{center}
\begin{center}
\includegraphics[scale=0.5]{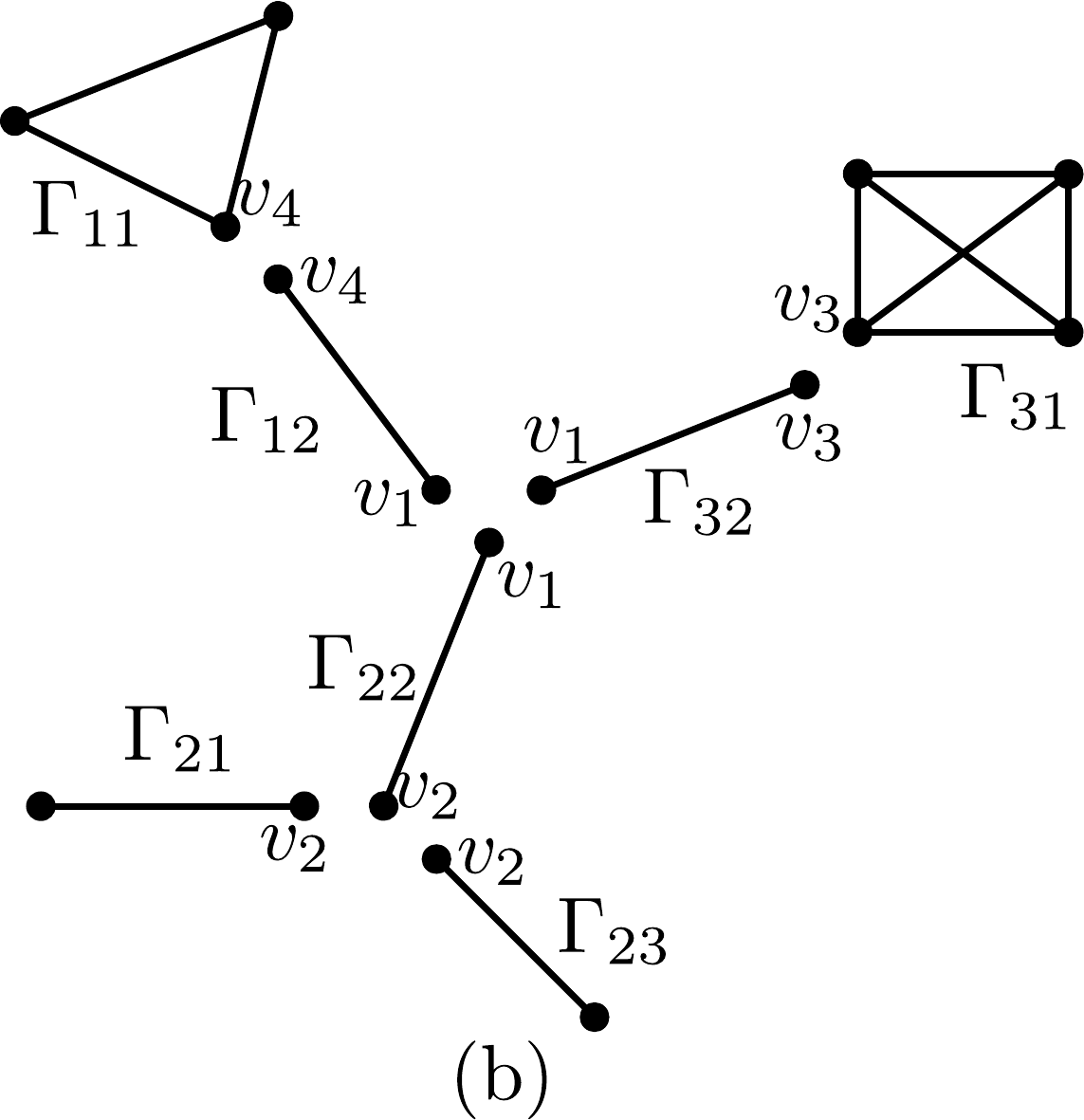}
\end{center}
\caption{\label{decomp1-2} (a) A graph, (b) Its topological components resulting from $v_1$-cut, (c) The full block decomposition.}
\end{figure}
If the $2$-connected components obtained from the above decomposition are not $3$-connected they can be further decomposed into the set of $3$-connected components and perhaps cycles. This is done by considering the set of $2$-cuts. For each $2$-cut $\{x,y\}\subset V$ we take all its topological components. In order to ensure that these components are $2$-connected we add an additional edge between vertices $\{x,y\}$ and call them the {\it marked components}. Repeating this process iteratively we arrive at marked components which are either $3$-connected or topological cycles (see figure \ref{decomp2-3} for an example of this kind of decomposition). Although the final set of marked components is typically not unique, one can show that the numbers of $3$-connected components and cycles do not depend on the order in which one applies $2$-cuts.

\begin{figure}[h]
\begin{center}
\includegraphics[scale=0.5]{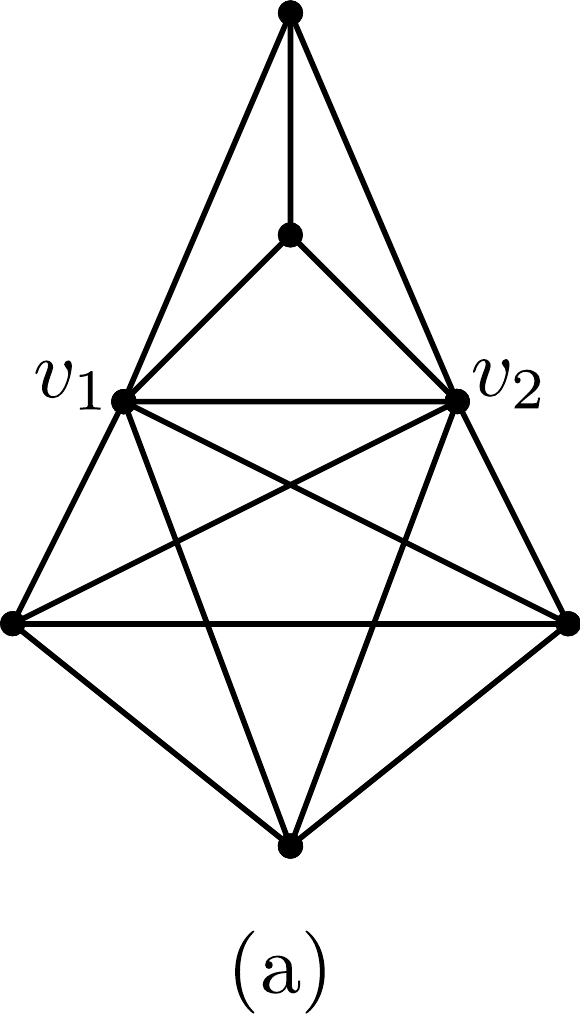}~~~~~~~\includegraphics[scale=0.5]{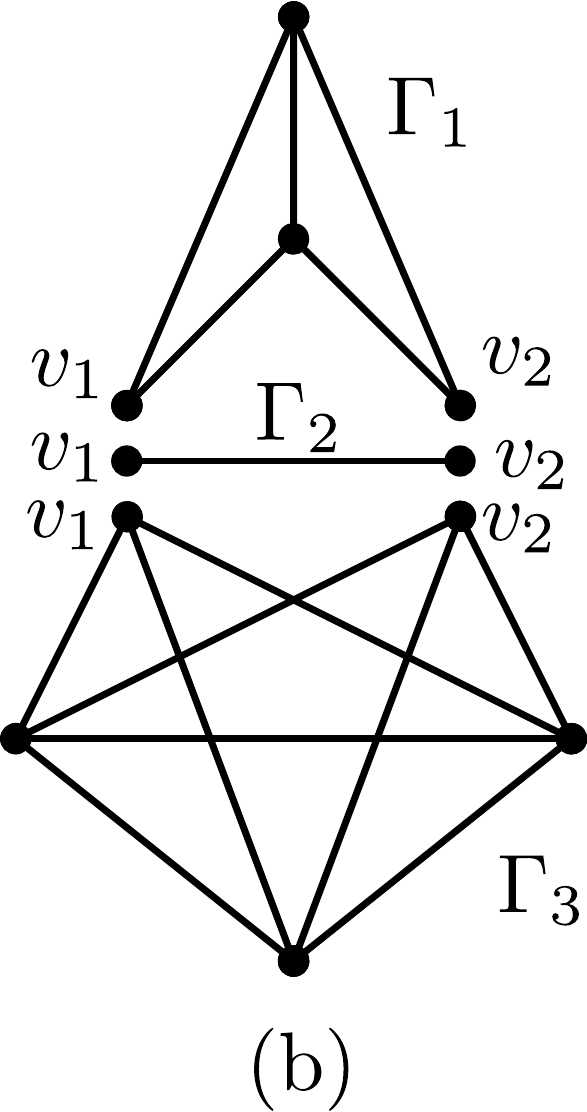}~~~~~~~\includegraphics[scale=0.5]{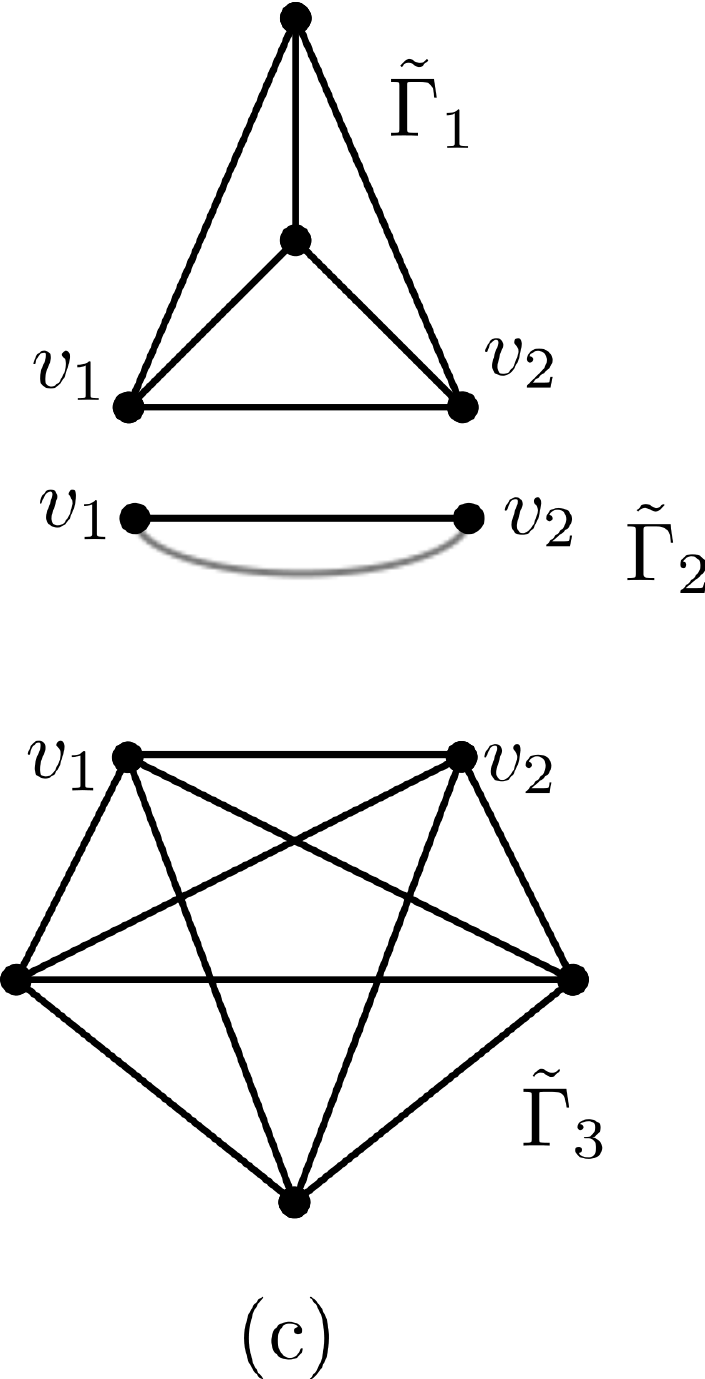}
\end{center}
\caption{\label{decomp2-3} (a) A $2$-connected graph, (b) Components resulting from $\{v_1,v_2\}$-cut, (c) Marked components resulting from $\{v_1,v_2\}$-cut.}
\end{figure}

\section{The fundamental group}\label{fund-group}
In this section we introduce the fundamental group of a topological space. As we will see, up to continuous deformations, the elements of this group are loops of the considered space.

Let $X$ be a topological space. A path in $X$ is a continuous map $f:[0,1]\rightarrow X$. Consider the family of paths $f_t:[0,1]\rightarrow X$, where $t\in [0,1]$ and:
\begin{itemize}
  \item the endpoints $f_t(0)=x_0$ and $f_t(1)=x_1$ are fixed, i.e. do not depend on $t$,
  \item the map $F(s,t):=f_t(s)$, that is the map $F:[0,1]\times [0,1]\rightarrow X$ is continuous.
\end{itemize}
The family satisfying these conditions is called a {\it homotopy of paths} in $X$.
\begin{definition}
Two paths $f_0$ and $f_1$ with fixed endpoints, i.e. $f_0(0)=f_1(0)$ and $f_0(1)=f_1(1)$, are homotopic if they can be connected by a homotopy of paths.
\end{definition}
\noindent The homotopy equivalent paths will be denoted by $f_0\simeq f_1$. One can show that the relation of homotopy of paths with fixed points, i.e. $\simeq$ is an equivalence relation and therefore divides paths into disjoint classes. We next define the product of two paths.

\begin{definition}
Let $f,g:[0,1]\rightarrow X$ be such that $f(1)=g(0)$. The product path $f\cdot g$ is the path given by
\begin{displaymath}
f\cdot g(t)=\left\{
              \begin{array}{cc}
                f(2t) & 0\leq t\leq \frac{1}{2}, \\
                g(2t-1) & \frac{1}{2} \leq t\leq 1.  \\
              \end{array}\right.
\end{displaymath}
\end{definition}
\noindent It is easy to see that the product of paths behaves well with respect to homotopy classes of paths, i.e. if $f_0\simeq f_1$ and $g_0\simeq g_1$ then $f_0\cdot g_0\simeq f_1\cdot g_1$.

Let us next consider loops, that is paths whose starting and ending points are the same (we call it basepoint). We denote by $\pi_1(X,x_0)$ the set of all homotopy classes of loops with basepoint $x_0$. One can show that $\pi_1(X,x_0)$ is a group with respect to the product of homotopy classes of loops defined by $[f][g]=[f\cdot g]$ called {\it the fundamental group} of $X$ at the basepoint $x_0$. Moreover, if two basepoints $x_0$ and $x_1$ lie in the same path-component of $X$ the groups $\pi_1(X,x_0)$ and $\pi_1(X,x_1)$ are isomorphic. Therefore for path-connected spaces we often write $\pi_1(X)$ instead of $\pi(X,x_0)$.

We next describe the fundamental group of a simple connected graph\footnote{The topology we use is a topology of a cell complex which is defined in section \ref{cell-complexes}}. Let $T\subset\Gamma$ be a spanning tree of $\Gamma$. Choose $v_0$ to be any vertex of $\Gamma$ (hence of $T$). Each edge $e_i$ of $\Gamma \setminus T$, which we will call a deleted edge, defines a loop in $\Gamma$. To see this note that there is a unique simple path joining $v_0$ with each of the endpoints of $e_i$. The announced loop, which we denote by $e_i$, starts from $v_0$ goes through the path in $T$ to one of the endpoints of $e_i$, then through $e_i$ and then returns to $v_0$ across the path in $T$. The homotopy classes of these loops generate $\pi_1(\Gamma)$. More precisely:
\begin{theorem}\label{fundgroupgraph}
Let $\Gamma$ be a connected simple graph and $T$ its spanning tree. Then the fundamental group $\pi_1(\Gamma)$ is a free group whose basis is given by classes $[e_i]$ corresponding to deleted edges $e_i\in \Gamma\setminus T$.
\end{theorem}

\section{Cell complexes}\label{cell-complexes}
An example of a topological space is a cell complex which we discuss in the following. 

Let $B_{n}=\{x\in\mathbb{R}^{n}\,:\,\|x\|\leq1\}$
be the standard unit-ball. The boundary of $B_{n}$ is the unit-sphere
$S^{n-1}=\{x\in\mathbb{R}^{n}\,:\,\|x\|=1\}$. A cell complex $X$
is a nested sequence of topological spaces
\begin{eqnarray}
X^{0}\subseteq X^{1}\subseteq\dots\subseteq X^{n},
\end{eqnarray}
where the $X^{k}$'s are the so-called $k$ - skeletons defined as follows:
\begin{itemize}
\item The $0$ - skeleton $X^{0}$ is a discrete set of points.
\item For $\mathbb{N}\ni k>0$, the $k$ - skeleton $X^{k}$ is the result
of attaching $k$ - dimensional balls $B_{k}$ to $X^{k-1}$ by gluing
maps
\begin{eqnarray}
\sigma:S^{k-1}\rightarrow X^{k-1}.
\end{eqnarray}

\end{itemize}
By $k$-cell $\alpha^{(k)}$ we understand the interior of the ball $B_{k}$ attached
to the $(k-1)$-skeleton $X^{k-1}$. We will denote by $\overline{\alpha}^{(k)}$ the cell $\alpha^{(k)}$ together with its boundary. The $k$-cell is regular if its
gluing map is an embedding (i.e., a homeomorphism onto its image). Finally we say that $X$ is $n$-dimensional if $n$ is the highest dimension of the cells in $X$.

Notice that every simple graph $\Gamma$ can be treated as a regular cell complex with
vertices as $0$-cells and edges as $1$-cells. If a graph contains
loops, these loops are irregular $1$-cells (the two points that comprise the boundary of
$B_{1}$ are attached to a single vertex of  the $0$-skeleton). The product
$\Gamma^{\times n}$ inherits a cell-complex structure; its cells
are cartesian products of cells of $\Gamma$.


\begin{figure}[h]
~~~~~~~~~~~~~~~~~\includegraphics[scale=0.5]{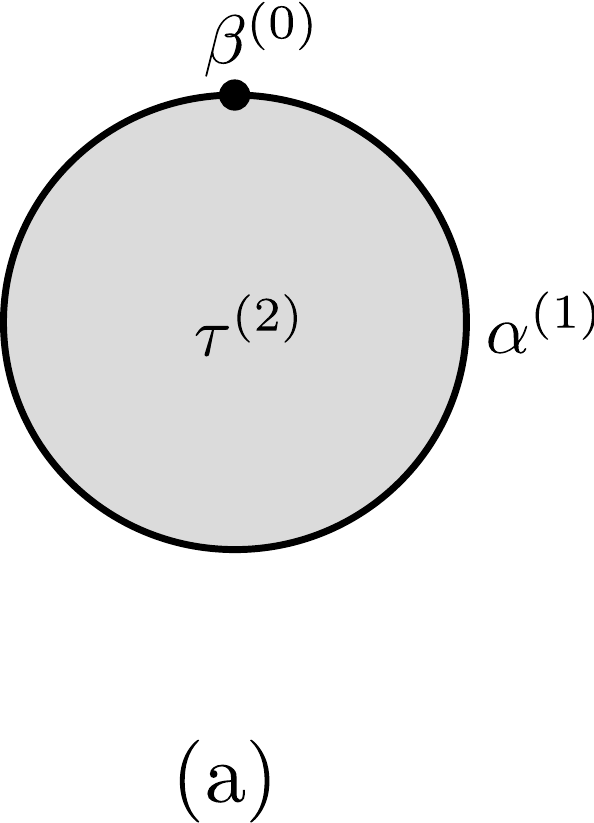}~~~~~~~~~~~~~~~~~~~~~~\includegraphics[scale=0.5]{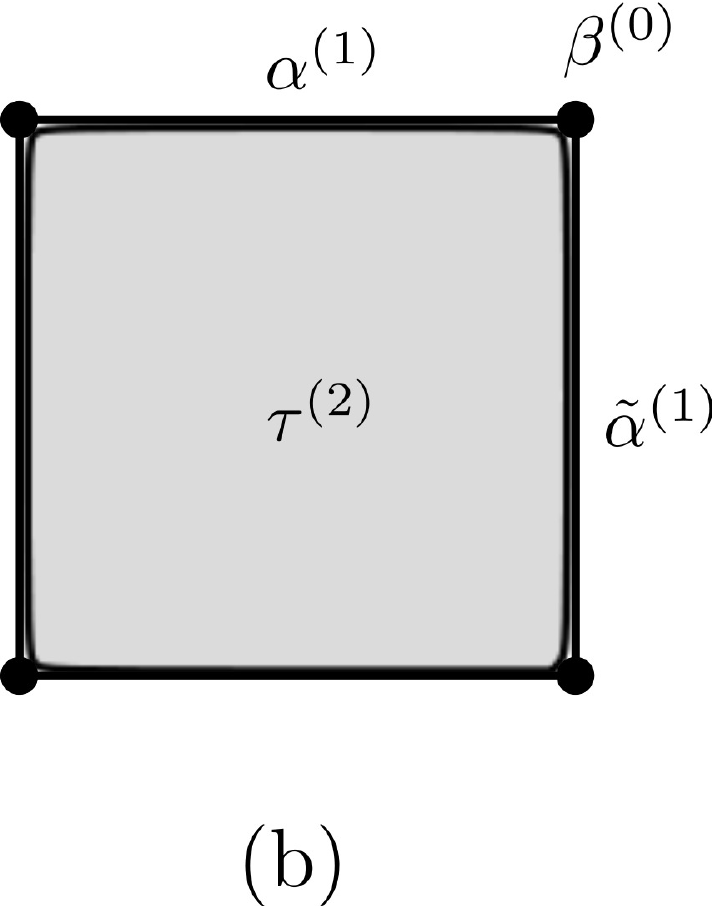}

\caption{Examples of (a) an irregular cell complex.  $\alpha^{(1)}$ is an irregular
1-cell and $\beta^{(0)}$ is an irregular face of $\alpha^{(1)}$, $\tau^{(2)}$ is a regular $2$-cell. (b) A
regular cell complex.}
\end{figure}

\section{Homology groups}
In this section we define homology groups, $H_k(X,\mathbb{Z})$ over the integers of the $n$-dimensional cell complex $X$.  As we will always speak about integer homology we will often write $H_k(X)$ instead of $H_k(X,\mathbb{Z})$.

The construction of homology groups goes in several steps which we now describe.  First, we assign an arbitrary orientation on the cells of $X$. Let $m_p$ be the number of $p$-cells in $X$. The oriented $p$-cells will be denoted by $\{e_1^{(p)},e_2^{(p)}\ldots e_{m_p}^{(p)}\}$. The $p$-chain is a formal linear combination
\begin{equation}\label{p-chain}
c=a_1e_1^{(p)}+a_2e_2^{(p)}+\ldots+ a_{m_p}e_{m_p}^{(p)},
\end{equation}
where coefficients $a_i$ are integers. We denote by $C_p(X)$ the set of all $p$-chains. This set can be given the structure of an abelian group. The addition in this group is defined by
\begin{eqnarray*}
  c_1&=&a_1e_1^{(p)}+a_2e_2^{(p)}+\ldots+ a_{m_p}e_{m_p}^{(p)}, \\\nonumber
  c_2 &=&b_1e_1^{(p)}+b_2e_2^{(p)}+\ldots+ b_{m_p}e_{m_p}^{(p)}, \\\nonumber
  c_1+c_2 &=& (a_1+b_1)e_1^{(p)}+a_2e_2^{(p)}+\ldots+ (a_{m_p}+b_{m_p})e_{m_p}^{(p)}.\\\nonumber
\end{eqnarray*}
In fact $C_p(X)$ is isomorphic to the direct sum of $m_p$ copies of $\mathbb{Z}$, i.e. $C_p(X)\simeq \mathbb{Z}^{m_p}$. We next consider the boundary map
\begin{equation}
\partial_p:C_p(X)\rightarrow C_{p-1}(X),
\end{equation}
which assigns to an oriented $p$-cell $e_i^{p}$ its boundary (see \cite{Hatcher} for a discussion on boundary maps). The boundary map satisfies $\partial_{p-1}\partial_{p}=0$, that is the boundary of the boundary is zero. This way we arrive at the {\it chain complex} $(C(X)_{\bullet},\partial_{\bullet})$, that is a sequence of abelian groups $C_n(X), \ldots ,C_0(X)$ connected by boundary homomorphisms $\partial_p:C_p(X)\rightarrow C_{p-1}(x)$, such that that composition of any two consecutive maps is zero $\partial_{p-1}\partial_{p}=0$. The standard way to denote a chain complex is the following:
\begin{equation}
C_n(X)\xrightarrow{\partial_n} C_{n-1}\xrightarrow{\partial_{n-1}}\cdots C_2(X)\xrightarrow{\partial_{2}}C_1(X)\xrightarrow{\partial_{1}}C_0(X)\xrightarrow{\partial_{0}}{0}.
\end{equation}
Let us denote by $\mathrm{Ker}(\partial_p)$ and $\mathrm{Im}(\partial_{p})$ the kernel and the image of the boundary map $\partial_p$. The elements of $\mathrm{Ker}(\partial_p)$ are called $p$-cylces and the elements of $\mathrm{Im}(\partial_{p})$ are called $p$-boundaries. The $p$-th homology group is defined as
\begin{equation}\label{pth-homology}
H_p(X)=\mathrm{Ker}(\partial_p)/\mathrm{Im}(\partial_{p+1}),
\end{equation}
so it is a quotient space of $p$-cycles by $p$-boundaries. 

We next present an example of a calculation of homology groups, i.e. we calculate the first homology group of the $2$-torus, $T^2$. To this end we consider the cell complex $X_T$ shown in figure \ref{cell-torus}. In fact, this is the simplest triangularization of $T^2$. The cell complex $X_T$ consists of two $2$-cells, three $1$-cells and one $0$-cell, i.e.
\begin{gather}
C_2(X_T)=\mathrm{Span}_{\mathbb{Z}}\{f_1,f_2\},\\\nonumber
C_1(X_T)=\mathrm{Span}_{\mathbb{Z}}\{e_1,e_2,e_3\},\\\nonumber
C_0(X_T)=\mathrm{Span}_{\mathbb{Z}}\{v\}.
\end{gather}
Therefore we have the following chain complex
\begin{equation}
C_2(X_T)\xrightarrow{\partial_{2}}C_1(X)\xrightarrow{\partial_{1}}C_0(X)\xrightarrow{\partial_{0}}{0}.
\end{equation}
In order to calculate $H_1(X_T,\mathbb{Z})$ we need to calculate the kernel and the image of boundary maps $\partial_1$ and $\partial_2$ respectively. Taking into account the orientation denoted in figure \ref{cell-torus} we have
\begin{gather}
\partial_1 e_1=\partial_1 e_2=\partial_1 e_3 = v-v=0, \\\nonumber
\partial_2 f_1=\partial_2 f_2=e_1-e_2+e_3.
\end{gather}
Therefore
\begin{gather}
\mathrm{Ker}\partial_1 = \mathrm{Span}_{\mathbb{Z}}\{e_1,e_2,e_3\}=\mathrm{Span}_{\mathbb{Z}}\{e_1,e_2,e_1-e_2+e_3\}, \\\nonumber
\mathrm{Im}\partial_2=\mathrm{Span}_{\mathbb{Z}}\{e_1-e_2+e_3\}.
\end{gather}
Using (\ref{pth-homology}) one easily sees that
\begin{equation}
H_1(X_T,\mathbb{Z})=\mathrm{Ker}(\partial_1)/\mathrm{Im}(\partial_{2})=\mathrm{Span}_{\mathbb{Z}}\{[e_1],[e_2]\}\simeq\mathbb{Z}\oplus\mathbb{Z}.
\end{equation}
\begin{figure}[H]
~~~~~~~~~~~~~~~~~\includegraphics[scale=0.5]{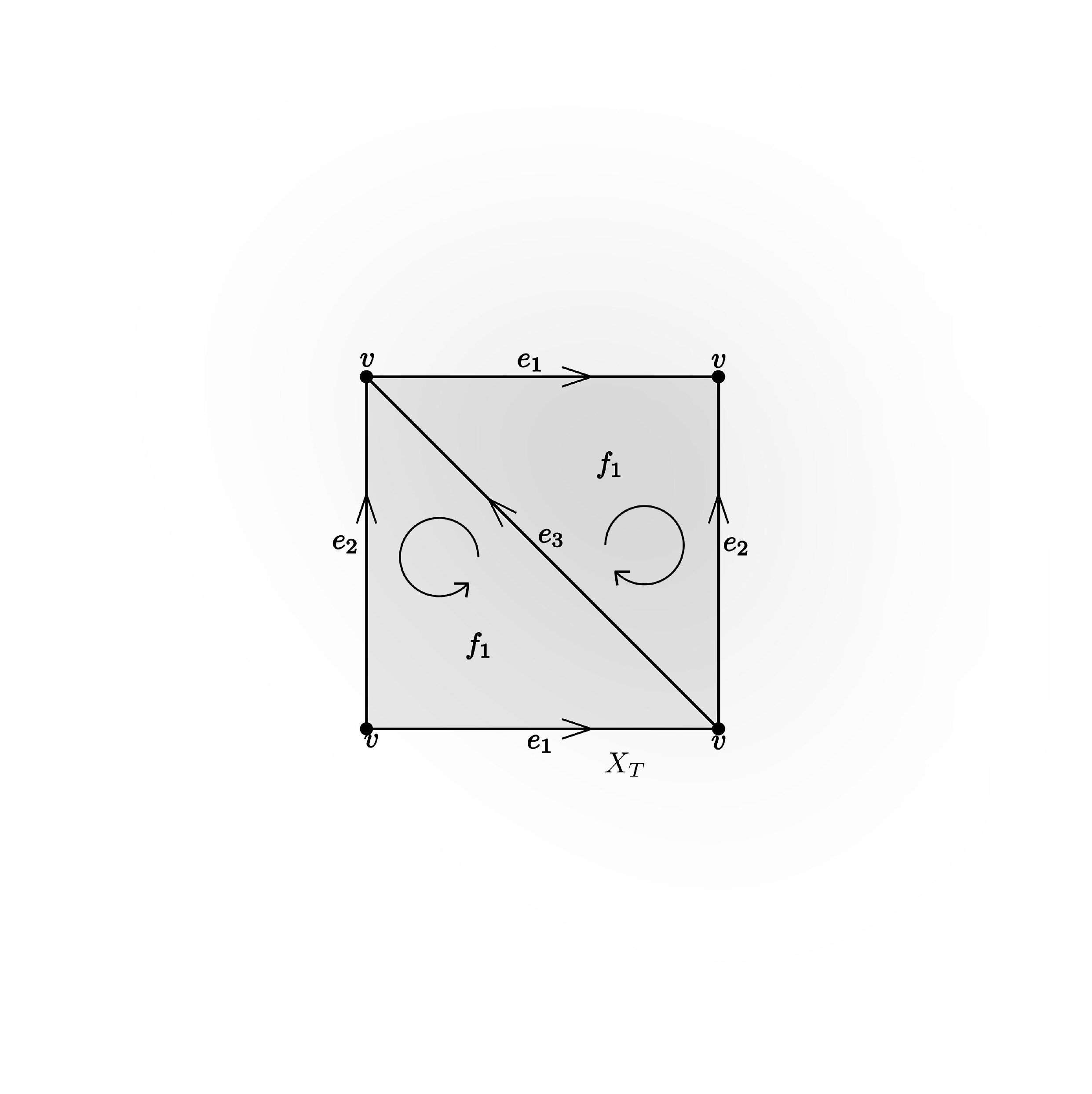}
\caption{\label{cell-torus} The cell complex $X_T$ of the $2$-torus $T^2$}
\end{figure}
\noindent We finish this section by making a statement about the connection between the first homology group $H_1(X)$ and the fundamental group $\pi_1(X)$. This connection is based on the fact that the map $f:[0,1]\rightarrow X$ can be viewed as either a path or a $1$-cell. Moreover if $f$ represents a loop the considered $1$-cell is a cycle. We need to first define the commutator subgroup of $\pi_1(X)$.

Let $G$ be a group and denote by $e$ its identity element. The commutator of two elements $g,h\in G$ is defined as
\begin{equation}
[g,h]=g^{-1}h^{-1}gh.
\end{equation}
Note that $[g,h]=e$ if and only $gh=hg$. The subgroup generated by all commutators of $G$, which we denote by $[G,G]$, is a normal subgroup of $G$. To se this note that if $u\in [G,G]$ and $g\in G$ then $g^{-1}ug=u(u^{-1}g^{-1}ug)=u[u,g]$. By normality of $[G, G]$ the quotient $G/[G,G]$ is a group. This group is called {\it abelianization} of $G$ as it is necessarily abelian. In fact if group $G$ is abelian itself then $G/[G,G]=G$.

The following theorem relates $H_1(X)$ to $\pi_1(X)$
\begin{theorem}\label{abelianization1}
The first homology group $H_1(X)$ is the abelianization of the fundamental group $\pi_1(X)$
\end{theorem}
By theorem \ref{fundgroupgraph} the fundamental group of a connected graph $\Gamma$ is a free group generated by loops through the deleted edges. Combining this result with theorem \ref{abelianization1} one easily obtains
\begin{theorem}\label{abelianization}
The first homology group $H_1(X)$ of a graph $\Gamma$ is $H_1(\Gamma,\mathbb{Z})=\mathbb{Z}^{\beta_1(\Gamma)}$, where $\beta_1(\Gamma)$ is the number of independent cycles of $\Gamma$ given by formula \ref{adjacency}.
\end{theorem}

\section{Structure theorem for finitely generated Abelian groups}
Homology groups of a finite cell complex $X$ are finitely generated Abelian groups. Therefore, in this section we discuss the structure theorem for these kind of groups.
\subsection{Finitely generated Abelian groups}
Let $\{G,+\}$ be an Abelian group. We put $0$ as the neutral element of $G$ and $-x$ denotes the inverse of $x$. For $x\in G$ and $n\in \mathbb{Z}$ we put
\begin{displaymath}
nx=\left\{
              \begin{array}{cc}
               \underbrace{ x+\ldots+x}_{n} &  n>0, \\
                \underbrace{ -x+\ldots+(-x)}_{n} & n<0, \\
                nx=0 & n=0.\\
              \end{array}\right.
\end{displaymath}
It is easy to see that for any chosen $x_1,\ldots,x_k \in G$ the set of elements
\begin{equation}
\{n_1x_1+\ldots +n_k x_k:\,\forall i\, n_i\in \mathbb{Z}\},
\end{equation}
is the subgroup of $G$ generated by $x_1,\ldots,x_k \in G$. We say that $x_1,\ldots,x_k \in G$ are linearly independent if $n_1x_1+\ldots +n_k x_k=0$ if and only if $\forall i\, n_i=0$.
\begin{definition}
A group $G$ is a finitely generated Abelian group if and only if there are elements $x_1,\ldots,x_k \in G$ such that $G=\{n_1x_1+\ldots +n_k x_k:\,\forall i\, n_i\in \mathbb{Z}\}$. Moreover, if the generating elements can be chosen to be linearly independent the group $G$ is a finitely generated {\it free} abelian group isomorphic to
\begin{equation}
G\simeq\underbrace{\mathbb{Z}\oplus\ldots \oplus\mathbb{Z}}_k=\mathbb{Z}^{k}. 
\end{equation} 
The number $k$ is called the rank of $G$. 
\end{definition}
\noindent For a finitely generated Abelian group which is not free there are some relations between generating elements. In order to describe these relations we will need the following
\begin{theorem}\label{subgroupOfZ}
Let $G$ be a finitely generated free abelian group of rank $p$ and $H$ a subgroup. There always exists choice of generators $x_1,\ldots,x_l$, $l\leq p$ in $G$ such that $H=k_1x_1+\dots+k_lx_l$, i.e. $H$ can be expressed in the form:
\begin{equation}
H=k_1\mathbb{Z}\oplus\ldots\oplus k_l\mathbb{Z},
\end{equation}
where $k_1\mathbb{Z}=\{x\in \mathbb{Z}:\,k_1|x\}$ and $\forall i\, k_i|k_{i+1}$ and $l\leq p$.
\end{theorem}
\noindent The structure theorem of an arbitrary finitely generated Abelian group reads:
\begin{theorem}
Let $G$ be a finitely generated Abelian group. Then
\begin{equation}
G\simeq\underbrace{\mathbb{Z}\oplus\ldots \oplus\mathbb{Z}}_k\oplus \mathbb{Z}_{k_1}\oplus\ldots\oplus\mathbb{Z}_{k_l}=\mathbb{Z}^{k}\oplus\mathbb{Z}_{k_1}\oplus\ldots\oplus\mathbb{Z}_{k_l},
\end{equation}
where $k_i|k_{i+1}$ for all $i\in\{1,\ldots,l\}$.
\end{theorem}
\begin{proof}
Let $x_1,\ldots,x_p$ be generating elements of $G$. The map
\begin{gather}
f:\underbrace{\mathbb{Z}\oplus\ldots \oplus\mathbb{Z}}_p\rightarrow G\\
f(n_1,\ldots,n_p)=n_1x_1+\ldots+n_px_p,
\end{gather}
is a surjective homomorphism between Abelian groups and therefore by the first isomorphism theorem $G=\underbrace{\mathbb{Z}\oplus\ldots \oplus\mathbb{Z}}_p/\mathrm{Ker}(f)$. But the kernel of $f$ is a subgroup of $\underbrace{\mathbb{Z}\oplus\ldots \oplus\mathbb{Z}}_p$. By theorem \ref{subgroupOfZ}
\begin{equation}
\mathrm{Ker}(f)=k_1\mathbb{Z}\oplus\ldots\oplus k_l\mathbb{Z},
\end{equation}
for some $l\leq p$. Hence,
\begin{gather}
G=\underbrace{\mathbb{Z}\oplus\ldots \oplus\mathbb{Z}}_p/k_1\mathbb{Z}\oplus\ldots\oplus k_l\mathbb{Z}=\\
\underbrace{\mathbb{Z}\oplus\ldots \oplus\mathbb{Z}}_k\oplus \mathbb{Z}_{k_1}\oplus\ldots\oplus\mathbb{Z}_{k_l}=\\
\mathbb{Z}^{k}\oplus\mathbb{Z}_{k_1}\oplus\ldots\oplus\mathbb{Z}_{k_l},
\end{gather}
where $k=p-l$.
\end{proof}

\section{Topology of configuration spaces and quantum statistics}\label{topologyAndStaistics}
In this section we define configuration spaces, discuss their basic properties and relate them to quantum statistics.

Let us denote by $M$ the one-particle classical configuration space (e.g., an $m$-dimensional manifold)
and by
\begin{eqnarray}
F_{n}(M)=\{(x_{1},\, x_{2},\ldots,\, x_{n})\,:x_{i}\in X,\, x_{i}\neq x_{j}\},
\end{eqnarray}
the space of $n$ distinct points in $M$. The $n$-particle
configuration space is defined as an orbit space
\begin{eqnarray}
C_{n}(M)=\nicefrac{F_{n}(M)}{S_{n}},
\end{eqnarray}
where $S_{n}$ is the permutation group of $n$ elements and the action
of $S_{n}$ on $F_{n}(M)$ is given by
\begin{eqnarray}
\sigma(x_{1}\,,\ldots\, ,x_{n})=(x_{\sigma^{-1}(1)}\,,\ldots\,, x_{\sigma^{-1}(n)}),\,\,\,\,\forall\sigma\in S_{n}.
\end{eqnarray}
Any closed loop in $C_{n}(M)$ represents a process in which particles start at some particular configuration and end up in the same configuration modulo that they might have been exchanged. As explained in section \ref{fund-group} the space of all loops up to continuous deformations equipped with loop composition is the fundamental group $\pi_{1}(C_{n}(M))$.




\noindent The abelianization of the fundamental group is the first homology group $H_{1}(C_{n}(M))$, and its structure plays an important role in the characterization of quantum statistics. In order to clarify this idea we will first consider the well-known problem of quantum statistics of many particles in $\mathbb{R}^{m}$, $m\geq2$. We will describe fully both the fundamental and homology groups of $C_{n}(\mathbb{R}^{m})$ for $m\geq2$, showing that for $m\geq3$, the only possible statistics are bosonic and fermionic, while for $m=2$ anyon statistics emerges. 
\subsection{Quantum statistics for $C_{n}(\mathbb{R}^{m})$}
\paragraph{The case $M=\mathbb{R}^{m}$ and $m\geq3$.}

When $M=\mathbb{R}^{m}$ and $m\geq3$ the fundamental group $\pi_{1}(F_{n}(\mathbb{R}^{m}))$
is trivial, since there are enough degrees of freedom to avoid coincident configurations during the continuous contraction of any loop. Let us recall that we have a natural action of the permutation group $S_{n}$ on $F_{n}(\mathbb{R}^{m})$
which is free%
\footnote{The action of a group $G$ on $X$ is free iff the stabilizer of any $x\in X$
is the neutral element of $G$ .%
}. In such a situation the following theorem holds \cite{Hatcher}.

\begin{theorem}If an action of a finite group $G$ on a space $Y$ is
free then $G$ is isomorphic to $\nicefrac{\pi_{1}(\nicefrac{Y}{G})}{p_{\ast}(\pi_{1}(Y))}$, where $p:Y\rightarrow\nicefrac{Y}{G}$ is the natural projection and
$p_{\ast}:\pi_{1}(Y)\rightarrow\pi_{1}(\nicefrac{Y}{G})$ is the induced map of fundamental groups.
\end{theorem}

\noindent Notice that in particular if $\pi_{1}(Y)$ is trivial we
get $G=\pi_{1}(\nicefrac{Y}{G})$. In our setting $Y=F_{n}(\mathbb{R}^{m})$
and $G=S_{n}$. The triviality of $\pi_{1}(F_{n}(\mathbb{R}^{m}))$
implies that  the fundamental group of $C_{n}(\mathbb{R}^{m})$ is given
by
\begin{eqnarray}
\pi_{1}(\nicefrac{F_{n}(\mathbb{R}^{m})}{S_{n}})=\pi_{1}(C_{n}(\mathbb{R}^{m}))=S_{n}.
\end{eqnarray}
The homology group $H_{1}(C_{n}(\mathbb{R}^{m})\,,\,\mathbb{Z})$
is the abelianization of $\pi_{1}(C_{n}(\mathbb{R}^{m}))$.  Hence,
\begin{eqnarray}
H_{1}(C_{n}(\mathbb{R}^{m})\,,\,\mathbb{Z})=\mathbb{Z}_{2}.\label{eq:homr3}
\end{eqnarray}
Notice that $H_{1}(C_{n}(\mathbb{R}^{m})\,,\,\mathbb{Z})$ might also be
represented as $(\{1\,,\, e^{i\pi}\},\cdot)$. This result can explain
why we have only bosons and fermions in $\mathbb{R}^{m}$ when $m\geq3$ (see, e.g. \cite{D85} for a detailed discussion).

\paragraph{The case $M=\mathbb{R}^{2}$.}

The case of $M=\mathbb{R}^{2}$ is different as $\pi_{1}(F_{n}(\mathbb{R}^{m}))$
is no longer trivial and it is hard to use Theorem 1 directly. In
fact it can be shown (see \cite{Fox}) that for $M=\mathbb{R}^{2}$
the fundamental group $\pi_{1}(C_{n}(\mathbb{R}^{2}))$ is an Artin braid
group $\mathrm{Br}_{n}(\mathbb{R}^{2})$
\begin{eqnarray}
\mathrm{Br}_{n}(\mathbb{R}^{2})=\langle\sigma_{1},\sigma_{2},\ldots,\sigma_{n-1}\,|\,\sigma_{i}\sigma_{i+1}\sigma_{i}=\sigma_{i+1}\sigma_{i}\sigma_{i+1},\,\sigma_{i}\sigma_{j}=\sigma_{j}\sigma_{i}\rangle,
\end{eqnarray}
where in the first group of relations we take $1\leq i\leq n-2$,
and in the second, we take $|i-j|\geq2.$ Although this group
has a complicated structure, it is easy to see that its abelianization is
\begin{eqnarray}
H_{1}(C_{n}(\mathbb{R}^{2})\,,\,\mathbb{Z})=\mathbb{Z}.
\end{eqnarray}
This simple fact gives rise to a phenomena called anyon statistics \cite{LM77,W90},
i.e., particles in $\mathbb{R}^{2}$ are no longer fermions or bosons
but instead any phase $e^{i\phi}$ can be gained when they are exchanged \cite{D85}.

\section{Graph configuration spaces}\label{graph-statistics}
Here we consider the main problem of this thesis, namely $M=\Gamma$ is a graph. We describe the combinatorial structure of $C_{n}(\Gamma)$.

Let $\Gamma=(V\,,\, E)$ be a metric\footnote{A graph is metric if its edges have assign lengths.} connected simple graph on $|V|$
vertices and $|E|$ edges. Similarly to the previous cases we
define
\begin{eqnarray}
F_{n}(\Gamma)=\{(x_{1},\, x_{2},\ldots,\, x_{n})\,:x_{i}\in\Gamma,\, x_{i}\neq x_{j}\},
\end{eqnarray}
and
\begin{eqnarray}
C_{n}(\Gamma)=F_{n}(\Gamma)/S_{n},
\end{eqnarray}
where $S_{n}$ is the permutation group of $n$ elements. Notice also
that the group $S_{n}$ acts freely on $F_{n}(\Gamma)$, which means that
$F_{n}(\Gamma)$ is the covering space of $C_{n}(\Gamma)$. It seems
{\it a priori} a difficult task to compute $H_{1}(C_{n}(\Gamma))$. Fortunately,
this problem can be reduced to the computation of the first homology group of a cell
complex, which we define now.

\noindent Following \cite{Ghrist}
we define  the $n$-particle combinatorial configuration space as
\begin{eqnarray}
\mathcal{D}^n(\Gamma)=(\Gamma^{\times n}-\tilde{\Delta})/S_{n},
\end{eqnarray}
where $\tilde{\Delta}$ denotes all cells whose closure intersects
with $\Delta$. The space $\mathcal{D}^n(\Gamma)$ possesses a natural cell - complex
structure with vertices as $0$-cells, edges  as $1$-cells, $2$-cells corresponding to moving two particles
along two disjoint
edges in $\Gamma$, and $k$ - cells defined  analogously. The existence of a cell - complex structure happens to be very helpful for investigating
the homotopy structure of the underlying space. Namely, we have the following theorem:
\begin{theorem}\label{Abrams_thm1}\cite{Abrams,PS09}(Abrams) For any graph $\Gamma$
with at least $n$ vertices, the inclusion $\mathcal{D}^n(\Gamma)\hookrightarrow C_{n}(\Gamma)$
is a homotopy equivalence iff the following hold:
\begin{enumerate}
\item Each path between distinct vertices of valence not equal to two passes
through at least $n-1$ edges.
\item Each closed path in $\Gamma$ passes through at least $n+1$ edges.
\end{enumerate}
\end{theorem}

\noindent For $n=2$  these conditions are automatically satisfied (provided $\Gamma$ is simple).  Intuitively, they can be understood as follows:
\begin{enumerate}
\item In order to have homotopy equivalence between $\mathcal{D}^n(\Gamma)$ and
$C_{n}(\Gamma)$, we need to be able to accommodate $n$ particles on
every edge of graph $\Gamma$.
\item For every cycle there is at least one free (not occupied) vertex which
enables the exchange of particles along this cycle.
\end{enumerate}
Using Theorem \ref{Abrams_thm1}, the problem of finding  $H_{1}(C_{n}(\Gamma))$
is reduced to the problem of computing $H_{1}(\mathcal{D}^n(\Gamma))$. In the following two chapters of the thesis we will discuss
how to determine $H_{1}(\mathcal{D}^n(\Gamma))$. Meanwhile,  to clarify the idea behind theorem \ref{Abrams_thm1} let us consider the following example.
\begin{example}Let $\Gamma$ be a star graph on four vertices (see figure
\ref{fig1}(a)). The two-particle configuration spaces $C_{2}(\Gamma)$ and
$\mathcal{D}^2(\Gamma)$ are shown in figures \ref{fig1}(b),(c). Notice that $C_{2}(\Gamma)$
consists of six $2$ - cells (three are interiors of triangles
and the other three are interiors of squares), eleven $1$ - cells and
six $0$ - cells. Vertices $(1,1)$, $(2,2)$, $(3,3)$ and $(4,4)$
do not belong to $C_{2}(\Gamma)$. Similarly dashed edges, i.e.~$(1,1)-(2,2)$,
$(2,2)-(4,4)$, $(2,2)-(3,3)$ do not belong to $C_{2}(\Gamma)$.
This is why $C_{2}(\Gamma)$ is not a cell complex - not every cell
has its boundary in $C_{2}(\Gamma)$. Notice that cells of $C_{2}(\Gamma)$ whose closures
intersect $\Delta$ (denoted by dashed lines and diamond points) do not
influence the homotopy type of $C_{2}(\Gamma)$ (see  figures \ref{fig1}(b),(c)). Hence, the space $\mathcal{D}^2(\Gamma)$ has
the same homotopy type as $C_{2}(\Gamma)$, but consists of six $1$
- cells and six $0$ - cells. $\mathcal{D}^2(\Gamma)$ is subspace of $C_{2}(\Gamma)$
denoted by dotted lines in figure \ref{fig1}(b). In particular, one can also easily calculate that $H_1(C_{2}(\Gamma))=H_1(\mathcal{D}^2(\Gamma))$.
\begin{figure}[H]
\includegraphics[scale=0.35]{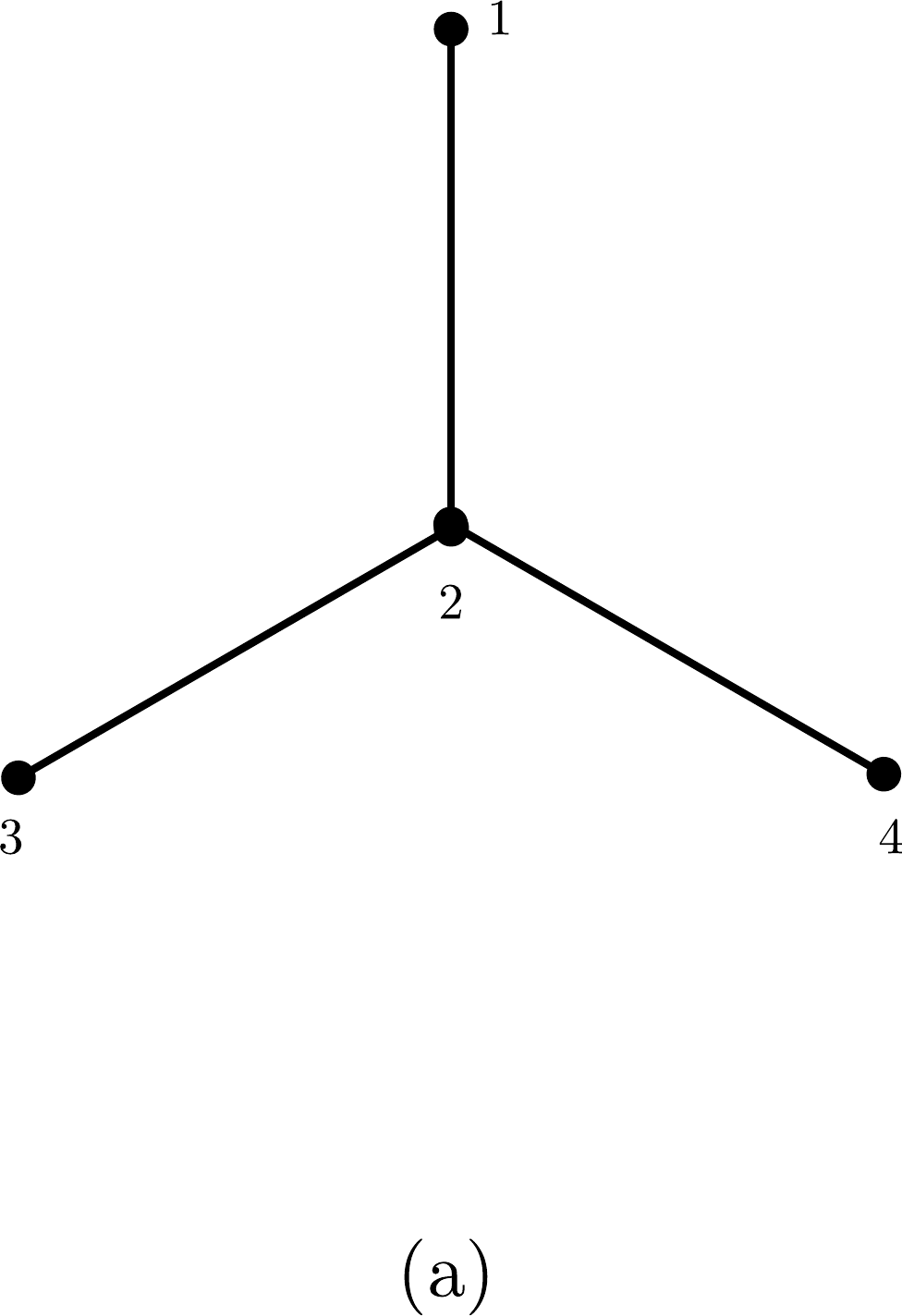}~~~\includegraphics[scale=0.35]{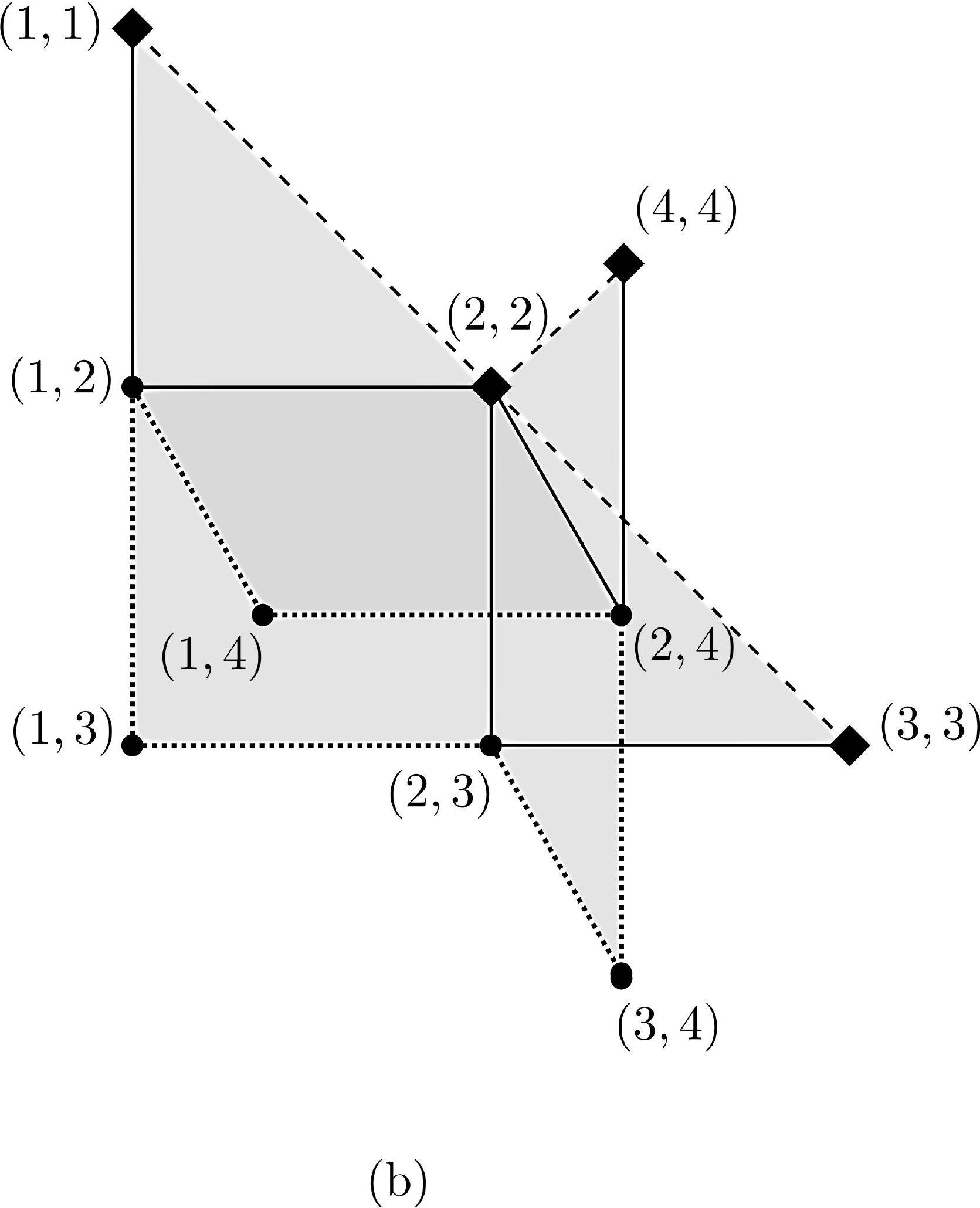}~~~\includegraphics[scale=0.35]{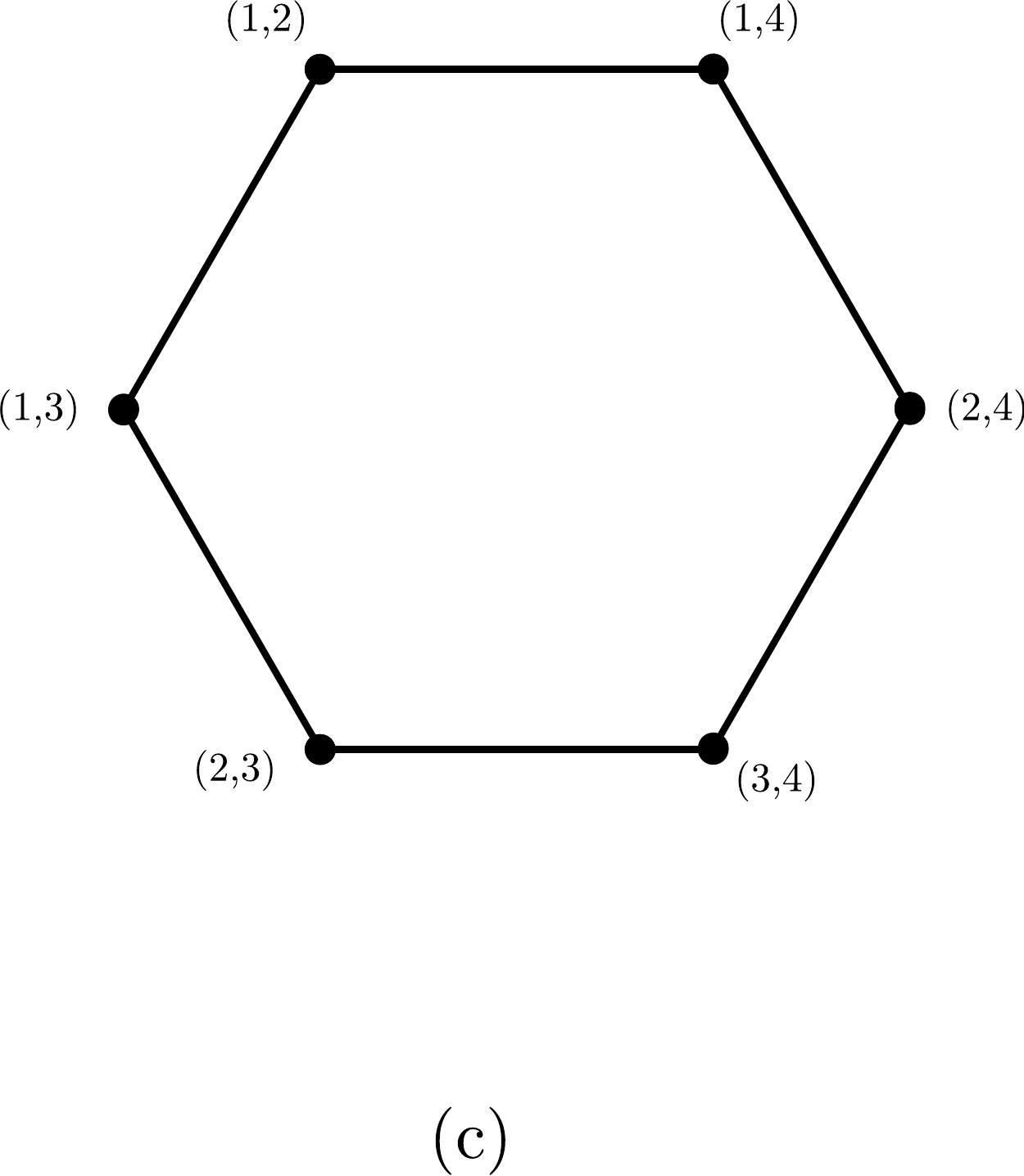}
\caption{\label{fig1} (a) The star graph $\Gamma$, (b) the two-particle configuration space
$C_{2}(\Gamma)$, (c) the two-particle discrete configuration space
$\mathcal{D}^2(\Gamma)$.}
\end{figure}
\end{example}

\section{Quantum graphs}\label{model}

The main attraction of quantum graphs is that they are simple models to study complicated phenomena. 
A metric graph $\Gamma=(V,E)$ is a graph whose edges have assigned lengths. One can consider single-particle quantum mechanics on a metric graph. A particle is described by a collection of wavefunctions, $\{\Psi_e\}_{e\in E}$,  living on the edges of $\Gamma$. On each edge a Hamiltonian, $H_e$, is defined. Typically it is of the form:
\begin{equation}
H_e=\frac{1}{2}(-i\frac{d}{dx}-A_e(x_e))^2+V_e(x_e).
\end{equation}  
In order to ensure selfadjointness of the Hamiltonian, boundary conditions at the vertices are introduced. For a free particle, i.e when $H_e=-\frac{1}{2}\frac{d^2}{dx^2}$ an example is Neumann boundary conditions:
\begin{enumerate}
\item Wavefunctions are continous at vertices.
\item The sum of outgoing derivatives vanishes at each vertex.
\end{enumerate}
The full characterization of boundary conditions for a free particle on a quantum graph was given in \cite{KS99} and is expressed in terms of Lagrangian planes of a certain complex symplectic form. Quantum mechanics of a single particle on a quantum graph is rather well understood. It has bee been recently extensively investigated from various angeles, e.g. superconductivity \cite{G81} quantum chaos \cite{KS97} or Anderson localization \cite{ASW06}. 

For understanding quantum statistics on a metric graph $\Gamma$, which is a topological property of the underlying configuration space $C_n(\Gamma)$, it is irrelevant to know the lengths of the edges of $\Gamma$. That is why in the following we will always treat $\Gamma$ as a $1$-dimensional cell complex. Using the fact that spaces $C_n(\Gamma)$ and $\mathcal{D}^n(\Gamma)$ are homotopic and by definition of the first homology group we need to only study the $2$-skeleton of $\mathcal{D}^n(\Gamma)$. In order to provide some physical intuition we will now describe the following model.

\subsection{Topological gauge potential}

Assume that $\Gamma$ is, as explained in section \ref{graph-statistics}, sufficiently subdivided. Note that $0$-skeleton of $\mathcal{D}^n(\Gamma)$ consists of unordered collections of $n$ distinct vertices of $\Gamma$, $\{v_{i_1},\ldots,v_{i_n}\}$. The connections between $0$-cells are described by $1$-cells of $\mathcal{D}^n(\Gamma)$. Two $0$-cells are connected by an edge iff they have $n-1$ vertices in common and the remaining two vertices are connected by an edge in $\Gamma$. In other words $1$-cells of $\mathcal{D}^n(\Gamma)$ are of the form $v_1\times\ldots\times v_{n-1}\times e$ up to permutations, where $v_j$ are vertices of $\Gamma$ and $e=j\rightarrow k$ is an edge of $\Gamma$ whose endpoints are not $\{v_1,\ldots, v_{n-1}\}$. For simplicity we will use the following notation
\[
\{v_1,\ldots,v_{n-1},j\rightarrow k\}:=v_1\times\ldots\times v_{n-1}\times e.
\]
An $n$-particle gauge potential is a function $\Omega^{(n)}$ defined on the directed edges of $\mathcal{D}^n(\Gamma)$ with the values in $\mathbb{R}$ modulo $2\pi$ such that
\begin{equation}\label{asym1}
\Omega^{(n)}(\{v_1,\ldots,v_{n-1},k\rightarrow j\})=-\Omega^{(n)}(\{v_1,\ldots,v_{n-1},j\rightarrow k\}).
\end{equation}
In order to define $\Omega$ on linear combinations of directed edges we extend (\ref{asym1}) by linearity.

For a given gauge potential, $\Omega^{(n)}$ the sum of its values calculated on the directed edges of an oriented cycle $C$ will be called the flux of $\Omega$ through $C$ and denoted $\Omega(C)$. Two gauge potentials $\Omega_1^{(n)}$ and $\Omega_2^{(n)}$ are called equivalent if for any oriented cycle $C$ the fluxes $\Omega_1^{(n)}(C)$ and $\Omega_2^{(n)}(C)$ are equal modulo $2\pi$.

The $n$-particle gauge potential $\Omega^{(n)}$ is called a {\it topological gauge potential} if for any contractible oriented cycle $C$ in $\mathcal{D}^n(\Gamma)$ the flux $\Omega^{(n)}(C)=0\,\mathrm{mod}\,2\pi$. It is thus clear that equivalence classes of topological gauge potentials are in 1-1 correspondence with the equivalence classes in $H_1(\mathcal{D}^n(\Gamma))$. Therefore, characterization of quantum statistics is characterization of all possible topological gauge potentials. These potentials can be incorporated into the Hamiltonian of a so-called tight-binding model. In short, it is a model whose underlying Hilbert space is spanned by elements of the $0$-skeleton of $\mathcal{D}^{n}(\Gamma)$ and the Hamiltonian is given by the adjacency matrix of the $1$-skeleton of $\mathcal{D}^{n}(\Gamma)$. As $\Omega^{(n)} $ is defined on the edges of $\mathcal{D}^{n}(\Gamma)$ it can be added to the Hamiltonian by changing: $H_{\{v_1,\ldots,v_{n-1},k\rightarrow j\}}\rightarrow H_{\{v_1,\ldots,v_{n-1},k\rightarrow j\}} e^{i\Omega^{(n)}(\{v_1,\ldots,v_{n-1},k\rightarrow j\})}$ (see \cite{JHJKJR} for more detailed discussion in case of two particles).

\chapter{Quantum Statistics on graphs}\label{chQS}

In this chapter we raise the question of what  quantum statistics are possible on quantum graphs. In particular we develop a full characterization of abelian quantum statistics on graphs. We explain how the number of anyon phases is related to connectivity.  For $2$-connected graphs the independence of quantum statistics with respect to the number of particles is proven. For non-planar $3$-connected graphs we identify bosons and fermions as the only possible statistics, whereas for planar  $3$-connected graphs we show that one anyon phase exists. Our approach also yields an alternative proof of the structure theorem for the first homology group of $n$-particle graph configuration spaces. Finally, we determine the topological gauge potentials for $2$-connected graphs.

In order to explore how the quantum statistics picture depends on topology, the case of two indistinguishable particles on a graph was  studied in \cite{JHJKJR} (see also \cite{BE92}). Recall that a graph $\Gamma$ is a network consisting of vertices (or nodes) connected by edges.  Quantum mechanically, one can either consider the one-dimensional Schr\"{o}dinger operator acting on the edges, with matching conditions for the wavefunctions at the vertices, or a discrete Schr\"{o}dinger operator acting on connected vertices (i.e.~a tight-binding model on the graph).  Such systems are of considerable independent interest and their single-particle quantum mechanics has been studied extensively in recent years \cite{Berkolaiko13}.  The extension of this theory to many-particle quantum graphs was another motivation for \cite{JHJKJR} (see also \cite{Bolte13}).  The discrete case turns out to be significantly easier to analyse, and in this situation it was found that a rich array of anyon statistics are kinematically possible.   Specifically, certain graphs were found to support anyons while others can only support fermions or bosons.  This was demonstrated by analysing the topology of the corresponding configuration graphs $C_2(\Gamma)=(\Gamma^{\times 2}-\Delta)/S_{2}$ in various examples.  It opens up the problem of determining general relations between the quantum statistics of a graph and its topology.

As explained in the previous section, mathematically the determination of quantum statistics reduces to finding the first homology group $H_{1}$ of the appropriate classical configuration space, $C_n(M)$.  Although the calculation for $C_n(\mathbb{R}^N)$ is relatively elementary, it becomes a non-trivial task when $\mathbb{R}^{N}$ is replaced by a general graph $\Gamma$. One possible route  is to use discrete Morse theory, as developed by Forman \cite{Forman98}. This is a combinatorial counterpart of classical Morse theory, which applies to cell complexes. In essence, it reduces the problem of finding $H_{1}(M)$, where $M$ is a cell complex, to the construction of certain discrete Morse functions, or equivalently discrete gradient vector fields. Following this line of reasoning
Farley and Sabalka \cite{FS05} defined the appropriate discrete vector fields and gave a formula for the first homology groups of tree graphs.  Recently, making extensive use of discrete Morse theory and some graph invariants, Ko and Park \cite{KP11} extended the results of \cite{FS05} to an arbitrary graph $\Gamma$.  However, their approach relies on a suite of relatively elaborate techniques -- mostly connected to a proper ordering of vertices and choices of trees to reduce the number of critical cells -- and the relationship to, and consequences for, the physics of quantum statistics are not easily identified.

In this chapter we give a full characterization of all possible abelian quantum statistics on graphs. In order to achieve this we develop a new set of ideas and methods which lead to an alternative proof of the structure theorem for the first homology group of the $n$-particle configuration space obtained by Ko and Park \cite{KP11}. Our reasoning, which is more elementary in that it makes minimal use of discrete Morse theory, is based on a set of simple combinatorial relations which stem from the analysis of some canonical small graphs.   The advantage for us of this approach is that it is explicit and direct.  This makes the essential physical ideas much more transparent and so enables us to identify the key topological determinants of the quantum statistics.  It also enables us to develop some further physical consequences. In particular we give a full characterization of the topological gauge potentials on $2$-connected graphs, and  identify some examples of particular physical interest, in which the quantum statistics have features that are subtle.

The chapter is organized as follows. We start with a discussion, in section
\ref{sec:Examples}, of some physically interesting examples of quantum
statistics on graphs, in order to motivate the general theory that follows. In section \ref{sec:Graph-configuration-spaces}
we define some basic properties of graph configuration
spaces. In section \ref{sec:Two-particle-quantum-statistics}
we develop a full characterization of the first homology group for $2$-particle graph
configuration spaces. In section \ref{sec:N-particle-statistics-for}
we give a simple argument for the stabilization of quantum statistics
with respect to the number of particles for $2$-connected graphs.
Using this we obtain the desired result for $n$-particle graph configuration
spaces when $\Gamma$ is $2$-connected. In order to generalize the
result to $1$-connected graphs we consider star and fan graphs.
The  main result is obtained at the end of section \ref{sec:N-particle-statistics-on}.
The first homology group $H_1(C_n(\Gamma))$ is given by the direct sum of a free component, which corresponds to anyon phases and Aharonov-Bohm phases, and a torsion component, which is restricted to be a direct sum of copies of $\mathbb{Z}_2$.  The last part of the chapter is devoted to the characterization
of topological gauge potentials for $2$-connected graphs.

\section{Quantum statistics on graphs\label{sec:Examples}}

In this section we discuss several examples which illustrate  some interesting and surprising aspects of quantum statistics on graphs.  A determining factor turns out to be the {\it connectivity} of a graph.  We recall (cf \cite{tutte01}) that a graph is {\it $k$-connected} if it remains connected after removing any  $k-1$ vertices.  (Note that a $k$-connected graph is also $j$-connected for any $j<k$.)  
According to Menger's theorem \cite{tutte01}, a graph is $k$-connected if and only if every pair of distinct vertices can be joined by at least $k$ disjoint paths.
A $k$-connected graph can be decomposed into $(k+1)$-connected components, unless it is complete \cite{Holberg92}.  Thus, a graph may be regarded as being built out of more highly connected components.
Quantum statistics, as we shall see, depends on $k$-connectedness up to $k = 3$.  (Remark: in this thesis, quantum statistics refers specifically to phases involving cycles of two or more particles; phases associated with single-particle cycles, called Aharonov-Bohm phases, are introduced in Section 2.4 below).

\subsection{$3$-connected graphs}

Quantum statistics for any number of particles on a $3$-connected graph depends only on whether the graph is planar, and not on any additional structure.
 We recall that a graph is planar if it can be drawn in the plane without  crossings.
For planar $3$-connected graphs we will show that the statistics is characterised by a single anyon phase associated with cycles in which a pair of particles exchange positions.  For non-planar $3$-connected graphs, the statistics is either Bose or Fermi -- in effect, the anyon phase is restricted to be $0$ and $\pi$.  Thus, as far as quantum statistics is concerned, three- and higher-connected graphs behave like $\bbR^2$ in the planar case and $\bbR^d$, $d > 2$, in the nonplanar case. A new aspect for graphs is the possibility of combining planar and nonplanar components. The graph shown in  figure~\ref{fig:The-large-almost} consists of a large square lattice in which four cells have been replaced by a defect in the form of a $K_5$ subgraph, the (nonplanar) fully connected graph on five vertices. This local substitution makes the full graph nonplanar, thereby excluding anyon statistics.
\begin{figure}[h]
\begin{center}\includegraphics[scale=0.5]{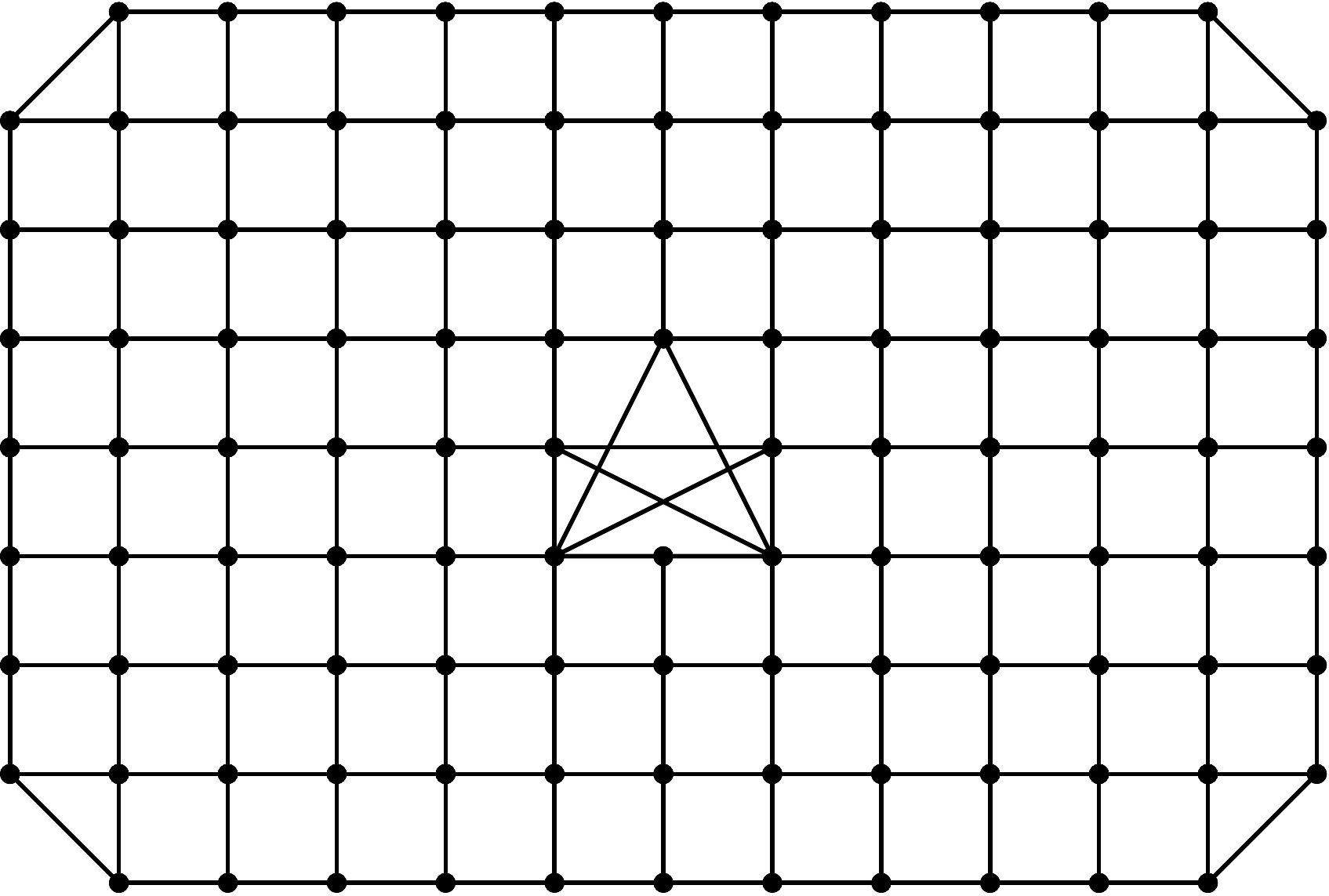}\end{center}
\caption{\label{fig:The-large-almost}The large almost planar $3$-connected
graph. }
\end{figure}

One of the simplest examples of this phenomenon is provided by the graph $G$ shown in figure~\ref{fig: GK5}. $G$ is planar $3$-connected, and therefore supports an anyon phase.  However, if  an additional edge $e$ is added, the resulting graph is $K_5$, and therefore supports only Bose or Fermi statistics.  One can continuously interpolate from a quantum Hamiltonian defined on $K_5$ to one defined by $G$ by introducing  an amplitude coefficient $\epsilon$ for transitions along $e$.  For $\epsilon = 0$, the edge $e$ is effectively absent, and the resulting Hamiltonian is  defined on $G$.  This situation might appear to be paradoxical; how could anyon statistics, well defined for $\epsilon = 0$, suddenly disappear for $\epsilon \neq 0$? The resolution lies in the fact that an anyon phase defined for $\epsilon = 0$ introduces, for $\epsilon \neq 0$, physical effects that cannot be attributed to quantum statistics (unless the phase is $0$ or $\pi$).  The transition between planar and nonplanar geometries, which is easily effected with quantum graphs, merits further study.
\begin{figure}[H]
\begin{center}\includegraphics[scale=0.15]{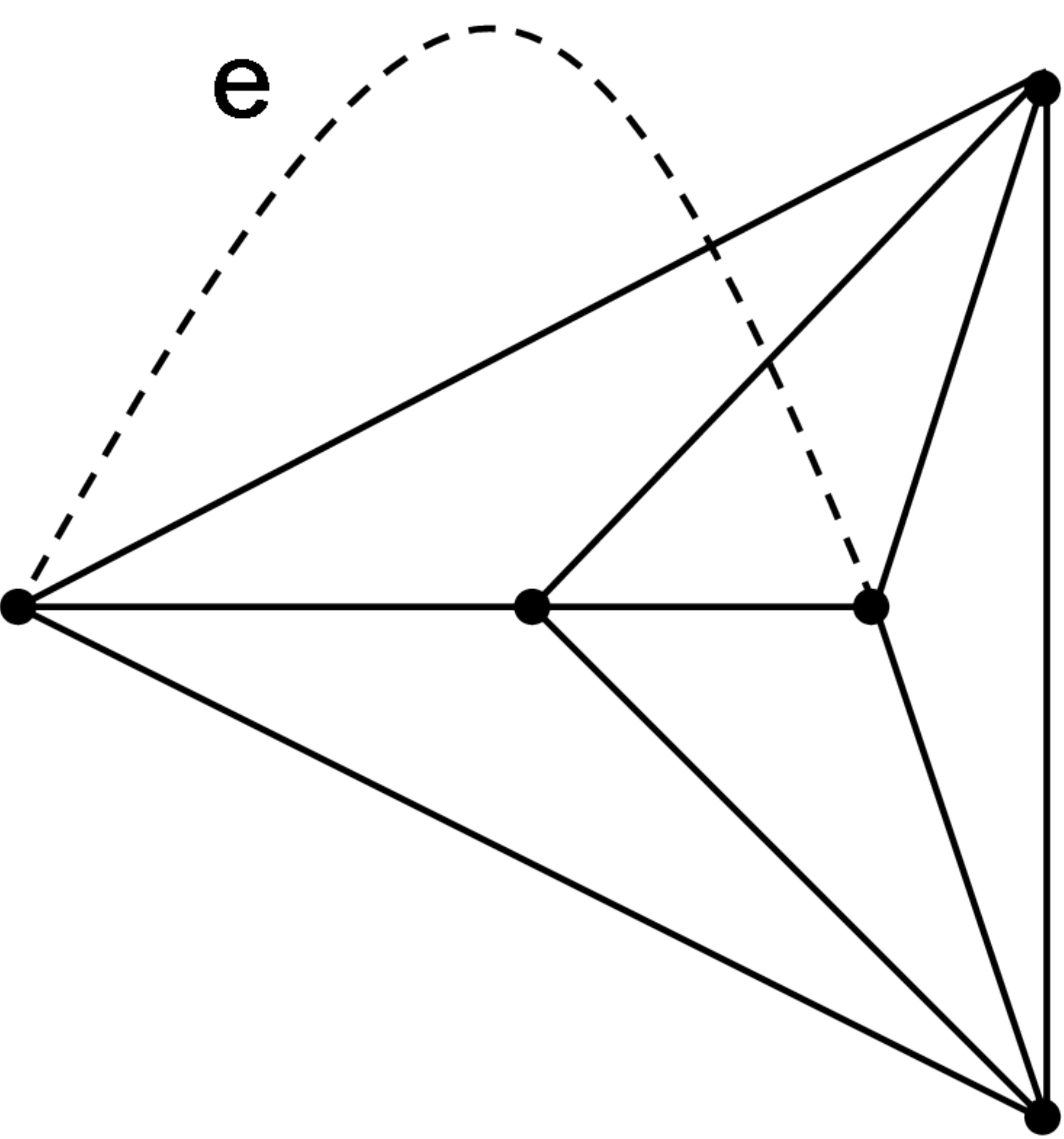}\end{center}
\caption{\label{fig: GK5}The graph $G$ (without  the edge $e$) is planar $3$-connected.  With $e$, the graph is $K_5$. }
\end{figure}

\subsection{$2$-connected graphs}

Quantum statistics on $2$-connected graphs is more complex, and depends on the decomposition of individual  graphs into cycles and $3$-connected components (see Section~\ref{subsec: two-connected}).
There may be multiple anyon and $\bbZ_2$ (or Bose/Fermi alternative) phases. But $2$-connected graphs share the following important property: their quantum statistics do not depend on the number of particles, and therefore can be regarded as a characteristic of the particle species.  This property is important physically; it means that there is a building-up principle for increasing the number of particles in the system.  This is described in detail in Section~\ref{sec: gauge potential},
 where we show how to construct an $n$-particle Hamiltonian from a two-particle Hamiltonian. Interesting examples are also obtained by building $2$-connected graphs out of higher-connected components.  Figure~\ref{fig: BF_chain} shows a chain of identical non-planar 3-connected components. The links between components, represented by lines in figure~\ref{fig: BF_chain}, consist of at least two edges, so that resulting graph is $2$-connected. In this case, the quantum statistics is in fact independent of the number of particles, and may be determined by specifying exchange phases ($0$ or $\pi$) for each component in the chain.  Thus, particles can act as bosons or fermions in different parts of the graph.
\begin{figure}[h]
\begin{center}\includegraphics[scale=1.3]{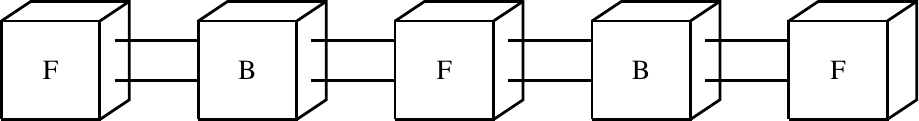}\end{center}
\caption{\label{fig: BF_chain} Linear chain of $3$-connected nonplanar components with alternating Bose and Fermi statistics.}
\end{figure}

\subsection{$1$-connected graphs}

Quantum statistics on graphs achieves its full complexity for 1-connected graphs, in which case it also depends on the number of particles $n$.  A representative example, treated in detail in Section~\ref{sub:The-star-graphs},
is a star graph with $E$ edges, for which the number of anyon phases is given by
\begin{gather*}
\beta_{n}^{E}={n+E-2 \choose E-1}\left(E-2\right)-{n+E-2 \choose E-2}+1,
\end{gather*}
and therefore depends on both $E$ and $n$.

\subsection{Aharonov-Bohm phases}\label{sec:ABphase}

Configuration-space cycles on which one particle moves around a circuit $C$ while the others remain stationary play an important role in the analysis of quantum statistics which follows.  We call these Aharonov-Bohm cycles, and the corresponding phases Aharonov-Bohm phases, because they correspond physically to magnetic fluxes threading $C$.  In many-body systems, Aharonov-Bohm phases and quantum statistics phases can interact in interesting ways.  In particular, Aharonov-Bohm phases can depend on the positions of the  stationary particles.  An example is shown in the two-particle  octahedron graph (see figure~\ref{fig: oct_AB}), in which the  Aharonov-Bohm phase associated with one particle going around the equator depends on whether the second particle is at the north or south pole.  For $3$-connected non-planar graphs, it can be shown that Aharonov-Bohm phases are independent of the positions of the stationery particles.  (The octahedron graph, despite appearances, is planar.)
\begin{figure}[h]
\begin{center}~~~~~\includegraphics[scale=0.2]{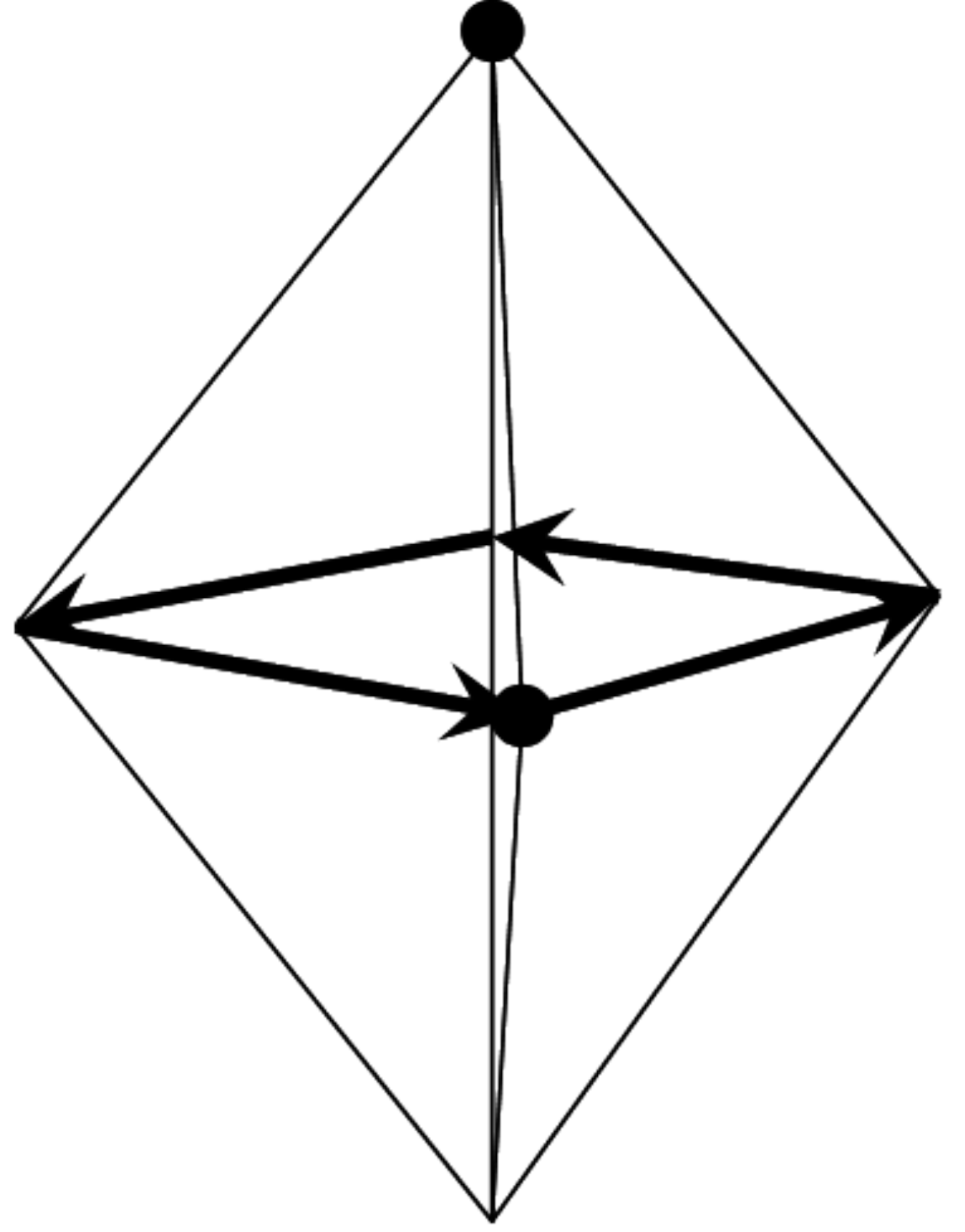}~~~~~~~~~~~~~~~~~~\includegraphics[scale=0.2]{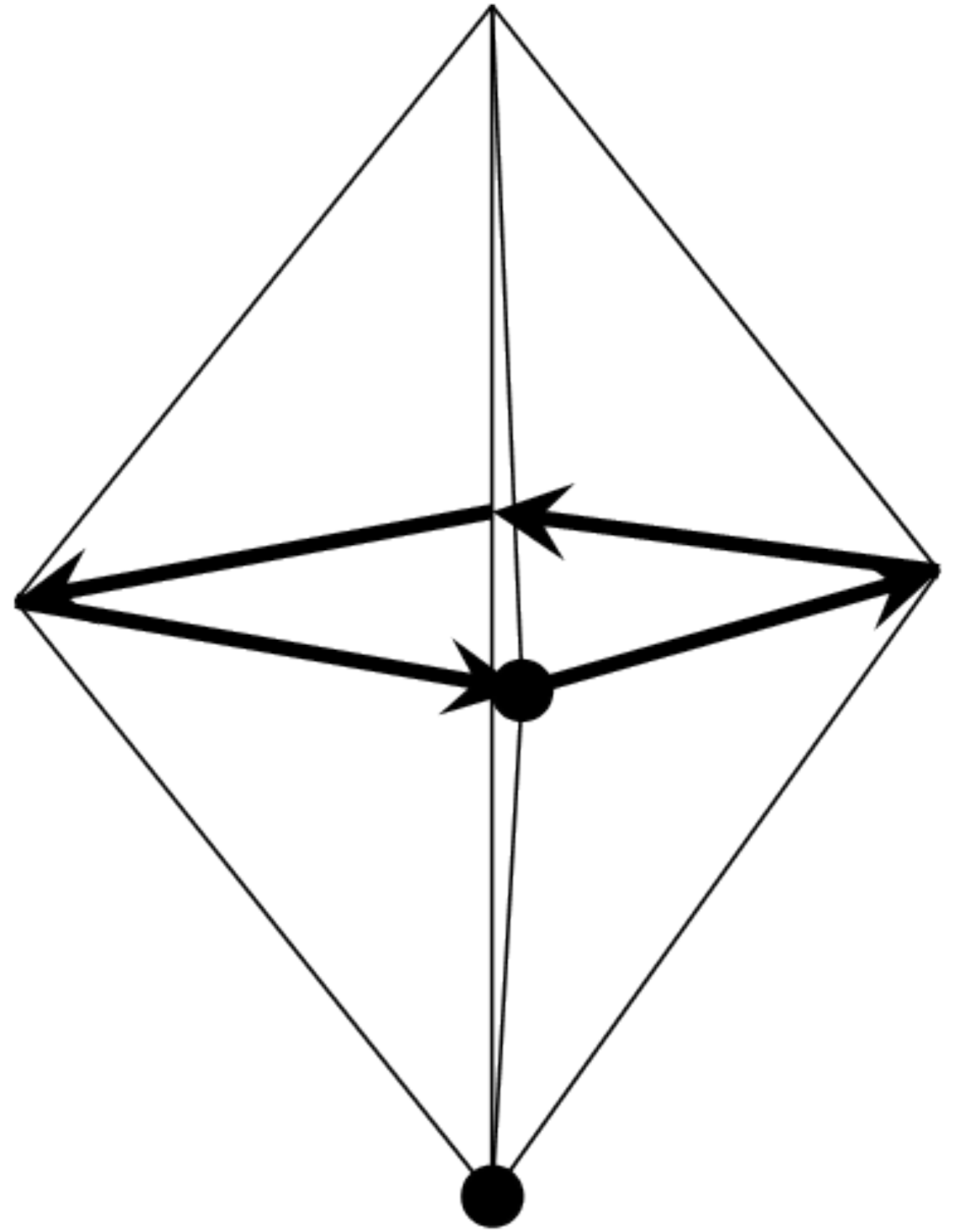}\end{center}
\caption{\label{fig: oct_AB} The Aharonov-Bohm phase for the equatorial cycle depends on whether the second particle is at the north or south pole.}
\end{figure}

\section{Graph configuration spaces\label{sec:Graph-configuration-spaces}}

In this section we repeat some definitions and theorems proved in the introduction. They will play a basic role in the current chapter. Let $\Gamma$ be a metric connected simple graph with $V$
vertices and $E$ edges. In a metric graph edges correspond to finite closed intervals of $\mathbb{R}$.  However, as we will be interested in the topology of the graph, the length of the edges will not play a role in the discussion.
As explained in the introduction an undirected edge between vertices $v_{1}$
and $v_{2}$ will be denoted by $v_{1}\leftrightarrow v_{2}$.  It will also be convenient to be able to label directed edges,
so $v_{1}\rightarrow v_{2}$ and $v_{2}\rightarrow v_{1}$ will denote the directed edges associated with $v_{1}\leftrightarrow v_{2}$.
A path joining two vertices $v_1$ and $v_m$ is then specified by a sequence of $m-1$ directed edges, written $v_1 \rightarrow v_2 \rightarrow \dots \rightarrow v_m$.

We define the \emph{$n$-particle configuration space} as the quotient space
\begin{eqnarray}
C_{n}(\Gamma)=\left(\Gamma^{\times n}-\Delta\right)/S_{n},
\end{eqnarray}
where $S_{n}$ is the permutation group of $n$ elements and
\begin{eqnarray}
\Delta=\{(x_{1},x_{2},\ldots,x_{n}):\exists_{i,j}\, x_{i}=x_{j}\},
\end{eqnarray}
is the set of coincident configurations.
We are interested in the calculation of the first homology group,
$H_{1}(C_{n}(\Gamma))$ of $C_{n}(\Gamma)$. The space $C_{n}(\Gamma)$
is not a cell complex. However, it is homotopy equivalent to the space
$\mathcal{D}^{n}(\Gamma)$, which is a cell complex, defined below.

Recall that a cell complex $X$ is a nested sequence of topological spaces
\begin{eqnarray}
X^{0}\subseteq X^{1}\subseteq\dots\subseteq X^{n},
\end{eqnarray}
where the $X^{k}$'s are the so-called $k$-skeletons defined as follows:
\begin{itemize}
\item The $0$ - skeleton $X^{0}$ is a finite set of points.
\item For $\mathbb{N}\ni k>0$, the $k$ - skeleton $X^{k}$ is the result
of attaching $k$ - dimensional balls $B_{k} = \{x\in\mathbb{R}^{k}\,:\,\|x\|\leq1\}$ to $X^{k-1}$ by gluing
maps
\begin{eqnarray}
\sigma:S^{k-1}\rightarrow X^{k-1},
\end{eqnarray}
where $S^{k-1}$ is the unit-sphere $S^{k-1}=\{x\in\mathbb{R}^{k}\,:\,\|x\|=1\}$.
\end{itemize}
A $k$-cell is the interior of the ball $B_{k}$ attached to the $(k-1)$-skeleton $X^{k-1}$.

Every simple graph $\Gamma$ is naturally a cell complex; the vertices are $0$-cells (points)
and edges are $1$-cells ($1$-dimensional balls whose boundaries are the $0$-cells). The product $\Gamma^{\times n}$ then naturally
inherits a cell complex structure. The cells of $\Gamma^{\times n}$ are
Cartesian products of cells of $\Gamma$. It is clear
that the space $C_{n}(\Gamma)$ is not a cell complex as
points belonging to $\Delta$ have been deleted. Following \cite{Abrams}
we define an \emph{$n$-particle combinatorial configuration space} as
\begin{eqnarray}
\mathcal{D}^{n}(\Gamma)=(\Gamma^{\times n}-\tilde{\Delta})/S_{n},
\end{eqnarray}
where $\tilde{\Delta}$ denotes all cells whose closure intersects
with $\Delta$. The space $\mathcal{D}^{n}(\Gamma)$ possesses a natural
cell complex structure. Moreover,
\begin{theorem}
\label{Abrams_thm}\cite{Abrams} For any graph $\Gamma$ with at
least $n$ vertices, the inclusion $\mathcal{D}^{n}(\Gamma)\hookrightarrow C_{n}(\Gamma)$
is a homotopy equivalence iff the following hold:
\begin{enumerate}
\item Each path between distinct vertices of valence not equal to two passes
through at least $n-1$ edges.
\item Each closed path in $\Gamma$ passes through at least $n+1$ edges.
\end{enumerate}
\end{theorem}
\noindent Following \cite{Abrams,FS05} we refer to a graph $\Gamma$ with properties 1 and 2 as \emph{sufficiently subdivided}. For $n=2$ these conditions are automatically satisfied
(provided $\Gamma$ is simple). Intuitively, they can be understood
as follows:
\begin{enumerate}
\item In order to have homotopy equivalence between $\mathcal{D}^{n}(\Gamma)$
and $C_{n}(\Gamma)$, we need to be able to accommodate $n$ particles
on every edge of graph $\Gamma$. This is done by introducing $n-2$ trivial vertices of degree $2$ to make a line subgraph between every adjacent pair of non-trivial vertices in the original graph $\Gamma$.
\item For every cycle there is at least one free (not occupied) vertex which
enables the exchange of particles around this cycle.
\end{enumerate}
For a sufficiently subdivided graph $\Gamma$ we can now effectively treat $\Gamma$ as a combinatorial graph where particles are accommodated at vertices and hop between adjacent unoccupied vertices along edges of $\Gamma$.  See Figure \ref{fig:(a)-The-Y} for a comparison of the configuration spaces $C_{2}(\Gamma)$ and $\mathcal{D}^{2}(\Gamma)$ of a Y-graph.

Using Theorem \ref{Abrams_thm}, the problem of finding $H_{1}(C_{n}(\Gamma))$
is reduced to the problem of computing $H_{1}(\mathcal{D}^{n}(\Gamma))$.
In the next sections we show how to determine $H_{1}(\mathcal{D}^{n}(\Gamma))$
for an arbitrary simple graph $\Gamma$. Note, however, that by the
structure theorem for finitely generated modules \cite{Nakahara}

\begin{gather}
H_{1}(\mathcal{D}^{n}(\Gamma))=\mathbb{Z}^{k}\oplus T_{l}\label{eq:finite-module}
\end{gather}
where $T_{l}$ is the torsion, i.e.
\begin{gather}
T_{l}=\mathbb{Z}_{n_1}\oplus\ldots\oplus\mathbb{Z}_{n_l},\label{eq:finite-module-torsion}
\end{gather}
and $n_{i}|n_{i+1}$. In other words $H_{1}(\mathcal{D}^{n}(\Gamma))$
is determined by $k$ free parameters $\{\phi_{1},\ldots,\phi_{k}\}$
and $l$ discrete parameters $\{\psi_{1},\ldots,\psi_{l}\}$ such that for
each $i\in\{1,\ldots l\}$

\begin{gather}
n_{i}\psi_{i}=0\,\,\mbox{mod}\,2\pi,\,\,n_i\in \mathbb{N}\,\,\,\:\mbox{and}\,\,n_{i}|n_{i+1}.\label{eq:finite-module-1}
\end{gather}
Taking into account their physical interpretation we will call the
parameters $\phi$ and $\psi$ continuous and discrete phases respectively.

\section{Two-particle quantum statistics\label{sec:Two-particle-quantum-statistics}}

In this section we fully describe the first homology group $H_{1}(\mathcal{D}^{2}(\Gamma))$
for an arbitrary connected simple graph $\Gamma$. We start with three simple examples: a cycle, a Y-graph and a lasso. The $2$-particle discrete configuration
space of the lasso reveals an important relation between the exchange
phase on the Y-graph and on the cycle. Combining this relation with
an ansatz for a perhaps over-complete spanning set of the cycle space of
$\mathcal{D}^{2}(\Gamma)$ and some combinatorial properties of $k$-connected
graphs, we give a formula for $H_{1}(\mathcal{D}^{2}(\Gamma))$. Our
argument is divided into three parts; corresponding to $3$-, $2$- and $1$-connected graphs respectively.

\paragraph{Three examples }
\begin{itemize}
\item Let $\Gamma$ be a triangle graph shown in figure \ref{fig:The-triangle}(a).
Its combinatorial configuration space $\mathcal{D}^{2}(\Gamma)$ is shown
in figure 1(b). The cycle $(1,2)\rightarrow(1,3)\rightarrow(2,3)\rightarrow(1,2)$
is not contractible and hence $H_{1}(\mathcal{D}^{2}(\Gamma))=\mathbb{Z}$.
In other words we have one free phase $\phi_{c}$ and no torsion.
\begin{figure}[h]
\begin{center}\includegraphics[scale=0.45]{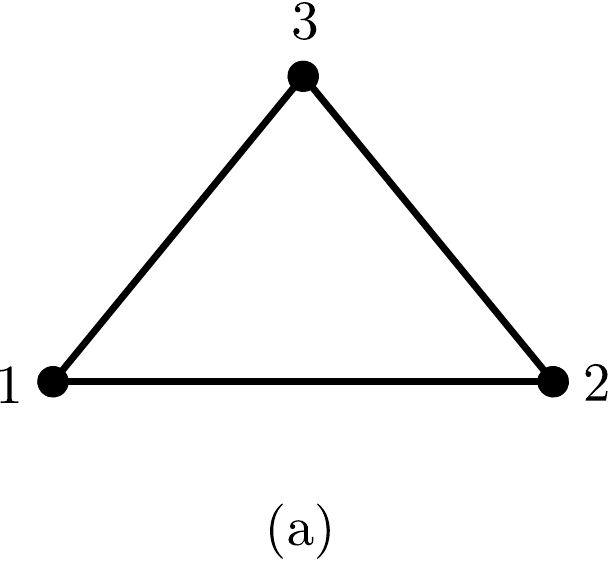}~~~~~~~~~~~~~~~~~~~~~~~~~~~~\includegraphics[scale=0.45]{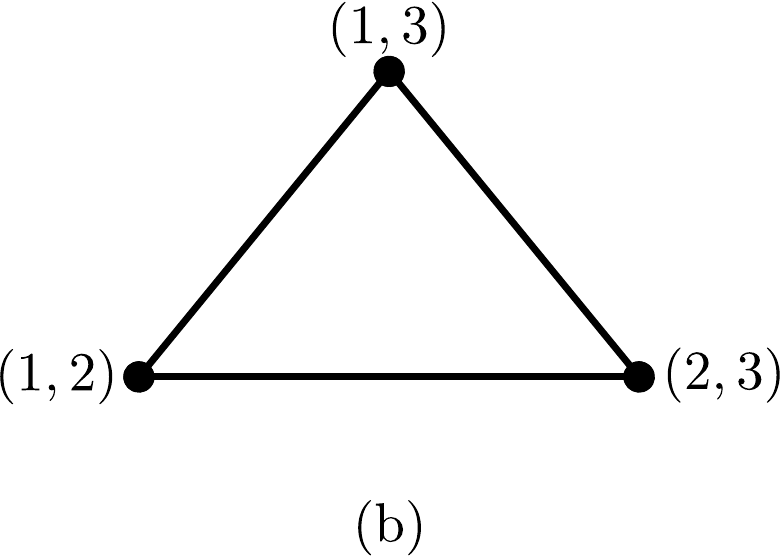}\end{center}

\caption{\label{fig:The-triangle}(a) The triangle graph $\Gamma$ (b) The
$2$-particle configuration space $\mathcal{D}^{2}(\Gamma)$.}
\end{figure}

\item Let $\Gamma$ be a Y-graph shown in figure \ref{fig:(a)-The-Y}(a).
Its combinatorial configuration space $\mathcal{D}^{2}(\Gamma)$ is shown
in figure \ref{fig:(a)-The-Y}(b). The cycle $(1,2)\rightarrow(1,3)\rightarrow(2,3)\rightarrow(3,4)\rightarrow(2,4)\rightarrow(1,4)\rightarrow(1,2)$
is not contractible and $H_{1}(\mathcal{D}^{2}(\Gamma))=\mathbb{Z}$.
Hence we have one free phase $\phi_{Y}$ and no torsion.  For comparison the configuration space $C_2(\Gamma)$ is shown in figure \ref{fig:(a)-The-Y}(c).  Contracting the triangular planes onto the hexagon and then contracting the surface of the hexagon to the boundary (expanding the empty vertex in the center) one obtains the combinatorial configuration space shown in figure \ref{fig:(a)-The-Y}(b).
\begin{figure}[h]
\begin{center}~~~~~~~~~~~~~~~~\includegraphics[scale=0.4]{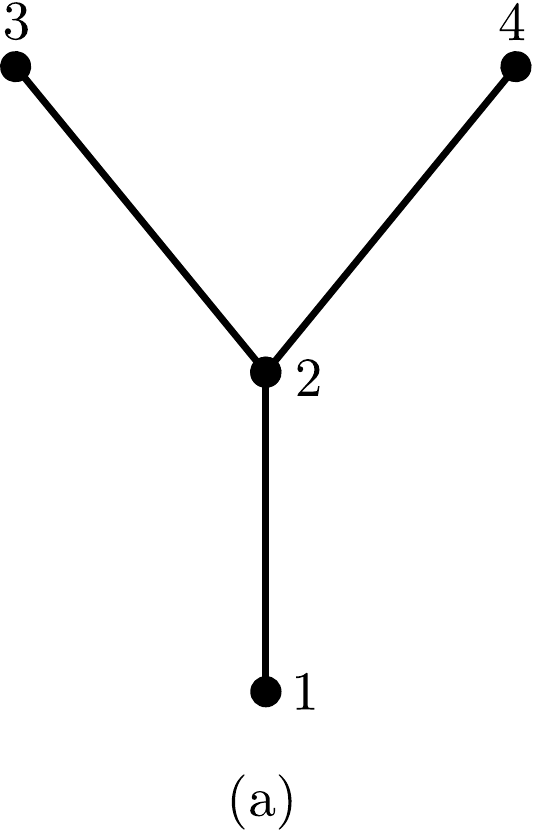}~~~~~~~~~~~~~~~~~~\includegraphics[scale=0.4]{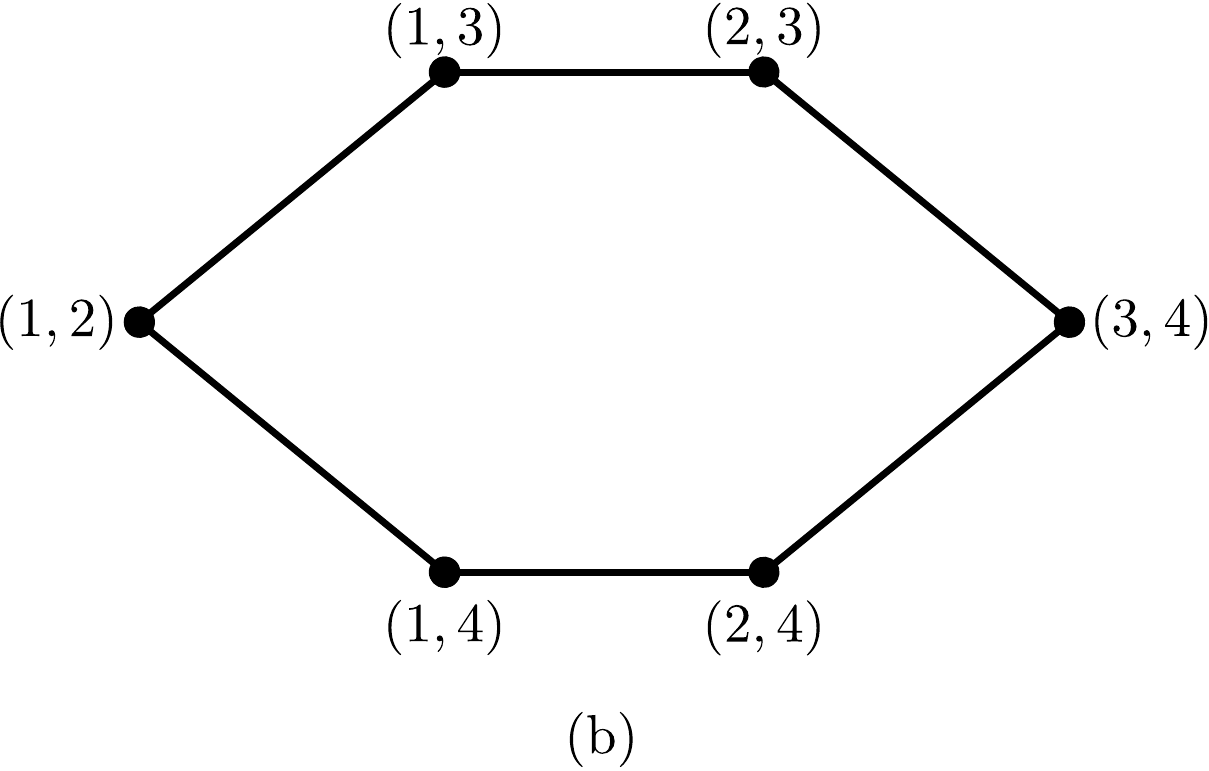}\end{center}
\begin{center}~~~~~~~~~~~~~~\includegraphics[scale=1]{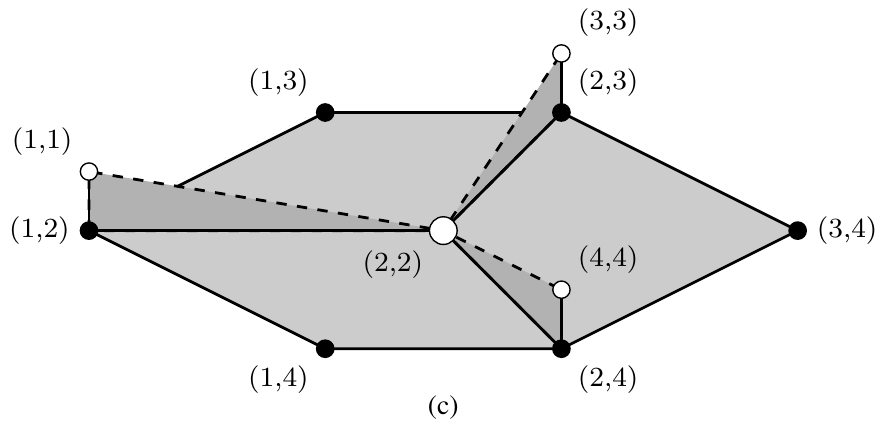}\end{center}






\caption{\label{fig:(a)-The-Y}(a) The Y-graph $\Gamma$. (b) The $2$-particle
combinatorial configuration space $\mathcal{D}^{2}(\Gamma)$. (c) The $2$-particle
configuration space $C_2(\Gamma)$; dashed lines and open vertices denote configurations where the particles are coincident.  Such configurations are excluded from $C_2(\Gamma)$.}
\end{figure}

\item Let $\Gamma$ be a lasso graph shown in figure \ref{fig:The-lasso}(a).
It is a combination of Y and triangle graphs. Its combinatorial configuration
space $\mathcal{D}^{2}(\Gamma)$ is shown in figure 3(b). The shaded
rectangle is a $2$-cell and hence the cycle $(1,3)\rightarrow(2,3)\rightarrow(2,4)\rightarrow(1,4)\rightarrow(1,3)$
is contractible. The cycle $(1,2)\rightarrow(1,4)\rightarrow(1,3)\rightarrow(1,2)$
corresponds to the situation when one particle is sitting at the vertex
$1$ and the other moves along the cycle $c=2\rightarrow4\rightarrow3\rightarrow2$
of $\Gamma$. We will call this cycle an Aharonov-Bohm cycle (AB-cycle) and
denote its phase $\phi_{c,1}^1$ 
(the subscript $c,1$ indicates that  $c$ is traversed by just 1 particle, and the superscript $1$ indicates the position of the stationary particle).  The cycle $(2,3)\rightarrow(3,4)\rightarrow(2,4)\rightarrow(2,3)$
represents the exchange of two particles around $c$. The corresponding phase will be denoted by $\phi_{c,2}$. Finally, for the cycle $(1,2)\rightarrow(1,3)\rightarrow(2,3)\rightarrow(3,4)\rightarrow(2,4)\rightarrow(1,4)\rightarrow(1,2)$, corresponding to exchange of two particles along a Y-graph, the phase is denoted $\phi_Y$. There is no torsion in $H_{1}(\mathcal{D}^{2}(\Gamma))$. Moreover,
\begin{gather}
\phi_{c,2}=\phi_{c,1}^{1}+\phi_{Y}.\label{eq:lasso-relation}
\end{gather}
Thus,  the Y-phase  $\phi_{Y}$ and the AB-phase $\phi_{c,1}^{1}$  determine  $\phi_{c,2}$.
\end{itemize}
\begin{remark}\label{remark1}
Any relation  between cycles on a graph $G$  holds between the corresponding cycles on a graph $F$ containing $G$ as a subgraph or a subgraph homotopic to $G$. It is for this reason that  (\ref{eq:lasso-relation}) will play a key role in relating Y-phases and AB-phases for general graphs.
\end{remark}

\begin{figure}[h]
\begin{center}~~~~~~~~~~~~~~\includegraphics[scale=0.4]{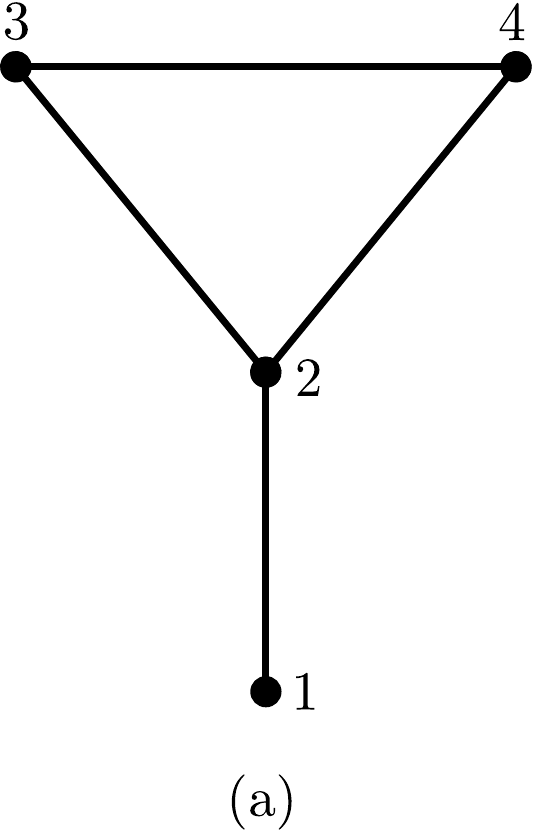}~~~~~~~~~~~~~~~~~~\includegraphics[scale=0.4]{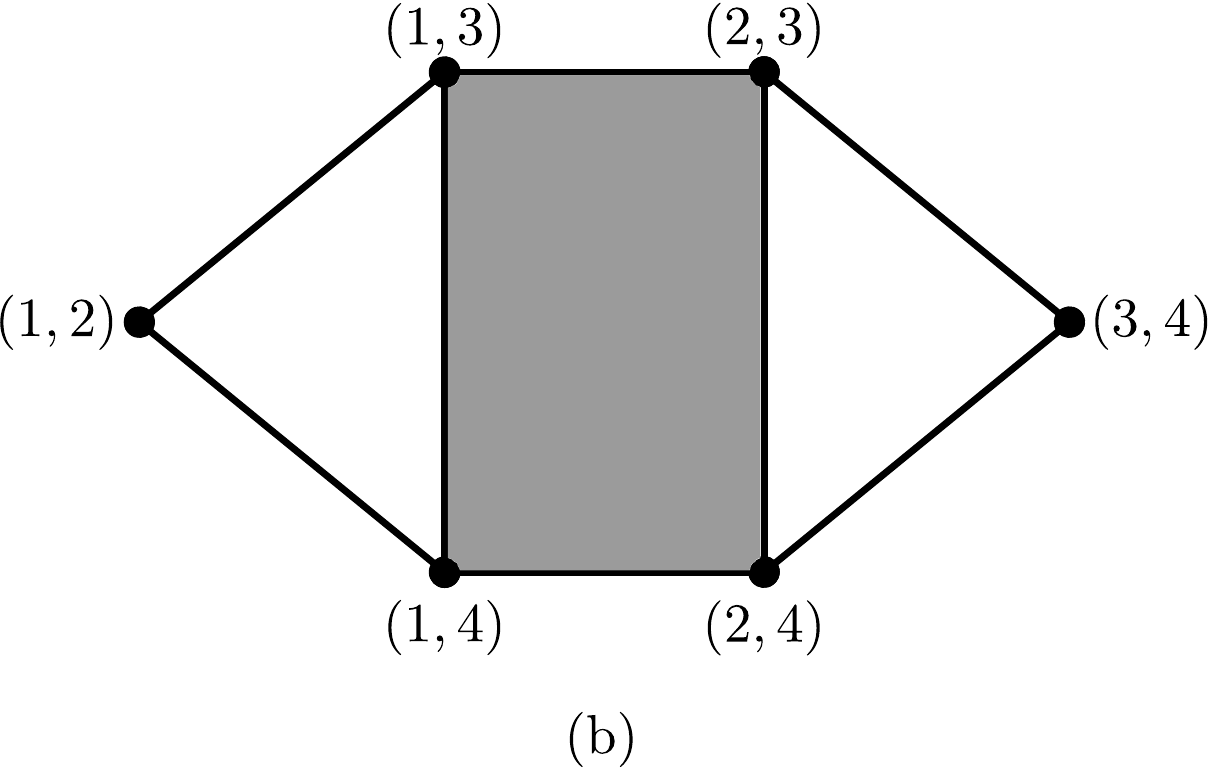}\end{center}

\caption{\label{fig:The-lasso}(a) The lasso graph $\Gamma$ (b) The $2$-particle
configuration space $\mathcal{D}^{2}(\Gamma)$.}
\end{figure}

\subsection{A spanning set of $H_1(\mathcal{D}^{2}(\Gamma))$\label{sub:An-over-complete-basis}}

In order to proceed with the calculation of $H_{1}(\mathcal{D}^{2}(\Gamma))$ for arbitrary $\Gamma$
we need a spanning set of $H_{1}(\mathcal{D}^{2}(\Gamma))$. Before we give
one, let us discuss the dependence of the AB-phase on the position
of the second particle. Suppose there is a cycle $c$ in $\Gamma$ with two vertices $v_{1}$
and $v_{2}$ not on the cycle. We want to know the relation between $\phi_{c,1}^{v_{1}}$ and $\phi_{c,1}^{v_{2}}$.
There are two possibilities to consider. The first is shown in
figure \ref{fig:The-dependence-of}(a) and represents the situation
when there is a path $P_{v_{1},v_{2}}$ which joins $v_{1}$ and $v_{2}$
and is disjoint with $c$. In this case both AB-cycles are homotopy
equivalent as they belong to the cylinder $c\times P_{v_{1},v_{2}}$. Therefore,
\begin{Fact}\label{fact1} Assume there is a cycle $c$ in $\Gamma$ with two vertices $v_{1}$
and $v_{2}$ not on the cycle. Suppose there is a path $P_{v_{1},v_{2}}$ which joins $v_{1}$ and $v_{2}$
and is disjoint with $c$. Then $\phi_{c,1}^{v_{1}}=\phi_{c,1}^{v_{2}}$.
\end{Fact}
Assume now that every path joining $v_1$ and $v_2$ passes through the cycle $c$ (see figure \ref{fig:The-dependence-of}(b)).
Noting that the graph contains two subgraphs homotopic to the lasso which in turn both contain $c$, and making use of Remark \ref{remark1}, we can repeat the argument leading to
relation (\ref{eq:lasso-relation}) for each lasso.  We obtain,
\begin{gather}
\phi_{c,2}=\phi_{c,1}^{v_{1}}+\phi_{Y_{1}},\,\,\,\,\phi_{c,2}=\phi_{c,1}^{v_{2}}+\phi_{Y_{2}},\label{eq:A-B-1}
\end{gather}
and hence
\begin{gather}
\phi_{c,1}^{v_{1}}-\phi_{c,1}^{v_{2}}=\phi_{Y_{2}}-\phi_{Y_{1}}.\label{eq:AB-2}
\end{gather}
Thus, for a fixed one-particle cycle $c$ in $\Gamma$,
the difference between any two AB-phases (corresponding to two different positions of the stationary particle) may be expressed in terms of the
Y-phases.
\begin{figure}[h]
\begin{center}~~~~~~\includegraphics[scale=0.5]{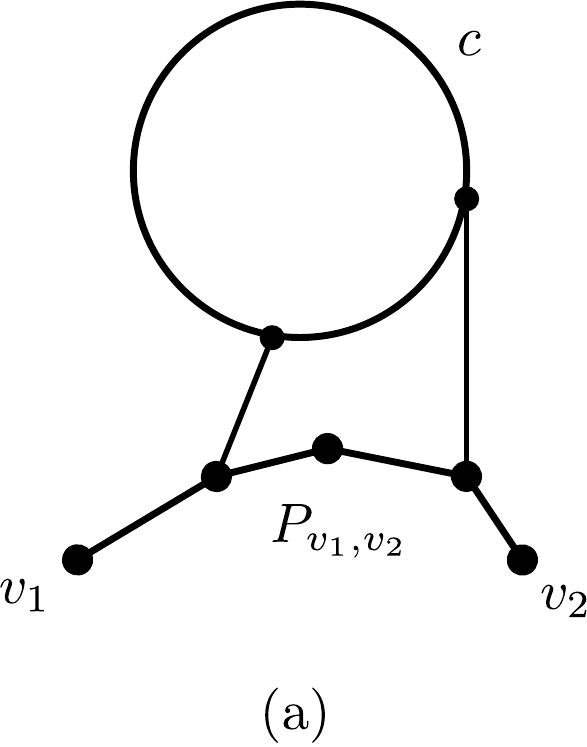}~~~~~~~~~~~~~~~~~~\includegraphics[scale=0.5]{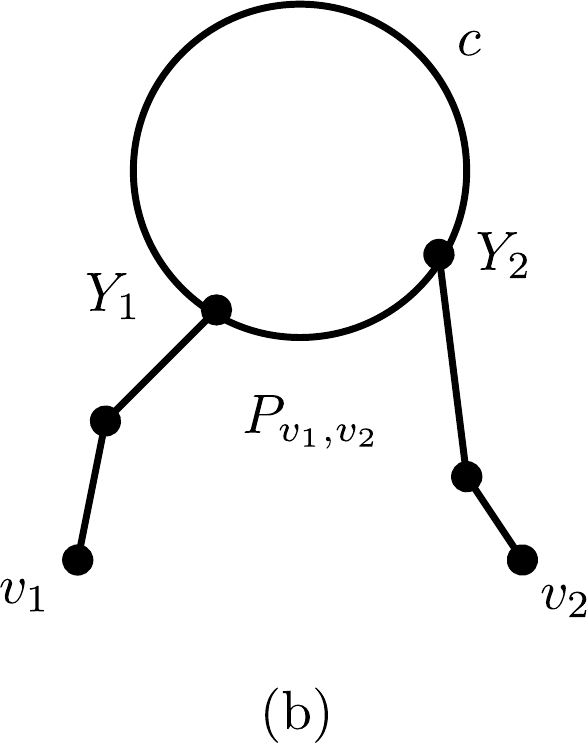}\end{center}

\caption{\label{fig:The-dependence-of}The dependence of the AB-phase for
cycle $c$ on the position of the second particle when (a) there is
a path between $v_{1}$ and $v_{2}$ disjoint with $c$, (b) every
path joining $v_{1}$ and $v_{2}$ passes through $c$.}
\end{figure}

As we show in section \ref{MT}, a spanning set of $H_{1}(\mathcal{D}^{2}(\Gamma))$ is given by all Y and AB-cycles. 
Note that from relations (\ref{eq:lasso-relation}) and (\ref{eq:AB-2}) , we can restrict the set of AB-cycles to belong to a basis for $H_1(\Gamma)$ (since all other AB-cycles can be expressed in terms of these and Y-cycles).  By Euler's formula, the dimension of $H_1(\Gamma)$ is given by the first Betti number,
\begin{gather}
\beta_{1}(\Gamma)=E-V+1,
\end{gather}
As a result, we will use a spanning set (which in general is over-complete) containing the following:
\begin{enumerate}
\item All $2$-particle cycles corresponding to the exchanges on Y subgraphs
of $\Gamma$. There may be relations between these cycles.
\item A set of $\beta_{1}(\Gamma)$ AB-cycles, one for each independent cycle in a basis for $H_1(\Gamma)$.
\end{enumerate}
Thus, $H_{1}(\mathcal{D}^{2}(\Gamma))=\mathbb{Z}^{\beta_1(\Gamma)}\oplus A$, where $A$ is determined by Y-cycles. Consequently, in order to determine $H_{1}(\mathcal{D}^{2}(\Gamma))$ one has to
study the relations between Y-cycles.

\subsection{$3$-connected graphs}

In this section we determine $H_{1}(\mathcal{D}^{2}(\Gamma))$
for $3$-connected graphs. Let $\Gamma$ be a connected graph. We define an $m$-separation of
$\Gamma$ \cite{tutte01}, where $m$ is a positive integer, as an
ordered pair $(\Gamma_{1},\Gamma_{2})$ of subgraphs of $\Gamma$
such that
\begin{enumerate}
\item The union $\Gamma_{1}\cup\Gamma_{2}=\Gamma$.
\item $\Gamma_{1}$ and $\Gamma_{2}$ are edge-disjoint and have exactly
$m$ common vertices, $V_{m}=\{v_{1},\ldots,v_{m}\}$.
\item $\Gamma_{1}$ and $\Gamma_{2}$ have each a vertex not belonging to
the other.
\end{enumerate}
It is customary to say that the $V_{m}$ separates vertices of $\Gamma_{1}$
and $\Gamma_{2}$ different from $V_{m}$.
\begin{definition}
\label{n-connected-def}A connected graph $\Gamma$ is $n$-connected
iff it has no $m$-separation for any $m<n$.
\end{definition}
The following theorem of Menger \cite{tutte01} gives
an additional insight into graph connectivity:
\begin{theorem}
\label{Menger} For an $n$-connected graph $\Gamma$ there are at
least $n$ internally disjoint paths between any pair of vertices.
\end{theorem}
The basic example of $3$-connected graphs are wheel graphs. A wheel
graph $W^{n}$ of order $n$ consists of a cycle with $n$ vertices
and a single additional vertex which is connected to each vertex of
the cycle by an edge. Following Tutte \cite{tutte01} we denote the
middle vertex by $h$ and call it the hub, and the cycle that does not include $h$ by $R$ and call
it the rim. The edges connecting the hub to the rim will be called
spokes. The importance of wheels in the theory of $3$-connected graphs
follows from the following theorem:
\begin{theorem}
\label{thm-Wheel-thorem}(Wheel theorem \cite{tutte01}) Let $\Gamma$ be a simple
$3$-connected graph different from a wheel. Then for some edge $e\in E(\Gamma)$,
either $\Gamma\setminus e$ or $\Gamma/e$ is simple and
$3$-connected.
\end{theorem}
Here $\Gamma\setminus e$ is constructed from $\Gamma$ by removing the edge $e$, and $\Gamma/e$ is obtained by contracting edge $e$ and identifying its vertices. These two operations will be called edge removal and edge contraction. The inverses will be called edge addition and vertex expansion. Note that vertex expansion requires specifying which edges are connected to which vertices after expansion. As we deal with $3$-connected graphs we will apply the vertex expansion only to vertices of degree at least four and split the edges between new vertices in a such way that they are at least $3$-valent.

As a direct corollary of Theorem \ref{thm-Wheel-thorem} any simple $3$-connected graph can be constructed in a finite number of steps
starting from a wheel graph $W^{k}$, for some $k$; that is, there exists a sequence of simple $3$-connected graphs
\begin{gather*}
W_{k}=\Gamma_{0}\mapsto\Gamma_{1}\mapsto\ldots\mapsto\Gamma_{n-1}\mapsto\Gamma_{n}=\Gamma,
\end{gather*}
where $\Gamma_{i}$ is constructed from $\Gamma_{i-1}$ by either
\begin{enumerate}
\item adding an edge between non-adjacent vertices, or
\item expanding at a vertex of valency at least four.
\end{enumerate}
Therefore, in order
to prove inductively some property of a $3$-connected graph, it is
enough to show that the property holds for an arbitrary wheel graph and that it persists under operations 1.~and 2.~above.
\begin{lemma}
For wheel graphs $W^{n}$ all phases $\phi_{Y}$ are equal up to a
sign.
\end{lemma}
\begin{proof}
The Y subgraphs of $W^{n}$ can be divided into two groups: (i) the center vertex of Y is on the rim, and
(ii) the center vertex of Y is the hub. For (i) let $v_{1}$ and $v_{2}$ be two adjacent vertices belonging to the
rim, $R$. Let $Y_{v_{1}}$ and $Y_{v_{2}}$ be the corresponding
Y-graphs whose central vertices are $v_{1}$ and $v_{2}$ respectively.
Evidently, the two edges of $Y_{v_{1}}$ and $Y_{v_{2}}$  which are spokes belong
to the same triangle cycle, $C$, i.e the cycle with vertices $v_{1}$,
$v_{2}$ and $h$ (see figure \ref{fig:wheel}(a)). Moreover, $b_{1}$
is connected to $b_{2}$ by a path which is disjoint with $\mbox{\ensuremath{C}}$. Using Fact \ref{fact1}, we have that
$\phi^{b_1}_{c,1} =  \phi^{b_2}_{c,1}$.  From this and  relation (\ref{eq:AB-2}), it follows that
 $\phi_{Y_{v_{1}}}=\phi_{Y_{v_{2}}}$. Repeating this reasoning
we obtain that all $\phi_{Y_{v_{i}}}$, with $v_{i}$ belonging to
the rim are equal (perhaps up to a sign). We are left with the Y-graphs whose central vertex is the hub. Similarly (see figure \ref{fig:wheel}(b))
we take a cycle, $C$, with two edges belonging to the chosen Y.
Then there is always a Y-graph with two edges belonging to $C$ and center on the rim.
Therefore, by Fact \ref{fact1} and  relation (\ref{eq:AB-2}) the phase on a
Y subgraph whose center vertex is the hub is the same as on the Y subgraphs
whose center vertex is on the rim.
\end{proof}
\begin{figure}[h]
\begin{center}~~~~~~~~~~\includegraphics[scale=0.45]{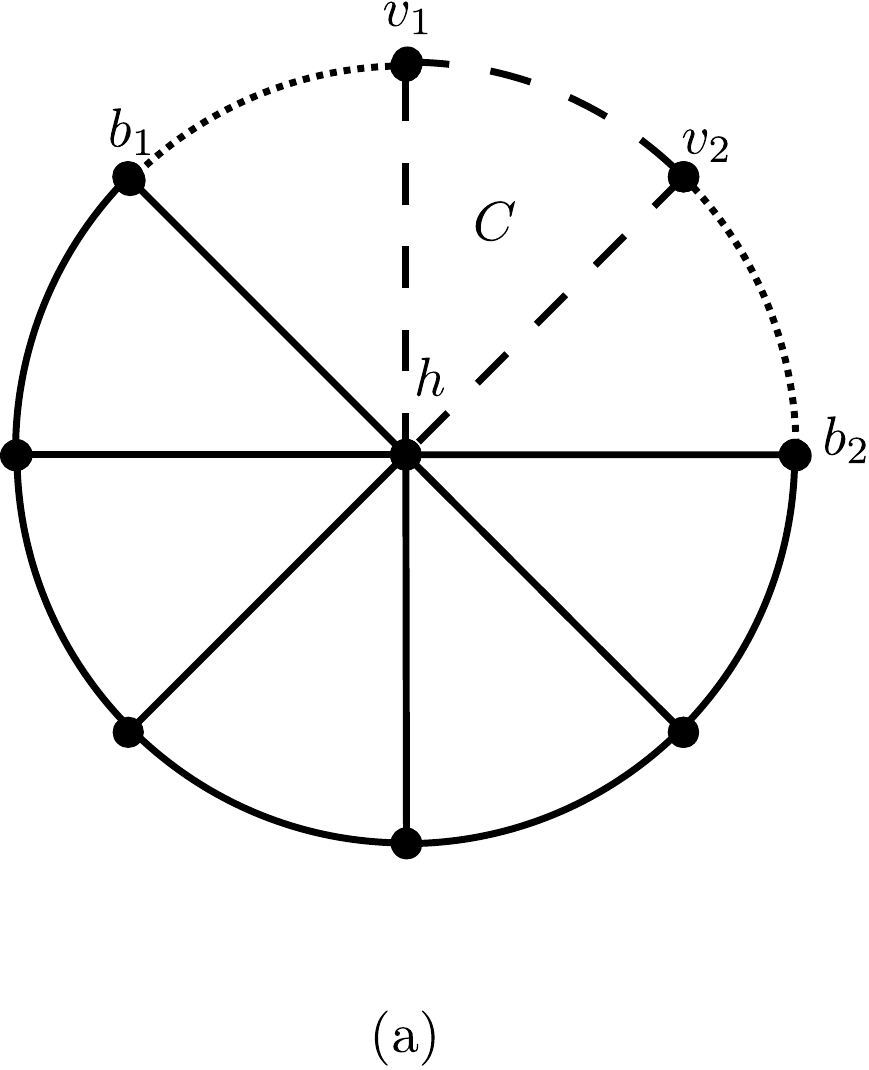}~~~~~~~~~~~~~~~\includegraphics[scale=0.45]{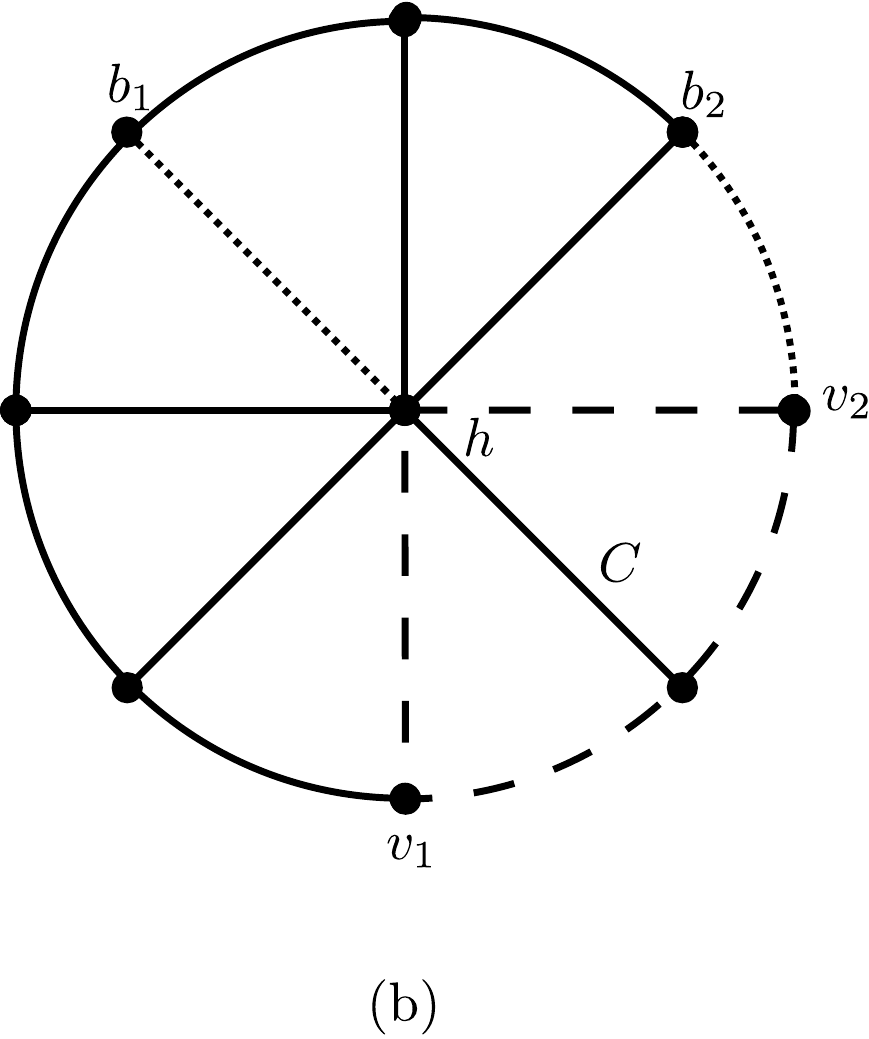}\end{center}

\caption{\label{fig:wheel}Wheel graphs.  (a) Dashed lines denote a pair of Y subgraphs $Y_{v_1}$ and $Y_{v_2}$ centered at adjacent vertices $v_1$ and $v_2$ on the rim.   The three shared edges of the Y subgraphs (long dashes) form a cycle $C$. (b) The Y subgraph $Y_h$ (edges are dashed) has three outer vertices $b_1$, $v_1$ and $v_2$.  Two of the edges of $Y_h$ together with a path on the rim joining $v_1$ and $v_2$ form a cycle $C$ (long dashes).
A second $Y$-graph $Y_{v_2}$ (edges are dashed) shares two edges of $C$. }
\end{figure}

\begin{lemma}
For $3$-connected simple graphs all phases $\phi_{Y}$ are equal
up to a sign.\end{lemma}
\begin{proof}
We prove by induction. By Lemma 1 the statement is true for all
wheel graphs.\\

\noindent 1.~Adding an edge: Assume that $v_{1}$ and $v_{2}$ are non-adjacent
vertices of the $3$-connected graph $\Gamma$.  Suppose that the relations on $\Gamma$ determine that all its $\phi_{Y}$
phases are equal (up to a sign).  These relations remain if we add an edge $e$ between the vertices
$v_{1}$ and $v_{2}$.  Therefore, on $\Gamma\cup e$, the phases $\phi_Y$ belonging to $\Gamma$ must still be equal.


However, the graph $\Gamma\cup e$ contains
new Y-graphs, whose central vertices are $v_{1}$ or $v_{2}$ and
one of the edges is $e$. We need to show that the phase $\phi_{Y}$
on these new Y's is the same as on the old ones. Let $\{e,f_{1},f_{2}\}$
be such a Y-graph (see figure \ref{fig:Adding-an-edge}(a)). Let $\alpha_1$ and $\alpha_2$ be endpoints of $f_1$ and $f_2$. By $3$-connectedness, there is a path between $\alpha_1$ and $\alpha_2$ which does not contain $v_1$ or $v_2$.  In this way we obtain a cycle $C$, as shown in figure \ref{fig:Adding-an-edge}(a). Again by $3$-connectedness, there is a path $P$ from $v_2$ to a vertex $\beta$ in $C$
 which does not contain $\alpha_1$ and $\alpha_2$. Let $Y'$ be the Y-graph with $\beta$ as its center and edges along $C$ and $P$, as shown in figure \ref{fig:Adding-an-edge}(a). Then $Y'$ belongs to $\Gamma$. Applying Fact \ref{fact1} and  relation (\ref{eq:AB-2}) (cf.~the proof of Lemma 1) to the cycle $C$ and
the two Y-graphs discussed, the result follows.\\


\noindent 2.~Vertex expansion: Let $\Gamma$ be a $3$-connected simple graph and
let $v$ be a vertex of degree at least four. Let $\tilde{\Gamma}$
be a graph derived from $\Gamma$ by expanding at the vertex $v$,
and assume that the new vertices, $v_{1}$ and $v_{2}$, are at least
$3$-valent. These assumptions are necessary for $\tilde{\Gamma}$
to be $3$-connected \cite{tutte01}. Note that $\Gamma$ and $\tilde{\Gamma}$ have
the same number of independent cycles. Moreover, by splitting at the
vertex $v$ we do not change the relations between the $\phi_{Y}$
phases of $\Gamma$. This is simply because if the equality of some
of the $\phi_{Y}$ phases required a cycle passing through $v$, one can now use the cycle with one more edge passing through $v_{1}$ and $v_{2}$
in $\tilde{\Gamma}$. The graph $\tilde{\Gamma}$ contains new
Y-graphs, whose central vertices are $v_{1}$ or $v_{2}$ and
one of the edges is $e=v_{1}\leftrightarrow v_{2}$. We need to show
that the phase $\phi_{Y}$ on these new Ys is the same as on the old
ones. Let $\{e,f_{1},f_{2}\}$ be such a graph and let $\alpha_1$ and $\alpha_2$ be endpoints of $f_1$ and $f_2$. By $3$-connectedness, there is a path between $\alpha_1$ and $\alpha_2$ which does not contain $v_1$ or $v_2$.  In this way we obtain a cycle $C$, as shown in figure \ref{fig:Adding-an-edge}(b).  Again by $3$-connectedness, there is a path $P$ from $v_2$ to a vertex $\beta$ in $C$
which does not contain $\alpha_1$ and $\alpha_2$.  Let $Y'$ be the Y-graph with $\beta$ as its center and edges along $C$ and $P$, as shown in figure \ref{fig:Adding-an-edge}(b).  Then $Y'$ belongs to $\Gamma$. Applying Fact \ref{fact1} and  relation (\ref{eq:AB-2})
to the cycle $C$ and the two Y-graphs discussed, the result follows.
\end{proof}

\begin{figure}[h]
\begin{center}~~~~~\includegraphics[scale=0.4]{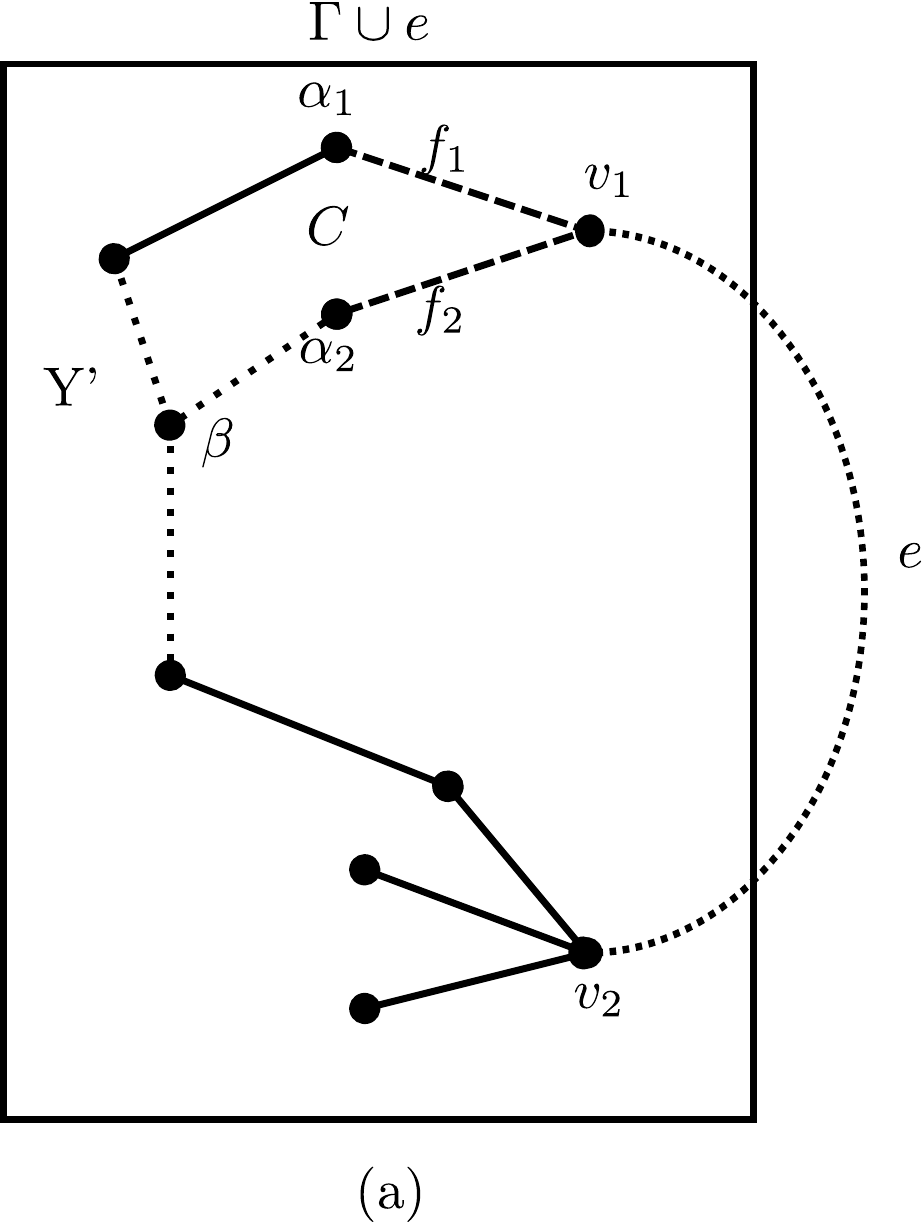}~~~~~~~~~~~~~~~~~~\includegraphics[scale=0.4]{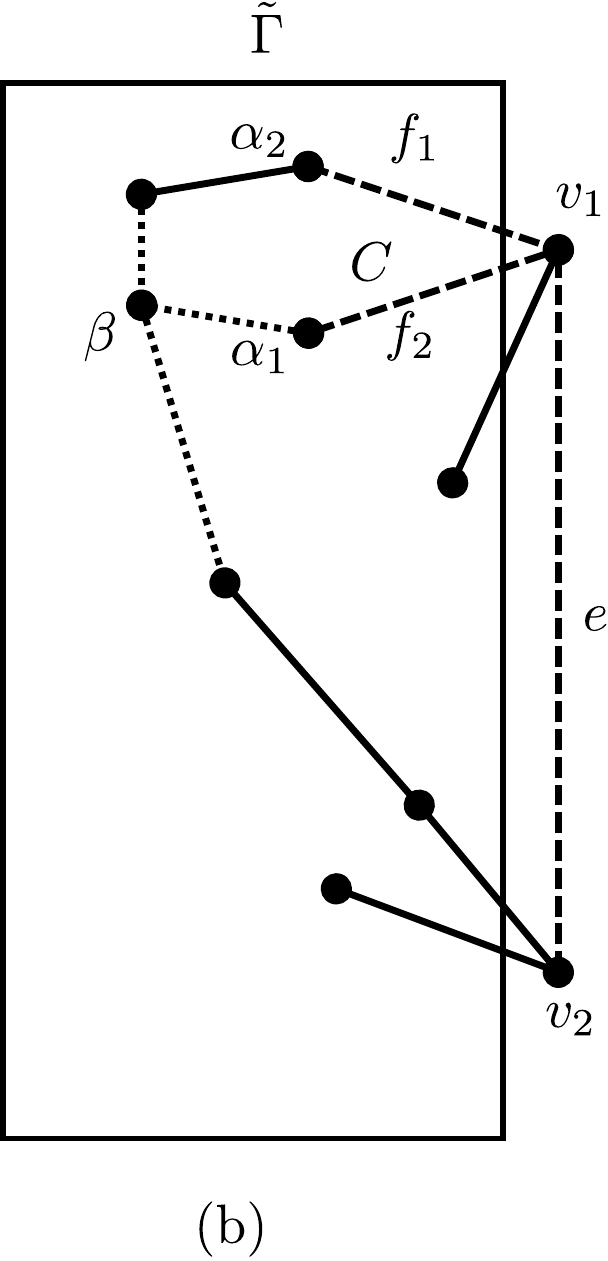}\end{center}

\caption{\label{fig:Adding-an-edge}(a) Adding an edge (b) Expanding at the vertex.}
\end{figure}

\begin{theorem}\label{3-final}
For a $3$-connected simple graph, $H_{1}(\mathcal{D}^{2}(\Gamma))=\mathbb{Z}^{\beta_1(\Gamma)}\oplus A$,
where $A=\mathbb{Z}_{2}$ for non-planar graphs and $A=\mathbb{Z}$
for planar graphs. \end{theorem}
\begin{proof}
By Lemmas 1 and 2 we only need to determine the phase $\phi_{Y}$. Using the construction in  \cite{JHJKJR},
it can be shown by elementary calculations that for the graphs $K_{5}$ and $K_{3,3}$,
$H_1(\mathcal{D}^2(\Gamma))=\mathbb{Z}^{\beta_1(\Gamma)}\oplus \mathbb{Z}_2$ (shorter calculations using discrete Morse theory are given in \cite{KP11}).
Therefore the phase $\phi_{Y}=0$ or $\pi$. By Kuratowski's
theorem \cite{Kuratowski30} every non-planar graph contains a subgraph
which is isomorphic to $K_{5}$ or $K_{3,3}$. This proves the statement
for non-planar graphs.

If $\Gamma$ is planar, then any phase $\phi_Y$ can be realised.  This can be demonstrated explicitly by appealing to the well-known anyon gauge potential for two particles in the plane, 
\[ {\bf A}({\bf r}) =   \frac{\alpha}{2\pi} {\bf \hat{z}}\times \frac{\bf r} {|{\bf r}|^2}. \]
The line integral of the one-form
\[  \omega = {\bf A}({\bf r_2} - {\bf r_1})\cdot {\bf d r_1} +  {\bf A}({\bf r_1} - {\bf r_2})\cdot {\bf d r_2} \]
around a primitive cycle in which the two particles are exchanged yields the anyon phase $\alpha$.
If $\Gamma$ is drawn in the plane and each edge of $\mathcal{D}^2(\Gamma)$ is assigned the phase given by the line integral of $\omega$, then the phase associated with exchanging the particles on a $Y$-subgraph is given by $\alpha$. 

%

\end{proof}

For a given cycle on a 3-connected graph, it follows from Theorem \ref{3-final} and  relation (\ref{eq:AB-2}) that the difference between AB-phases (corresponding to different positions of the stationary particle) is either $0$ or $2\phi_Y$.  If the graph is nonplanar, we have that $2\phi_Y = 0 \mod 2\pi$, so that
the AB-phases 
are independent of the position of stationary particle.

\subsection{$2$-connected graphs}\label{subsec: two-connected}

In this subsection we discuss $2$-connected graphs. First, by
considering a simple example we show that in contrast to $3$-connected
graphs it is possible to have more than one $\phi_{Y}$ phase. Using
a decomposition procedure of a $2$-connected graph into $3$-connected
graphs and topological cycles we provide the formula for $H_{1}(\mathcal{D}^{2}(\Gamma))$.

\begin{figure}[h]
\begin{center}\includegraphics[scale=0.5]{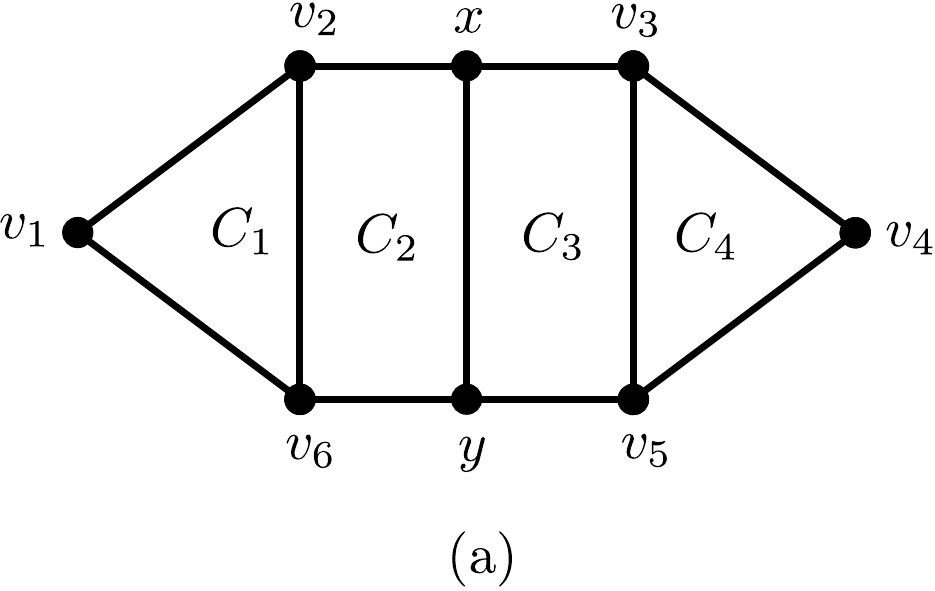}~~~\includegraphics[scale=0.5]{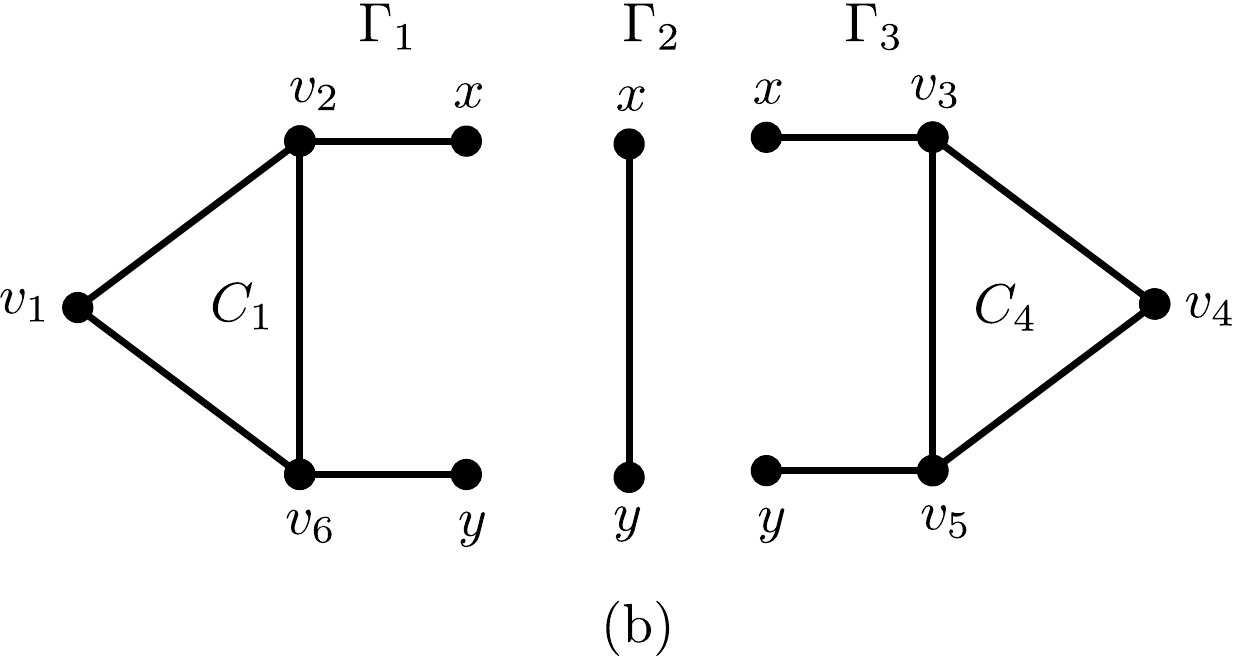}\end{center}

\begin{center}~~~~~~~~~~~~~~\includegraphics[scale=0.5]{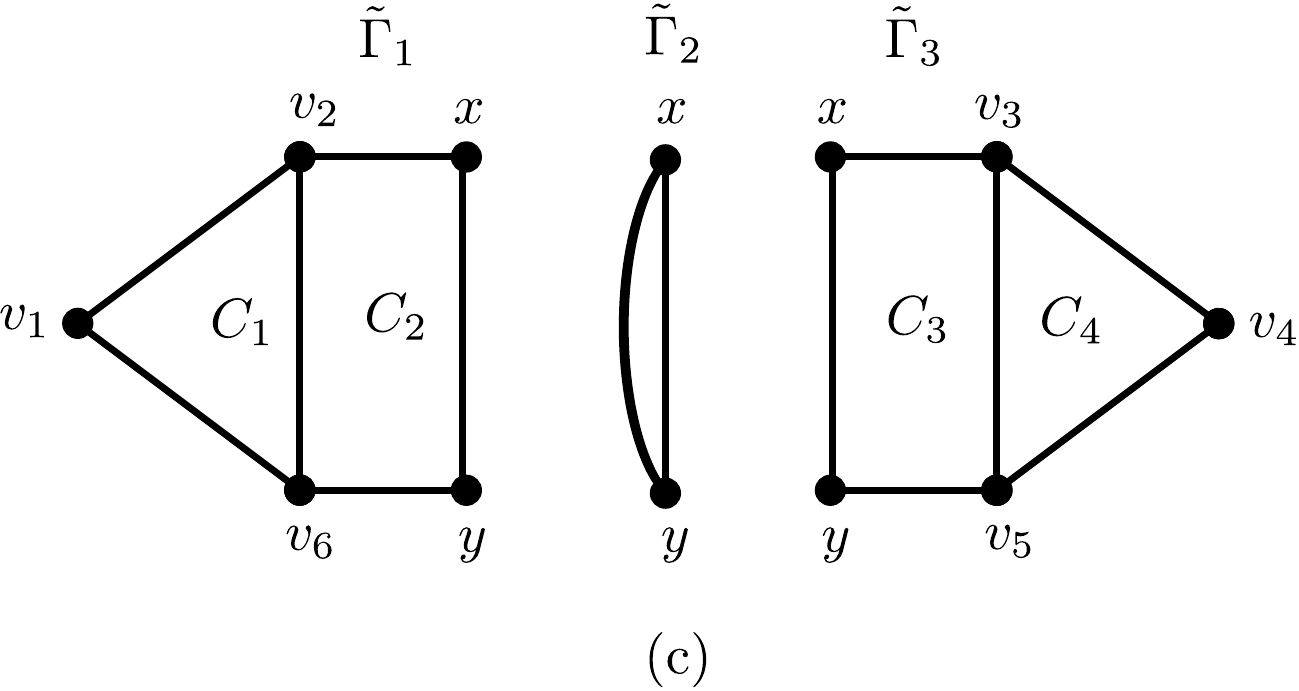}\end{center}

\caption{\label{fig:2-conn-ex}(a) An example of a $2$-connected graph, (b) the
components of the $2$-cut $\{x,y\}$, (c) the marked components. }
\end{figure}

\begin{example}
Let us consider graph $\Gamma$ shown in figure \ref{fig:2-conn-ex}(a).
Since vertices $v_{1}$ and $v_{4}$ are $2$-valent, $\Gamma$ is
not $3$-connected. It is however $2$-connected. Note that $\beta_1(\Gamma)=4$
and that there are six Y-graphs, with central vertices $v_{2}$, $v_{3}$,
$v_{5}$, $v_{6}$, $x$ and $y$ respectively. Using Fact \ref{fact1} and  relation (\ref{eq:AB-2})
we verify that
\begin{gather}
\phi_{Y_{v_{2}}}=\phi_{Y_{v_{6}}},\,\,\phi_{Y_{v_{3}}}=\phi_{Y_{v_{5}}},\,\,\phi_{Y_{x}}=\phi_{Y_{y}}.\label{eq:example4}
\end{gather}
One can also show that the phases $\phi_{Y_{v_{2}}}$, $\phi_{Y_{v_{3}}}$ and $\phi_{Y_{x}}$ are independent.

(For completeness, we give an explicit argument,  showing that each one of the phases $\phi_{Y_{v_{2}}}, \phi_{Y_{v_{3}}}, \phi_{Y_{x}}$ can be made to be nonzero while the other two are made to be zero. Following the procedure of  \cite{JHJKJR}, we can assign an arbitrary phase $\alpha$ to the edge
$(v_4,v_5)\leftrightarrow (v_3,v_4)$ of $\mathcal{D}^2(\Gamma)$, and zero phase to all its other edges.  This is because $(v_4,v_5)\leftrightarrow (v_3,v_4)$
does not  belong to a contractible square in  $\mathcal{D}^2(\Gamma)$ (no edge of
$\Gamma$ disjoint from $v_3\leftrightarrow v_5$ has $v_4$ as a vertex).  Since $(v_4,v_5)\leftrightarrow (v_3,v_4)$ uses the edge $v_3 \leftrightarrow v_5$ in $\Gamma$, which belongs to $Y_{v_3}$ but not to $Y_x$ or $Y_{v_2}$, the phase $\phi_{Y_{v_{3}}}$ associated with particle exchange on $Y_{v_3}$ is given by $\alpha$ (up to a sign) while $\phi_{Y_{v_{2}}} = \phi_{Y_{y}} = 0$.  A similar argument, based on the fact that the edge  $(v_1,v_2)\leftrightarrow (v_1,v_6)$ also does not belong to a contractible square in $\mathcal{D}^2(\Gamma)$, leads to an assignment of phases with  $\phi_{Y_{v_{2}}}$ arbitrary, $\phi_{Y_{v_{3}}} = \phi_{Y_{y}} = 0$. Finally,  one can assign edge phases in $\mathcal{D}^2(\Gamma)$ so that $\phi_{Y_{y}}$ is arbitrary.  Adjusting the phases of the edges $(v_4,v_5)\leftrightarrow (v_3,v_4)$
and $(v_1,v_2)\leftrightarrow (v_1,v_6)$ so that $\phi_{Y_{v_{2}}} = \phi_{Y_{v_{3}}}  = 0$ (which doesn't affect $\phi_{Y_x}$), we obtain an assignment of phases with  $\phi_{Y_x}$ arbitrary and $\phi_{Y_{v_{3}}} = \phi_{Y_{v_{2}}} = 0$.  Thus,
$\phi_{Y_{v_{2}}}$, $\phi_{Y_{v_{3}}}$ and $\phi_{Y_x}$ are linearly independent.)

%
Therefore we have three independent $\phi_{Y}$ phases and four AB-phases, and so
\begin{gather}
H_{1}(\mathcal{D}^{2}(\Gamma))=\mathbb{Z}^{7}.
\end{gather}
Vertices $\{x,y\}$ constitute a $2$-vertex cut of $\Gamma$, i.e.
after their deletion $\Gamma$ splits into three connected components
$\Gamma_{1}$, $\Gamma_{2}$, $\Gamma_{3}$ (see figure \ref{fig:2-conn-ex}(b)).
They are no longer $2$-connected. Moreover, for example, the two Y-subgraphs $Y_{v_{2}}$ and $Y_{v_{6}}$ for which $\phi_{Y_{v_{2}}}=\phi_{Y_{v_{6}}}$
in $\Gamma$ no longer satisfy this condition in $\Gamma_{1}$, i.e. $\phi_{Y_{v_{2}}}\neq\phi_{Y_{v_{6}}}$ in $\Gamma_1$. This
is because the AB-phases $\phi_{C_{1},1}^{x}$ and $\phi_{C_{1},1}^{y}$
are not necessarily equal.
(This can be readily seen by constructing the two-particle configuration space  ${\cal D}^2(\Gamma_1)$, an extension of the lasso in Figure~\ref{fig:The-lasso}(b), and recognising that the corresponding AB cycles are independent.)

To make components $\Gamma_i$ $2$-connected and at the same time keep the correct relations between the $\phi_{Y_{v_{i}}}$'s, it is enough to add to each
component $\Gamma_{i}$ an additional edge between vertices $x$ and $y$ (see figure \ref{fig:2-conn-ex}(c)).
The resulting graphs, which we call the marked components and denote by $\tilde{\Gamma}_{i}$ \cite{KP11}, are $2$-connected.  Moreover,
 the relations between the Y-graphs in each $\tilde{\Gamma}_{i}$
are the same as in $\Gamma$. The union of the three marked components
has, however, $\beta_1(\Gamma)+1$ independent cycles. On the other
hand, by splitting $\Gamma$ into marked components, the Y-cycles $Y_{x}$
and $Y_{y}$ have been lost. Since $\phi_{Y_{x}}=\phi_{Y_{y}}$ we
have lost one $\phi_{Y}$ phase. Summing up we can write $H_{1}(\mathcal{D}^{2}(\Gamma))\oplus\mathbb{Z}=\left [ \bigoplus_{i=1}^{3}H_{1}(\mathcal{D}^{2}(\tilde{\Gamma}_{i}))
\right ]\oplus\mathbb{Z}$.
\end{example}

\paragraph{2-vertex cut for an arbitrary $2$-connected graph $\Gamma$}

In figure \ref{fig:2-vertex-cut}(a) a more general $2$-vertex
cut is shown together with components $\Gamma_{i}$ red(note that $\Gamma_i$ consists of an interior $\gamma_i$, the edges connecting $\gamma_i$ to vertices $x$ and $y$,  and  $x$ and $y$ themselves). It is easy to
see that the marked components $\tilde{\Gamma}_{i}$ are $2$-connected
and the relations between the $\phi_{Y}$ phases in each $\tilde{\Gamma}_{i}$
are the same as in $\Gamma$.
Let $\mu(x,y)$ be the number of $\tilde{\Gamma}_{i}$ components into which $\Gamma$ splits after removal of vertices $x$ and $y$. By Euler's formula the union $\{\tilde{\Gamma}_{i}\}_{i=1}^{\mu(x,y)}$ of $\mu(x,y)$ marked components has
\begin{gather}
\beta=\# \mathrm{edges}-\# \mathrm{vertices}+\mu(x,y)\nonumber\\=E(\Gamma)+\mu(x,y)-\left(V(\Gamma)+2(\mu(x,y)-1)\right)+\mu(x,y)\nonumber \\
=E(\Gamma)-V(\Gamma)+2=\beta_1(\Gamma)+1,\label{eq:cycle-number-2-conn}
\end{gather}
independent cycles. By splitting $\Gamma$ into the marked components
we possibly lose $\phi_{Y}$ phases corresponding to the Y-graphs
with the central vertex $x$ or $y$. However
\begin{enumerate}
\item If three edges of a Y-graph are connected to the same component we do
not lose $\phi_{Y}$.\\

\item
  If two edges of a Y-graph are connected to the same component, we do not lose $\phi_Y$. The argument is as follows, referring to Figure~\ref{fig:2-vertex-cut}(b):  Let $Y_x$ denote a Y-graph centered at $x$ with vertices $u$ and $v$ in the interior $\gamma_2$ of the component $\Gamma_2$.  Since $\gamma_2$ is 1-connected, there is a path $P$  in $\gamma_2$ from $u$ to $v$ (short dashes in Figure~\ref{fig:2-vertex-cut}(b)). Together with the edges from $x$ to $u$ and $v$, $P$ forms a cycle $C$ in $\Gamma_2$ containing two edges of $Y_x$.  In addition, there is a path $Q$ in $\Gamma_2$ from $u$ to $y$.  Let $w$ denote the last vertex on $Q$ which belongs to $C$ ($w$ might coincide with $u$ or $v$, but need not).   Let $Y_w$ denote the $Y$-graph centred at $w$ with two edges along $C$ and one edge along $Q$.  Then $Y_w$ is contained in $\Gamma_2$, and by relation \eqref{eq:AB-2}, $\phi_{Y_x} = \phi_{Y_{w}}$.
 Therefore, $\phi_{Y_x} $ is not lost under splitting.
 %
\end{enumerate}
Hence the $\phi_{Y}$ phases we lose correspond to the Y-graphs for which
each edge is connected to a different component. First we want to show that any two Y-graphs with the central vertex $x$ (or $y$) whose edges are connected to three fixed components have the same phase. It is enough to show this for Y-graphs which share the same center and two edges. Let us consider two such Y-graphs (see figure \ref{fig:2-vertex-cut}(c) -- the dashed edges are common to both Y-graphs; the distinct edges are  dotted and dotted-dashed). Let  $a_1$, $a_2$ and $b_1$, $b_2$ be the endpoints of the two shared edges, and $\alpha_{1}$, $\alpha_{2}$  the endpoints of the two distinct edges. As the $\gamma_i$'s are connected, there are paths $P_{a_{1},a_{2}}$, $P_{b_{1},b_{2}}$ and $P_{\alpha_{1},\alpha_{2}}$ in $\gamma_1$, $\gamma_3$ and $\gamma_2$ respectively.
Therefore, we can apply Fact \ref{fact1} and  relation (\ref{eq:AB-2})
to the cycle $x\rightarrow a_{1}\cup P_{a_{1},a_{2}}\cup a_{2}\rightarrow y\rightarrow b_{2}\cup P_{b_{1},b_{2}}\cup b_{1}\rightarrow x$ and the two considered Y-graphs
to conclude that their $\phi_{Y}$ phases are the same. Therefore, for each choice of three distinct components, there is just one $\phi_{Y}$ phase. Moreover,
 for a given choice of distinct components, the phase for the Y-graph
with  central vertex $x$ is the same as for the Y-graph with
central vertex $y$ (see figure \ref{fig:2-vertex-cut}(d) where the considered Y-graphs are denoted by dashed and dotted lines). This is once again due to Fact \ref{fact1} and  relation (\ref{eq:AB-2}) applied to the cycle $x\rightarrow a_{1}\cup P_{a_{1},a_{2}}\cup a_{2}\rightarrow y\rightarrow \alpha_{2}\cup P_{\alpha_{1},\alpha_{2}}\cup \alpha_{1}\rightarrow x$ and the two considered Y-graphs.

Summing up, the number of phases we lose when splitting $\Gamma$ into $\mu(x,y)$ marked components, $N_{2}(x,y)$, is equal to the number
of independent Y-graphs in the star graph with $\mu(x,y)$ edges. This can be calculated (see for example \cite{JHJKJR}) to be $N_{2}(x,y)=\frac{1}{2}\left(\mu(x,y)-2\right)\left(\mu(x,y)-1\right)$.
Hence
\begin{gather}
H_{1}(\mathcal{D}^{2}(\Gamma))= \left [ \bigoplus_{i=1}^{\mu(x,y)}H_{1}(\mathcal{D}^{2}(\tilde{\Gamma}_{i}))\right]\oplus\mathbb{Z}^{N_{2}(x,y)-1}.\label{eq:2-connected}
\end{gather}
Note that the $-1$ in the exponent here is to get rid of the additional AB-phase stemming from the calculation (\ref{eq:cycle-number-2-conn}).
Also, it is straightforward to see that although introducing an additional edge to a marked component may give rise to a new $Y$-graph, the associated $Y$-phase is not new,  and is equal to a $Y$-phase of $Y$-graph inside the component.
\begin{figure}[h]
\begin{center}\includegraphics[scale=0.35]{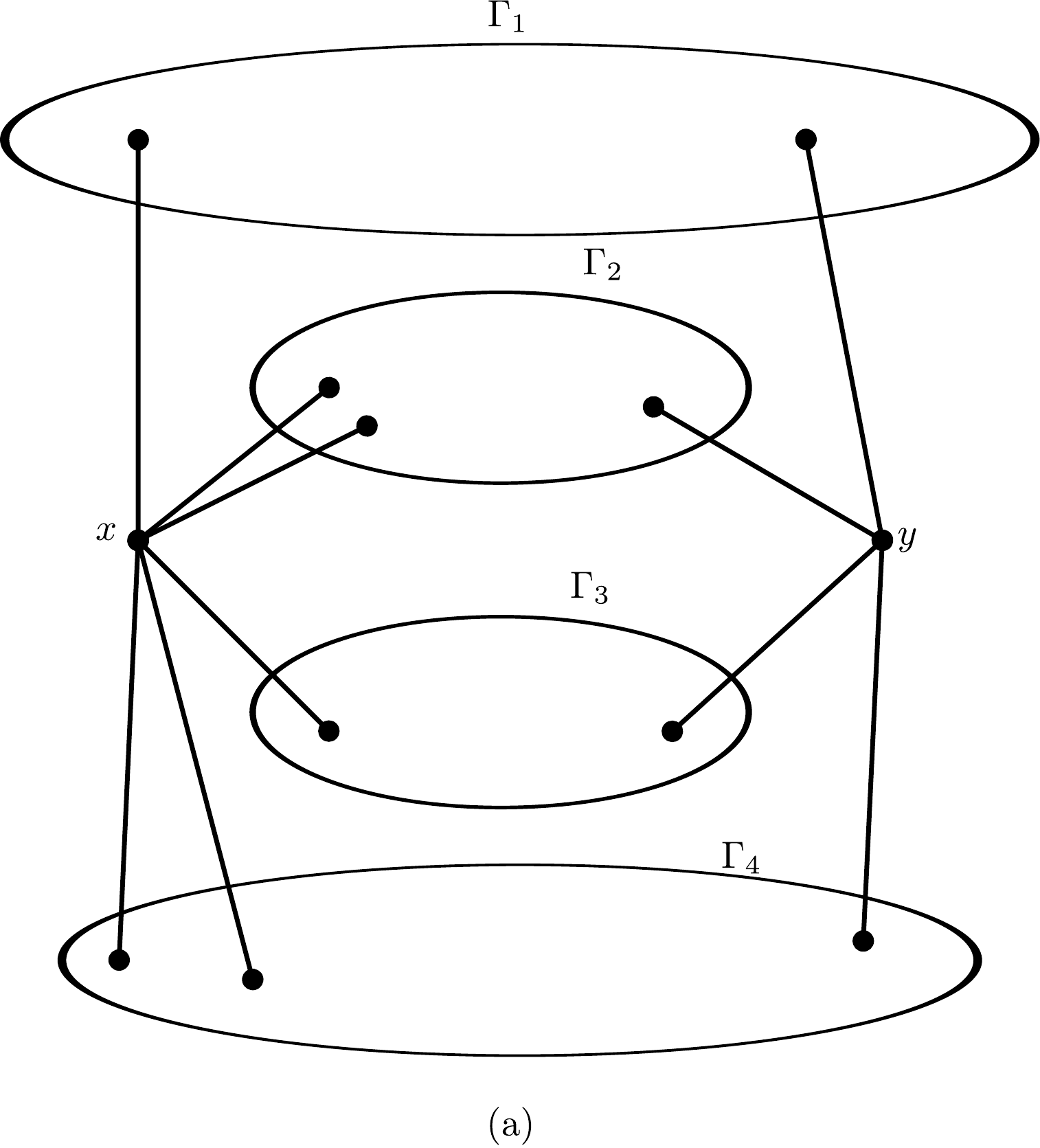}~~~ \includegraphics[scale=0.35]{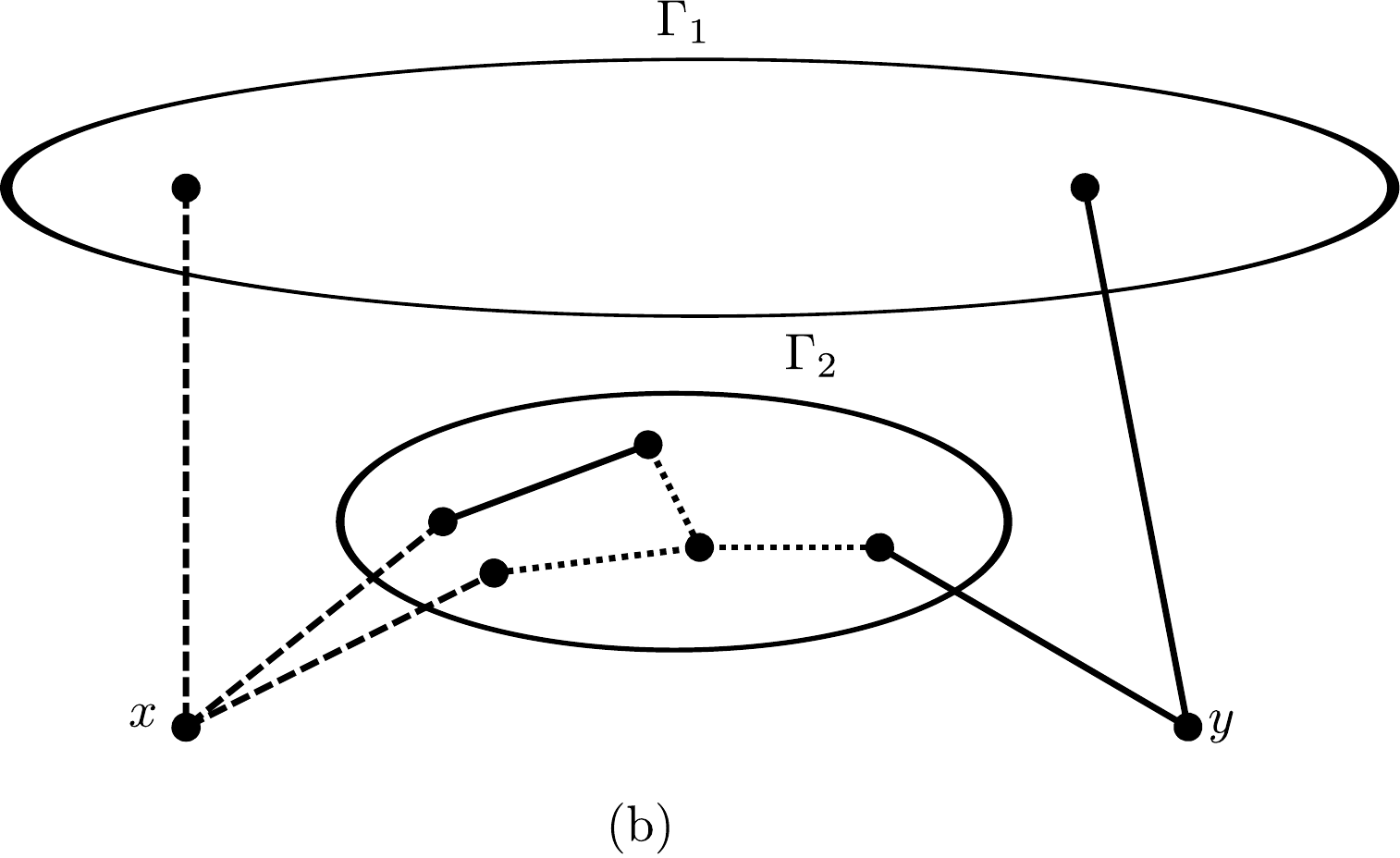}\end{center}
\bigskip{}
\begin{center}\includegraphics[scale=0.35]{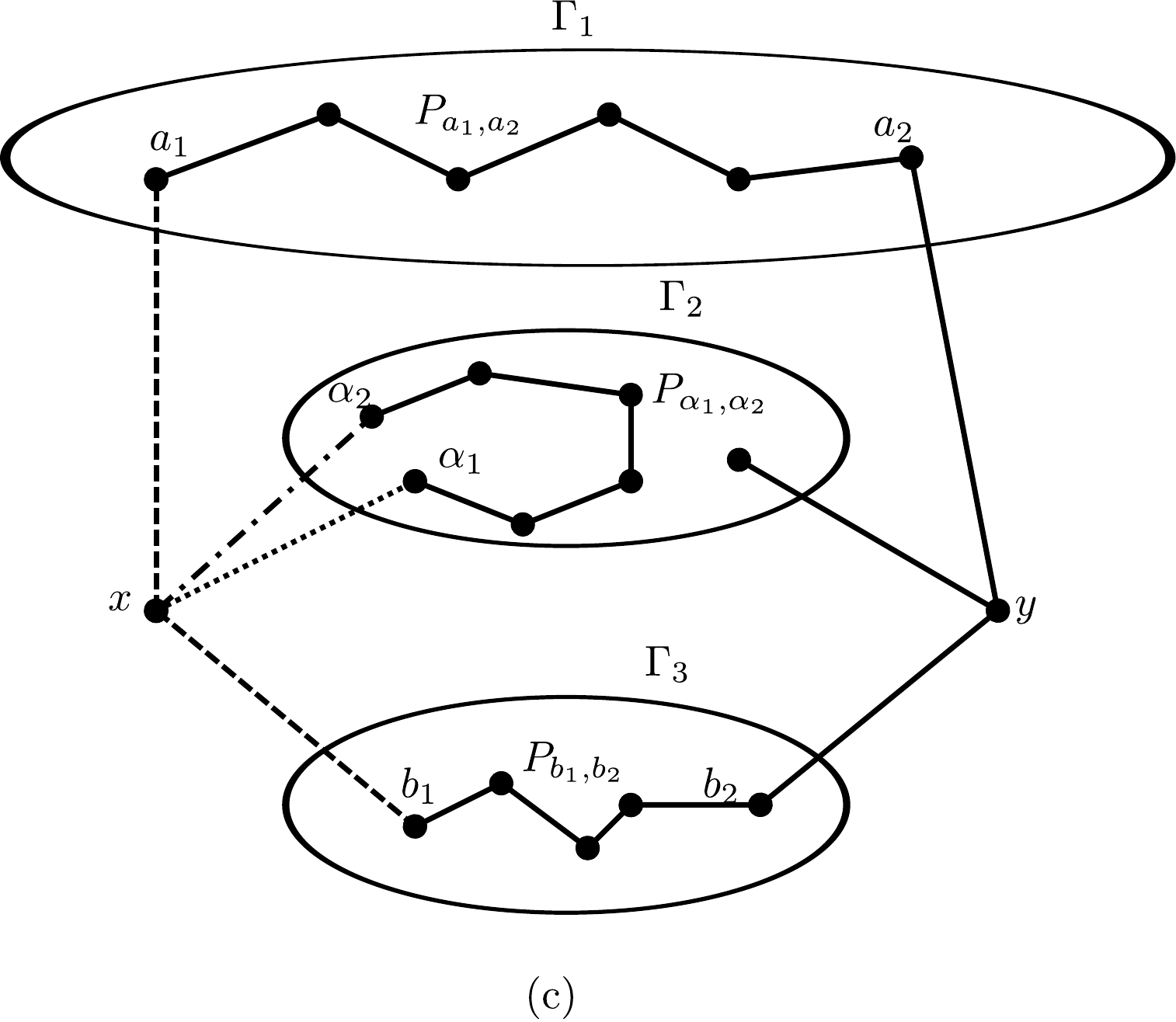}~~\includegraphics[scale=0.35]{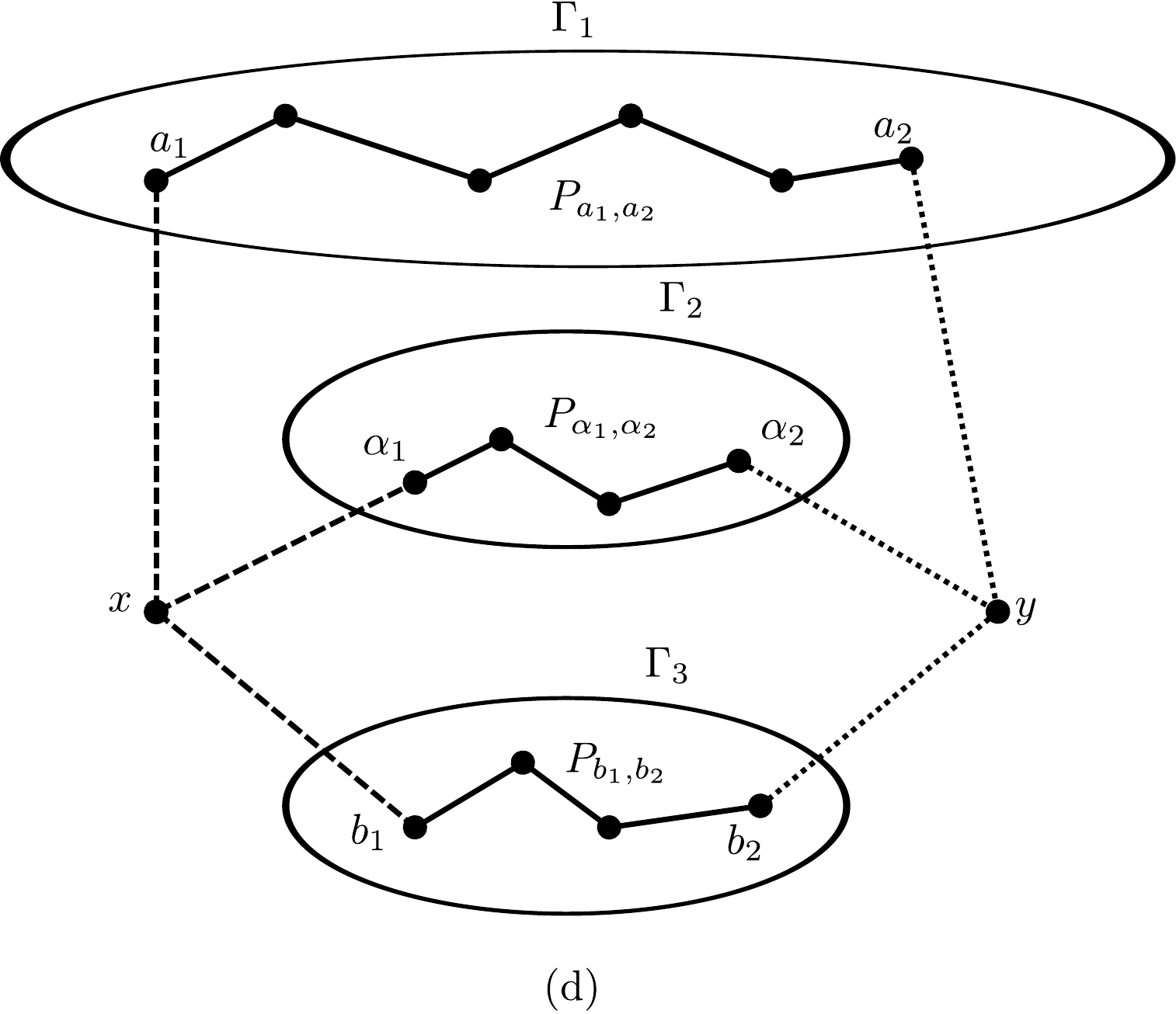}\end{center}
\caption{\label{fig:2-vertex-cut}(a) 2-vertex cut of $\Gamma$.  The $\gamma_i$'s are the interiors of the connected components $\Gamma_i$. 
(b) $Y_{x}$
with two edges connected to $\gamma_{2}$ (c) two Y-cycles with three
edges in three different components (d) the equality of $\phi_{Y_{x}}$
and $\phi_{Y_{y}}$. }
\end{figure}
Finally, it is known in graph theory that by the repeated application
of the above decomposition procedure the resulting marked components
are either topological cycles or $3$-connected graphs \cite{tutte01}. Let $n$ be
the number of $2$-vertex cuts which is needed to get such a decomposition,
$N_{2}=\sum_{\{x_{i},y_{i}\}}N_{2}(x_{i},y_{i})$, $N_{3}$ the number
of planar $3$-connected components, $N_{3}^{\prime}$ the number of
non-planar $3$-connected components and $N_{3}^{''}$ the number of
the topological cycles. Let $\mu=N_{3}+N_{3}^{'}+N_{3}^{''}$. Then
\begin{gather}
H_{1}(\mathcal{D}^{2}(\Gamma))=\left[\bigoplus_{i=1}^{\mu}H_{1}(\mathcal{D}^{2}(\tilde{\Gamma}_{i}))\right ]\oplus\mathbb{Z}^{N_{2}-n},\label{eq:2-conn-intermidiate}
\end{gather}
where
\begin{gather}
H_{1}(\mathcal{D}^{2}(\tilde{\Gamma}_{i}))=\mathbb{Z}^{\beta_1(\tilde{\Gamma}_{i})}\oplus\mathbb{Z},\,\,\, \tilde{\Gamma}_{i}-\mathrm{planar}\\\nonumber
H_{1}(\mathcal{D}^{2}(\tilde{\Gamma}_{i}))=\mathbb{Z}^{\beta_1(\tilde{\Gamma}_{i})}\oplus\mathbb{Z}_{2},\,\,\,\tilde{\Gamma}_{i}-\mathrm{nonplanar}\\\nonumber
H_{1}(\mathcal{D}^{2}(\tilde{\Gamma}_{i}))=\mathbb{Z},\,\,\,\tilde{\Gamma}_{i}-\mathrm{topological}\,\,\mathrm{cycle}\\\nonumber
\end{gather}
Note that $\sum_{i}\beta_1(\tilde{\Gamma}_{i})+N_{3}^{''}=\beta_1(\Gamma)+n$
and therefore
\begin{gather}
H_{1}(\mathcal{D}^{2}(\Gamma))=\mathbb{Z}^{\beta_1(\Gamma)+N_{2}+N_{3}}\oplus\mathbb{Z}_{2}^{N_{3}^{\prime}}.\label{eq:2-c}
\end{gather}

\subsection{$1$-connected graphs\label{sub:One-connected-graphs}}
In this subsection we focus on $1$-connected graphs. Assume that $\Gamma$ is $1$-connected but not $2$-connected. There exists a vertex $v\in V(\Gamma)$ such that after its deletion
$\Gamma$ splits into at least two connected components. Denote these components by
$\Gamma_{1},\ldots,\Gamma_{\mu(v)}$.
It is to be understood that each component $\Gamma_{i}$  contains the edges which connect it to $v$, along with a copy of the vertex $v$ itself.  Let $E_{i}$ denote the number of edges at $v$ which belong to $\Gamma_{i}$.
By Euler's formula the union of components $\{\Gamma_{i}\}_{i=1}^{{\mu(v)}}$ has
\begin{gather}
E(\Gamma)-\left(V(\Gamma)+\mu(v)-1\right)+\mu(v)=\beta_1(\Gamma)\label{eq:cycles-1-conn}
\end{gather}
independent cycles, hence the number of independent cycles does not change
compared to $\Gamma$. Moreover, the phases $\phi_{Y}$ inside each
of the components are the same as in $\Gamma$. Note, however, that
by splitting we lose Y-graphs whose three edges do not belong to
one fixed component $\Gamma_{i}$. Consequently, there are two cases to consider:

\begin{figure}[h]
\begin{center}~~\includegraphics[scale=0.35]{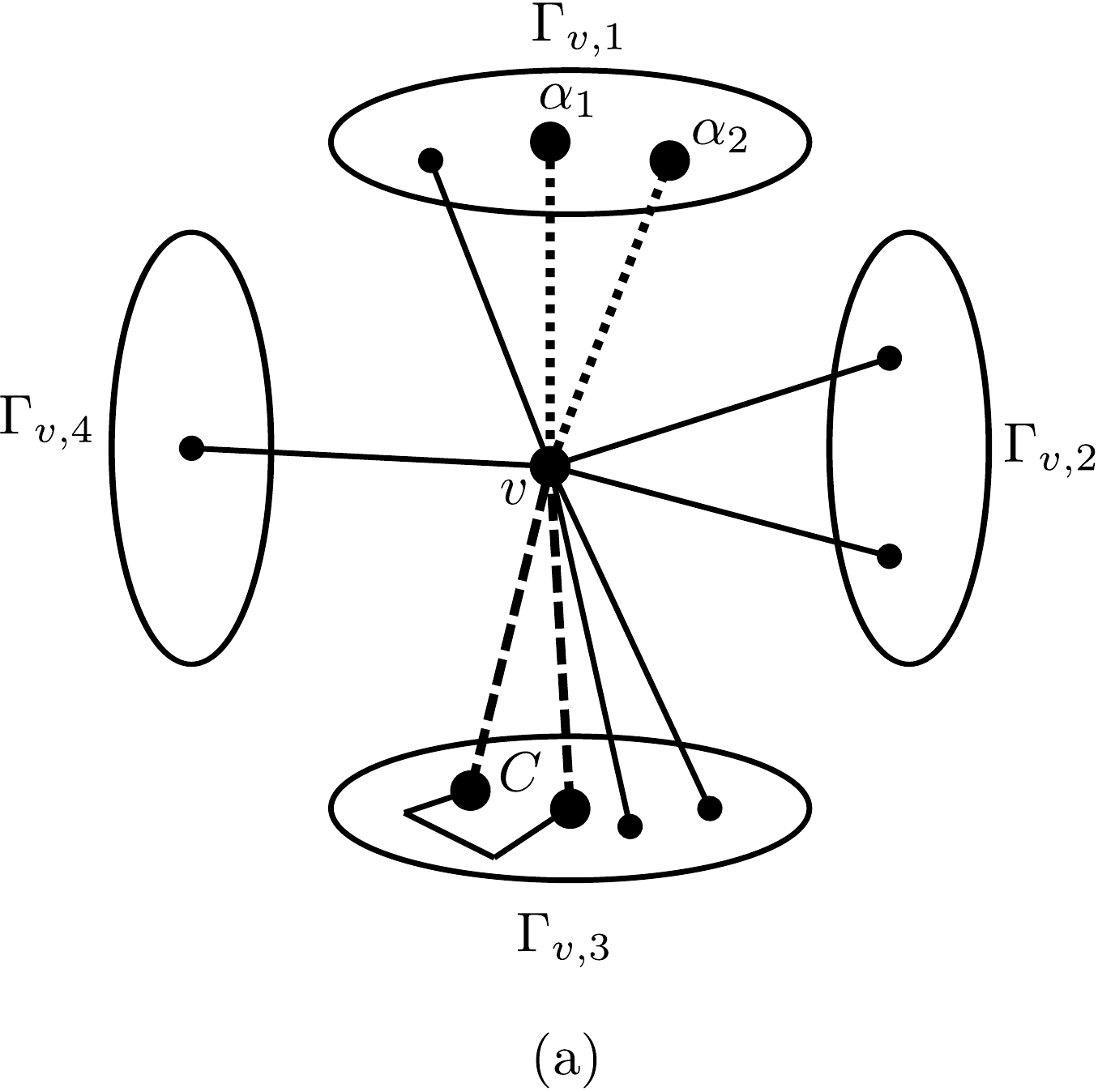}~~~~~~~~~~~~~~\includegraphics[scale=0.6]{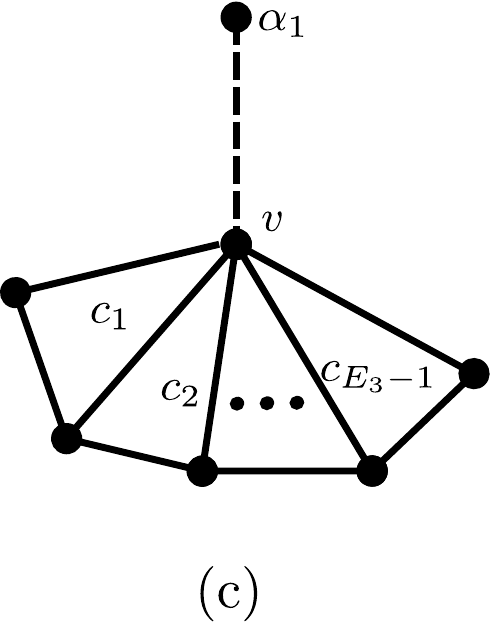}~\end{center}


\begin{center}\includegraphics[scale=0.35]{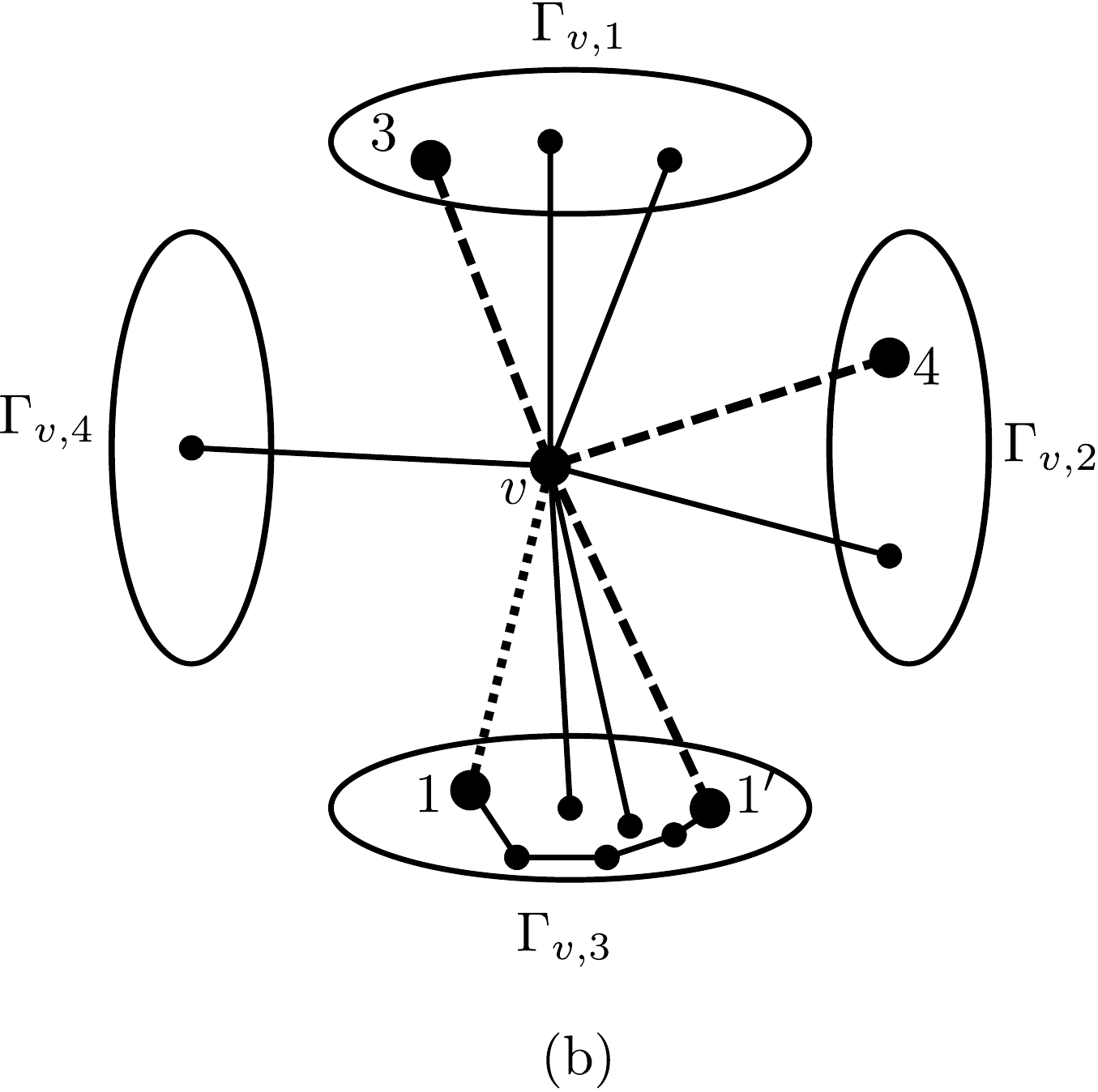}~~~\includegraphics[scale=0.5]{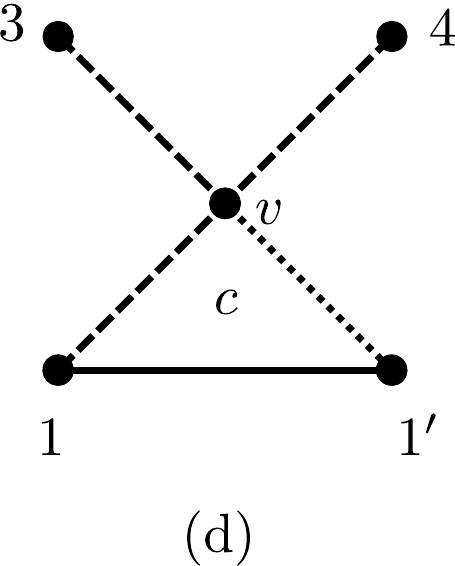}~~\includegraphics[scale=0.5]{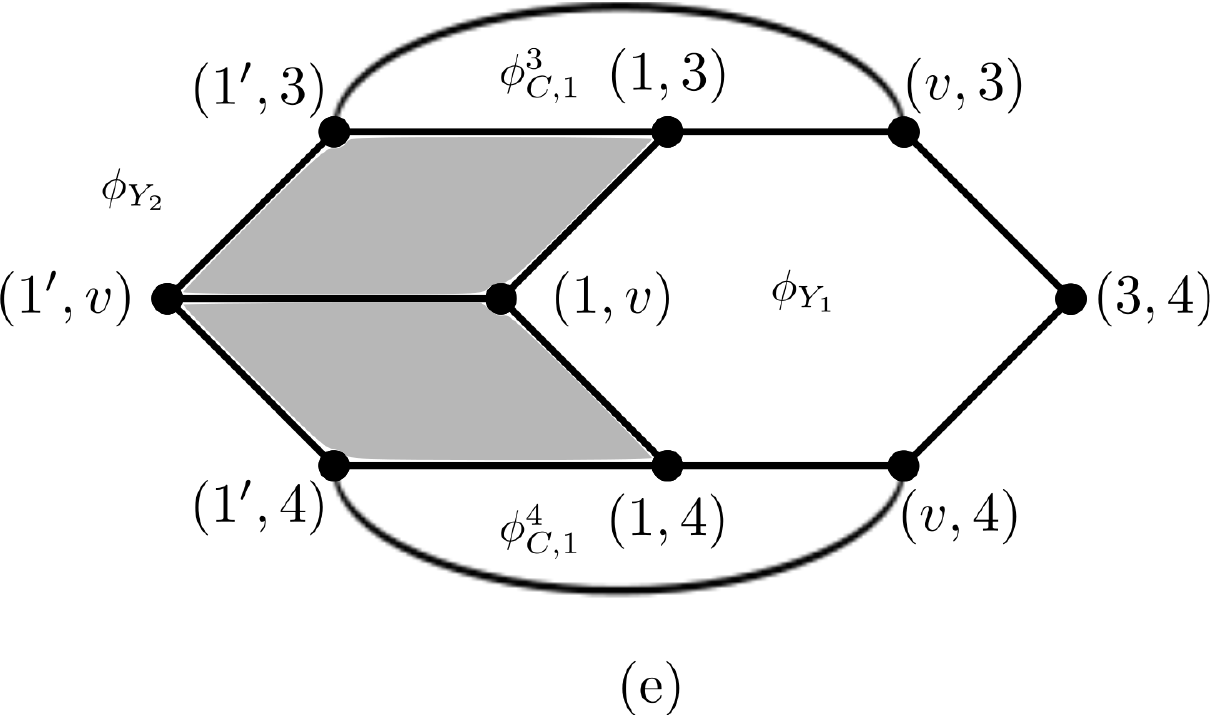}
\end{center}

\caption{\label{fig:1-vertex-cut}(a)  The $Y$-graphs $Y_1$ and $Y_2$ have central vertex $v$ and two common edges (long dashes) with vertices in $\Gamma_{v,3}$, but different edges (short dashes) with different vertices $\alpha_1$ and $\alpha_2$ in $\Gamma_{1,v}$.  Their exchange phases are the same.
(b) Each edge of the Y-graph is attached
to a different component. (c) Y-graphs with two edges in the same component (d) Two Y-graphs centered at $v$ with external vertices $\{1,3,4\}$ and $\{1',3,4\}$ respectively. (e) The relevant part of $2$-particle configuration space of (d).}
\end{figure}

\begin{enumerate}
\item Two edges of the Y-graph are attached to one component, for example $\Gamma_{v,3}$,
while the third one is attached to another component, $\Gamma_{v,1}$. We claim that the phase $\phi_{Y}$ does
not depend on the choice of the third edge, provided
it is attached to $\Gamma_{v,1}$. To see this consider two Y-graphs,
$Y_{1}$ and $Y_{2}$ shown in figure \ref{fig:1-vertex-cut}(a).
Since vertices $\alpha_{1}$ and $\alpha_{2}$ are connected by a
path, by Fact \ref{fact1} $\phi_{C,1}^{\alpha_{1}}=\phi_{C,1}^{\alpha_{2}}$. Next, relation (\ref{eq:AB-2}) applied to cycle $C$ and the two considered Y  graphs gives $\phi_{Y_{1}}=\phi_{Y_{2}}$. \\

After choosing one edge of Y in component $\Gamma_{v,1}$ (by the above argument it does not matter which), we can choose the  two other edges in $\Gamma_{v,3}$ in $E_{3}\choose{2}$ ways. Therefore, {\it a priori}, we have $E_{3}\choose{2}$ Y-graphs to consider. There are, however,  relations between them. In order to find the relevant relations consider the graph shown in figure \ref{fig:1-vertex-cut}(c).
We are interested in  Y-graphs with one edge given by $\alpha_1 \leftrightarrow v$ (dashed line) and two edges joining $v$ to vertices in
 $\Gamma_{v,3}$, say  $j$ and $k$.
 Each such Y-graph determines a cycle $c$ in $\Gamma_{v,3}$ containing vertices $v$, $j$ and $k$ (since $\Gamma_{v,3}$ is connected).
We have that
\begin{gather}\label{eq: c2 c1 Y for 1 connected}
\phi_{c,2}=\phi_{c,1}^{\alpha_1}+\phi_{Y}.
\end{gather}
Therefore, the $E_{3}\choose{2}$ $Y$-phases under consideration are determined by the AB- and two-particle phases, $\phi_{c,2}$ and $\phi_{c,1}^{\alpha_1}$, of the associated cycles $c$.  These cycles may be expressed as linear combinations of a basis  of $E_3-1$ cycles, denoted $c_1, \ldots, c_{E_3 - 1}$, 
as  in figure \ref{fig:1-vertex-cut}(c).  It is clear that if $c = \sum_{i=1}^{E_3} r_i c_i$, then 
\begin{gather}
\phi_{c,1}^{\alpha_1}=\sum_{i=1}^{E_3-1} r_i \phi_{c_i,1}^{\alpha_1}, \quad \,\,\,\phi_{c,2}=\sum_{i=1}^{E_3-1} r_i \phi_{c_i,2}.
\end{gather}
Thus, the $Y$-phases under consideration may be expressed in terms of the $2(E_3-1)$ phases $\phi_{c_i,2}$ and $\phi_{c_i,1}^{\alpha_1}$. \\

Let $Y_i$ be the $Y$-graph which determines the cycle $c_i$.  We may turn the preceding argument around; from  \eqref{eq: c2 c1 Y for 1 connected}, the AB-phase $\phi_{c_i,1}^{\alpha_1}$ can be expressed in terms of $\phi_{Y_i}$ and $\phi_{c_i,2}$.  Combining the preceding observations, we deduce that the $\binom{E_3}{2}$  Y-phases lost when the vertex $v$ is removed may be expressed in terms of the phases $\phi_{c_i,2}$ and $\phi_{Y_i}$.  The phases  $\phi_{c_i,2}$ remain when $v$ is removed.   It follows that phases $\phi_{Y_i}$ suffice to determine all of the lost phases, so that the number of independent $Y$-phases lost is $E_3-1$.
Repeating this argument for each component, the total  number of Y-phases lost is $\sum_{i=1}^{\mu(v)}(E_{i}-1)(\mu(v)-1)=(\mu(v)-1)(\nu(v)-\mu(v))$,
where   $\nu(v) = \sum_i E_{i}$ is the valency of $v$.\\

\item Each edge of the Y-graph is attached to a different component. We will show now that once
three different components have been chosen it does not matter which
of the edges attaching $\Gamma_{v,i}$ to $v$ we choose.  It suffices to consider the case  where the edges differ for only one component.
Let us consider the two Y-graphs shown in figure \ref{fig:1-vertex-cut}(b).
The first one consists of the three dashed edges and the second of
two dashed edges attached to $\Gamma_{v,1}$ and $\Gamma_{v,2}$ respectively and
the dotted edged attached to $\Gamma_{v,3}$.  The two Y-graphs are shown on their own in figure \ref{fig:1-vertex-cut}(d);   we let $Y_1$ and $Y_2$ denote the Y-graphs with vertices $\{1,3,4,v\}$ and $\{1',3,4,v\}$ respectively. A subgraph of the corresponding 2-particle configuration space is shown in figure  \ref{fig:1-vertex-cut}(e).   
There  we see that
\begin{gather}
\phi_{Y_2}=\phi_{Y_1}+\phi_{c,1}^3+\phi_{c,1}^4.
\end{gather}
In Step~1 above, we showed that the AB phases   $\phi_{c,1}^3$ and $\phi_{c,1}^4$
can be expressed in terms of $\phi_{c,2}$ and Y-phases already accounted for in Step 1.
Thus, the number of the independent
Y-phases we lose is equal to the number of independent
Y-cycles in the two-particle configuration space of the star graph
with $\mu(v)$ edges, that is, $(\mu(v)-1)(\mu(v)-2)/2$.
\end{enumerate}
Summing up we can write
\begin{gather}
H_{1}(\mathcal{D}^{2}(\Gamma))=\left[\bigoplus_{i=1}^{\mu(v)}H_{1}(\mathcal{D}^{2}(\Gamma_{v,i}))\right]\oplus\mathbb{Z}^{N_{1}(v)},\label{eq:1-cut-forula}
\end{gather}
where $N_{1}(v)=(\mu(v)-1)(\mu(v)-2)/{2}+(\mu(v)-1)(\nu(v)-\mu(v))$.
It is known in graph theory \cite{tutte01} that by the repeated application
of the above decomposition procedure the resulting components become finally
$2$-connected graphs. Let $v_{1},\ldots,v_{l}$ be the set of cut
vertices such that components $\Gamma_{v_{i},k}$ are $2$-connected.
Making use of formula (\ref{eq:2-c}) we can write

\begin{gather}
H_{1}(\mathcal{D}^{2}(\Gamma))=\mathbb{Z}^{\beta(\Gamma)+N_{1}+N_{2}+N_{3}}\oplus\mathbb{Z}_{2}^{N_{3}^{\prime}},\label{eq:2-particle-final}
\end{gather}
where $N_{1}=\sum_{i}N_{1}(v_{i})$.

\section{n-particle statistics for $2$-connected graphs\label{sec:N-particle-statistics-for}}

Having discussed $2$-particle configuration spaces, we switch to the
$n$-particle case, $\mathcal{D}^{n}(\Gamma)$, where $n>2$. We proceed in a
similar manner to the previous section. First we give a spanning set of $H_1(\mathcal{D}^n(\Gamma))$. Next we show that if $\Gamma$ is $2$-connected
the first homology group stabilizes with respect to $n$, that is, $H_{1}(\mathcal{D}^{n}(\Gamma))=H_{1}(\mathcal{D}^{2}(\Gamma))$.
Making use of formula (\ref{eq:2-c})

\begin{gather*}
H_{1}(\mathcal{D}^{n}(\Gamma))=\mathbb{Z}^{\beta(\Gamma)+N_{2}+N_{3}}\oplus\mathbb{Z}_{2}^{N_{3}^{\prime}}.
\end{gather*}

\subsection{A spanning set of $H_1(\mathcal{D}^{n}(\Gamma))$\label{sub:An-over-complete-basis-2}}

In order to calculate $H_1(\mathcal{D}^{n}(\Gamma))$ we first need to
subdivide the edges of $\Gamma$ appropriately. By Theorem \ref{Abrams_thm} each edge of $\Gamma$
must be able to accommodate $n$ particles and each cycle needs to
have at least $n+1$ vertices, that is, $\Gamma$ needs to be sufficiently subdivided. Before we specify a
spanning set of $H_1(\mathcal{D}^{n}(\Gamma))$ we first discuss
two interesting aspects of this space. The first one concerns the
relation between the exchange phase of $k$ particles, $k\leq n$
on the cycle $C$ of the lasso graph and its $\phi_{Y}$ phases (see
Lemma \ref{aspect1-1} ). The second gives the relation between the
AB-phases for fixed cycle $c$ of $\Gamma$ and the different possible positions
of the $n-1$ stationary particles.
\begin{lemma}
\label{aspect1-1}The exchange phase, $\phi_{C,n}$, of $n$ particles
on the cycle $c$ of the lasso graph is the sum of the exchange phase,
$\phi_{C,n-1}^{1}$, of $n-1$ particles on the cycle $C$ with the
last particle sitting at the vertex not belonging to $C$, e.g. vertex
$1$, and the phase $\phi_Y$ associated with the exchange of two particles on the $Y$ subgraph with $n-2$ particles
placed in the vertices $v_1,\ldots ,v_{n-2}$ of $C$ not belonging to the Y
\begin{gather*}
\phi_{C,n}=\phi_{C,n-1}^{1}+\phi_{Y}^{v_1,\ldots ,v_{n-2}}.
\end{gather*}
\end{lemma}
\begin{proof}
By (\ref{eq:lasso-relation}), the lemma is true for $n=2$. The proof for $n = 3$ particles is shown in figure~\ref{fig:The-relevant-parts}(a), and contains the essence of the
 argument for general $n$.  Indeed, the way to incorporate additional particles  is illustrated by the $n=4$ case, shown in figure~\ref{fig:The-relevant-parts}(b).
Note that figure~\ref{fig:The-relevant-parts} shows only the small portion of the $n=3$ and $n=4$  configuration spaces required to establish the lemma. These configuration spaces are derived from the $1$-particle lasso graphs shown in figures \ref{fig:-The-subdivided}(a)
and \ref{fig:-The-subdivided}(b) respectively; it is easy to see that these are indeed
sufficiently subdivided. The Y-graphs we consider for $n=3$ an $n=4$ are $\{2\leftrightarrow3,\,3\leftrightarrow4,\,3\leftrightarrow6\}$
and $\{3\leftrightarrow4,\,4\leftrightarrow5,4\leftrightarrow8\}$
respectively.

\end{proof}
\begin{figure}[h]
\begin{center}~~~~~~\includegraphics[scale=0.5]{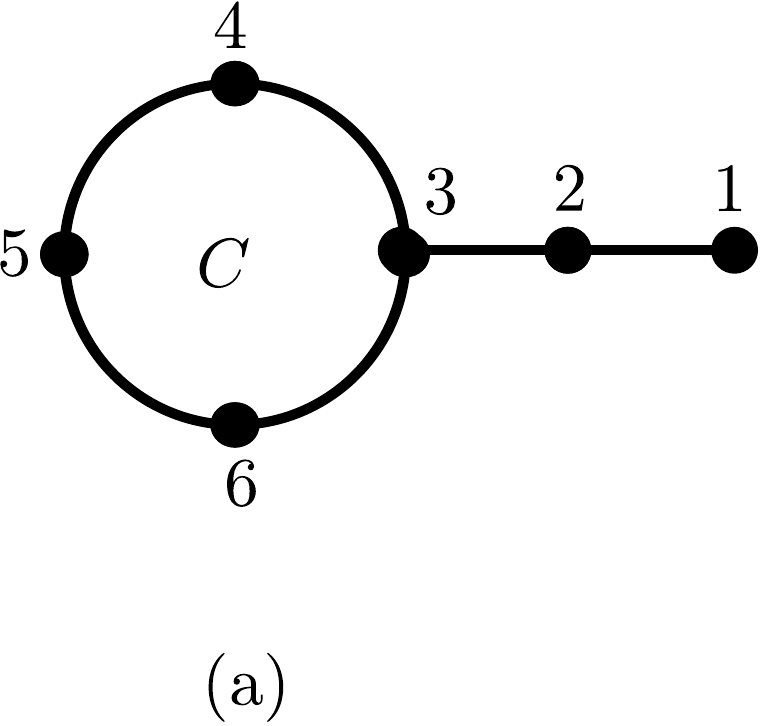}~~~~~~~~~~~\includegraphics[scale=0.5]{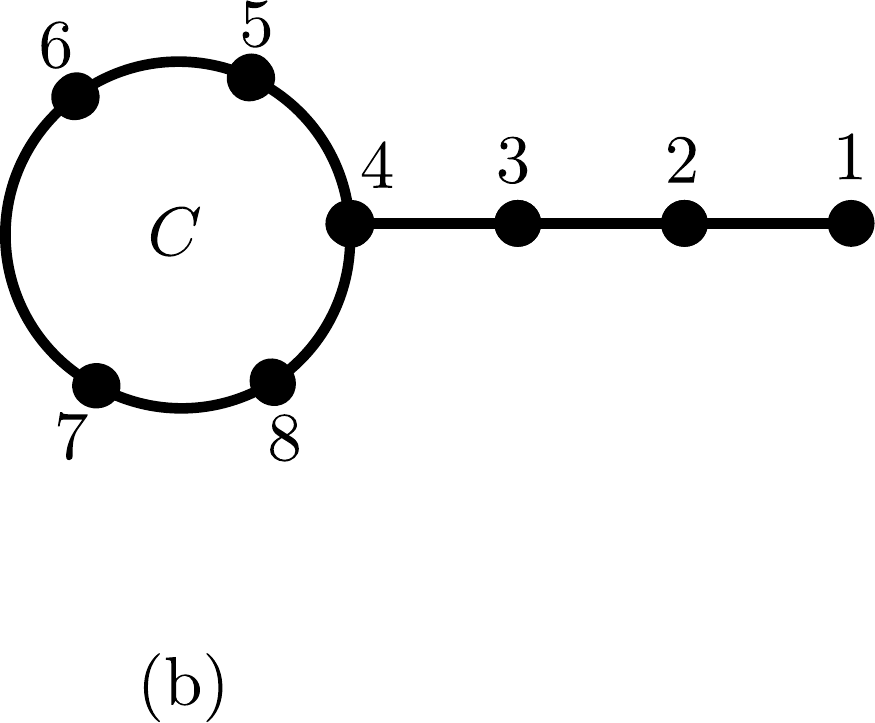}\end{center}

\caption{ \label{fig:-The-subdivided}The subdivided lasso for (a) $3$ particles,
(b) $4$ particles. }
\end{figure}

\begin{figure}[h]
\begin{center}\includegraphics[scale=0.5]{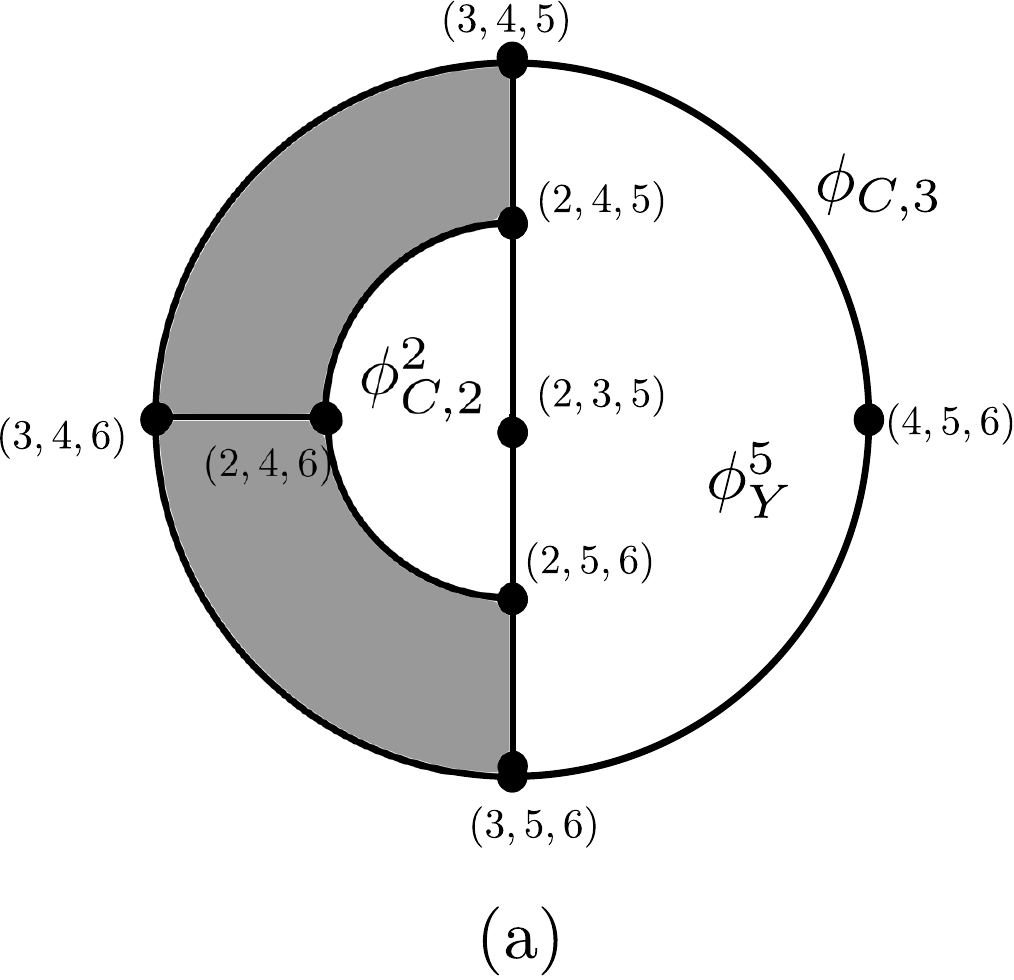}~~~\includegraphics[scale=0.5]{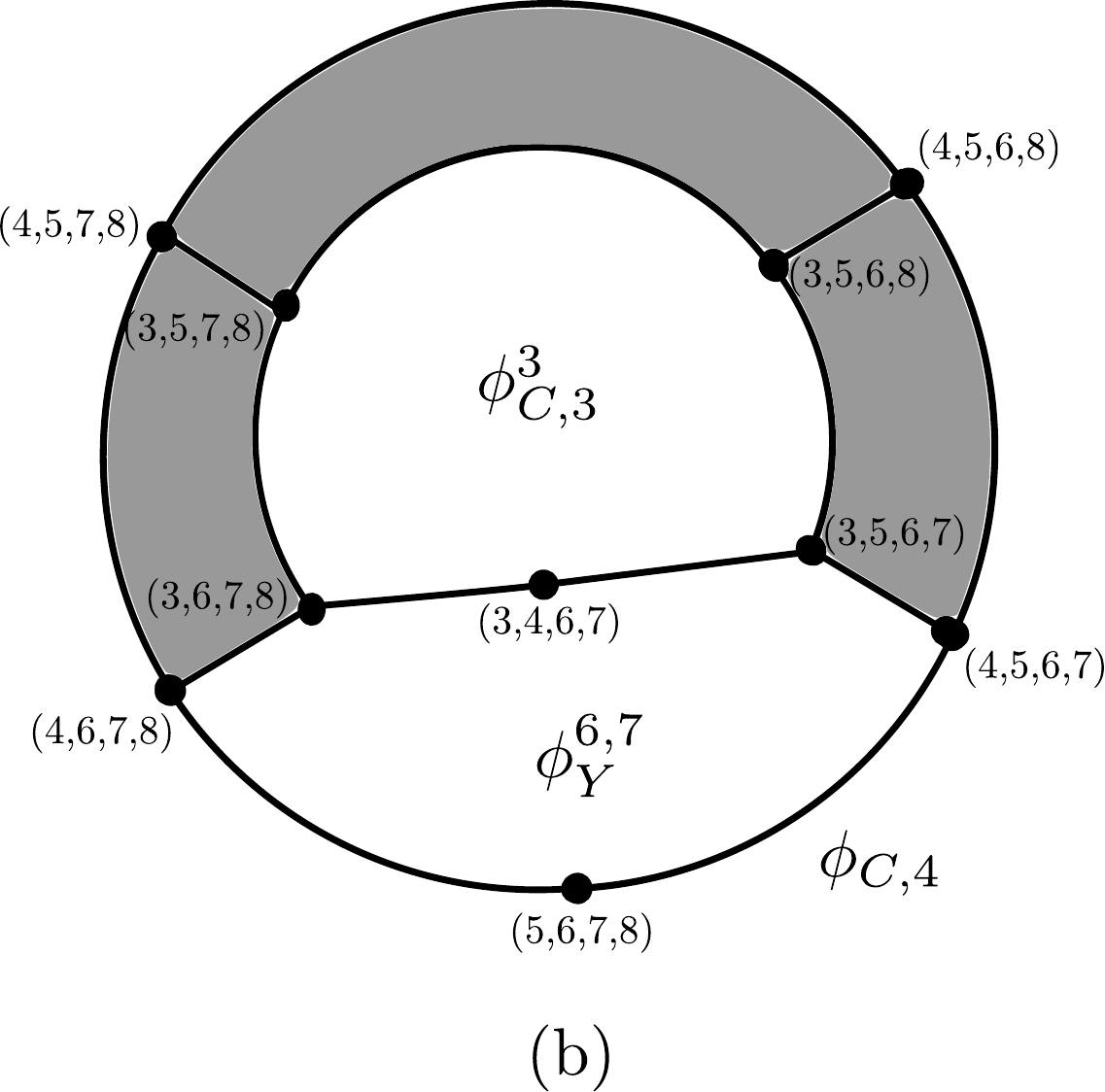}\end{center}

\caption{\label{fig:The-relevant-parts} 
 Subgraphs of the configurations
spaces for the lasso graphs with (a) $3$ particles: $\phi_{C,3} = \phi^2_{C,2} + \phi^5_Y$, (b) $4$ particles: $\phi_{C,4}  = \phi^3_{C,3} + \phi^{6,7}_Y$. }
\end{figure}

\noindent By repeated application of Lemma \ref{aspect1-1} we see
that $\phi_{C,n}$ can be expressed as a sum of an AB-phase and the
Y-phases corresponding to different positions of $n-2$ particles.
For example in the case of the graphs from figure \ref{fig:-The-subdivided}(a)
and \ref{fig:-The-subdivided}(b) we get

\begin{gather*}
\phi_{C,3}=\phi_{Y}^{5}+\phi_{C,2}^{2}=\phi_{Y}^{5}+\phi_{Y}^{1}+\phi_{C,1}^{1,2}\,,\\
\phi_{C,4}=\phi_{Y}^{6,7}+\phi_{C,3}^{3}=\phi_{Y}^{6,7}+\phi_{C,3}^{1}=\phi_{Y}^{6,7}+\phi_{Y}^{1,6}+\phi_{Y}^{1,2}+\phi_{C,1}^{1,2,3}\,.
\end{gather*}

\paragraph{Aharonov-Bohm phases}

Assume now that we have $n$ particles on $\Gamma$. Let $C$ be a
cycle of $\Gamma$ and $e_{1}$ and $e_{2}$ two sufficiently subdivided
edges attached to $C$ (see figure \ref{fig:stabilization}(a)). We
denote by $\phi_{C,1}^{k_{1},k_{2}}$ the AB-phase corresponding
to the situation where one particle goes around the cycle $C$ while $k_{1}$
particles are in the edge $e_{1}$ and $k_{2}$ particles are in the
edge $e_{2}$, $k_{1}+k_{2}=n-1$. For each distribution $(k_{1},k_{2})$
of the $n-1$ particles between the edges $e_{1}$ and $e_{2}$ we
get a (possibly) different AB-cycle and AB-phase in $\mathcal{D}^{n}(\Gamma)$.
We want to know how they are related. To this end notice that

\begin{gather}
\phi_{C,2}^{k_{1},k_{2}}=\phi_{C,1}^{k_{1}+1,k_{2}}+\phi_{Y_{1}}^{k_{1},k_{2}},\,\,\,\,\phi_{C,2}^{k_{1},k_{2}}=\phi_{C,1}^{k_{1},k_{2}+1}+\phi_{Y_{2}}^{k_{1},k_{2}},\label{eq:AB-1-1}
\end{gather}
and hence
\begin{gather}
\phi_{C,1}^{k_{1}+1,k_{2}}-\phi_{C,1}^{k_{1},k_{2}+1}=\phi_{Y_{2}}^{k_{1},k_{2}}-\phi_{Y_{1}}^{k_{1},k_{2}}.\label{eq:AB-2-1}
\end{gather}
The relations between different AB-phases for a fixed cycle $C$
of $\Gamma$ are therefore encoded in the $2$-particle phases $\phi_{Y}$,
albeit these phases can depend on the positions of the remaining $n-2$
particles.\\

\begin{figure}[h]
\begin{center}~~~~~~\includegraphics[scale=0.5]{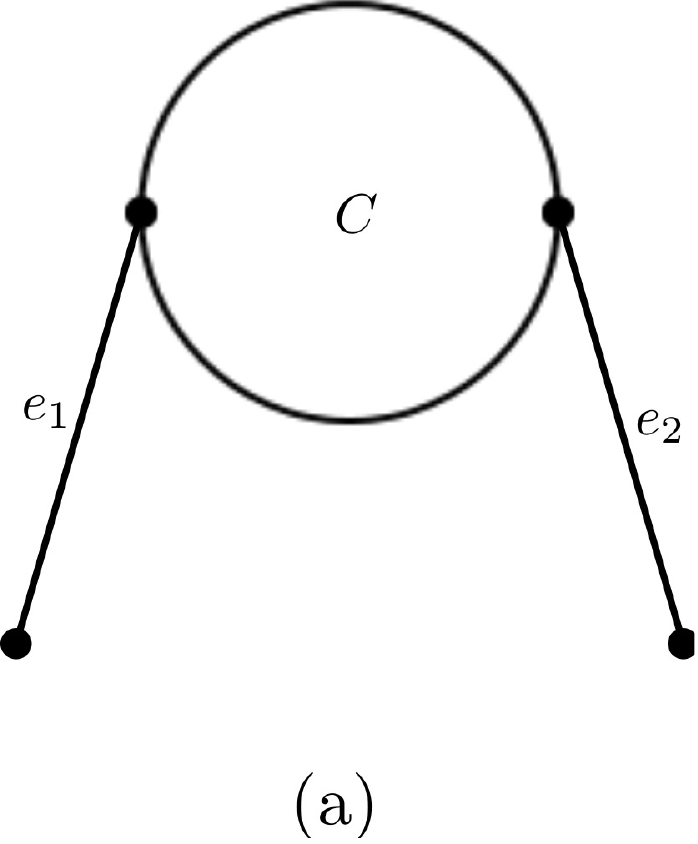}~~~~~~~~~~~~~~~~~~~~~~~~~~\includegraphics[scale=0.5]{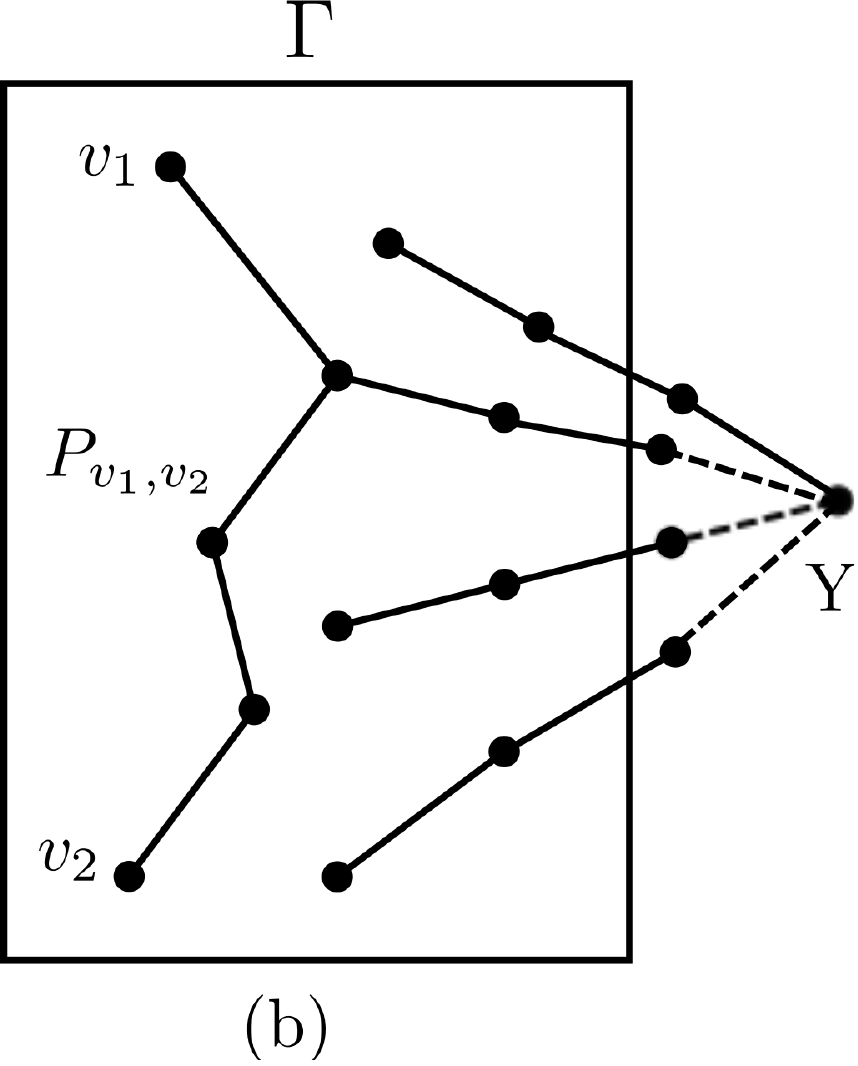}\end{center}

\caption{\label{fig:stabilization}(a) The relation between AB-phases, (b)
the stabilization of the first homology group.}
\end{figure}

\noindent A spanning set of
$H_1(\mathcal(D)^n(\Gamma))$ is given by the following (see section \ref{MT} for proof):
\begin{enumerate}
\item All $2$-particle cycles corresponding to the exchange of two particles
on the Y subgraph while $n-2$ particles are at vertices not belonging
to the considered Y-graph. In general the phases $\phi_{Y}$ depend
on the position of the remaining $n-2$ particles.
\item The set of $\beta_{1}(\Gamma)$ AB-cycles, where $\beta_{1}(\Gamma)$
is the number of the independent cycles of $\Gamma$. \end{enumerate}
\begin{theorem}\label{thm: 5}
\label{stabilization}For a $2$-connected graph $\Gamma$ the
first homology group stabilizes with respect to the number of particles,
i.e. $H_{1}(\mathcal{D}^{n}(\Gamma))=H_{1}(\mathcal{D}^{2}(\Gamma))$.\end{theorem}
\begin{proof}
Using our spanning set it is enough to show that phases on the
Y-cycles do not depend on the position of the remaining $n-2$ particles.
Notice that if any pair of the vertices not belonging to the chosen
Y-graph is connected by a path then clearly the corresponding Y-phases have this property. Since the graph $\Gamma$ is $2$-connected
it remains at least $1$-connected after removal of a vertex. Removing the central vertex of the Y (see figure
\ref{fig:stabilization}(b)), the theorem follows.
\end{proof}

\section{$n$-particle statistics on $1$-connected graphs\label{sec:N-particle-statistics-on}}

By Theorem~\ref{thm: 5}, in order to fully characterize the first homology group of $\mathcal{D}^{n}(\Gamma)$ for an arbitrary graph $\Gamma$ we are left to calculate $H_1(\mathcal{D}^{n}(\Gamma))$ for graphs which are 1-connected but not 2-connected. This is achieved by considering $n$-particle star and fan graphs.

\subsection{Star graphs\label{sub:The-star-graphs}}

In the following we consider a particular family of $1$-connected graphs,
namely the star graphs $S_{E}$ with $E$ edges (see figure \ref{fig:The-star and fan}(a)).
Our aim is to provide a formula for the dimension of the first homology
group, $\beta_{n}^{E}$, of the $n$-particle configuration space $\mathcal{D}^{n}(S_{E})$. Let us recall that a graph $\Gamma$ is $1$-connected iff after deletion of one vertex it splits into at least two connected components.

\paragraph{Star graph with non-subdivided edges}
It turns out that the computation of $\beta_n^E$ can be reduced to the case of $n$ particles on a  star graph with non-subdivided edges, so we consider this case first.
Let $\bar{S}_{E}$ denote the star graph with $E+1$ vertices and $E$ edges each connecting the central vertex to a single vertex of valency $1$; such a star graph is not sufficiently subdivided for $n>2$ particles. As there are no pairs of disjoint edges (every edge contains the central vertex), there are no contractible cycles.
Therefore, the $n$-particle configuration
space, $\mathcal{D}^{n}(\bar{S}_E)$ is a graph, i.e. a one-dimensional cell complex. The number of independent cycles
in $\mathcal{D}^{n}(\bar{S}_E)$, denoted here and in what follows by $\gamma_n^E$,
is given by the first Betti number, $E_n - V_n + 1$, where $E_n$ and $V_n$ are the number of edges and vertices in $\mathcal{D}^{n}(S_E)$. It is
easy to see that  $V_{n}={E+1 \choose n}$
and $E_{n}=E\cdot{E-1 \choose n-1}$. Hence
\begin{gather}
\gamma_{n}^{E}=E {E-1 \choose n-1}-{E+1 \choose n}+1.\label{eq:2-particleY}
\end{gather}

\paragraph{Y-graph }

The simplest case of a sufficiently subdivided star graph is a Y-graph where each arm has $n-1$ segments.  As there are no cycles on the Y-graph itself, cycles in the $n$-particle configuration space are generated by two-particle exchanges on the non-subdivided subgraph $\bar{Y}$ comprised of the three segments adjacent to the central vertex. A basis of independent cycles is obtained by taking all possible configurations of the $n-2$ particles amongst the three arms of the Y-graph.  As configurations which differ by shifting particles within the arms of the Y produce homotopic cycles, the number of distinct configurations is the number of partitions of $n-2$ indistinguishable particles amongst three distinguishable boxes, or ${(n-2) +(3-1) \choose n-2} = {n \choose n-2}$.  Therefore,

\begin{gather}
\beta_{n}^{3}= {n \choose n-2}\gamma_2^3=\frac{n(n-1)}{2}.\label{eq:Y}
\end{gather}

\paragraph{Star graph with five arms}

For star graphs with more than three arms, it is necessary to take account of relations between cycles involving two or more moving particles.  With this in mind, we introduce the following terminology: an $(n,m)$-cycle is a cycle of $n$ particles on which $m$ particles move and $(n-m)$ particles remain fixed.

The general case is well illustrated by considering the star graph with $E=5$ arms.  As above, we suppose that each arm of  $S_5$ has $(n-1)$ segments, and is therefore sufficiently subdivided to accommodate $n$ particles.  Let $\bar{S}_5$ denote the non-subdivided subgraph consisting of the five segments adjacent to the central vertex. As there are no cycles on $S_5$, a spanning set for the first homology group of the $n$-particle configuration space is provided by two-particle cycles on the Y's contained in $\bar{S}_5$.  The number of independent two-particle cycles on $\bar{S}_5$ is given by $\gamma_5^2$.  For each of these, we can distribute the remaining $(n-2)$ particles among the five edges of $S_5$  (cycles which differ by shifting particles within an edge  are homotopic).  Therefore, we obtain a spanning set consisting  of ${\beta''}_n^5$ $(n,2)$-cycles, where
\[ {\beta''}_n^5 :=  \binom{n+2}{4} \gamma_2^5 .\]

The preceding discussion of non-subdivided star graphs reveals that there are relations among the cycles in the spanning set.  In particular, a subset of the $(n,2)$-cycles can be replaced by a smaller number of $(n,3)$-cycles.

To see this, consider first the case of $n=3$ particles on the non-subdivided star graph $\bar{S}_5$.
By definition, the number of independent $(3,3)$-cycles is $\gamma_3^5$.  However, the number of  $(3,2)$-cycles on $\bar{S}_5$ is larger; it is given by $\binom{5}{1} \gamma^4_2$, where the first factor represents the number of positions of the fixed particle, and the second factor represents the number of independent $(2,2)$-cycles on the remaining four edges of $\bar{S}_5$.  It is easily checked that $\gamma_3^5 - \binom{5}{1} \gamma^4_2 = -3$, so that there are three relations amongst the $(3,2)$-cycles on $\bar{S}_5$. 

 We return to the case of $n$ particles.  For each $(3,3)$-cycle on $\bar{S}_5$, there are $\binom{n+1}{4}$ $(n,3)$-cycles on $S_5$; the factor  $\binom{n+1}{4}$ is the number of ways to distribute the $n-3$ fixed  particles on the five edges of $S_5$ outside of $\bar{S}_5$. Calculating  the number of $(n,2)$-cycles on $S_5$ obtained from $(3,2)$-cycles on $\bar{S}_5$ requires a bit more care. The  reasoning underlying the preceding count of $(n,3)$ cycles  would suggest that the number of such $(n,2)$-cycles is given by $\binom{n+1}{4} \binom{5 }{1} \gamma_2^4$.  However,  this expression introduces some double counting.  In particular,  $(n,2)$-cycles for which two of the fixed particles lie  in $\bar{S}_5$ are  counted twice, as each of these two fixed particles is separately regarded as the fixed particle in a $(3,2)$-cycle on ${\bar  S}_5$.  The correct expression is obtained by subtracting the number of doubly counted cycles; this is given by $\binom{n}{4} \binom{5}{2} \gamma_2^3$. Thus we may replace  this subset of $(n,2)$-cycles by the $(n,3)$-cycles to which they are related to obtain a smaller spanning set with ${\beta'}_n^5 $ elements, where
\[ {\beta'}_n^5 =  {\beta''}_n^5 +  \binom{n+1}{4} \gamma_3^5 - \left( \binom{n+1}{4} \binom{5} {1} \gamma_2^4 - \binom{n}{4} \binom{5} {2}\gamma_2^3\right).\]

Finally, we must account for relations among the $(n,3)$-cycles.  Consider first the case of just four particles on $\bar{S}_5$.  The number of independent $(4,4)$-cycles is $\gamma_4^5$.  The number of $(4,3)$-cycles 
is $\binom{5}{1} \gamma^4_3$, where the first factor represents the number of positions of the fixed particle, and the second factor represents the number of independent $(3,3)$-cycles on the remaining four edges of $\bar{S}_5$.  For each  $(4,4)$-cycle on $\bar{S}_5$, there are $\binom{n}{4}$ $(n,4)$ cycles  on $S_5$. Similarly, for each $(4,3)$-cycle on $\bar{S}_5$, there are $\binom{n}{4}$ $(n,3)$-cycles on $S_5$ (there is no over-counting, as there are no five-particle cycles on $\bar{S}_5$).  Replacing this subset of $(n,3)$-cycles by the $(n,4)$-cycles to which they are related, we get a smaller spanning set of $\beta_n^5$ elements, where
\[ \beta_n^5 =  {\beta'}_n^5 +  \binom{n}{4}\left( \gamma_4^5 - \binom{5}{1} \gamma_3^4\right) = 6\binom{n+2}{4} - 4\binom{n+1}{4} + \binom{n}{4}.\]
As there are no five-particle cycles on $\bar{S}_5$, there are no additional relations, and the resulting spanning set constitutes a basis.

\paragraph{$n$ particles on a star graph with $E$ arms}

The formula in the general case of $E$ edges is obtained following a similar argument.  We start with a spanning set of $\binom{n+E-3}{E-1} \gamma_2^E$ $(n,2)$-cycles on $S_E$.  We then replace a subset of $(n,2)$-cycles by a smaller number of $(n,3)$-cycles, then replace a subset of these $(n,3)$-cycles by a smaller number of $(n,4)$-cycles, and so on, proceeding to $(n,E-1)$-cycles, thereby obtaining a basis. The number of elements in the basis is given by
\begin{equation} \label{eq: first beta expression} \beta_n^E = \sum_{m=2}^{E-1}\left( \binom{n - m + E -1}{E-1} \gamma_m^E + \sum_{j = 1}^{E-m}  (-1)^j \binom{n - m-j + E}{E-1} \binom{E}{j} \gamma_{m-1}^{E-j}\right).\end{equation}
The outer $m$-sum is taken over $(n,m)$-cycles.  The $m$th term is the difference between the number of $(n,m)$-cycles and the number of $(n,m-1)$-cycles to which they are related.  The inclusion-exclusion sum over $j$ compensates for over-counting $(n,m-1)$-cycles with $j$ fixed particles in ${\bar S}_E$.

It  turns out to be convenient to rearrange the sums in \eqref{eq: first beta expression} to obtain the following equivalent expression:
\begin{equation} \label{eq: second beta expression}
\beta_{n}^{E}=\sum_{k=2}^{E-1}{n-k+E-1 \choose E-1} \alpha^E_{k} \end{equation}
where
\begin{equation}\label{eq:alpha_k}
\alpha_{k}^{E}=\sum_{i=0}^{k-2}(-1)^{i}{E \choose i}\cdot\gamma_{k-i}^{E-i}.
\end{equation}
This is because the coefficients $\alpha_k^E$ turn out to have a simple expression.
First, straightforward manipulation yields
\begin{gather}
\alpha_{k}^{E}=\gamma_{k}^{E}-\sum_{i=1}^{k-2}{E \choose i}\alpha_{k-i}^{E-i}.\label{eq:alpha_k1}
\end{gather}
We then have the following:
%
%
%
%
%
%
%
%

\begin{lemma}
\label{lemma4}The coefficients $\alpha_{k}^{E}=(-1)^{k}{E-1 \choose k}$.\end{lemma}
\begin{proof}
We proceed by induction. Direct calculations give $\alpha_{2}={E-1 \choose 2}$.
Assume that $\alpha_{i}^{E}=(-1)^{i}{E-1 \choose i}$ for $i\in\{2,\ldots,k-1\}$
and $k\leq E$. Using this assumption and (\ref{eq:alpha_k1})

\begin{gather*}
\alpha_{k}=\gamma_{k}^{E}-(-1)^{k}\sum_{i=1}^{k-2}(-1)^{i}{E \choose i}{E-i-1 \choose k-i}.
\end{gather*}
Making use of the identity ${r \choose k}=(-1)^{k}{k-r-1 \choose k}$
and Vandermonde's convolution $\sum_{i=0}^{k}{E \choose i}{k-E \choose k-i}=1$, we get

\begin{gather*}
(-1)^{k}\sum_{i=1}^{k-2}(-1)^{i}{E \choose i}{E-i-1 \choose k-i}=\sum_{i=1}^{k-2}{E \choose i}{k-E \choose k-i}\\
=1-(-1)^{k}{E-1 \choose k}+(E-k){E \choose k-1}-{E \choose k}\,.
\end{gather*}
Using (\ref{eq:2-particleY}) for $\gamma_k^E$, we get
\begin{gather*}
\alpha_{k}=(-1)^{k}{E-1 \choose k}+E{E-1 \choose k-1}-{E+1 \choose k}-(E-k){E \choose k-1}+{E \choose k}\,.
\end{gather*}
Expanding ${E+1 \choose k}={E \choose k}+{E \choose k-1}$ and straightforward manipulations show
\begin{gather*}
\alpha_{k}=(-1)^{k}{E-1 \choose k},
\end{gather*}
which completes the argument.
\end{proof}
\noindent By Lemma \ref{lemma4}
\begin{gather*}
\beta_{n}^{E}=\sum_{k=2}^{E-1}{n-k+E-1 \choose E-1}\cdot\alpha_{k}=\sum_{k=2}^{E-1}\left(-1\right)^{k}{E-1 \choose k}{n-k+E-1 \choose E-1}\\
=\sum_{k=2}^{E-1}\left(-1\right)^{k}{E-1 \choose k}{n-k+E-1 \choose n-k}=(-1)^{n}\sum_{k=2}^{E-1}{E-1 \choose k}{-E \choose n-k}\,.
\end{gather*}
By Vandermonde's convolution

\begin{gather*}
\sum_{k=0}^{E-1}{E-1 \choose k}{-E \choose n-k}=\sum_{k=0}^{n}{E-1 \choose k}{-E \choose n-k}={-1 \choose n}=(-1)^{n}.
\end{gather*}
Therefore

\begin{gather*}
\beta_{n}^{E}=1-{n+E-1 \choose E-1}+{n+E-2 \choose E-1}\left(E-1\right).
\end{gather*}
Notice that ${n+E-1 \choose E-1}={n+E-2 \choose E-1}+{n+E-2 \choose E-2}$
and thus

\begin{gather}
\beta_{n}^{E}={n+E-2 \choose E-1}\left(E-2\right)-{n+E-2 \choose E-2}+1.\label{eq:star-n}
\end{gather}

\begin{figure}[h]
\begin{center}~~~~~~~~~\includegraphics[scale=0.6]{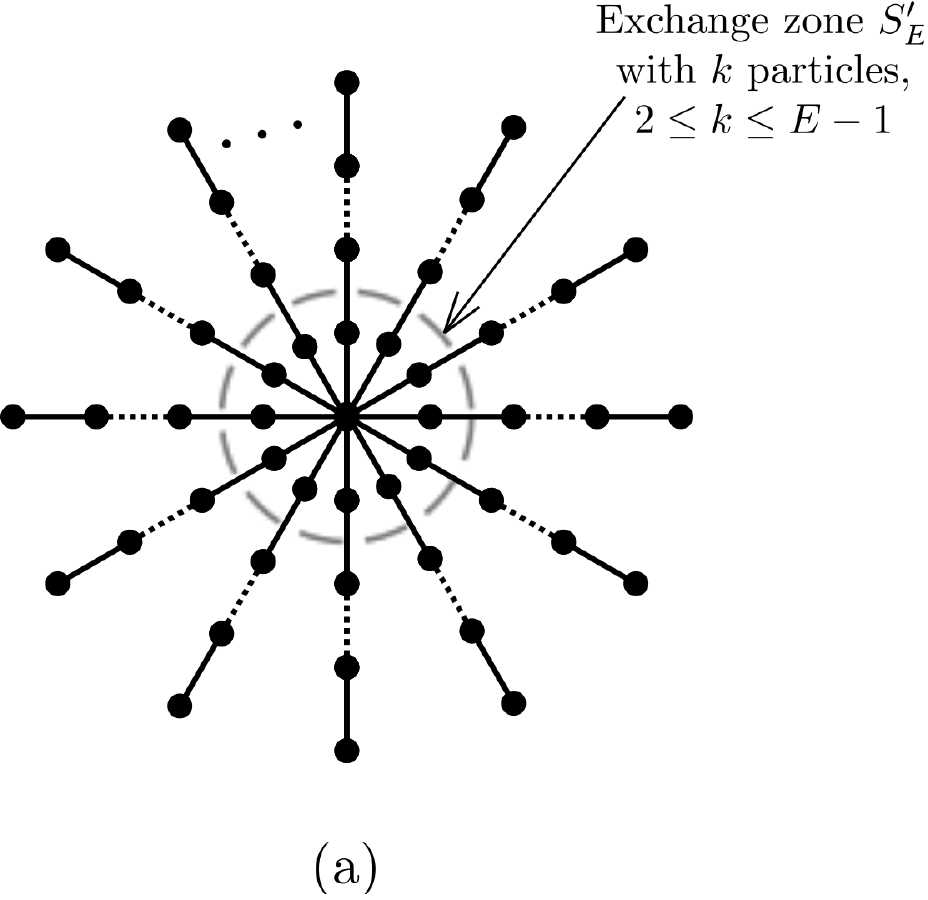}~~~~~~~~~~~\includegraphics[scale=0.6]{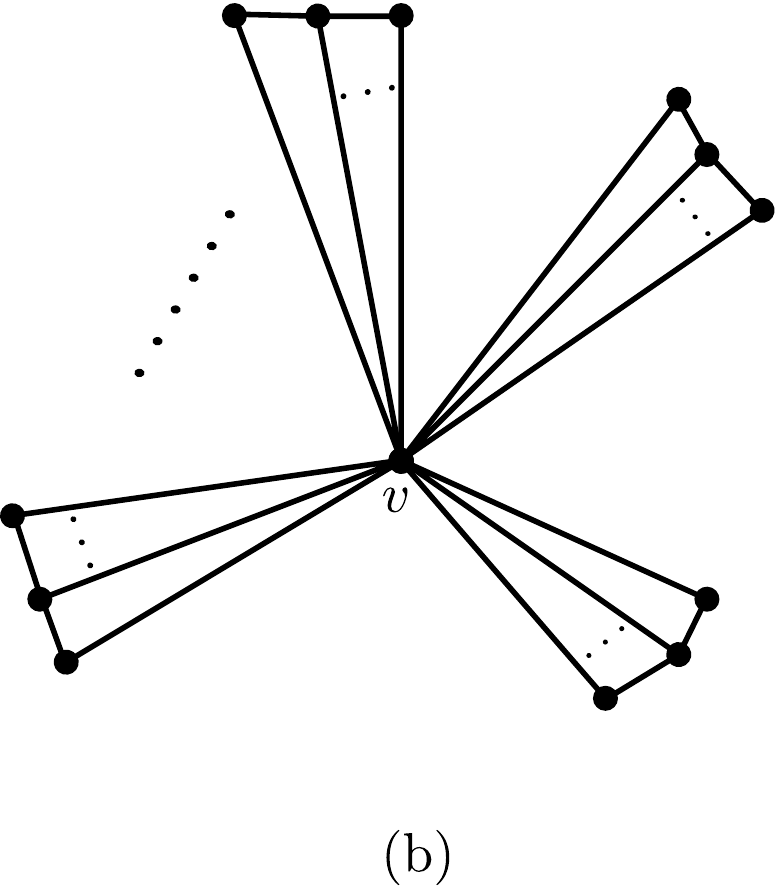}\end{center}

\caption{\label{fig:The-star and fan}(a) The star graph with $E$ arms and
$n$ particles. Each arm has $n$ vertices. The exchange zone $S_{E}^{\prime}$
can accommodate $2$, $3$,...,$E-1$ particles. (b) The fan graph
$F$.}
\end{figure}

\noindent Note finally that in contrast with $2$-connected graphs,
formula (\ref{eq:star-n}) indicates a strong dependence of the
quantum statistics on the number of particles, $n$.

\subsection{The fan graphs }

Following the argument presented in section \ref{sub:One-connected-graphs}
in order to treat a one-vertex cut $v$ we need to count the number
of the independent Y-phases which are lost due to the removal of $v$. As in Section \ref{sub:One-connected-graphs}, let $\mu = \mu(v)$ denote the number of connected components following the deletion of $v$, and denote these components by $\Gamma_1, \ldots, \Gamma_\mu$.  For Y-cycles with edges in three distinct components, the number  of independent phases, $\beta_n^\mu$, is given by the expression (\ref{eq:star-n}) for star graphs,
\begin{gather}
\beta_{n}^{\mu}={n+\mu-2 \choose \mu-1}\left(\mu-2\right)-{n+\mu-2 \choose \mu-2}+1.\label{eq:fan-star}
\end{gather}
We must also determine the number of independent Y-cycles with two edges in the same component $\Gamma_{i}$, denoted  $\gamma_n(v)$ .

Let us first consider a simple example, namely
 the graphs shown in figures \ref{fig:fan-ex}(a)
and \ref{fig:fan-ex}(b). Assume there are three particles. We 
calculate $\gamma_3(v)$ as follows. 
The $Y$ subgraphs we are interested in are denoted
by dashed lines and are $Y_{1}$ and $Y_{2}$ respectively. Note that
each of them contributes three phases corresponding to different
positions of the third particle $\{\phi_{Y_{1}}^{A},\phi_{Y_{1}}^{B},\phi_{Y_{1}}^{C},\phi_{Y_{2}}^{A},\phi_{Y_{2}}^{B},\phi_{Y_{2}}^{C}\}$.
They are, however, not independent. To see this, note that using Lemma
\ref{aspect1-1} we can write
\begin{gather*}
\phi_{c,3}=\phi_{Y_{1}}^{A}+\phi_{Y_{1}}^{B}+\phi_{c,1}^{B,B^{\prime}},\,\,\,\phi_{c,3}=\phi_{Y_{2}}^{A}+\phi_{Y_{2}}^{C}+\phi_{c,1}^{C,C^{\prime}}\,,\\
\phi_{c,2}^{B}=\phi_{Y_{1}}^{B}+\phi_{c,1}^{B,B^{\prime}},\,\,\,\phi_{c,2}^{B}=\phi_{Y_{2}}^{B}+\phi_{c,1}^{B,C}\,,\\
\phi_{c,2}^{C}=\phi_{Y_{1}}^{C}+\phi_{c,1}^{B,C},\,\,\,\phi_{c,2}^{C}=\phi_{Y_{2}}^{C}+\phi_{c,1}^{C,C^{\prime}}\,.
\end{gather*}
The phase $\phi_{c,3}$ is not lost when $v$ is cut. 
On the other hand, the five phases
\begin{gather}
\{\phi_{c,1}^{C,C^{\prime}},\,\phi_{c,1}^{B,B^{\prime}},\,\phi_{c,1}^{B,C},\,\phi_{c,2}^{B},\,\phi_{c,2}^{C}\},
\end{gather}
are lost. The knowledge of them and $\phi_{c}^{3}$ determines all
six $\phi_{Y}$ phases. Therefore, $\gamma_3(v)$ is the number of $1$ and $2$-particle exchanges on cycle
$c$ (which is $5$) rather than the number of $Y$ phases (which
is $6$).

\begin{figure}[h]
\begin{center}~~~~~~~~~~~~\includegraphics[scale=0.6]{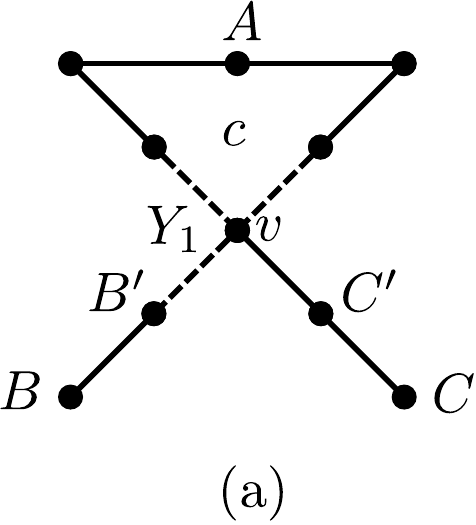}~~~~~~~~~~~~~~~~~~\includegraphics[scale=0.6]{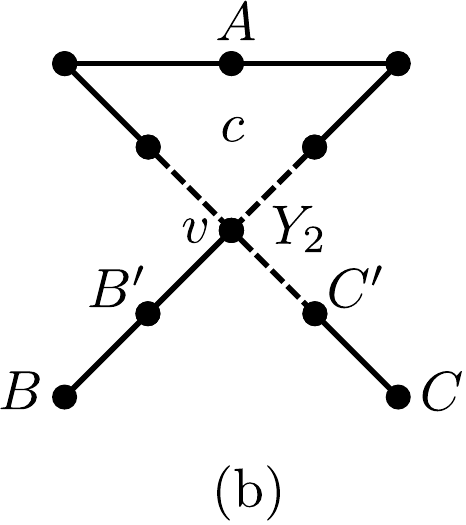}\end{center}

\caption{\label{fig:fan-ex}The $Y$ subgraphs (a) $Y_{1}$ and (b)$Y_{2}$.}
\end{figure}

For the general case, let $\nu_i$ denote the number of edges at $v$ which belong to $\Gamma_i$.  Since the $\Gamma_i$ are connected, there exist $\nu_i - 1$ independent cycles in $\Gamma_i$ which connect these edges.  Denote these by $C_{i,1},\ldots, C_{1,\nu_1-1}$.  Fan graphs (see Fig~\ref{fig:The-star and fan} (b)) provide the simplest realization.  Using arguments similar to those in the above example, one can show that Y-cycles with two edges in the same component  can be expressed in terms of two sets of cycles. The first set contains cycles which are wholly contained in just one of the connected components. These cycles are not lost when $v$ is cut, and therefore do not contribute to $\gamma_n(v)$.  The second type of cycle is characterised as follows:  Consider a partition $\{n_i\}_{i=1}^{\mu}$ of the particles amongst the components $\Gamma_i$.  For each partition, we can construct cycles where all of the particles in $\Gamma_i$ -- assuming $\Gamma_i$ contains at least one particle, i.e.~that $n_i > 0$ -- are taken to move once around  $C_{i,j}$ while the other particles remain fixed.  Excluding the cases in which all of the particles belong to a single component,
 the number of such cycles is given by the following sum over partitions $n_1 +\cdots + n_\mu = n$:
\[ \gamma_n(v) =  \sum_{\scriptstyle
n_1,\ldots,n_\mu = 0 
\\ \atop \scriptstyle
n_1 + \cdots + n_\mu = n}^n \
\sum_{\scriptstyle
i = 1   
\\ \atop \scriptstyle
0 < n_i < n}^\mu (\nu_i - 1).\]
Noting that
\[ \sum_{\scriptstyle
i = 1   
\\ \atop \scriptstyle
0 < n_i < n}^\mu  =
\sum_{i = 1}^\mu \ - \
\sum_{\scriptstyle
i = 1   
\\ \atop \scriptstyle  n_i = 0}^\mu
\ - \
\sum_{\scriptstyle
i = 1  
\\ \atop \scriptstyle n_i = n}^\mu\]
and $\sum_{i = 1}^\mu (\nu_i - 1) = \nu - \mu$,
we readily obtain
\[  \gamma_n(v) = \left( \binom{n+\mu -1}{n} - \binom{n+\mu-2}{n} - 1\right)(\nu-\mu) = \left( \binom{n+\mu -2}{n-1} - 1\right)(\nu-\mu).\]
Hence the  number of the phases lost when $v$ is cut is given by
\begin{gather}
N_{1}(v,n)= \beta^\mu_n + \gamma_n(v) = {n+\mu-2 \choose \mu-1}\left(\nu-2\right)-{n+\mu-2 \choose \mu-2}-\left(\nu-\mu-1\right).\label{eq:one-connected-n}
\end{gather}

\paragraph{The final formula for $H_1(\mathcal{D}^{n}(\Gamma))$}

By the repeated application of the one-vertex cuts the resulting components
of $\Gamma$ become finally $2$-connected graphs. Let $v_{1},\ldots,v_{l}$
be the set of cut vertices such that components $\Gamma_{v_{i},k}$
are $2$-connected. Making use of formula (\ref{eq:2-connected})
we write

\begin{gather}
H_{1}(\mathcal{D}^{n}(\Gamma))=\mathbb{Z}^{\beta(\Gamma)+N_{1}+N_{2}+N_{3}}\oplus\mathbb{Z}_{2}^{N_{3}^{\prime}},\label{eq:2-particle-final-1}
\end{gather}
 where $N_{1}=\sum_{i}N_{1}(v_{i},n)$, the coefficients $N_{1}(v_i,n)$
are given by (\ref{eq:one-connected-n}) and $N_{2}$, $N_{3}$, $N_{3}^{\prime}$
are defined as in section \ref{sec:Two-particle-quantum-statistics}.

\section{Gauge potentials for $2$-connected graphs}

In this section we give a prescription for the $n$-particle topological gauge potential on $\mathcal{D}^n(\Gamma)$ in terms of the $2$-particle topological gauge potential. For $2$-connected graphs all choices of $n$-particle topological gauge potentials on $\mathcal{D}^n(\Gamma)$  are realized by this prescription. The discussion is divided into three parts: i) separation of a $2$-particle topological gauge potential into AB and quantum statistics components, ii) topological gauge potentials for 2-particles on a subdivided graph,  iii) $n$-particle topological gauge potentials.

We start with some relevant background. Assume as previously that $\Gamma$ is sufficiently subdivided. Recall that directed edges or $1$-cells of $\mathcal{D}^n(\Gamma)$ are of the form $v_1\times\ldots\times v_{n-1}\times e$ up to permutations, where $v_j$ are vertices of $\Gamma$ and $e=j\rightarrow k$ is an edge of $\Gamma$ whose endpoints are not $\{v_1,\ldots, v_{n-1}\}$. For simplicity we will use the following notation
\[
\{v_1,\ldots,v_{n-1},j\rightarrow k\}:=v_1\times\ldots\times v_{n-1}\times e.
\]
An $n$-particle gauge potential is a function $\Omega^{(n)}$ defined on the directed edges of $\mathcal{D}^n(\Gamma)$ with the values in $\mathbb{R}$ modulo $2\pi$ such that
\begin{equation}\label{asym}
\Omega^{(n)}(\{v_1,\ldots,v_{n-1},k\rightarrow j\})=-\Omega^{(n)}(\{v_1,\ldots,v_{n-1},j\rightarrow k\}).
\end{equation}
In order to define $\Omega$ on linear combinations of directed edges we extend (\ref{asym}) by linearity.

For a given gauge potential, $\Omega^{(n)}$ the sum of its values calculated on the directed edges of an oriented cycle $C$ will be called the flux of $\Omega$ through $C$ and denoted $\Omega(C)$. Two gauge potentials $\Omega_1^{(n)}$ and $\Omega_2^{(n)}$ are called equivalent if for any oriented cycle $C$ the fluxes $\Omega_1^{(n)}(C)$ and $\Omega_2^{(n)}(C)$ are equal modulo $2\pi$.

The $n$-particle gauge potential $\Omega^{(n)}$ is called a {\it topological gauge potential} if for any contractible oriented cycle $C$ in $\mathcal{D}^n(\Gamma)$ the flux $\Omega^{(n)}(C)=0\,\mathrm{mod}\,2\pi$. It is thus clear that equivalence classes of topological gauge potentials are in 1-1 correspondence with the equivalence classes in $H_1(\mathcal{D}^n(\Gamma))$.

\paragraph{Pure Aharonov-Bohm  and pure quantum statistics topological gauge potentials}
Let $\Gamma$ be a graph with $V$ vertices. We say that a 2-particle gauge potential $\Omega^{(2)}_{AB}$ is a {\it pure Aharonov-Bohm  gauge potential} if and only if
\begin{equation}
\label{ eq: AB gauge potential}
\Omega^{(2)}_{AB}(\{i, j\rightarrow k\}) = \omega^{(1)}(j \rightarrow k), \text { for all distinct vertices $i,j,k$ of $\Gamma$}.
\end{equation}
Here $\omega^{(1)}$ can be regarded as a gauge potential on $\Gamma$. Thus, for a pure AB gauge potential, the phase associated with one particle moving from $j$ to $k$ does not depend on where the other particle is. We say that a 2-particle gauge potential $\Omega^{(2)}_{S}$ is a {\it pure statistics gauge potential} if and only if
\begin{equation}
\label{ eq: AB gauge potential}
\sum_{\scriptstyle i\atop \scriptstyle i \neq j,k} \Omega^{(2)}_{S}(\{i, j\rightarrow k\}) = 0, \text { for all distinct vertices $j,k$ of $G$}.
\end{equation}
That is, the phase associated with one particle moving from $j$ to $k$ averaged over all possible positions of the other particle is zero.  It is clear that an arbitrary gauge potential $\Omega^{(2)}$ has a unique decomposition into a pure AB and pure statistics gauge potentials, i.e.
\begin{equation}
\label{eq: decomposition }
\Omega^{(2)} = \Omega^{(2)}_{AB} + \Omega^{(2)}_{S},
\end{equation}
where
\begin{equation}
\label{eq: decomposition 2}
 \Omega^{(2)}_{AB}(\{i, j\rightarrow k\}) = \frac{1}{V-2} \sum_{\scriptstyle p\atop \scriptstyle p \neq j,k} \Omega^{(2)}(\{p,j\rightarrow k\}), \quad \Omega^{(2)}_S =\Omega^{(2)} - \Omega^{(2)}_{AB}.
\end{equation}
It is straightforward to verify that if $\Omega^{(2)}$ is a topological gauge potential, then so are $\Omega^{(2)}_{AB}$ and $\Omega^{(2)}_S$, and vice versa. Moreover, one can easily check that $\Omega_{AB}^{(2)}$ vanishes on any Y-cycle of $\mathcal{D}^2(\Gamma)$. Note, however, that for a given cycle $C$ of $\Gamma$ the AB-phase, $\phi^{v}_{C,1}$ considered in the previous sections is not $\Omega^{(2)}_{AB}(v\times C)$ but rather $\Omega^{(2)}(v\times C)$ as AB-phases can depend on the position of the stationary particle.

\paragraph{Gauge potential for a subdivided 2-particle graph}\label{sec: gauge potential}
Let $\Gbar$ be a graph with vertices $\Vcalbar = \{1,\ldots, \Vbar\}$. Let $\Omegabar^{(2)}$ be a gauge potential on $\mathcal{D}^2(\Gbar)$.

We assume that $\Omegabar^{(2)}$ is topological, that is, for every pair of disjoint edges of $\Gbar$, $i\leftrightarrow k$ and $j\leftrightarrow l$ we have
\begin{equation}
\label{eq: relation 0 }
\Omegabar^{(2)}(i,j\rightarrow l) + \Omegabar^{(2)}(l,i\rightarrow k) +\Omegabar^{(2)}(k,l\rightarrow j) + \Omegabar^{(2)}(j,k\rightarrow i) = 0.
\end{equation}
Assume we add a vertex to $\Gbar$ by subdividing an edge.  Let $p$ and $q$ denote the vertices of this edge, and denote the new graph by $\Gamma$ and  the added vertex by $a$. Since subdividing an edge does not change the topology of a graph, it is clear that we can find a gauge potential, $\Omega^{(2)}$, on $\mathcal{D}^2(\Gamma)$ that is, in some sense, equivalent to $\Omegabar^{(2)}$.


For the sake of  completeness, we first give a precise definition of what it means for gauge potentials on $\mathcal{D}^2(\Gamma)$ and $\mathcal{D}^2(\Gbar)$ to be equivalent. Given a path $\Cbar$ on $\mathcal{D}^2(\Gbar)$, we can construct a path $P$ on $\mathcal{D}^2(\Gamma)$ by making the replacements
\begin{align}\label{eq: subs C_0 to C}
\{i,p\rightarrow q\} &\mapsto \{i,p\rightarrow a\rightarrow q\},\nonumber\\
\{i,q\rightarrow p\} &\mapsto \{i,q\rightarrow a\rightarrow p\}.
\end{align}
Similarly, given a path $P$ on $\mathcal{D}^2(\Gamma)$  we can construct a path $\Cbar$ on $\mathcal{D}^2(\Gbar)$ by making the following substitutions:
\begin{align}\label{eq: subs C to C_0}
\{i,p\rightarrow a\rightarrow p\} &\mapsto  \{i,p\},\nonumber\\
\{i,p\rightarrow a\rightarrow q\} &\mapsto \{i,p\rightarrow q\},\nonumber\\
\{i,q\rightarrow a\rightarrow p\} &\mapsto \{i,q\rightarrow p\},\nonumber\\
\{i,q\rightarrow a\rightarrow q\} &\mapsto \{i,q\}.
\end{align}
We say that $\Omega^{(2)}$ and $\Omegabar^{(2)}$ are equivalent if
\begin{equation}
\label{eq: equiv Omega Omega^0 }
 \Omega^{(2)}(P) = \Omegabar^{(2)}(\Cbar)
\end{equation}
whenever $P$ and $\Cbar$ are related as above.

Next we give an explicit prescription for $\Omega^{(2)}$.  For edges in $\mathcal{D}^2(\Gamma)$ that do not involve vertices on the subdivided edge, we take $\Omega^{(2)}$ to coincide with $\Omegabar^{(2)}$.  That is,  for $i, j, k$ all distinct from $p,a,q$, we take
\begin{equation}
\label{eq: Omega ijk }
\Omega^{(2)}(\{i,j\rightarrow k\}) = \Omegabar^{(2)}(\{i,j\rightarrow k\}).
\end{equation}
As $p$ and $q$ are not adjacent on $\Gamma$, we take
\begin{equation}
\label{eq: Omega ijk }
\Omega^{(2)}(\{i,p\rightarrow q\}) = 0.
\end{equation}
For edges on $\mathcal{D}^2(\Gamma)$ involving the subdivided segments $p\rightarrow a$ and $a \rightarrow q$,  we require that $\Omega^{(2)}(\{i,p\rightarrow a\})$ and  $\Omega^{(2)}(\{i,a\rightarrow q\})$ add up to give the  phase  $\Omegabar^{(2)}(i,p\rightarrow q)$ on the original edge.  The partitioning of the original phase between the subdivided segments amounts to a choice of gauge.  For definiteness, we will take the phases on the two halves of the subdivided edge to be the same, so that
\begin{equation} \label{eq: Omega p,q-> a}
\Omega^{(2)}(\{i, p\rightarrow a\})  =
\Omega^{(2)}(\{i, a\rightarrow q\})  = \frac12 \Omegabar^{(2)}(\{i,p\rightarrow q\}). 
\end{equation}

It remains to determine $\Omega^{(2)}$ for edges of $C_2(G)$ on which the stationary particle sits at the new vertex $a$.  This follows from requiring that $\Omega^{(2)}$ satisfy the relations
\begin{align}
\label{eq: relations}
\Omega^{(2)}(\{a,i\rightarrow j\}) + \Omega^{(2)}(\{j,a\rightarrow p\}) + \Omega^{(2)}(\{p,j\rightarrow i\}) + \Omega^{(2)}(\{i,p\rightarrow a\}) &= 0,\nonumber\\
\Omega^{(2)}(\{a,i\rightarrow j\}) + \Omega^{(2)}(\{j,a\rightarrow q\}) + \Omega^{(2)}(\{q,j\rightarrow i\}) + \Omega^{(2)}(\{i,q\rightarrow a\}) &= 0.
\end{align}
From  (\ref{eq: Omega p,q-> a}) and the antisymmetry property $\Omega^{(2)}(\{i,j\rightarrow k\}) = -\Omega(\{i,k\rightarrow j\})$, along with the relations \eqref{eq: relation 0 } satisfied by $\Omegabar^{(2)}$, it follows that these conditions are equivalent, and both are satisfied by taking
\begin{equation}\label{eq: Omega a i-> j}
    \Omega^{(2)}(a, i\rightarrow j)  = \half \left(\Omegabar^{(2)}(p,i\rightarrow j) + \Omegabar^{(2)}(q,i\rightarrow j)\right).
\end{equation}
Finally, when $i$ or $j$ coincide with one of the vertices $p$ or $q$ the expression should be
\begin{equation}\label{eq: Omega a p-> j}
    \Omega^{(2)}(\{a, q\rightarrow j\})  = \left(\Omegabar^{(2)}(\{p,q\rightarrow j\}) + \half \Omegabar^{(2)}(\{j,q\rightarrow p\})\right).
\end{equation}
It is then straightforward to verify that $\Omega^{(2)}(P) = \Omegabar^{(2)}(\Cbar)$ whenever $P$ and $\Cbar$ are related as in \eqref{eq: subs C_0 to C} and \eqref{eq: subs C to C_0} and that $\Omega^{(2)}$ is a topological gauge potential.

\paragraph{Construction of $n$-particle topological gauge potential}

Let $\Omegabar^{(2)}$ be a gauge potential on $\mathcal{D}^2(\Gbar)$.  By repeatedly applying the procedure from the previous paragraph, we can construct an equivalent gauge potential $\Omega^{(2)}$ on $\mathcal{D}^2(\Gamma)$, where $\Gamma$ is a sufficiently subdivided version of $\Gbar$, in which $n-2$ vertices are added to each edge of $\Gbar$.
We resolve $\Omega^{(2)}$ into its AB and statistics components $\Omega^{(2)}_{AB}$ and $\Omega^{(2)}_S$, as in \eqref{eq: decomposition }.  Suppose the pure AB component is described by the gauge potential $\omega^{(1)}$ on $\Gamma$.  We define the $n$-particle gauge potential, $\Omega^{(n)} $, on $\mathcal{D}^{n}(\Gamma)$ as follows. Given $(n+1)$ vertices of $\Gamma$, denoted $\{ v_1,\ldots, v_{n-1},i,j\}$, with $i\sim j$, we take
\begin{equation}
\label{eq: Omega^n }
\Omega^{(n)} \left(
\{ v_1,\ldots,v_{n-1}, i\rightarrow j\}
\right)
= \omega^{(1)}(i\rightarrow j) +
\sum_{r=1}^{n-1} \Omega^{(2)}_S(\{v_r, i\rightarrow j\}).
\end{equation}
That is, the phase associated with the one-particle move $i\rightarrow j$ is the sum of the AB-phase $\omega^{(1)}(i,j)$ and the two-particle statistics phases $\Omega_S^{(2)}(\{v_r, i\rightarrow j\})$ summed over the positions of the other particles.

Given that $\Omega^{(2)}$ is a topological gauge potential, let us verify that $\Omega^{(n)} $ is a topological gauge potential. Let $i\rightarrow k$ and $j\rightarrow l$ be distinct edges of $\Gamma$, and let $\{v_1,\ldots, v_{n-2}\}$ denote $(n-2)$ vertices of $\Gamma$ that  are distinct from $i$, $j$, $k$, $l$.  We need to verify if
\begin{gather*}
\Omega^{(n)} \left( \{v_1,\ldots, v_{n-2},i, j\rightarrow l\}\right) +
\Omega^{(n)} \left( \{v_1,\ldots, v_{n-2},l, i\rightarrow k\}\right) + \\
+\Omega^{(n)} \left( \{v_1,\ldots, v_{n-2},k, l\rightarrow j\}\right) +
\Omega^{(n)} \left( \{v_1,\ldots, v_{n-2},j, k\rightarrow i\}\right)=0.
\end{gather*}
Using \eqref{eq: Omega^n } it reduces to
\begin{small}
\begin{gather*}
\omega^{(1)} (i\rightarrow k) + \omega^{(1)} (k\rightarrow i) + \omega^{(1)} (j \rightarrow l) + \omega^{(1)} (l\rightarrow k) + \\
+
\left( \sum_{r = 1}^{n-2} \Omega^{(2)}_S(\{v_r, j\rightarrow l\}) +  \Omega^{(2)}_S(\{i, j\rightarrow l\}) \right) +
\left( \sum_{r = 1}^{n-2} \Omega^{(2)}_S(\{v_r, i\rightarrow k\}) +  \Omega^{(2)}_S(\{l, i\rightarrow k\}) \right)+ \\
+\left( \sum_{r = 1}^{n-2} \Omega^{(2)}_S(\{v_r, l\rightarrow j\}) +  \Omega^{(2)}_S(\{k, l\rightarrow j\}) \right) +
\left( \sum_{r = 1}^{n-2} \Omega^{(2)}_S(\{v_r, k\rightarrow i\}) +  \Omega^{(2)}_S(\{j, k\rightarrow i\}) \right). \\
\end{gather*}
\end{small}
Next, using the antisymmetry property $\Omega^{(2)}_S(\{v_r, i\rightarrow k\}) = -\Omega^{(2)}_S(\{v_r, k\rightarrow i\})$ and the fact that $\Omega^{(2)}_S$ is a topological gauge potential we get\begin{small}
\begin{gather*}
\sum_{r=1}^{n-2}  \left(\Omega^{(2)}_S(\{v_r, j\rightarrow l\}) + \Omega^{(2)}_S(\{v_r, l\rightarrow j\})\right) +
\left(\Omega^{(2)}_S(\{v_r, i\rightarrow k\}) + \Omega^{(2)}_S(\{v_r, k\rightarrow i\})\right) +\\
+ \Omega^{(2)}_S(\{i, j\rightarrow l\}) +  \Omega^{(2)}_S(\{l, i\rightarrow k\}) +  \Omega^{(2)}_S(\{k, l\rightarrow j\}) + \Omega^{(2)}_S(\{j, k\rightarrow i\})  = 0.
\end{gather*}\end{small}

Therefore, the gauge potential defined by \eqref{eq: Omega^n } is topological. Equivalence classes of n-particle topological gauge potentials are essentially elements of the first homology group $H_1(\mathcal{D}^2(\Gamma))$. By Theorem~\ref{thm: 5} the equivalence classes in $H_1(\mathcal{D}^n(\Gamma))$ are in 1-1 correspondence with equivalence classes in $H_1(\mathcal{D}^2(\Gamma))$. Hence, for $2$-connected graphs all choices of $n$-particle topological gauge potential on $\mathcal{D}^n(\Gamma)$  can be realized  by \eqref{eq: Omega^n }. Finally, note  that, as explained in \cite{JHJKJR}, having an $n$-particle topological gauge potential one can easily construct a tight-binding Hamiltonian which supports quantum statistics represented by it (see \cite{JHJKJR} for more details).

\section{Morse theory argument} \label{MT}

We present an argument which shows the $n$-particle cycles given in sections \ref{sub:An-over-complete-basis} and
\ref{sub:An-over-complete-basis-2} form an over-complete spanning set of the first homology group $H_{1}(\mathcal{D}^{n}(\Gamma))$.  The argument follows the characterization of the fundamental group using discrete Morse theory by Farley and Sabalka \cite{FS05,FS08,FS12} or alternatively the characterization of the discrete Morse function for the $n$-particle graph \cite{S12}.  Here, however, we present the central idea in a way that does not assume a familiarity with discrete Morse theory in order to remain accessible.  For a rigorous proof we refer to the articles cited above.

Given a sufficiently subdivided graph $\Gamma$ we identify some maximal spanning subtree $T$ in $\Gamma$; $T$ is obtained by omitting exactly $\beta_1 (\Gamma)$ of the edges in $\Gamma$ such that $T$ remains connected but contains no loops.  The tree can then be drawn in the plane to fix an orientation.  A single vertex of degree $1$ in $T$ is identified as the root and the vertices of $T$ are labeled $1,2,\dots,|V|$ starting with $1$ for the root and labeling each vertex in turn traveling from the root around the boundary of $T$ clockwise, see figure \ref{fig:graphtotree}.

\begin{figure}[tbh]
\begin{center}
\includegraphics[scale=1.3]{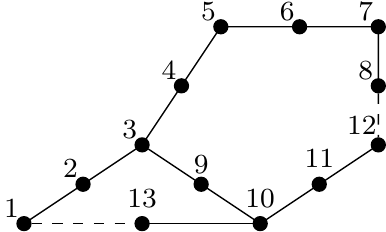}
\end{center}
\caption{\label{fig:graphtotree}A sufficiently subdivided graph for $3$ particles, edges in a maximal spanning tree are shown with solid lines and edges omitted to obtain the tree are shown with dashed lines.  Vertices are labeled following the boundary of the tree clockwise from the root vertex $1$.}
\end{figure}

To characterize a spanning set of $n$-particle cycles for the first homology group we fix a root configuration $\vx_0=\{1,2,\dots,n\}$ where the particles are lined up as close to the root as possible, see figure \ref{fig:Ybasis}(a).  The tree $T$ is used to establish a set of contractable paths between $n$-particle configurations on the graph (a discrete vector field). Given an $n$-particle configuration $\vx=\{v_1,\dots,v_n\}$ on the graph a path from $\vx$ to $\vx_0$ is a sequence of one-particle moves, where a single particle hops to an adjacent vacant vertex with the remaining $n-1$ particles remaining fixed.  This is a $1$-cell $\{v_1,\dots,v_{n-1},u\rightarrow v\}$ where $u$ and $v$ are the locations of the moving particle.  The labeling of the vertices in the tree provides a discrete vector field on the configuration space.  A particle moves according to the vector field if $n+1\rightarrow n$, i.e. the particle moves towards the root along the tree.  This allows a particle to move through a non-trivial vertex (a vertex of degree $\geq 3$) if the particle is coming from the direction clockwise from the direction of the root.  To define a flow that takes any configuration back to $\vx_0$ we also define a set of priorities at the non-trivial vertices that avoids $n$-particle paths crossing.  A particle may also move onto a non-trivial vertex $u$ according to the vector field if the $1$-cell $\{v_1,\dots,v_{n-1},u\rightarrow v\}$ does not contain a vertex $v_j$ with $v<v_j<u$; i.e. moving into a nontrivial vertex particles give way (yield) to the right.  So a particle can only move into the nontrivial vertex if there are no particles on branches of the graph between the branch the particle is on and the root direction clockwise from the root.  With this set of priorities it is clear that a path (sequence of $1$-cells) exists that takes any configuration $\vx$ to $\vx_0$ using only $1$-cells in the discrete vector field.  Equivalently by reversing the direction of edges in $1$-cells we can move particles from the reference configuration $\vx_0$ to any configuration $\vx$ against the flow.  As $n$-particle paths following this discrete flow do not cross these paths are contractible; equivalently, the phase around closed loops combining paths following and against the discrete flow is zero.  Note, we will describe paths either in the direction of the flow or against it as according to the vector field.

It remains to find a spanning set for the cycles that use $1$-cells not in the discrete vector field (that is, cells that are neither in the direction of the flow or against it).  We see now that there are only two types of $1$-cells that are excluded; those where the edge $u\leftrightarrow v$ is one of the $\beta_1 (\Gamma)$ edges omitted from $\Gamma$ to construct $T$, and those where a particle moves through a non-trivial vertex out of order - without giving way to the right.

We first consider a $1$-cell $c_{u\rightarrow v}=\{v_1, \dots,v_{n-1}, u\rightarrow v\}$ where $u\leftrightarrow v$ is an omitted edge.  Such a $1$-cell is naturally associated with a cycle where the particles move from $\vx_0$ to $\{v_1,\dots,v_{n-1},u\}$ against the flow, then follow $c_{u\rightarrow v}$ and finally move back from $\{v_1,\dots,v_{n-1},v\}$ to $\vx_0$ following the flow.  These $n$-particle cycles are typically the AB-cycles where one particle moves around a loop in $\Gamma$ with the other particles at a given configuration.  We saw in section \ref{sub:An-over-complete-basis} that while the phase associated with an AB-cycle can depend on the position of the other particles, these phases can be parameterized by only $\beta_1(\Gamma)$ independent parameters; one parameter for those cycles using each omitted edge.

We now consider, instead, cycles that include a $1$-cell $c=\{v_1, \dots,v_{n-1},u \rightarrow v)\}$ where a particle moves out of order at a nontrivial vertex.  Again each such $1$-cell is naturally associated to a cycle $C$ through $\vx_0$ where the particle moves according to the vector field except when it uses the $1$-cell $c$.  Such a cycle is shown in figure \ref{fig:Ybasis}.

\begin{figure}[tbh]
\begin{center}
\includegraphics[scale=1.1]{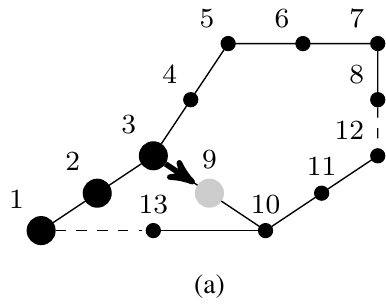}\includegraphics[scale=1.1]{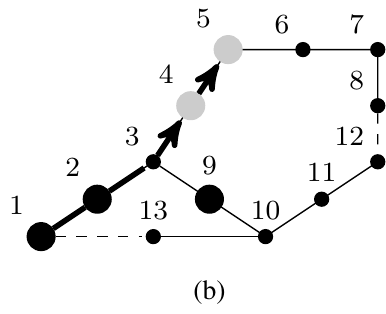}\includegraphics[scale=1.1]{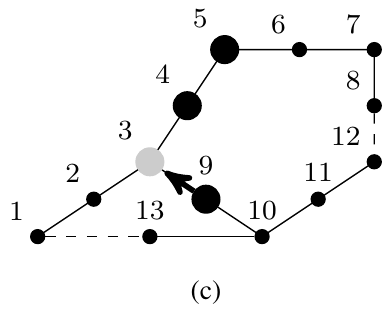}
\end{center}

\begin{center}
\includegraphics[scale=1.1]{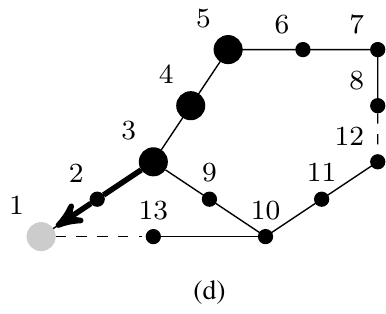}\includegraphics[scale=1.1]{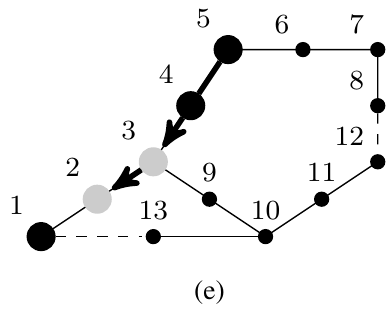}
\end{center}
\caption{\label{fig:Ybasis} An exchange cycle starting from the root configuration $\{1,2,3\}$ and using a single $1$-cell (c) that does not respect the flow at the non-trivial vertex $3$.  Large bold nodes indicate the initial positions of particles and light nodes their final positions.  In paths (a),(b),(d) and (e) particles move according to the vector field.}
\end{figure}

Such a cycle can be broken down into a product of $Y$-cycles in which pairs of particles are exchanged using three arms of the tree connected to the nontrivial vertex $v$ identified by $u,1$ and some $v_j$ where $v_j$ is a vertex in $c$ with $v<v_j<u$.  Figure \ref{fig:Ybasis2} shows a cycle homotopic to the cycle in figure \ref{fig:Ybasis} broken into the product of two $Y$-cycles; paths (a) through (c) and (d) through (e) respectively.  Notice that moving according to the vector field one returns from the initial configuration in figure \ref{fig:Ybasis2}(a) to the root configuration in figure \ref{fig:Ybasis}(a) and similarly one returns from the final configuration in figure \ref{fig:Ybasis2}(e) to the final configuration figure \ref{fig:Ybasis2}(d).  Then by contracting adjacent $1$-cells in the paths where the direction of the edge has been reversed it is straightforward to verify that the cycles in figures \ref{fig:Ybasis} and \ref{fig:Ybasis2} are indeed homotopic.

\begin{figure}[tbh]
\begin{center}
\includegraphics[scale=1.1]{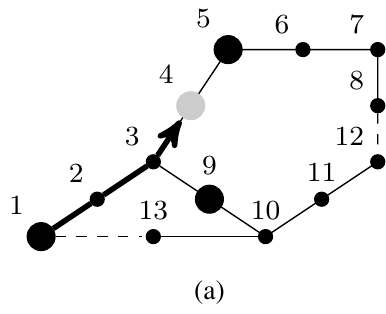}\includegraphics[scale=1.1]{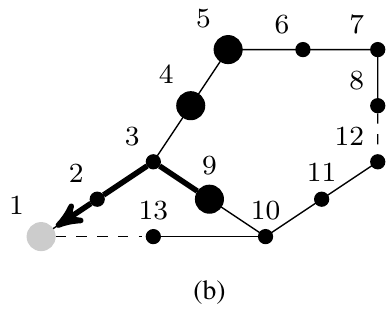}\includegraphics[scale=1.1]{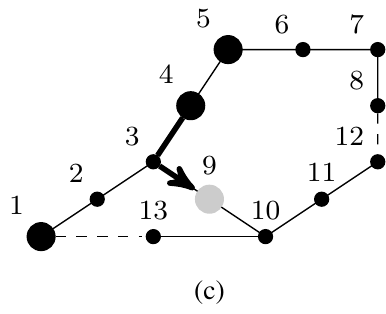}
\end{center}
\begin{center}
\includegraphics[scale=1.1]{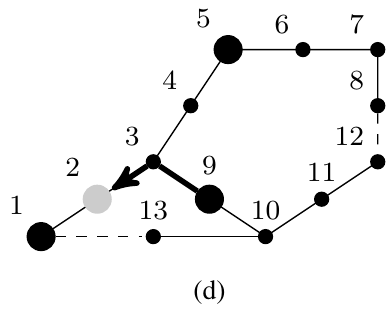}\includegraphics[scale=1.1]{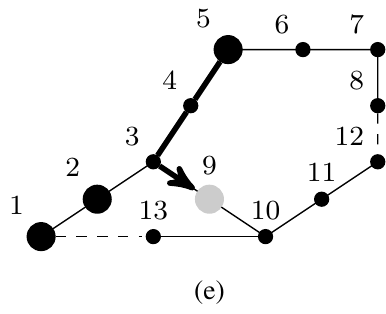}\includegraphics[scale=1.1]{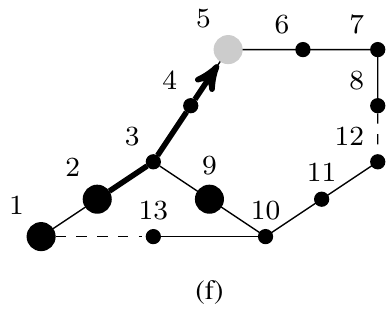}
\end{center}
\caption{\label{fig:Ybasis2} Examples of paths that form $Y$-cycles in the over-complete spanning set; large bold nodes indicate the initial positions of particles on the path and light nodes the final position a particle moves to. (a),(b) and (c) together form a $Y$-cycle, exchanging two particles at the non-trivial vertex $3$, similarly (c),(d) and (e) also form a $Y$-cycle.  Paths (a) through (e) together in order is a cycle homotopic to the exchange cycle starting from the root configuration shown in figure \ref{fig:Ybasis}.}
\end{figure}

Given a cycle $C$ from $\vx_0$ associated with a $1$-cell $c$ that does not respect the ordering at a nontrivial vertex to obtain a factorization of $C$ as a product of $Y$-cycles one need only start from $c$ and follow $C$ until it is necessary to move a third particle.  Instead of moving the third particle close the path to make a $Y$-cycle, which requires moving only one of the two particles moved so far.  Then retrace ones steps to rejoin $C$ and move the third particle through the nontrivial vertex again close a $Y$-cycle and repeat.  As any permutation can be written as the product of exchanges any such cycle $C$ can be factored as a product of $Y$-cycles.

Finally, as any $n$-particle cycle can be written as a closed sequence of $1$-cells and between $1$-cells we can add contractable paths according to the vector field without changing the phase associated with a cycle, we see that the AB-cycles and the cycles associated with $Y$ subgraphs centered at the nontrivial vertices form a spanning set for the $n$-particle cycles.  Clearly this spanning set will, in general, be over-complete as many relations between these cycles exist in a typical graph, in fact the full discrete Morse theory argument shows that all such relations are determined by critical $2$-cells \cite{FS05}.

\chapter{Discrete Morse functions for graph configuration spaces}
\section{Introduction}

In the last section of chapter \ref{chQS} some ideas of discrete Morse theory has been already introduced. In this chapter we present an alternative application of discrete Morse theory for two-particle graph configuration spaces. In contrast to previous constructions, which are based on discrete Morse vector fields, our approach is through Morse functions, which have a nice physical interpretation as two-body potentials constructed from one-body potentials. We also give a brief introduction to discrete Morse theory. 

Recently there has been significant progress in understanding topological properties of configuration spaces of many particles on metric graphs \cite{FS08,KP11}. This  was enabled by the foundational development of discrete Morse theory by Forman during  the late 1990's \cite{Forman98}. This theory reduces the calculation of homology groups to an essentially combinatorial problem, namely the construction of certain discrete Morse functions, or equivalently discrete gradient vector fields. Using this idea Farley and Sabalka \cite{FS08} gave a recipe for the construction of such a discrete gradient vector field \cite{FS08} on many-particle graphs and classified the first homology groups for tree graphs.
In 2011 Ko and Park \cite{KP11} significantly extended these results to arbitrary graphs by incorporating graph-theoretic theorems concerning the decomposition of a graph into its two and three-connected components.

In this chapter we give an alternative application of discrete Morse theory for two-particle graph configuration spaces. In contrast to the construction given in \cite{FS08}, which is based on discrete Morse vector fields, our approach is through discrete Morse functions. Our main goal is to provide an intuitive way of constructing a discrete Morse function and hence a discrete Morse gradient vector field. The central object of the construction is the `trial Morse' function. It may be understood as two-body potential constructed from one-body potential, a perspective which is perhaps more natural and intuitive from a physics point of view. Having a perfect Morse function\footnote{For the definition of a perfect Morse function see section \ref{sec:Classical-Morse-theory}.} $f_1$ on a graph $\Gamma$ we treat it as a one-body potential. The value of the trial Morse function at each point of a two-particle configuration space is the sum of the values of $f_{1}$ corresponding to the two particles positions in $\Gamma$. The trial Morse function is typically not a Morse function, i.e. it might not satisfy some of the relevant conditions. Nevertheless, we find that it is always possible to modify it and obtain a proper Morse function out of it. In fact, the trial Morse function is not `far' from being a Morse function and the number of cells at which it needs fixing is relatively small. Remarkably, this simple idea leads to similar results as those obtained in \cite{FS08}. We demonstrate it in Section \ref{sec:Main-example} by calculating two simple examples. We find that in both cases the trial Morse function has small defects which can be easily removed and a proper Morse function is obtained. The corresponding discrete Morse vector field is equivalent to the one stemming from the Farley and Sabalka method \cite{FS08}. As is shown in Section \ref{sec:General-consideration-for}, it is always possible to get rid of defects of the trial Morse function. The argument is rather technical. However, since the problem is of a certain combinatorial complexity we believe it cannot be easily simplified. We describe in details how the final result, i.e set of discrete Morse functions along with rules for identifying the critical cells and constructing the boundary map of the associated Morse complex, is built in stages from this simple idea. Our main purpose is hence to present an approach which we believe is conceptually simple and physically natural. It would be interesting to check if the presented constructions can give any simplification in understanding the results of \cite{KP11} but we do not pursue this here.

The chapter is organized as follows. In section \ref{sec:Morse-theory-in} we give a brief introduction to discrete Morse theory. Then in sections \ref{sec:One-particle-graph} and~\ref{sec:Main-example}, for two examples we present a definition of a `trial' Morse function $\tilde{f}_2$ for two-particle graph configuration space. We notice that the trial Morse function typically does not satisfy the conditions required of a Morse function according to Forman's theory. Nevertheless, we show in Section \ref{sec:General-consideration-for} that with small modifications, which we explicitly identify, the trial Morse function can be transformed into a proper Morse function $f_2$. In theorem \ref{theorem1} we give an explicit definition of the function $f_2$ and theorem \ref{theorem2} specifies its critical cells. Since the number of critical cells and hence the size of the associated Morse complex is small compared with the size of configuration space the calculation of homology groups are greatly simplified. The technical details of the proofs are given in the section \ref{proofs}. In section \ref{sec:topological gauge potentials} we discuss more specifically how the techniques of discrete Morse theory apply to the problem of quantum statistics on graphs.

\section{Morse theory in the nutshell\label{sec:Morse-theory-in}}

In this section we briefly present both classical and discrete Morse
theories. We focus on the similarities between them and illustrate
the ideas by several simple examples.

\subsection{Classical Morse theory}\label{sec:Classical-Morse-theory}

The concept of classical Morse theory is essentially very similar to its discrete counterpart. Since the former is better known we have found it beneficial to first discuss the classical version. A good reference is the monograph by
Milnor \cite{milnor}. Classical Morse theory is a useful tool
to describe topological properties of compact manifolds. Having such
a manifold $M$ we say that a smooth function $f:M\rightarrow\mathbb{R}$
is a Morse function if its Hessian matrix at every critical point
is nondegenerate, i.e.,
\begin{eqnarray}
df(x)=0\,\,\Rightarrow\,\,\mathrm{det}\left(\frac{\partial^{2}f}{\partial x_{i}\partial x_{j}}\right)(x)\neq0.
\end{eqnarray}
It can be shown that if $M$ is compact then $f$ has a finite number
of isolated critical points \cite{milnor}. The classical Morse theory
is based on the following two facts:
\begin{enumerate}
\item Let $M_{c}=\{x\in M\,:\, f(x)\leq c\}$ denote a sub level set of $f$.  Then $M_{c}$  is homotopy equivalent to $M_{c^{\prime}}$ if there is no critical value\footnote{A critical value of $f$  is the value of $f$ at one of its critical points.} between
the interval $(c,c\prime)$.
\item The change in topology when $M_{c}$ goes through a critical value
is determined by the index (i.e., the number of negative eigenvalues) of the Hessian
matrix at the associated critical point.
\end{enumerate}
The central point of classical Morse theory are the so-called Morse
inequalities, which relate the Betti numbers $\beta_{k}=\mathrm{dim}H_{k}(M)$, i.e. the dimensions of k-homology groups \cite{Hatcher}, to the numbers $m_{k}$ of critical points of index $k$, i.e.,
\begin{eqnarray}
\sum_{k}m_{k}t^{k}-\sum_{k}\beta_{k}t^{k}=(1+t)\sum_{k}q_{k}t^{k},\label{eq:Morse_ineq}
\end{eqnarray}
where $q_{k}\geq0$ and $t$ is an arbitrary real number. In particular
(\ref{eq:Morse_ineq}) implies that $\beta_{k}\leq m_{k}$. The function
$f$ is called a perfect Morse function iff $\beta_{k}=m_{k}$ for every
$k$. Since there is no general prescription it is typically hard to find a perfect Morse
function for a given manifold $M$. In fact a perfect Morse function may even not exist \cite{Ayala11}. However, even if $f$ is not perfect we can still encode the topological properties of $M$
in a quite small cell complex. Namely it follows from Morse theory that given a Morse function $f$,
one can show that $M$ is homotopic to a cell complex with $m_k$ $k$-cells,
and the gluing maps can be constructed in terms of the gradient paths of $f$.
We will not discuss this as it is far more complicated than in the discrete case.

\subsection{Discrete Morse function\label{sub:Discrete-Morse-function}}

In this section we discuss the concept of discrete Morse functions
for cell complexes as introduced by Forman \cite{Forman98}. Let $\alpha^{(p)}\in X$ denote a
$p$ - cell. A discrete Morse function on a regular cell complex $X$ is a function $f$
which assigns larger values to higher-dimensional cells with `local'
exceptions.
\begin{definition}
\label{Morse-fuction}A function $f\,:\, X\rightarrow\mathbb{R}$
is a discrete Morse function iff for every $\alpha^{(p)}\in X$ we
have
\begin{eqnarray}
\#\{\beta^{(p+1)}\supset\alpha\,:\, f(\beta)\leq f(\alpha)\}\leq1,\\
\#\{\beta^{(p-1)}\subset\alpha\,:\, f(\beta)\geq f(\alpha)\}\leq1.
\end{eqnarray}
\end{definition}
In other words, definition \ref{Morse-fuction} states that for any
$p$ - cell $\alpha^{(p)}$, there can be $\mathbf{at\,\, most}$
one $(p+1)$ - cell $\beta^{(p+1)}$ containing  $\alpha^{(p)}$ for which $f(\beta^{(p+1)})$ is less than or equal to  $f(\alpha^{(p)})$.  Similarly, there can be $\mathbf{at\,\, most}$ one $(p-1)$ - cell $\beta^{(p-1)}$ contained in  $\alpha^{(p)}$ for which $f(\beta^{(p-1)})$ is greater than or equal to  $f(\alpha^{(p)})$. Examples of a Morse function and a non-Morse function are shown in figure \ref{fig3}. The most important part of discrete Morse theory is the definition of a critical cell:
\begin{definition}
\label{criticalcell}A cell $\alpha^{(p)}$ is critical iff
\begin{eqnarray}
\#\{\beta^{(p+1)}\supset\alpha\,:\, f(\beta)\leq f(\alpha)\}=0,\,\,\mathrm{and}\\
\#\{\beta^{(p-1)}\subset\alpha\,:\, f(\beta)\geq f(\alpha)\}=0.
\end{eqnarray}
\end{definition}
That is,  $\alpha$ is critical if  $f(\alpha)$ is greater than the value of $f$ on all of the faces of $\alpha$, and $f(\alpha)$ is greater than the value of $f$ on all cells containing $\alpha$ as a face. From definitions
\ref{Morse-fuction} and \ref{criticalcell}, we get that a cell $\alpha$ is noncritical iff either
\begin{enumerate}
\item $\exists \ {\rm unique}\ \tau^{(p+1)}\supset\alpha\,\,\,\,\, \ {\rm with}\  f(\tau)\leq f(\alpha),$
or
\item $\exists \ {\rm unique}\  \beta^{(p-1)}\subset\alpha\,\,\,\,\,\ {\rm with}\  f(\beta)\geq f(\alpha).$
\end{enumerate}
It is quite important to understand that these two conditions cannot
be simultaneously fulfilled, as we now explain. Let us assume
on the contrary that both conditions (i) and (ii) hold. We have the following
sequence of cells:
\begin{eqnarray}
\tau^{(p+1)}\supset\alpha^{(p)}\supset\beta^{(p-1)}.
\end{eqnarray}
Since $\alpha^{(p)}$ is regular there is necessarily an $\tilde{\alpha}^{(p)}$
such that $\tau^{(p+1)}\supset\tilde{\alpha}^{(p)}\supset\beta^{(p-1)}$
(see figures \ref{fig2}(a),(b) for an intuitive explanation). Since $f(\tau)\leq f(\alpha)$,
by definition \ref{Morse-fuction} we have
\begin{eqnarray}
f(\tilde{\alpha})<f(\tau).
\end{eqnarray}
We also know that $f(\beta)\geq f(\alpha)$ which, once again by definition
\ref{Morse-fuction}, implies $f(\beta)<f(\tilde{\alpha})$. Summing
up we get
\begin{eqnarray}
f(\alpha)\leq f(\beta)< f(\tilde{\alpha})<f(\tau)\leq f(\alpha),
\end{eqnarray}
which is a contradiction.
\begin{figure}[h]
~~~~~~~~~~~~~~~~~\includegraphics[scale=0.5]{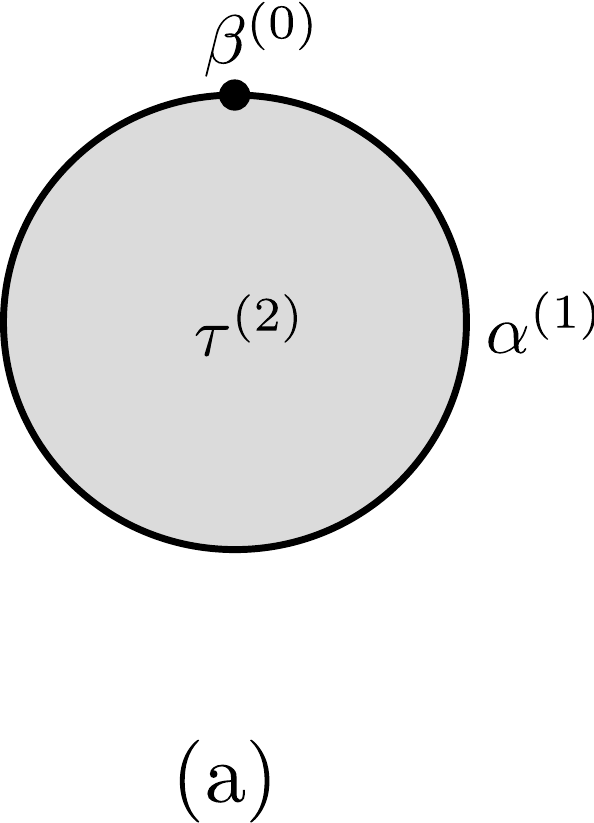}~~~~~~~~~~~~~~~~~~~~~~\includegraphics[scale=0.5]{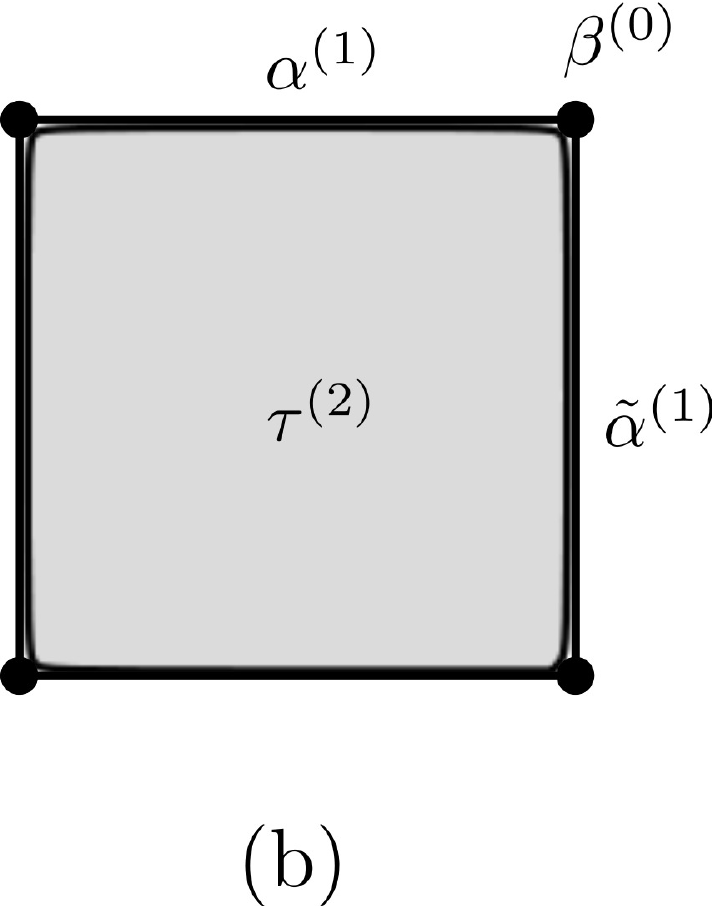}

\caption{\label{fig2} Examples of (a) an irregular cell complex.  $\alpha^{(1)}$ is an irregular
1 - cell and $\beta^{(0)}$ is an irregular face of $\alpha^{(1)}$. (b) A
regular cell complex with $\tau^{(2)}\supset\alpha^{(1)}\supset\beta^{(0)}.$}
\end{figure}

\noindent Following the path of classical Morse theory we define
next the level sub-complex $K(c)$ by
\begin{eqnarray}
K(c)=\cup_{f(\alpha)\leq c}\cup_{\beta\subseteq\alpha}\beta.\label{eq:levelsubcomplex}
\end{eqnarray}
That is,  $K(c)$ is the sub-complex containing all cells on which $f$
is less or equal to $c$, $\mathbf{{together\,\, with\,\, their\,\, faces}}$%
\footnote{Notice that the value of $f$ on some of these faces might be bigger than
$c$.%
}.
Notice that by definition (\ref{Morse-fuction}) a Morse function does not have to be a bijection.  However, we have the following  \cite{Forman98}:
\begin{lemma}\label{lem: Morse 1-1}
For any Morse function $f_1$, there exist another Morse function $f_2$ which
is 1-1 (injective) and which has the same critical cells as  $f_1$.
\end{lemma}
The process of  attaching cells is accompanied by two important lemmas
which describe the change in homotopy type of level sub-complexes
when critical or noncritical cells are attached. Since, from lemma~\ref{lem: Morse 1-1}, we can assume that a given Morse function is 1-1, we can always choose the intervals $[a,b]$
below so that $f^{-1}([a,b])$ contains exactly one cell.
%
%
\begin{lemma}
\label{lem:3} \cite{Forman98} If there are no critical cells $\alpha$
with $f(\alpha)\in[a,b]$, then $K(b)$ is homotopy equivalent to $K(a)$.
\end{lemma}

\begin{lemma}
\label{lem:4}\cite{Forman98} If there is a single critical cell $\alpha^{(p)}$
with $f(\alpha)\in[a,b]$, then $K(b)$ is homotopy equivalent to
\begin{eqnarray}
K(b)=K(a)\cup\alpha
\end{eqnarray}
and $\partial\alpha\subset K(a)$.
\end{lemma}
The above two lemmas lead to the following conclusion:
\begin{theorem}
\label{thm:1}\cite{Forman98} Let $X$ be a cell complex and $f\,:\, X\rightarrow\mathbb{R}$
be a Morse function. Then $X$ is homotopy equivalent to a cell complex
with exactly one cell of dimension $p$ for each critical cell $\alpha^{(p)}$

\begin{figure}[h]\label{figure3}
~~~~~~~~~~~~~~~~~\includegraphics[scale=0.5]{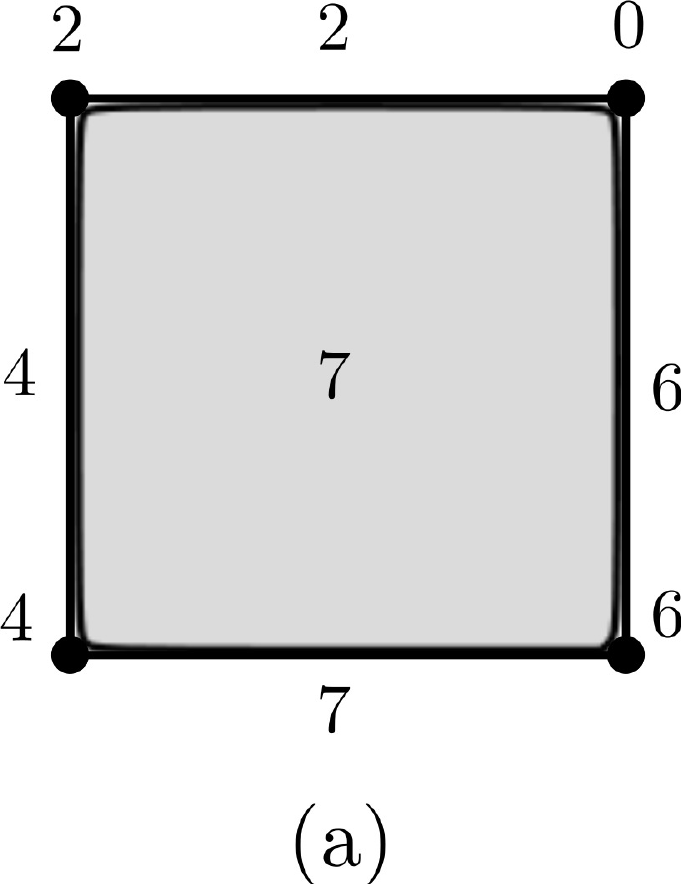}~~~~~~~~~~~~~~~~~~~~\includegraphics[scale=0.5]{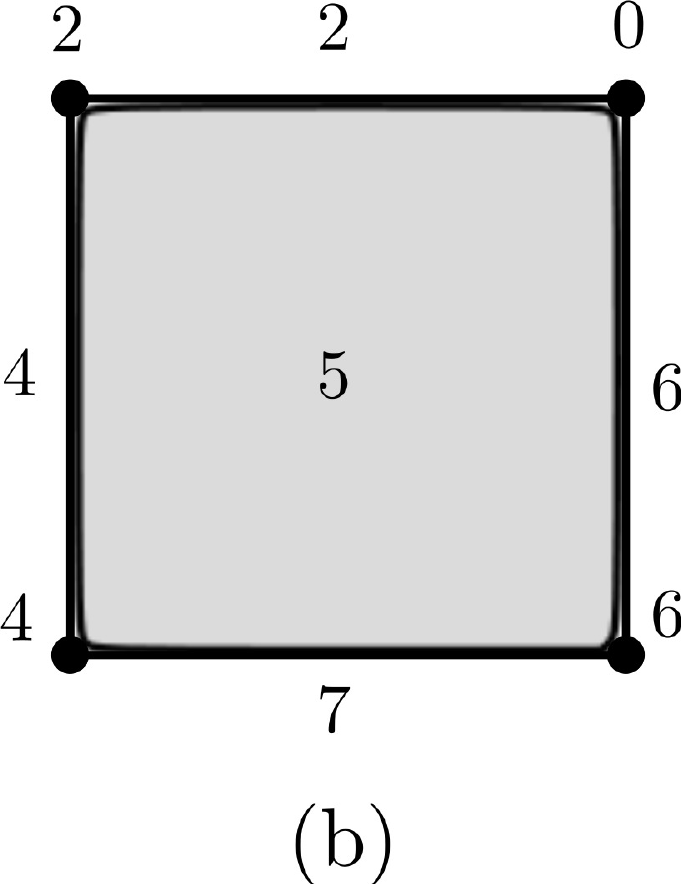}

\caption{\label{fig3} Examples of (a) a Morse function, and (b) a non-Morse function, since the 2-cell has value $5$ and there are two $1$-cells in its boundary
with higher values assigned ($6$, $7$).}
\end{figure}

\end{theorem}

\subsection{Discrete Morse vector field}
From theorem \ref{thm:1} it follows that a given cell complex is homotopy equivalent to a  cell complex containing only its critical cells, the so-called Morse complex. The construction of the Morse complex, in particular its boundary map (as well as the proof of theorem~\ref{thm:1}), depends crucially on the concept of a discrete vector field, which  we define next.
We know from definition \ref{Morse-fuction} that the noncritical
cells can be paired. If a $p$-cell is noncritical, then it is paired
with either the unique noncritical $(p+1)$-cell on which $f$ takes
an equal or smaller value, or the unique noncritical $(p-1)$-cell on
which $f$ takes an equal or larger value. In order to indicate this
pairing we draw an arrow from the $(p-1)$-cell to the $p$-cell in the first
case or from the $p$-cell to the $(p+1)$-cell in the second case (see figure
\ref{fig4}). Repeating this for all cells we get the so-called discrete gradient
vector field of the Morse function. It also follows from section \ref{sub:Discrete-Morse-function}
that for every cell $\alpha$ exactly one of the following is true:
\begin{enumerate}
\item $\alpha$ is the tail of one arrow,
\item $\alpha$ is the head of one arrow,
\item $\alpha$ is neither the tail nor the head of an arrow.
\end{enumerate}
Of course $\alpha$ is critical iff it is neither the tail nor the
head of an arrow. Assume now that we are given a collection of arrows
on some cell complex satisfying the above three conditions. The question
we would like to address is whether it is a gradient vector field
of some Morse function. In order to answer this question we need to
be more precise. We define
\begin{definition}
A discrete vector field $V$ on a cell complex $X$ is a collection
of pairs $\{\alpha^{(p)}\subset\beta^{(p+1)}\}$ of cells such that
each cell is in at most one pair of $V$.
\end{definition}
Having a vector field it is natural to consider its `integral lines'.
We define the $V$ - path as a sequence of cells
\begin{eqnarray}
\alpha_{0}^{(p)},\,\beta_{0}^{(p+1)},\,\alpha_{1}^{(p)},\,\beta_{1}^{(p+1)},\,\ldots,\alpha_{k}^{(p)},\,\beta_{k}^{(p+1)}\label{eq:Vpath}
\end{eqnarray}
such that$\{\alpha_{i}^{(p)}\subset\beta_{i}^{(p+1)}\}\in V$ and
$\beta_{i}^{(p+1)}\supset\alpha_{i+1}^{(p)}$. Assume now that $V$
is a gradient vector field of a discrete Morse function $f$ and consider
a $V$ - path (\ref{eq:Vpath}). Then of course we have
\begin{eqnarray}
f(\alpha_{0}^{(p)})\geq f(\beta_{0}^{(p+1)})>f(\alpha_{1}^{(p)})\geq f(\beta_{1}^{(p+1)})>\ldots>f(\alpha_{1}^{(p)})\geq f(\beta_{k}^{(p+1)}).
\end{eqnarray}
This implies that if $V$ is a gradient vector field of
the Morse function then $f$ decreases along any $V$-path which in
particular means that there are no closed $V$-paths. It happens that
the converse is also true, namely a discrete vector field $V$ is a gradient
vector field of some Morse function iff there are no closed $V$ -
paths \cite{Forman98}.

\begin{figure}[h]
~~~~~~~~~~~~~~~~~\includegraphics[scale=0.5]{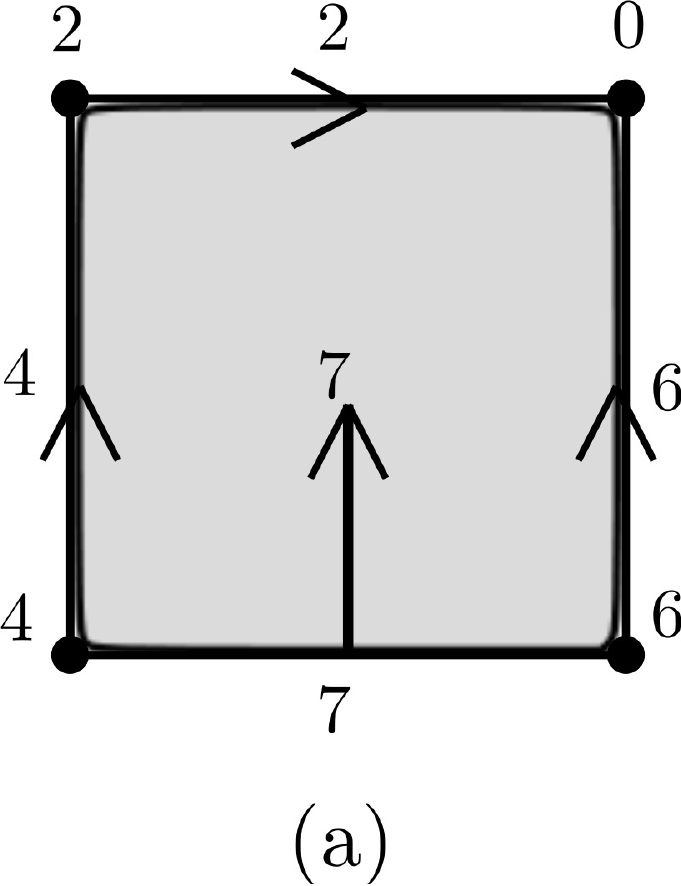}~~~~~~~~~~~~~~~~~~~~~~~~\includegraphics[scale=0.5]{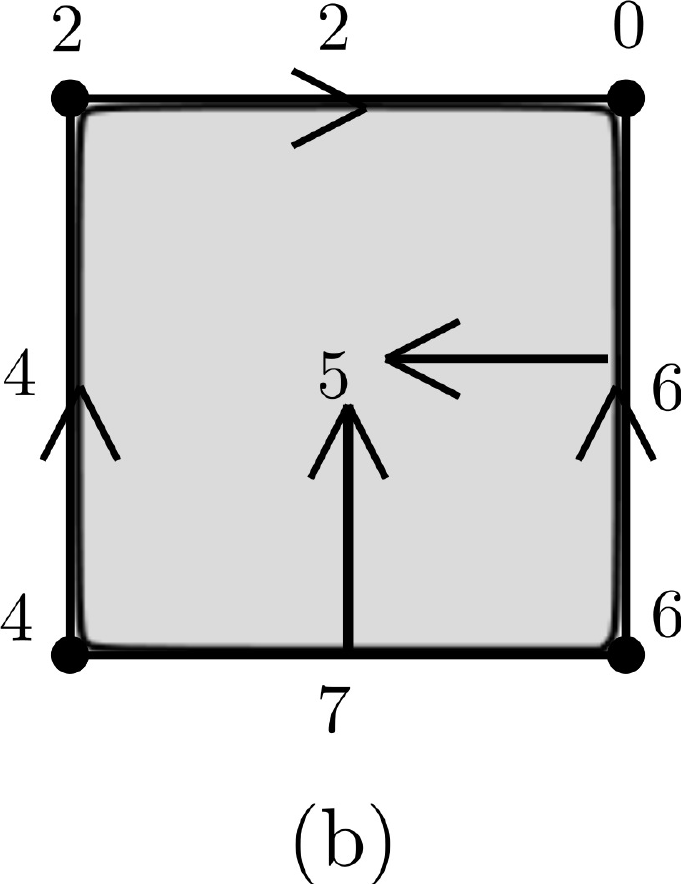}

\caption{\label{fig4} Examples of (a) a correct and (b) an incorrect discrete gradient vector fields; the 2-cell is the head of two arrows and the 1-cell is the head and tail of one arrow.}
\end{figure}

\subsection{The Morse complex}

Up to now we have learned how to reduce the number of cells of the original
cell complex to the critical ones. However, it is still not clear
how these cells are `glued' together, i.e. what is the boundary map
between the critical cells? The following result relates the concept of critical cells with discrete gradient vector fields \cite{Forman98}.
\begin{theorem}
Assume that orientation has been chosen for each cell in the cell complex
$X$. Then for any critical $(p+1)$-cell $\beta$ we have
\begin{eqnarray}
\tilde{\partial}\beta=\sum_{critical\,\alpha^{(p)}}c_{\beta,\alpha}\alpha,\label{eq:boundary}
\end{eqnarray}
where $\tilde{\ensuremath{\partial}}$ is the boundary map
in the cell complex consisting of the critical cells, whose existence
is guaranteed by theorem \ref{thm:1}, and
\begin{eqnarray}
c_{\beta,\alpha}=\sum_{\gamma\in P(\beta,\alpha)}m(\gamma),
\end{eqnarray}
where $P(\beta,\alpha)$ is the set of all $V$ - paths from
the boundary of $\beta$ to cells whose boundary contains $\alpha$ and
$m(\gamma)=\pm1$, depending on whether the orientation induced from
$\beta$ to $\alpha$ through $\gamma$ agrees with the one chosen for
$\alpha$.
\end{theorem}
The collection of critical cells together with the boundary map $\tilde{\partial}$ is called the Morse complex of the function $f$ and we will denote it by $M(f)$. Examples of the computation of boundary maps for Morse complexes will be given in section \ref{sec:Main-example}.

\section{A perfect Morse function on $\Gamma$ and its discrete vector field.\label{sec:One-particle-graph}}

In this section we present a construction of a perfect discrete Morse function
on a $1$ - particle graph. It is defined analogously as in the classical case, i.e. the number of critical cells in each dimension is equal to the corresponding dimension of the homology group. The existence of such a function will
be used in section \ref{sec:Main-example} to construct a `good' but not necessarily perfect
Morse function on a $2$-particle graph.

Let $\Gamma=(V\,,\, E)$ be a graph with $v=|V|$ vertices and $e=|E|$
edges. In the following we assume that $\Gamma$ is connected and
simple. Let $T$ be a spanning tree of $\Gamma$, i.e. $T$ is a connected
spanning subgraph of $\Gamma$ such that $V(T)=V(\Gamma)$ and for
any pair of vertices $v_{i}\neq v_{j}$ there is exactly one path
in $T$ joining $v_{i}$ with $v_{j}$. We naturally have $|E(\Gamma)|-|E(T)|\geq0$.
The Euler characteristic of $\Gamma$ treated as a cell
complex is given by
\begin{eqnarray}
\chi(\Gamma)=v-e=\mathrm{dim}H_{0}(\Gamma)-\mathrm{dim}H_{1}(\Gamma)=b_{0}-b_{1}.
\end{eqnarray}
Since $\Gamma$ is connected, $H_{0}(\Gamma)=\mathbb{Z}$. Hence we get
\begin{eqnarray}
b_{0}=1,\\
b_{1}=e-v+1.
\end{eqnarray}
On the other hand it is well known that $b_{1}=|E(\Gamma)|-|E(T)|$. Summing up from
the topological point of view $\Gamma$ is homotopy equivalent to
a wedge sum of $b_{1}$ circles. Our goal is to construct a perfect
Morse function $f_{1}$ on $\Gamma$, i.e. the one with exactly $b_{1}$
critical $1$ - cells and one critical $0$ - cell. To this end we choose a vertex $v_{1}$ of valency one in
$T$ (it always exists) and travel through the tree anticlockwise
from it labeling vertices by $v_{k}$. The value of $f$ on the vertex
$v_{k}$ is $f_{1}(k)=2k-2$ and the value of $f_{1}$ on the edge
$(i,j)\in T$ is $f_{1}((i,j))=\mathrm{max}\left(f_{1}(i),\, f_{1}(j)\right)$.
The last step is to define $f_{1}$ on the deleted edges $(i,j)\in E(\Gamma)\setminus E(T)$.
We choose $f_{1}((i,j))=\mathrm{max}(f_{1}(i),\, f_{1}(j))+2$, where
$v_{i},v_{j}$ are the boundary vertices of $(i,j)$. This way we obtain
that all vertices besides $v_{1}$ and all edges of $T$ are not critical
cells of $f_{1}$. The critical $1$ - cells are exactly the deleted edges.
The following example clarifies this idea (see figure \ref{fig5}).
\begin{example}
Consider the graph $\Gamma$ shown in figure \ref{fig5}(a). Its spanning tree
is denoted by solid lines and the deleted edges by dashed lines. For each vertex and
edge the corresponding value of a perfect discrete Morse function
$f_1$ is explicitly written. Notice that according to definition
\ref{criticalcell} we have exactly one critical $0$ - cell (denoted
by a square) and four critical $1$ - cells which are deleted
edges. The discrete vector field for $f_1$ is represented by arrows.
The contraction of $\Gamma$ along this field yields the contraction of $T$ to a single
point and hence the Morse complex $M(f_1)$ is the wedge sum of four circles (see figure \ref{fig5}(b))
\end{example}
\begin{figure}[H]
~~~~~~~~~~~\includegraphics[scale=0.43]{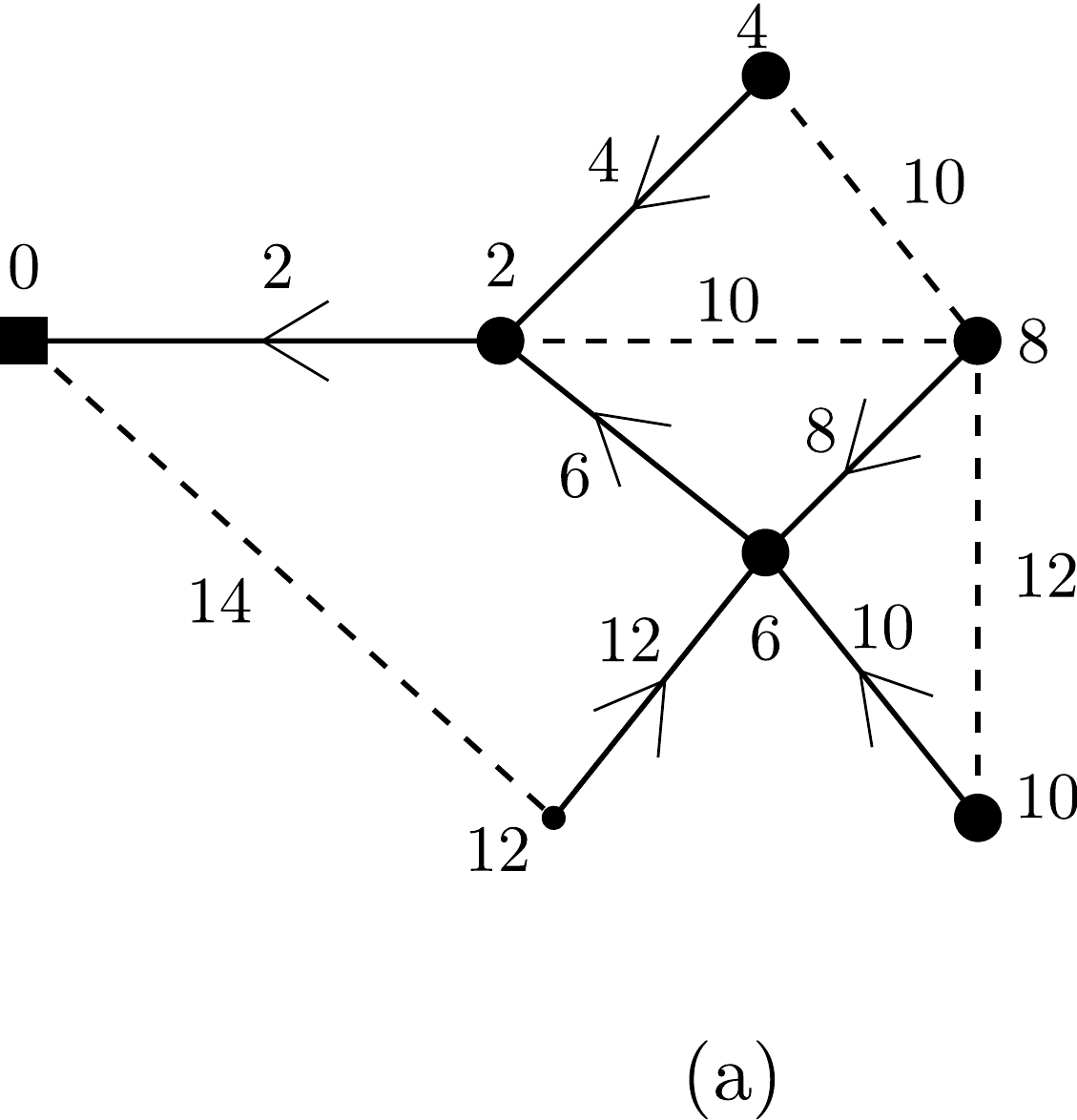}~~~~~~~~~~~~~~~~~\includegraphics[scale=0.5]{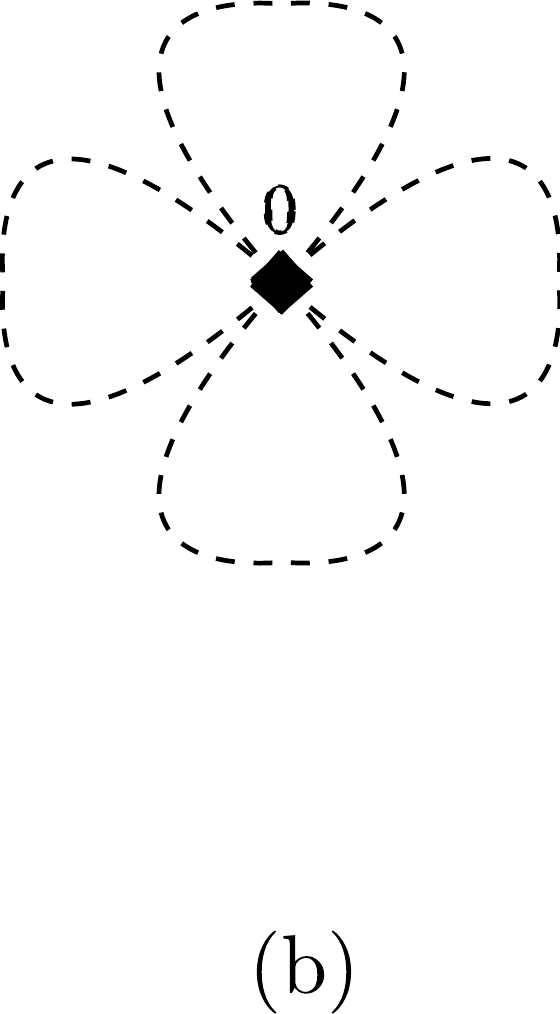}

\caption{\label{fig5} (a) The perfect discrete Morse function $f_{1}$ on the graph $\Gamma$
and its discrete gradient vector field. (b) The Morse complex $M(f_{1})$.}
\end{figure}

\section{The main examples \label{sec:Main-example}}

In this section we present a method of construction of a `good' Morse
function on the two particle configuration space $\mathcal{D}^2(\Gamma_{i})$
for two different graphs $\Gamma_{i}$ shown in figures \ref{fig6}(a) and \ref{fig8}(a).
We also demonstrate how to use the tools described in section \ref{sec:Morse-theory-in}
in order to derive a Morse complex and compute the first homology group.
We begin with a graph $\Gamma_{1}$ which we will refer to as lasso
(see figure \ref{fig6}(a)). The spanning tree of $\Gamma_{1}$ is denoted in
black in figure \ref{fig6}(a). In figure \ref{fig6}(b) we see an example of the perfect
Morse function $f_{1}$ on $\Gamma_{1}$ together with its gradient
vector field. They were constructed according to the procedure explained
in section \ref{sec:One-particle-graph}. The Morse complex of $\Gamma_{1}$
consists of one $0$-cell (the vertex $1$) and one $1$-cell (the
edge $(3,4)$).

\begin{figure}[h]
~~~~~~~~~~~~~\includegraphics[scale=0.4]{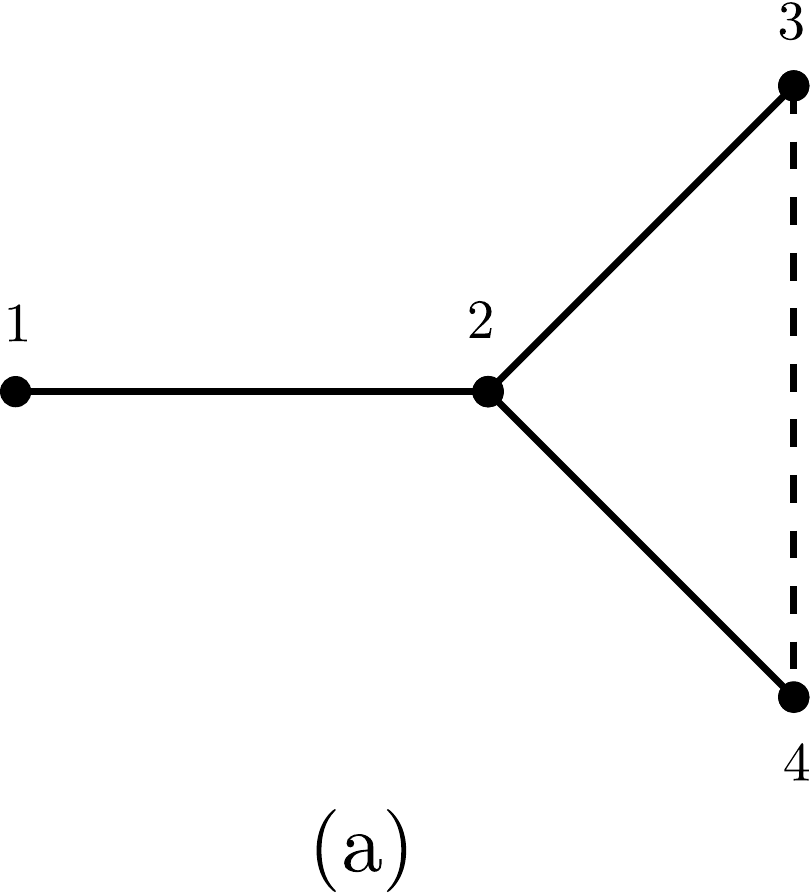}~~~~~~~~~~~~~~~~~~~~~~~~~~~~~~~~~~~~~\includegraphics[scale=0.4]{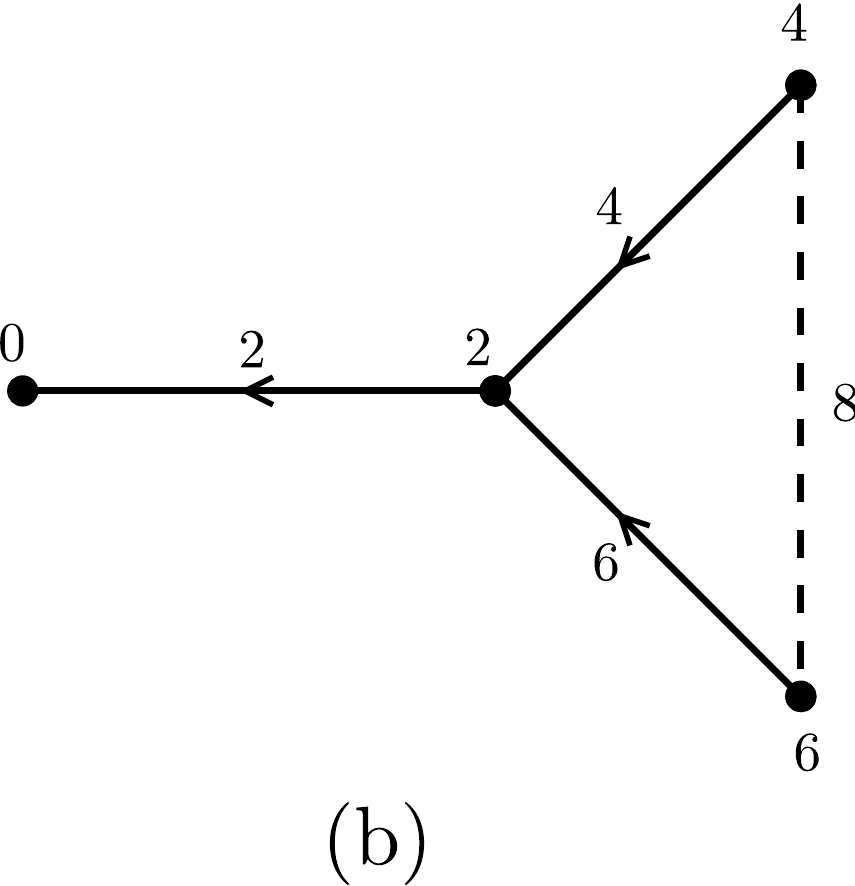}

\caption{\label{fig6} (a) One particle on lasso, (b) The perfect discrete Morse function
$f_{1}$}
\end{figure}

\noindent The two particle configuration space $\mathcal{D}^2(\Gamma_{1})$
is shown in figure \ref{fig7}(a). Notice that $\mathcal{D}^2(\Gamma_{1})$ consists
of one $2$ - cell $(3,4)\times(1,2)$\footnote{This notation should be understood as the Cartesian product of edges $(3,4)$ and $(1,2)$, hence a square.}, six $0$ - cells and eight
$1$ - cells. In order to define the Morse function $f_{2}$ on $\mathcal{D}^2(\Gamma_{1})$
we need to specify its value for each of these cells. We begin with
a trial function $\tilde{f}_{2}$ which is completely determined once
we know the perfect Morse function on $\Gamma_{1}$. To this end
we treat $f_{1}$ as a kind of `potential energy' of one particle.
The function $\tilde{f}_{2}$ is simply the sum of the energies of both
particles, i.e. the value of $\tilde{f}_{2}$ on a cell corresponding
to a particular position of two particles on $\Gamma_{1}$ is the
sum of the values of $f_{1}$ corresponding to this position. To
be more precise we have for
\begin{eqnarray}
\mathrm{0-cells:\,\,\,\,\,\,\,\,\,\,\,\,\,}\,\,\,\,\,\,\,\,\,\,\,\,\tilde{f}_{2}(i\times j) & = & f_{1}(i)+f_{1}(j),\nonumber \\
\mathrm{1-cells:}\,\,\,\,\,\,\,\,\,\,\,\,\tilde{f}_{2}\left(i\times(j,k)\right) & = & f_{1}(i)+f_{1}\left((j,k)\right),\nonumber \\
\mathrm{2-cells:}\,\,\,\tilde{f}_{2}\left((i,j)\times(k,l)\right) & = & f_{1}\left((i,j)\right)+f_{1}\left((k,l)\right).\label{eq:rules}
\end{eqnarray}
In figure \ref{fig7}(b) we can see $\mathcal{D}^2(\Gamma_{1})$ together with $\tilde{f}_{2}$.
Observe that $\tilde{f}_{2}$ is not a Morse function since the value
of $\tilde{f}_{2}\left((3,4)\right)$ is the same as the value of
$\tilde{f}_{2}$ on edges $4\times(2,3)$ and $3\times(2,4)$ which
are adjacent to the vertex $(3,4)$. The rule that $0$ - cell can
be the face of at most one $1$ - cell with smaller or equal value of $\tilde{f}_{2}$
is violated. In order to have Morse function $f_{2}$ on $\mathcal{D}^2(\Gamma_{1})$
we introduce one modification, namely
\begin{eqnarray}
f_{2}\left(3\times(2,4)\right)=\tilde{f}_{2}\left(3\times(2,4)\right)+1,\label{eq:mod2}
\end{eqnarray}
and $f_{2}$ is $\tilde{f_{2}}$ on the other cells.
\begin{figure}[h]
~~~~\includegraphics[scale=0.4]{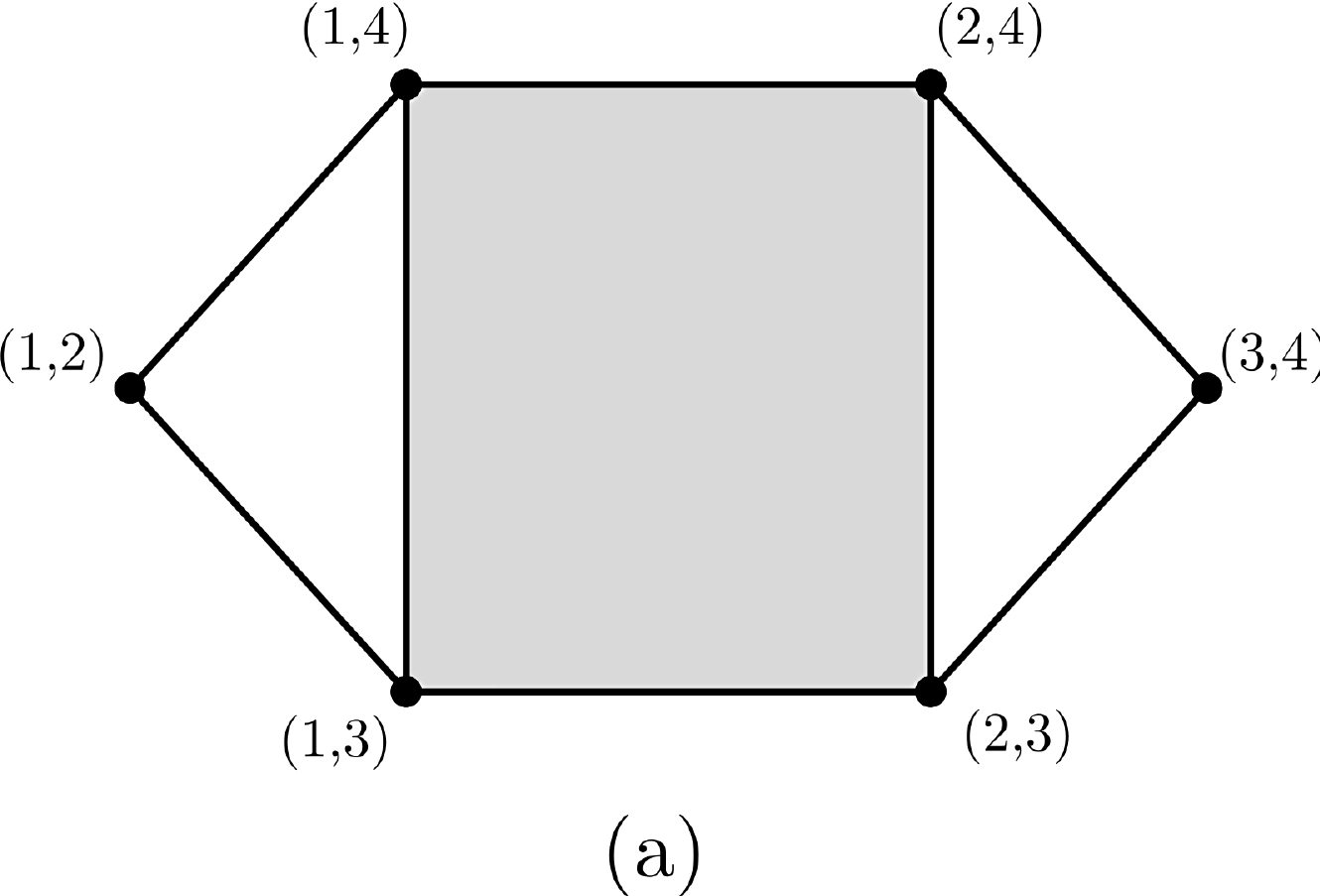}~~~~~~~~~\includegraphics[scale=0.4]{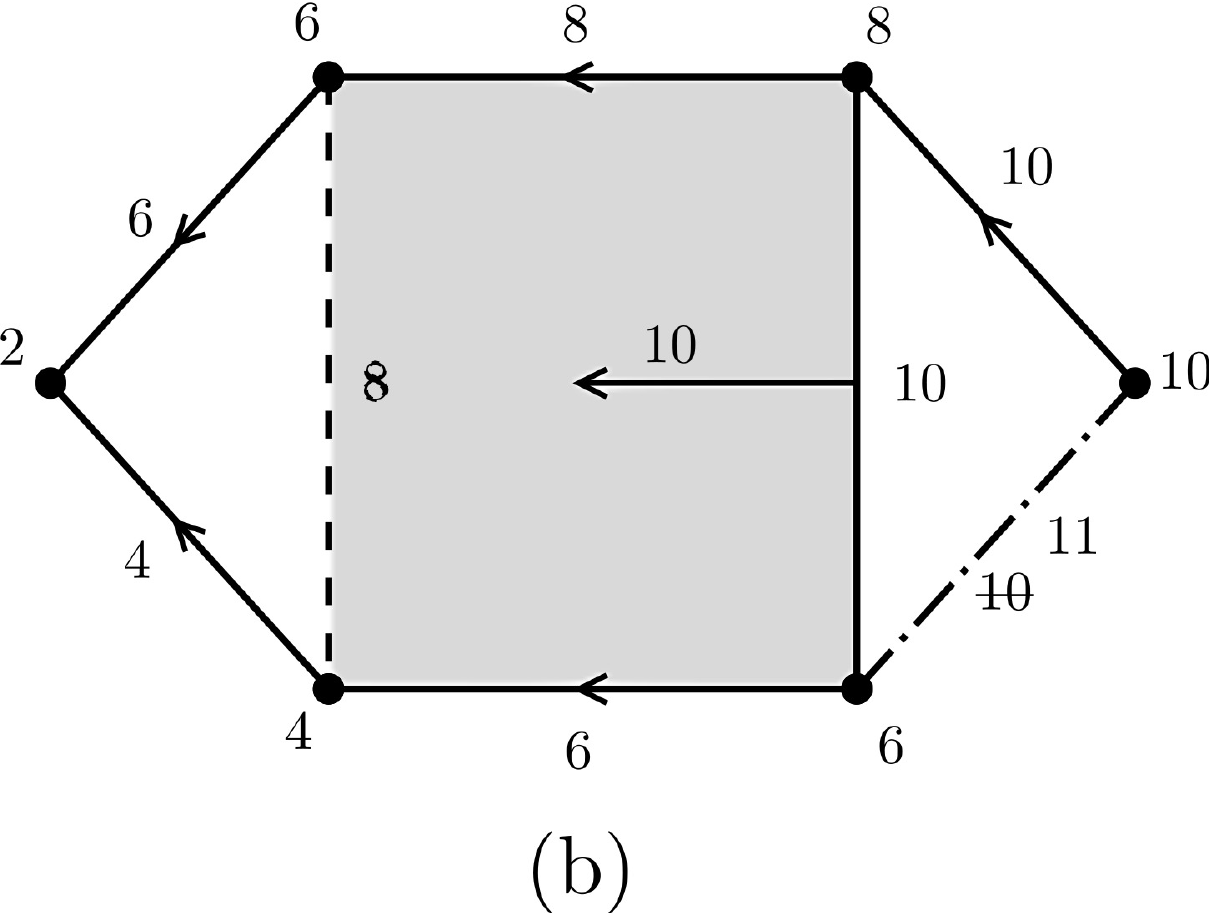}~~~~~~~~

\medskip{}

~~~~~~~~\includegraphics[scale=0.35]{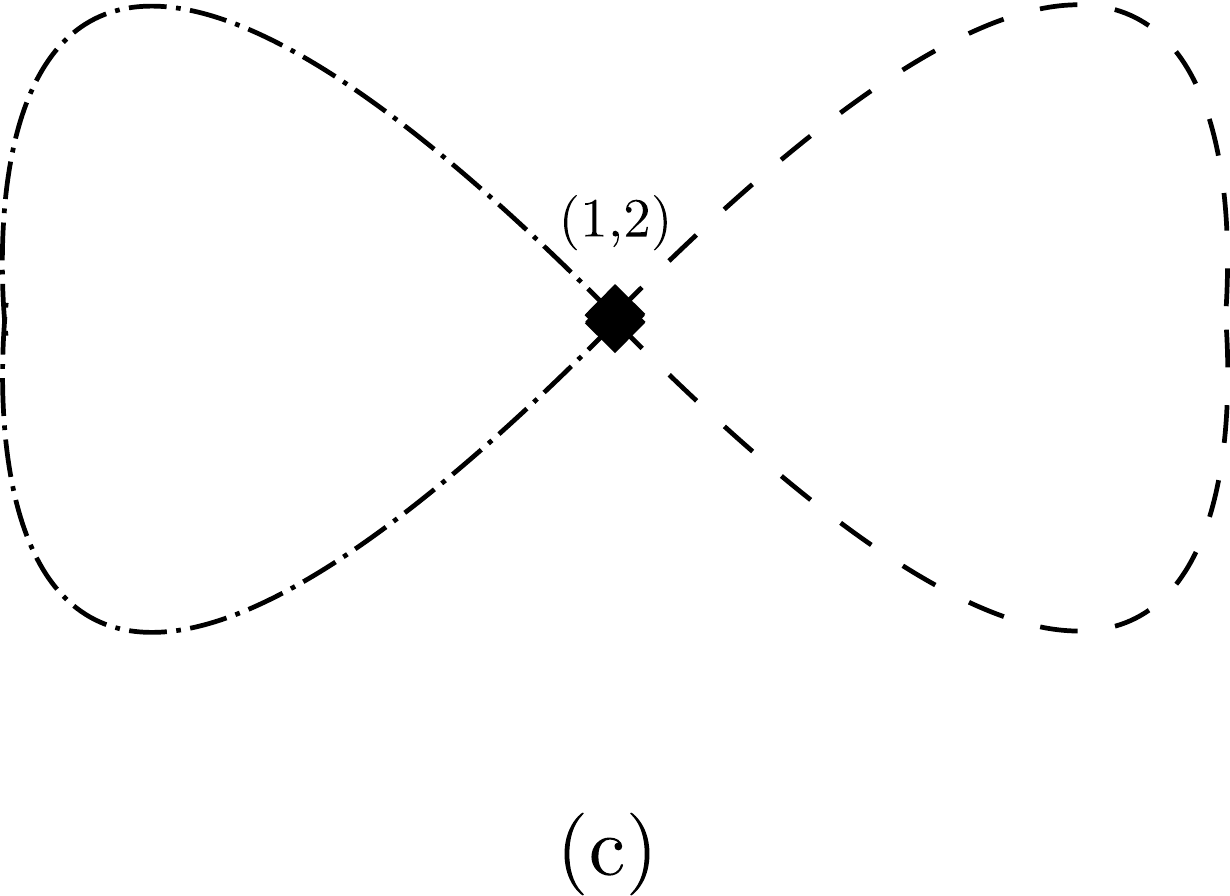}~~~~~~~~~~~\includegraphics[scale=0.4]{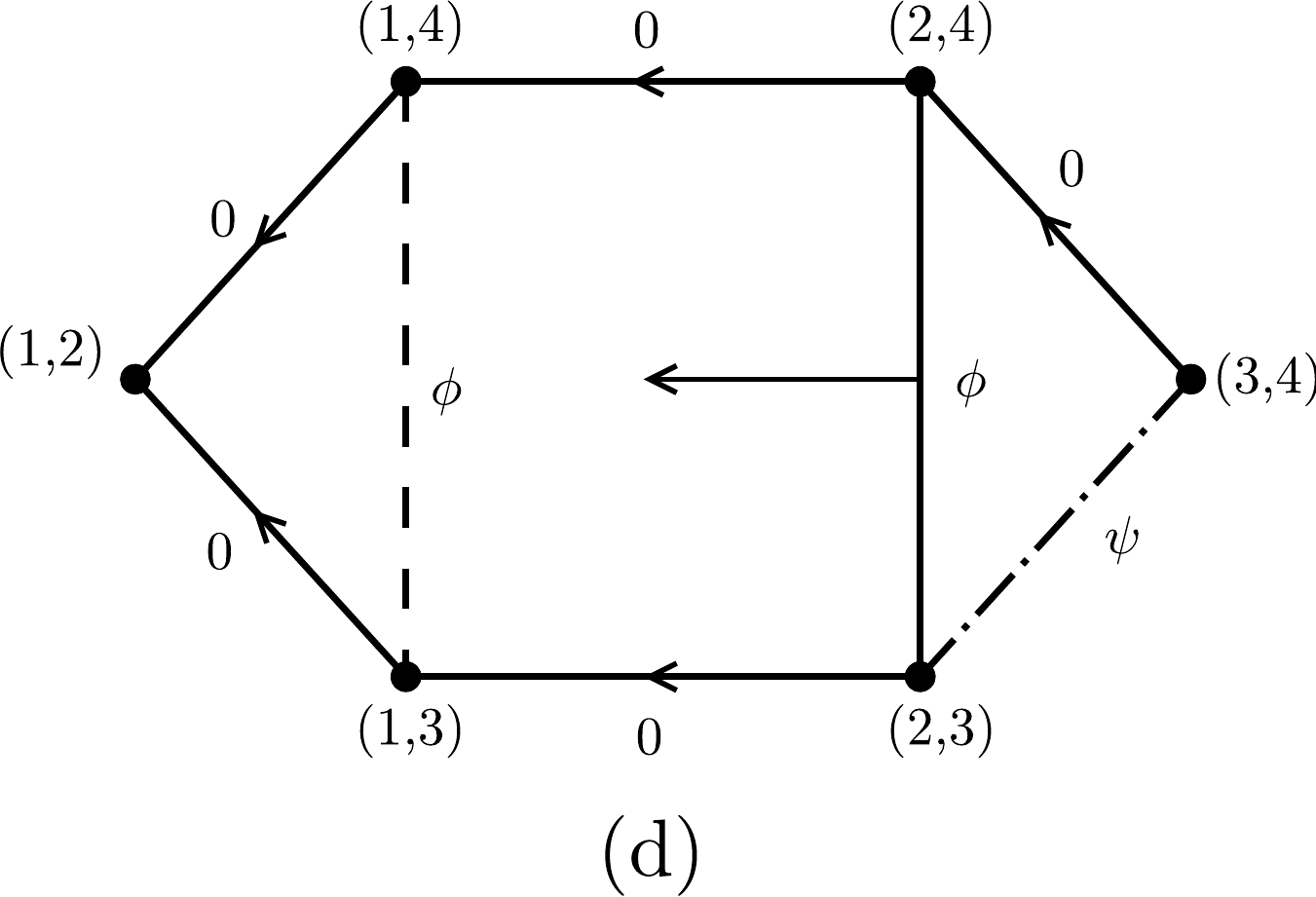}

\caption{\label{fig7}(a) The two particles on lasso, $\mathcal{D}^2(\Gamma_{1})$, (b) the discrete Morse
function and its gradient vector field (c) the Morse complex (d) the topological gauge potential $\Omega$}
\end{figure}

\noindent Notice that the choice we made is not unique. We could have changed $\tilde{f}_{2}\left(4\times(2,3)\right)$
in a similar way and leave $\tilde{f}_{2}\left(3\times(2,4)\right)$
untouched. After the modification (\ref{eq:mod2}) we construct the corresponding
discrete vector field for $f_{2}$. The Morse complex of $f_{2}$
consists of one critical $0$-cell (vertex $(1,2)$) and two critical
$1$ - cells (edges $3\times(2,4)$ and $1\times(3,4)$). Observe
that there are two different mechanisms responsible for criticality
of these $1$ - cells. The cell $1\times(3,4)$ is critical due to
the definition of trial Morse function $\tilde{f}_{2}$ and $3\times(2,4)$ has been chosen to be critical in order to make
$\tilde{f}_{2}$ the well defined Morse function $f_{2}$. We will
see later that these are in fact the only two ways giving rise to the critical
cells. Notice finally that function $f_{2}$ is in fact a perfect
Morse function and the Morse inequalities for it are equalities.

\noindent We will now consider a more difficult example. The one particle
configuration space, i.e. graph $\Gamma_{2}$ together with the perfect
Morse function and its gradient vector field are shown in figure \ref{fig8}(a)
and \ref{fig8}(b).

\begin{figure}[H]
~~~~~~~~~~~~~~~~~\includegraphics[scale=0.5]{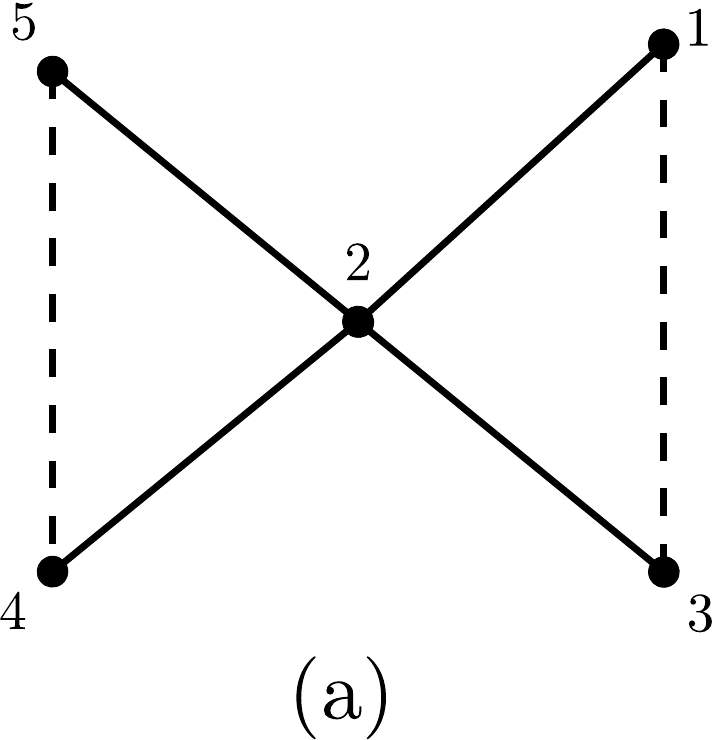}
~~~~~~~~~~~~~~~~~\includegraphics[scale=0.5]{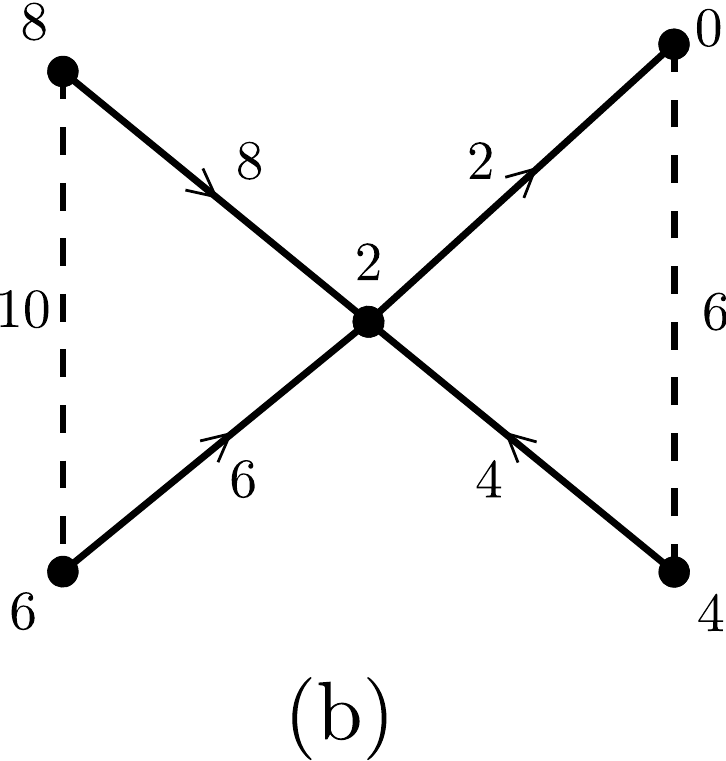}~~~~~~~~

\caption{\label{fig8}(a) One particle on bow-tie (b) Perfect discrete Morse function}
\end{figure}

\noindent The construction of two particle configuration space is
a bit more elaborate than in the lasso case and the result is shown
in figure \ref{fig9}(a). Using rules given in (\ref{eq:rules}) we obtain
the trial Morse function $\tilde{f}_{2}$ which is shown in figure
\ref{fig9}(b). The critical cells of $\tilde{f}_{2}$ and the cells causing
$\tilde{f}_{2}$ to not be a Morse function are given in table 3.1.
\begin{table}
\caption{\label{table1}The critical cells of $\tilde{f}_{2}$ and the vertices and edges
causing $\tilde{f}_{2}$ to not be a Morse function.}
\begin{tabular}{|c|c|}
\hline
\multicolumn{2}{|c|}{Critical cells of the trial Morse fuction $\tilde{f}_{2}$}\tabularnewline
\hline
0 - cells & $1\times2$\tabularnewline
\hline
1 - cells & $1\times(4,5)$, $2\times(1,3)$ \tabularnewline
\hline
2 - cells  & $(1,3)\times(4,5)$\tabularnewline
\hline
\end{tabular}
\begin{tabular}{|c|c|c|}
\hline
\multicolumn{3}{|c|}{$\tilde{f}_{2}$ is not Morse function because}\tabularnewline
\hline
vertex  & edges & value\tabularnewline
\hline
$(3,4)$  & $3\times(2,4)$, $4\times(2,3)$ & $\tilde{f}_{2}=10$\tabularnewline
\hline
$(3,5)$ & $5\times(2,3)$, $3\times(2,5)$ & $\tilde{f}_{2}=12$\tabularnewline
\hline
$(4,5)$ & $5\times(2,4)$, $4\times(2,5)$ & $\tilde{f}_{2}=14$\tabularnewline
\hline
\end{tabular}
\end{table}

In figure \ref{fig9}(b) we have chosen $1$ - cells: $3\times(2,4)$,
$3\times(2,5)$ and $4\times(2,5)$ to be critical, although we should
emphasize that it is one choice out of eight possible ones. We will
now determine the first homology group of the Morse complex $M(f_{2})$
and hence $H_{1}(\mathcal{D}^2(\Gamma_2))$. The Morse complex $M(f_{2})$ is
the sum of $M_{0}(f_{2})$ consisting of one $0$-cell (vertex $1\times2$),
$M_{1}(f_{2})$ which consists of five critical $1$-cells and $M_{2}(f_{2})$
which is one critical $2$-cell $c_{2}=(1,3)\times(4,5)$.
\[
\xymatrix{M_{2}(f_{2})\ar[r]^{\tilde{\partial}_{2}} & M_{1}(f_{2})\ar[r]^{\tilde{\partial}_{1}} & M_{0}(f_{2}).}
\]
The first homology is given by
\begin{eqnarray}
H_{1}(M(f_{2}))=H_{1}(\mathcal{D}^2(\Gamma_2))=\frac{\mbox{Ker}\tilde{\partial}_{1}}{\mbox{Im}\tilde{\partial}_{2}}.
\end{eqnarray}
It is easy to see that $\tilde{\partial}_{1}c_{1}=0$ for any $c_{1}\in M_{1}(f_{2})$
and hence $\mbox{Ker}\tilde{\partial}_{1}=\mathbb{Z}^{5}$. What is
left is to find $\tilde{\partial}c_{2}$ which is a linear combination
of critical $1$-cells from $M_{1}(f_{2})$. According to formula
(\ref{eq:boundary}) we take the boundary of $c_{2}$ in $C_{2}(\Gamma_{2})$
and consider all paths starting from it and ending at the $2$-cells
containing critical $1$-cells (see table 3.2).
\begin{table}
\caption{\label{table2}The boundary of $c_{2}$.}
\begin{tabular}{|c|c|c|c|}
\hline
boundary of $c_{2}$  &  path & critical $1$ - cells & orientation\tabularnewline
\hline
$1\times(4,5)$ & $\emptyset$ & $1\times(4,5)$ & +\tabularnewline
\hline
$5\times(1,3)$ & %
\begin{tabular}{c}
$5\times(1,3)$, $(2,5)\times(1,3)$, $2\times(1,3)$.\tabularnewline
$5\times(1,3)$, $(2,5)\times(1,3)$, $3\times(2,5)$.\tabularnewline
\end{tabular} & %
\begin{tabular}{c}
$2\times(1,3)$\tabularnewline
$3\times(2,5)$\tabularnewline
\end{tabular} &
\begin{tabular}{c}
-\tabularnewline
-\tabularnewline
\end{tabular}\tabularnewline
\hline
$3\times(4,5)$ & %
\begin{tabular}{c}
$3\times(4,5)$, $(4,5)\times(2,3)$, $2\times(4,5)$, \tabularnewline
$(1,2)\times(4,5)$, $1\times(4,5)$.\tabularnewline
\end{tabular} & $1\times(4,5)$ & -\tabularnewline
\hline
$4\times(1,3)$ & %
\begin{tabular}{c}
$4\times(1,3)$, $(1,3)\times(2,4)$, $2\times(1,3)$.\tabularnewline
$4\times(1,3)$, $(1,3)\times(2,4)$, $3\times(2,4)$.\tabularnewline
\end{tabular} & %
\begin{tabular}{c}
$2\times(1,3)$\tabularnewline
$3\times(2,4)$\tabularnewline
\end{tabular} &
\begin{tabular}{c}
+\tabularnewline
+\tabularnewline
\end{tabular}\tabularnewline
\hline
\end{tabular}

\end{table}
\noindent Eventually taking into account orientation we get
\begin{eqnarray}
\tilde{\partial}_{2}(c_{2})=1\times(4,5)-3\times(2,5)-2\times(1,3)-1\times(4,5)+\\+3\times(2,4)+2\times(1,3)=-3\times(2,5)+3\times(2,4).
\end{eqnarray}
Hence,
\begin{eqnarray}
H_{1}(\mathcal{D}^2(\Gamma_2))=\frac{\mbox{Ker}\tilde{\partial}_{1}}{\mbox{Im}\tilde{\partial}_{2}}=\mathbb{Z}^{4}.
\end{eqnarray}
The Morse complex $M(f_{2})$ is shown explicitly in figure \ref{fig9}(c).
It is worth mentioning that in this example $f_{2}$ is not a perfect
Morse function.
\begin{figure}[H]
\includegraphics[scale=0.37]{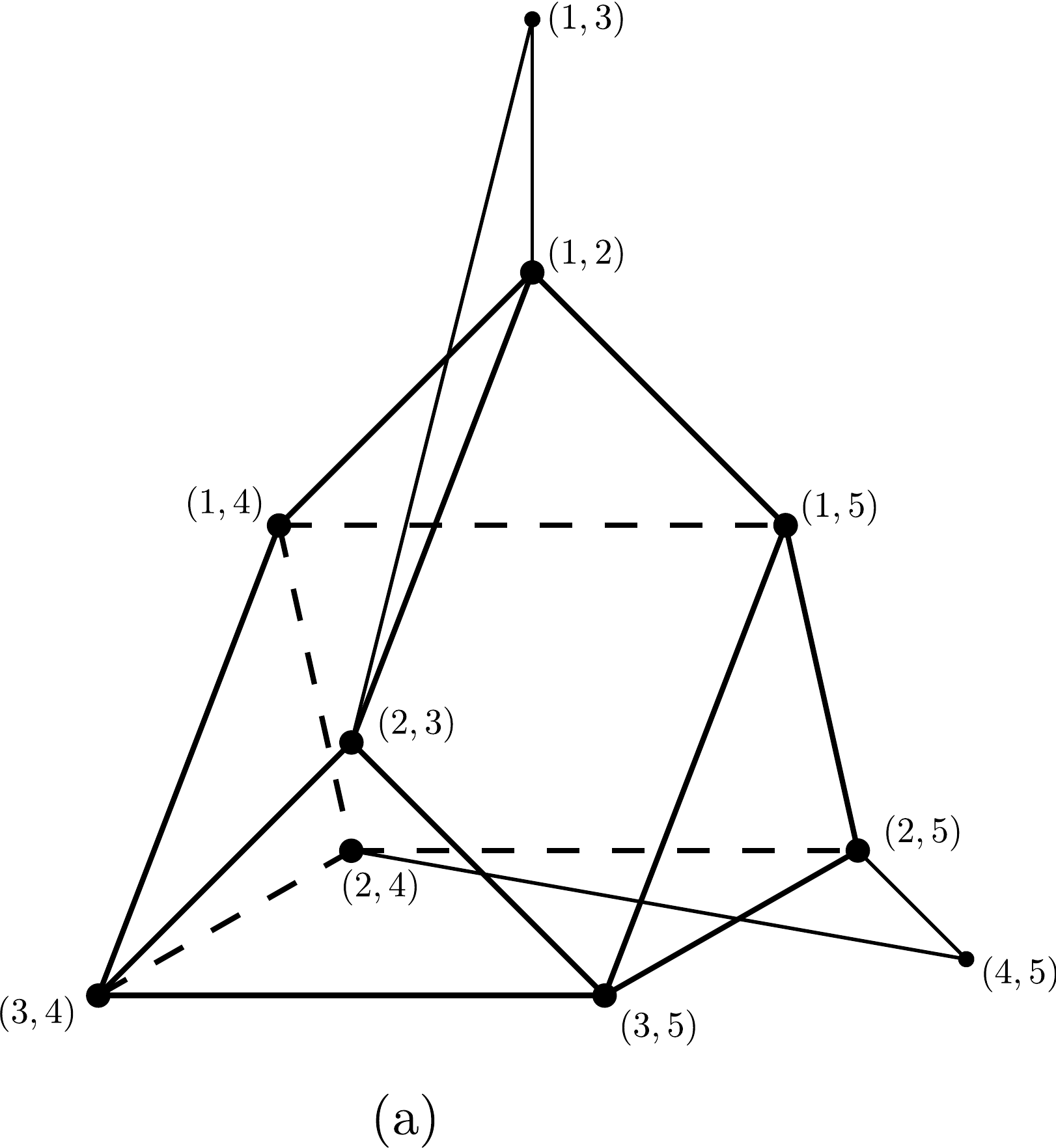}~~~\includegraphics[scale=0.4]{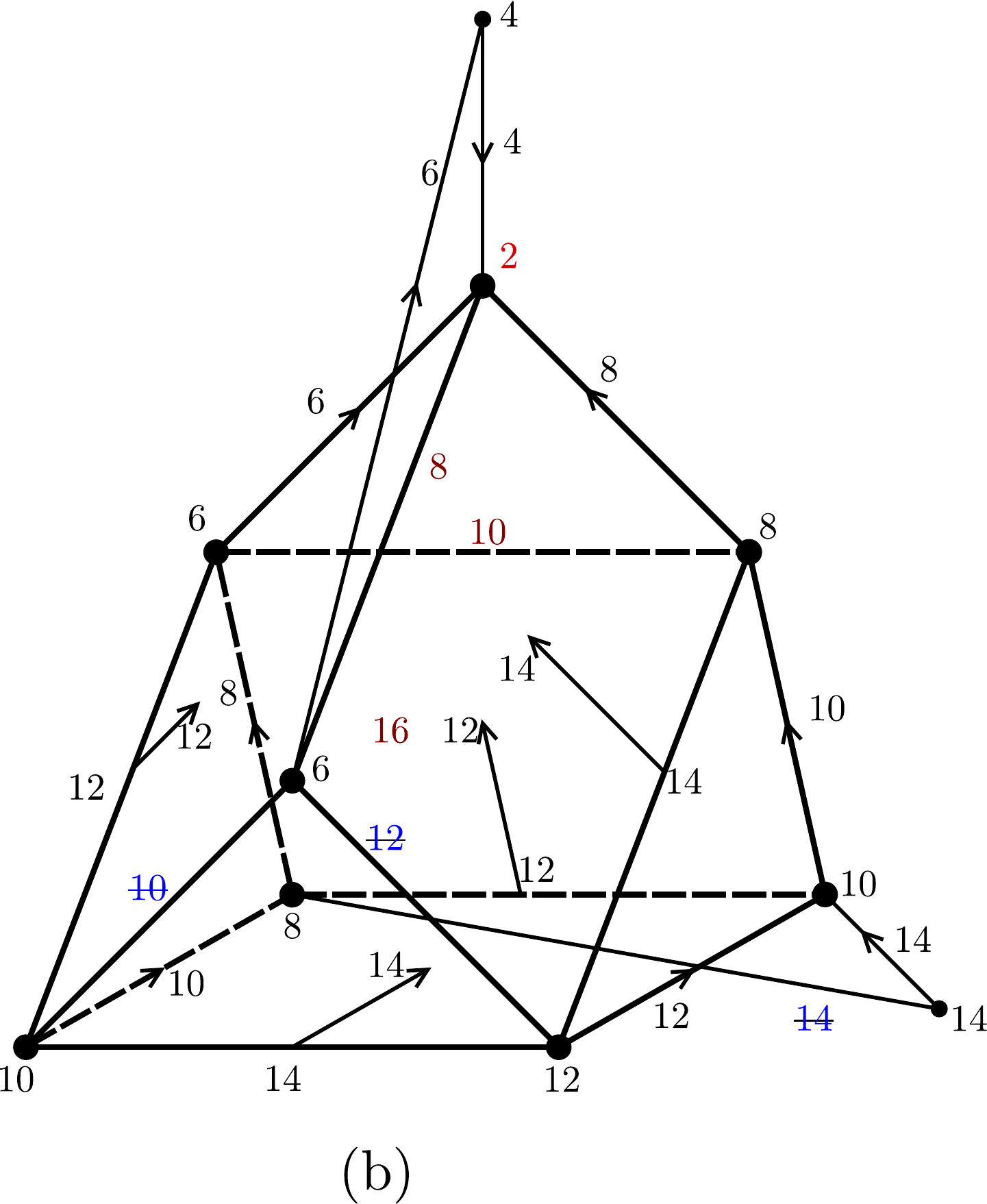}

\medskip{}

~~~~~~~~~~~~~~~~~~~~~~~~~\includegraphics[scale=0.4]{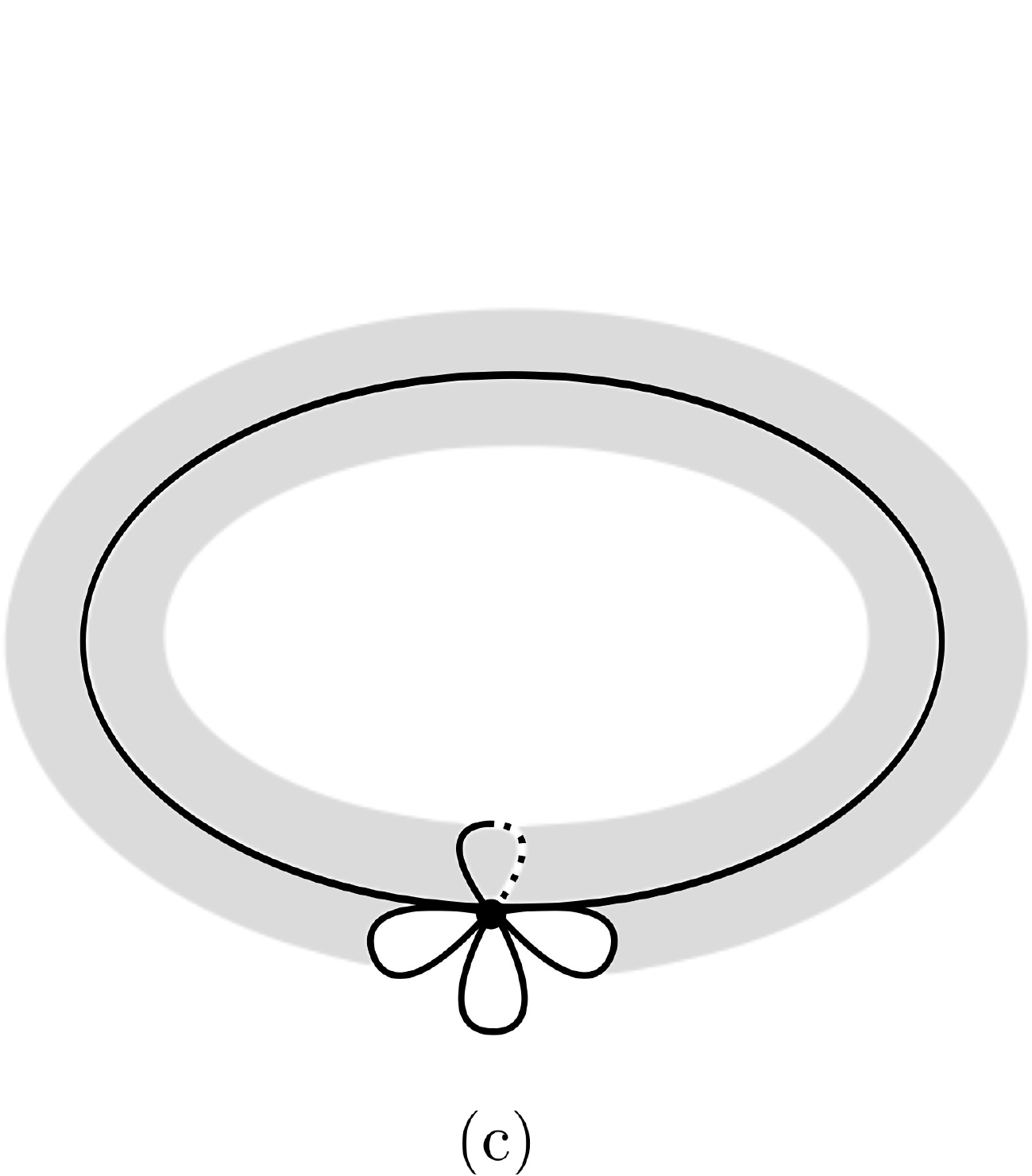}

\caption{\label{fig9}(a) Two particles on bow-tie (b) the discrete Morse function and its gradient vector
field, (c) the Morse complex $M(f_{2})$. }
\end{figure}

\section{Discrete Morse theory and topological gauge potentials}\label{sec:topological gauge potentials}

In this section we describe more specifically how the techniques of discrete Morse theory apply to the problem of quantum statistics on graphs. A more general discussion of the model can be found in \cite{JHJKJR}. Here we describe a particular representative example, highlighting the usefulness of discrete Morse theory.

Let $\Gamma$ be a graph shown in figure \ref{fig6}(a). The Hilbert space associated to $\Gamma$ is $\mathcal{H}=\mathbb{C}^4$ and is spanned by the vertices of $\Gamma$. The dynamics is given by Schr\"{o}dinger equation where the Hamiltonian $H$ is a hermitian matrix, such that $H_{jk}=0$ if $j$ is not adjacent to $k$ in $\Gamma$. As discussed in \cite{JHJKJR} this corresponds to the so-called tight binding model of one-particle dynamics on $\Gamma$. One can add to the model an additional ingredient, namely whenever the particle hops between adjacent vertices of $\Gamma$ the wavefunction gains an additional phase factor. This can be incorporated to the Hamiltonian by introducing a gauge potential. It is an antisymmetric real matrix $\Omega$ such that each $\Omega_{jk}\in [0,\,2\pi[$ and $\Omega_{jk}=0$ if $j$ is not adjacent to $k$ in $\Gamma$. The modified Hamiltonian is then $H_{jk}^\Omega=H_{jk}e^{i\Omega_{jk}}$. The flux of $\Omega$ through any cycle of $\Gamma$ is the sum of values of $\Omega$ on the directed edges of the cycle. It can be given a physical interpretation in terms of the Aharonov-Bohm phase.

In order to describe in a similar manner the dynamics of two indistinguishable particles on $\Gamma$ we follow the procedure given in \cite{JHJKJR}. The structure of the Hilbert space and the corresponding tight binding Hamiltonian are encoded in $\mathcal{D}^2(\Gamma)$. Namely, we have $\mathcal{H}_2=\mathbb{C}^6$ and is spanned by the vertices of $\mathcal{D}^2(\Gamma)$. The Hamiltonian is given by a hermitian matrix, such that $H_{j,k\rightarrow l}=0$ if $k$ is not adjacent to $l$ in $\Gamma$. The notation $j,k\rightarrow l$ describes two vertices $(j,k)$ and $(j,l)$ connected by an edge in $\mathcal{D}^2(\Gamma)$. The additional assumption which we add in this case stems from the topological structure of $\mathcal{D}^2(\Gamma)$ and is reflected in the condition on the gauge potential. Namely, since the 2-cell $c_2=(1,2)\times(3,4)$ is contractible we require that the flux through its boundary vanishes, i.e.
\begin{equation}\label{Omega}
\Omega(\partial c_2)=\Omega_{1,3\rightarrow 4}+\Omega_{4,1\rightarrow 2}+\Omega_{2,4\rightarrow 3}+\Omega_{3,2\rightarrow 1}=0\, \mathrm{mod}\,2\pi.
\end{equation}
Our goal is to find the parametrization of all gauge potentials satisfying (\ref{Omega}), up to a so-called trivial gauge, i.e. up to addition of $\Omega^\prime$ such that $\Omega^\prime(c)=0\,\mathrm{mod}\,2\pi$, for any cycle $c$. To this end we use discrete Morse theory. We first notice that the edges of $\mathcal{D}^2(\Gamma)$ which are heads of an arrow of the discrete Morse vector field form a tree. Without loss of generality we can put $\Omega_{j,k\rightarrow l}=0$ whenever $j\times(k,l)$ is a head of an arrow. Next, on the edges corresponding to the critical $1$-cells we put arbitrary phases $\Omega_{1,3\rightarrow 4}=\phi$ and $\Omega_{3,2\rightarrow 4}=\psi$. Notice that since $f_2$ is a perfect Morse function these phases are independent. The only remaining edge is $2\times(3,4)$ which is a tail of an arrow. In order to decide what phase should be put on it we follow the gradient path of the discrete Morse vector field which leads to edge $1\times(3,4)$. Hence $\Omega_{2,3\rightarrow 4}=\phi$. The effect of our construction is the topological gauge potential $\Omega$ which is given by two independent parameters (see figure \ref{fig7}(d)) and satisfies (\ref{Omega}). The described reasoning can be \emph{mutatis mutandis} applied to any graph $\Gamma$, albeit the phases on edges corresponding to the critical cells are not independent if $f_2$ is not a perfect Morse function. Finally notice, that in the considered example, the phase  $\phi$ can be interpreted as an Aharonov-Bohm phase and $\psi$ as the exchange phase. The latter gives rise to anyon statistics.

\section{General consideration for two particles\label{sec:General-consideration-for} }
In this section we investigate the first Homology group $H_{1}(C_{2}(\Gamma))$
by means of discrete Morse theory. In section \ref{sec:Main-example} the idea of a trial Morse function
was introduced. Let us recall here that the trial Morse function is
defined in two steps. The first one is to define a perfect Morse function
on $\Gamma$. To this end one chooses a spanning tree $T$ in $\Gamma$.
The vertices of $\Gamma$ are labeled by $1,\,2,\ldots,|V|$ according
to the procedure described in section \ref{sec:One-particle-graph}. The perfect Morse function
$f_{1}$ on $\Gamma$ is then given by its value on the vertices and
edges of $\Gamma$, i.e.
\begin{eqnarray}\label{1-particle-perfect}
f_{1}(i)=2i-2,\\
f_{1}((j,k))=\mathrm{max}(f_{1}(j),\, f_{1}(k)),\,\,(j,k)\in T,\\
f_{1}((j,k))=\mathrm{max}(f_{1}(j),\, f_{1}(k))+2,\,\,(j,k)\in\Gamma\setminus T\label{eq:f1-1}
\end{eqnarray}
When $f_{1}$ is specified the trial Morse function on $\mathcal{D}^2(\Gamma)$
is given by the formula
\begin{eqnarray}
\mathrm{0-cells:\,\,\,\,\,\,\,\,\,\,\,\,\,}\,\,\,\,\,\,\,\,\,\,\,\,\tilde{f}_{2}(i\times j) & = & f_{1}(i)+f_{1}(j),\nonumber \\
\mathrm{1-cells:}\,\,\,\,\,\,\,\,\,\,\,\,\tilde{f}_{2}\left(i\times(j,k)\right) & = & f_{1}(i)+f_{1}\left((j,k)\right),\nonumber \\
\mathrm{2-cells:}\,\,\,\tilde{f}_{2}\left((i,j)\times(k,l)\right) & = & f_{1}\left((i,j)\right)+f_{1}\left((k,l)\right).\label{eq:rules-1}
\end{eqnarray}
Let us emphasize that the trial Morse function is typically not a Morse function, i.e., the conditions of definition \ref{Morse-fuction} might not be satisfied. Nevertheless, we will show that it is always possible to modify the function $\tilde{f}_{2}$ and obtain a Morse function $f_2$ out of it. In fact the function $\tilde{f_2}$ is not 'far' from being a Morse function and, as we will see, the number of cells at which it needs fixing is relatively small. In the next paragraphs we localize the obstructions causing $\tilde{f}_2$ to not be a Morse function and explain how to overcome them.

\noindent The cell complex $\mathcal{D}^2(\Gamma)$ consists of $2$, $1$, and $0$-cells which we will denote by $\alpha$, $\beta$ and $\kappa$ respectively. For all these cells we have to verify the conditions of definition \ref{Morse-fuction}. Notice that checking these conditions for any cell involves looking at its higher and lower dimensional neighbours. In case of $2$-cell $\alpha$ we have only the former ones, i.e., the $1$-cells $\beta$ in the boundary of $\alpha$. For the $1$-cell $\beta$ both $2$-cells $\alpha$ and $0$-cells $\kappa$ are present. Finally for the $0$-cell $\kappa$ we have only $1$-cells $\beta$.

Our strategy is the following. We begin with the trial Morse function $\tilde{f}_2$ and go over all $2$-cells checking the conditions of definition \ref{Morse-fuction}. The outcome of this step is a new trial Morse function $\bar{f}_2$ which has no defects on $2$-cells. Next we consider all $1$-cells and verify the conditions of definition \ref{Morse-fuction} for $\bar{f}_2$. It happens that they are satisfied. Finally we go over all $0$-cells. The result of this three-steps procedure is a well defined Morse function $f_2$. Below we present more detailed discussion. The proofs of all statements are in section \ref{proofs}.
\begin{enumerate}
    \item \textbf{Step 1} We start with a trial Morse function $\tilde{f}_2$. We notice first that for any edge $e\in T$ there is a unique vertex $v$ in its boundary such that $f_1(e)=f_1(v)$. In other words every vertex $v$, different from $v=1$, specifies exactly one edge $e\in T$ which we will denote by $e(v)$. Next we divide the set of $2$-cells into three disjoint classes. The first one contains $2$-cells $\alpha=e_i\times e_j$, where both $e_i,e_j\notin T$. The second one contains  $2$-cells $\alpha=e_i\times e(v)$, where $e(v)\in T$ and $e_i\notin T$, and the last one contains $2$-cells $\alpha=e(u)\times e(v)$, where both $e(u),e(v) \in T$. Now, since there are no $3$-cells, we have only to check that for each $2$-cell $\alpha$
\begin{eqnarray}\label{alphacond}
    \#\{\beta\subset\alpha\,:\, \tilde{f}_2(\beta)\geq \tilde{f}_2(\alpha)\}\leq1
\end{eqnarray}
The following results are proved in section \ref{proofs}
\begin{enumerate}
    \item For the $2$-cells $\alpha=e_i\times e_j$ where both $e_i,e_j\notin T$ the condition (\ref{alphacond}) is satisfied (see fact \ref{fact1}).
    \item For the $2$-cells $\alpha=e_i\times e(v)$ where $e_i\notin T$ and $e(v)\in T$ the condition (\ref{alphacond}) is satisfied (see fact \ref{fact2}).

\item  For the $2$-cells $\alpha=e(u)\times e(v)$ where both $e(u),e(v)\in T$ the condition (\ref{alphacond}) is not satisfied. There are exactly two $1$-cells $\beta_1,\beta_2\subset\alpha$ such that $\tilde{f}_2(\beta_1)=\tilde{f}_2(\alpha)=\tilde{f}_2(\beta_2)$. They are of the form $\beta_1=u\times e(v)$ and $\beta_2=v\times e(u)$. The function $\tilde{f}_2$ can be fixed in two ways (see fact \ref{fact3}). We put $\bar{f}_2(\alpha)=\tilde{f}_2(\alpha)+1$ and either $\bar{f}_2(\beta_1):=\tilde{f}_2(\beta_1)+1$ or $\bar{f}_2(\beta_2):=\tilde{f}_2(\beta_2)+1$.
In both cases $\{\beta_i,\,\alpha\}$ is the pair of noncritical cells.
\end{enumerate}
The result of this step is a new trial Morse function $\bar{f}_2$, which satisfies (\ref{alphacond}).
\item \textbf{Step 2} We divide the set of $1$-cells into two disjoint classes. The first one contains $1$-cells $\beta=v\times e$, where $e\notin T$ and the second one contains $\beta=v\times e(u)$, where $e(u)\in T$. For the $1$-cells within each of this classes we introduce additional division with respect to condition $e(v)\cap e=\emptyset$ (or $e(v)\cap e(u)= \emptyset$). Notice that all $1$-cells $\beta$ which were modified in \textbf{Step 1} belong to the second class and satisfy $e(v)\cap e(u)=\emptyset$. Next we take a trial Morse function $\bar{f}_2$ and go over all $1$-cells $\beta$ checking for each of them if
\begin{eqnarray}\label{betacondition1}
\#\{\alpha\supset\beta\,:\, \bar{f}_2(\alpha)\leq \bar{f}_2(\beta)\}\leq1,\\
\#\{\kappa\subset\beta\,:\, \bar{f}_2\geq \bar{f}_2(\beta)\}\leq1.\label{betacondition2}
\end{eqnarray}
What we find out is
\begin{enumerate}
        \item For the $1$-cells $\beta=v\times e(u)$, where $e(u)\in T$ and $e(v)\cap e(u)\neq\emptyset$ the conditions (\ref{betacondition1}, \ref{betacondition2}) are satisfied (see fact \ref{fact4}). 

    \item For the $1$-cells $\beta=v\times e$, where $e\notin T$ and $e(v)\cap e\neq\emptyset$ the conditions (\ref{betacondition1}, \ref{betacondition2}) are satisfied (see fact \ref{fact5}).
    \item For the $1$-cells $\beta=v\times e(u)$, where $e(u)\in T$ and $e(v)\cap e(u)=\emptyset$ the conditions (\ref{betacondition1}, \ref{betacondition2}) are satisfied (see fact \ref{fact6}).
    \item For the $1$-cells $\beta=v\times e$, where $e\notin T$ and $e(v)\cap e=\emptyset$ the conditions (\ref{betacondition1}, \ref{betacondition2}) are satisfied (see fact \ref{fact7}). 
\end{enumerate}
Summing up the trial Morse function $\bar{f}_2$, obtained in \textbf{Step 1} satisfies both (\ref{alphacond}) and (\ref{betacondition1}), (\ref{betacondition2}). We switch now to the analysis of $0$-cells.
\item \textbf{Step 3} We divide the set of $0$-cells into four disjoint classes in the following way. We denote by $\tau(v)\neq v$ the vertex to which $e(v)$ is adjacent and call it the terminal vertex of $e(v)$. For any $0$-cell $\kappa=v\times u$ we have that either
\begin{enumerate}
    \item $e(v)\cap e(u)\neq\emptyset$ and the terminal vertex $\tau(v)$ of $e(v)$ is equal to $u$.\label{k1}
    \item $e(v)\cap e(u)\neq\emptyset$ and the terminal vertex $\tau(u)$ of $e(u)$ is equal to the terminal vertex $\tau(v)$ of $e(v)$.\label{k2}
    \item $e(v)\cap e(u)=\emptyset$.\label{k3}
    \item $\kappa=1\times u$.\label{k4}
\end{enumerate}
What is left is checking the following condition for any $0$-cell $\kappa$ :
\begin{eqnarray}\label{betacond}
\#\{\beta\supset\kappa\,:\, \bar{f}_2(\beta)\leq \bar{f}_2(\kappa)\}\leq1
\end{eqnarray}
We find out that
\begin{enumerate}
    \item For the $0$-cell $\kappa=u\times v$ belonging to \ref{k1} the condition (\ref{betacond}) is satisfied (see fact \ref{fact8}).
     \item For the $0$-cell $\kappa=u\times v$  belonging to \ref{k2} the condition (\ref{betacond}) is not satisfied. There are exactly two $1$-cells $\beta_1,\beta_2\supset\kappa$ such that $\bar{f}_2(\beta_1)=\bar{f}_2(\kappa)=\bar{f}_2(\beta_2)$. They are of the form $\beta_1=u\times e(v)$ and $\beta_2=v\times e(u)$. The function $\bar{f}_2$ can be fixed in two ways. We put $f_2(\beta_1):=\bar{f}_2(\beta_1)+1$ or $f_2(\beta_2):=\bar{f}_2(\beta_2)+1$ (see fact \ref{fact9}). Moreover, this change does not violate the Morse conditions at any $2$-cell containing $\beta_i$.
     \item For the $0$-cell $\kappa=u\times v$ belonging to \ref{k3} the condition  (\ref{betacond}) is satisfied (see fact \ref{fact10})
     \item For the $0$-cell $\kappa=u\times v$ belonging to \ref{k4} the condition  (\ref{betacond}) is satisfied (see fact \ref{fact11})
\end{enumerate}
\end{enumerate}

\noindent As a result of the above procedure we obtain the Morse function $f_2$. The following theorem summarizes the above described procedure.
\begin{theorem}\label{theorem1} Let $f_1$ be a perfect Morse function on a $1$-particle graph $\Gamma$ defined by (\ref{1-particle-perfect}). Define a trial Morse function $\tilde{f}_2$ on $\mathcal{D}^2(\Gamma)$ by $\tilde{f}_2(\alpha\times\beta):=f_1(\alpha)+f_1(\beta)$. A Morse function $f_2$ on $\mathcal{D}^2(\Gamma)$ is the modification of $\tilde{f}_2$ obtained in the following way:
\begin{enumerate}
  \item For $2$-cells of the form $\alpha=e(u)\times e(v)$ where both $e(u),e(v)\in T$, increment $\tilde{f}_2(\alpha)$ by $1$ and increment either $\tilde{f}_2(u\times e(v))$ or $\tilde{f}_2(e(u)\times v)$ by $1$ as well.
  \item For $0$-cells of the form $\kappa=u\times v$ where $\tau(e(u))=\tau(e(v))$, increment either $\tilde{f}_2(u\times e(v))$ or $\tilde{f}_2(e(u)\times v)$ by 1.
\end{enumerate}
\end{theorem}

We can now ask the question which cells of $\mathcal{D}^2(\Gamma)$ are critical cells of $f_2$. Careful consideration of the arguments given in facts \ref{fact1}-\ref{fact11} lead to the following conclusions:
\begin{theorem}\label{theorem2} The conditions for the critical cells of $f_2$ are
\begin{itemize}
\item The $0$-cell is critical if and only if it is $1\times2$
\item The $1$-cell is critical if and only if

\begin{enumerate}
\item It is $v\times e$ where $e\notin T$ and $e(v)\cap e\neq\emptyset$ or $v=1$.
\item Assume that $e(v)\cap e(u)\neq\emptyset$ and the terminal vertex $\tau(u)$ of $e(u)$
is equal to the terminal vertex $\tau(v)$ of $e(v)$. Then either the $1$-cell $v\times e(u)$ or the $1$-cell $u\times e(v)$ is critical, but not both.
\end{enumerate}
\item The $2$-cell is critical if and only if it is $e_{1}\times e_{2}$ where both
$e_{i}\notin T$.
\end{itemize}
\end{theorem}
These rules are related to those given by Farley and Sabalka in \cite{FS08}. As pointed out by an anonymous referee the freedom in choosing noncritical $1$-cells (see fact 3 in section \ref{proofs}) and critical $1$-cells (see fact 9 in section \ref{proofs}) is also present in Farley and Sabalka's \cite{FS08} construction. Moreover, a perfect Morse function on a $1$-particle graph used in our construction stems from the labeling of the tree discussed in \cite{FS08}.

\section{Summary}
We have presented a description of topological properties of two-particle
graph configuration spaces in terms of discrete Morse theory. Our
approach is through discrete Morse functions, which may be regarded as two-particle potential energies.
We proceeded by introducing a trial Morse function on the full two-particle cell complex, $\mathcal{D}^2(\Gamma)$, which is simply the sum of single-particle potentials on the
one-particle cell complex, $\Gamma$.
We showed that the trial Morse function is close to being a true Morse function
provided that the single-particle potential is a perfect Morse function on $\Gamma$. Moreover, we give an explicit prescription for removing  local defects. The  fixing process is unique modulo the freedom described in facts
\ref{fact3} and \ref{fact9}. The  construction was demonstrated by two examples. A future goal would be to see if these constructions can provide any simplification in understanding of the results of \cite{KP11}. It will be also interesting to verify if the presented techniques can be extended to $N$-particle graphs and if they lead to analogous results as in \cite{FS08}. The preliminary calculations indicate that the answer is positive, however small modifications of a perfect Morse function on a $1$-particle graph are needed.

Finally, notice that using our analogy with the potential energy a trial Morse function is constructed as if particles do not interact. The modification of a trail Morse function can be hence viewed as introducing an interaction. On the other hand, in the considered graph setting, quantum statistics or anyons can be regarded as fermions which interact in some particular way. Remarkably, the modifications of a trial Morse function in particular these described in point 2 of theorem \ref{theorem1} correspond to situations when two particles come close together. 
\section{Proofs}\label{proofs}
In this section we give the proofs of the statements made in section \ref{sec:General-consideration-for}. The following notation will be used. We denote by $D_v$ all edges of $\Gamma$ which are adjacent to $v$ and belong to $\Gamma-T$. Similarly by $T_v$ we denote all edges of $\Gamma$ which are adjacent to $v$ and belong to $T$, except one distinguished edge $e(v)\in T$, but not in $T_v$.

\begin{Fact}\label{fact1}
Let $\alpha=e_{1}\times e_{2}$ be a $2$-cell such that both $e_{1}$ and $e_{2}$
do not belong to $T$. The condition (\ref{alphacond}) is satisfied and
$\alpha$ is a critical cell.\end{Fact}
Proof.
The two cell $e_{1}\times e_{2}$ is shown in the figure \ref{figure10}, where $e_{1}=(i,j)$
and $e_{2}=(k,l)$ and $i>j$, $k>l$. The result follows immediately from this figure.

\begin{figure}[H]
~~~~~~~~~~~~~~~~~~~~~~~~~~~~~~~~\includegraphics[scale=0.5]{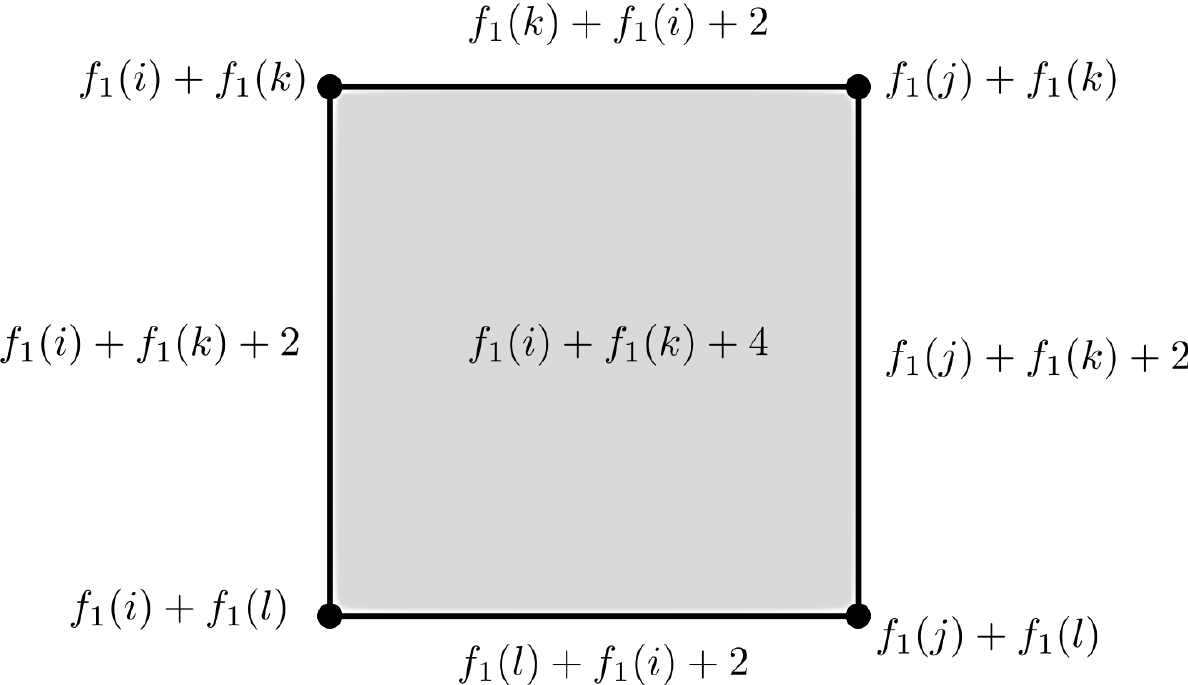}

\caption{The critical cell $e_{1}\times e_{2}$ where both $e_{1}$ and $e_{2}$
do not belong to $T$}\label{figure10}

\end{figure}

\begin{Fact}\label{fact2}
Let $\alpha=e\times e(v)$ be a $2$-cell, where $e\notin T$ and $e(v)\in T$.
Condition (\ref{alphacond}) is satisfied and
$\alpha$ is a noncritical cell.\end{Fact}
Proof.
We of course assume that $e(v)\cap e=\emptyset$. The $2$-cell $\alpha$ is shown on figure \ref{figure11}, where we denoted $e(v)=(v,\tau(v))$ and $e=(j,k)$. The result follows immediately from this figure.
\begin{figure}[h]


\includegraphics[scale=0.5]{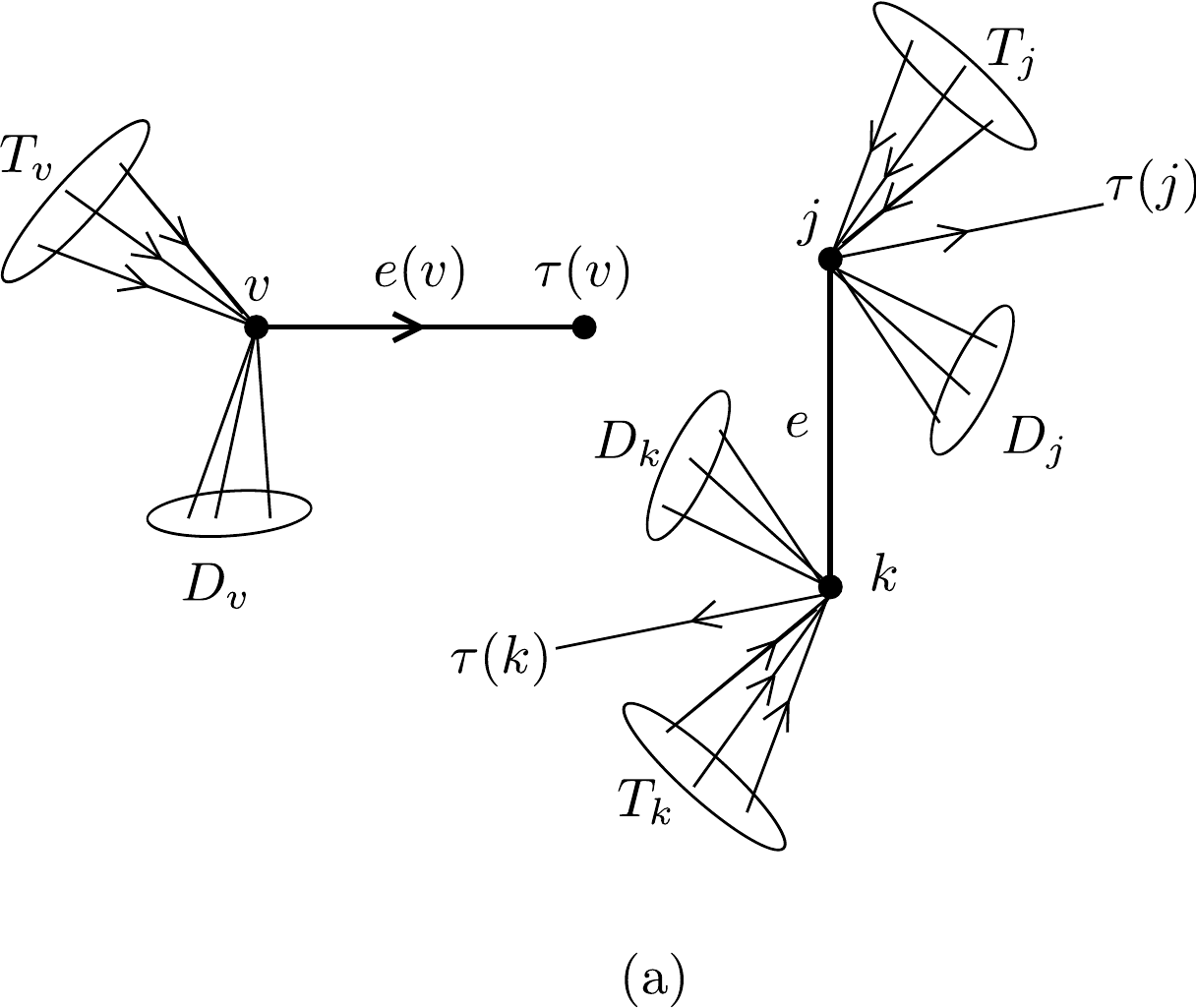}~~\includegraphics[scale=0.5]{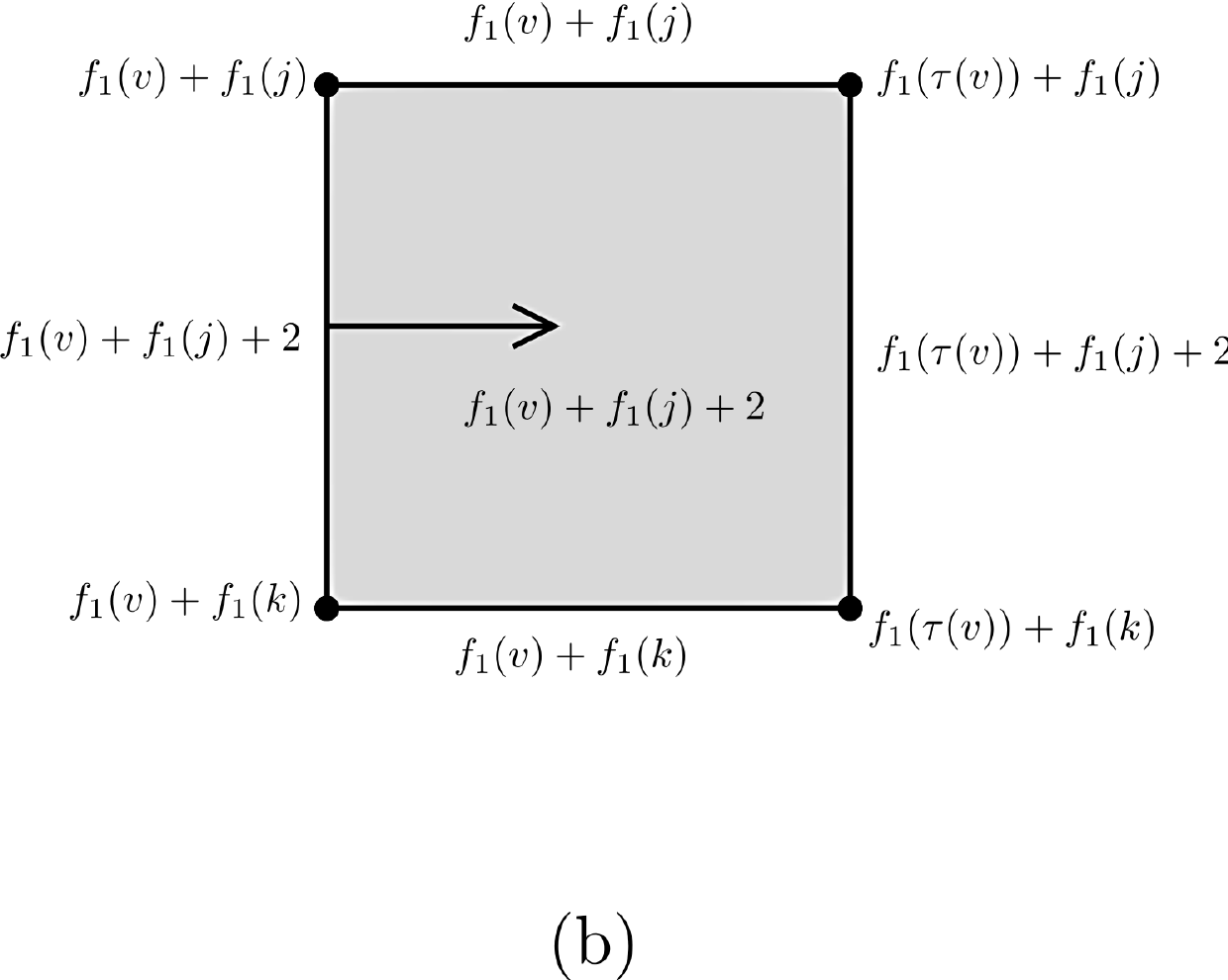}

\caption{(a) $e(v)\cap e=\emptyset$ and $e\notin T$, (b) The
noncritical cells $v\times e$ and $e(v)\times e$. }\label{figure11}
\end{figure}

\begin{Fact}\label{fact3}
Let $\alpha=e(u)\times e(v)$ be the $2$-cells, where both $e(u),e(v)\in T$. Condition (\ref{alphacond}) is not satisfied. There are exactly two $1$-cells $\beta_1,\beta_2\subset\alpha$ such that $\tilde{f}_2(\beta_1)=\tilde{f}_2(\alpha)=\tilde{f}_2(\beta_2)$. They are of the form $\beta_1=u\times e(v)$ and $\beta_2=v\times e(u)$. The function $\tilde{f}_2$ can be fixed in two ways. We put $\bar{f}_2(\alpha)=\tilde{f}_2(\alpha)+1$ and either $\bar{f}_2(\beta_1):=\tilde{f}_2(\beta_1)+1$ or $\bar{f}_2(\beta_2):=\tilde{f}_2(\beta_2)+1$. \end{Fact}

\noindent Proof.
The $2$-cell $e(v)\times e(u)$ when $e(v)\cap e(u)=\emptyset$
is presented in figure \ref{figure12}(a),(b). The trail Morse function $\tilde{f}_{2}$
requires fixing and two possibilities are shown on figure \ref{figure12}(c),(d).
Notice that in both cases we get a pair of noncritical cells. Namely
the $1$-cell $v\times e(u)$ and $2$-cell $e(v)\times e(u)$ for the situation
presented in figure \ref{figure12}(c) and $1$-cell $u\times e(v)$, $2$-cell $e(v)\times e(u)$
for the situation presented in figure \ref{figure12}(d).

\begin{figure}[h]\label{figure12}
\includegraphics[scale=0.6]{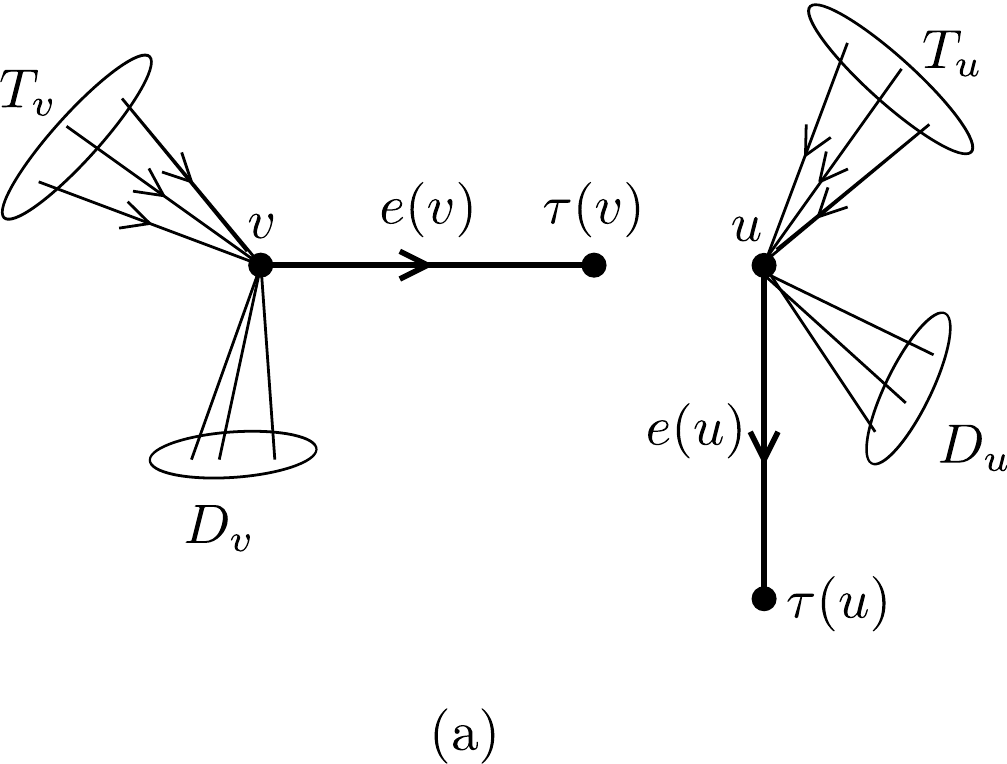}~~~~\includegraphics[scale=0.5]{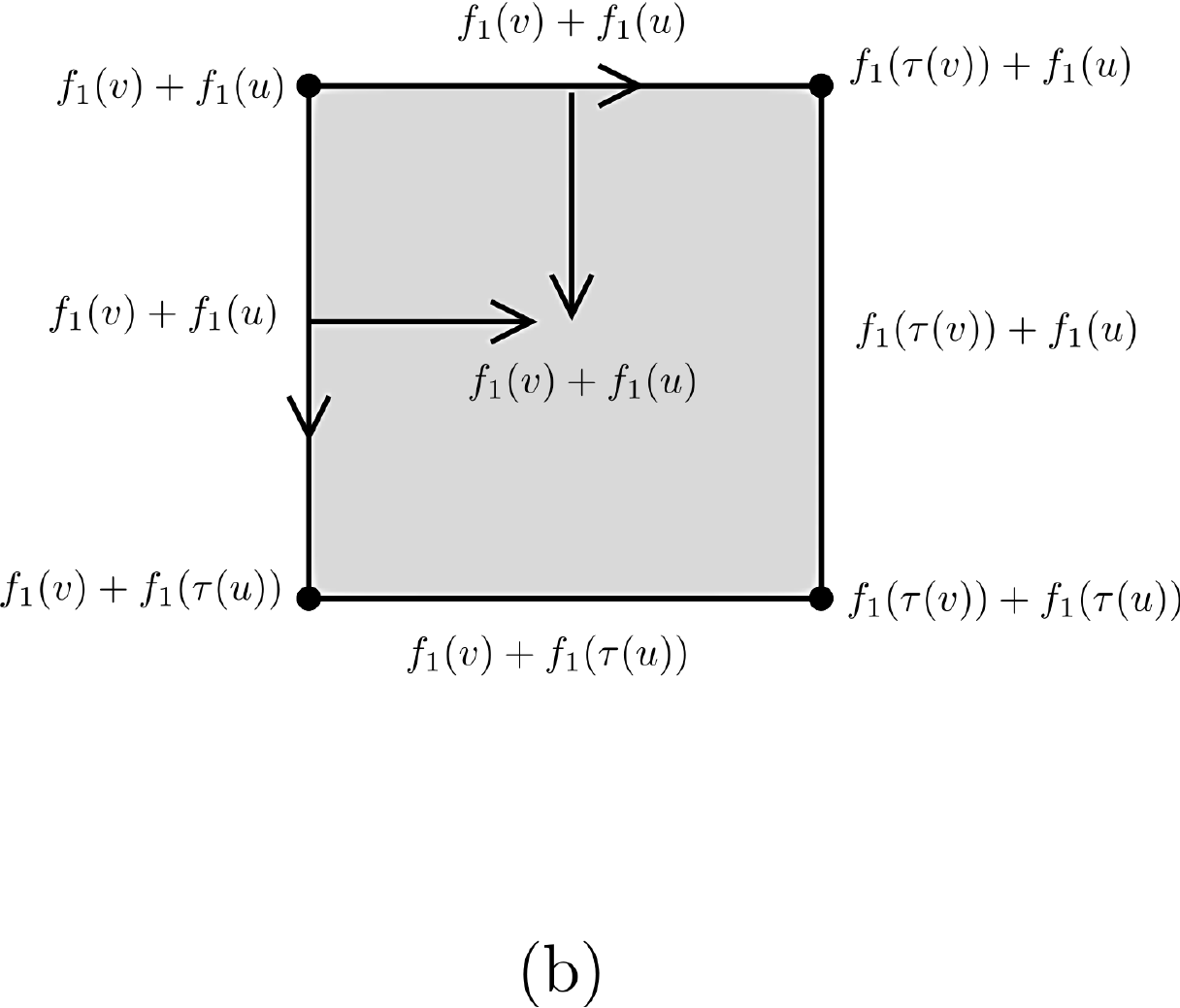}

\bigskip{}

\includegraphics[scale=0.5]{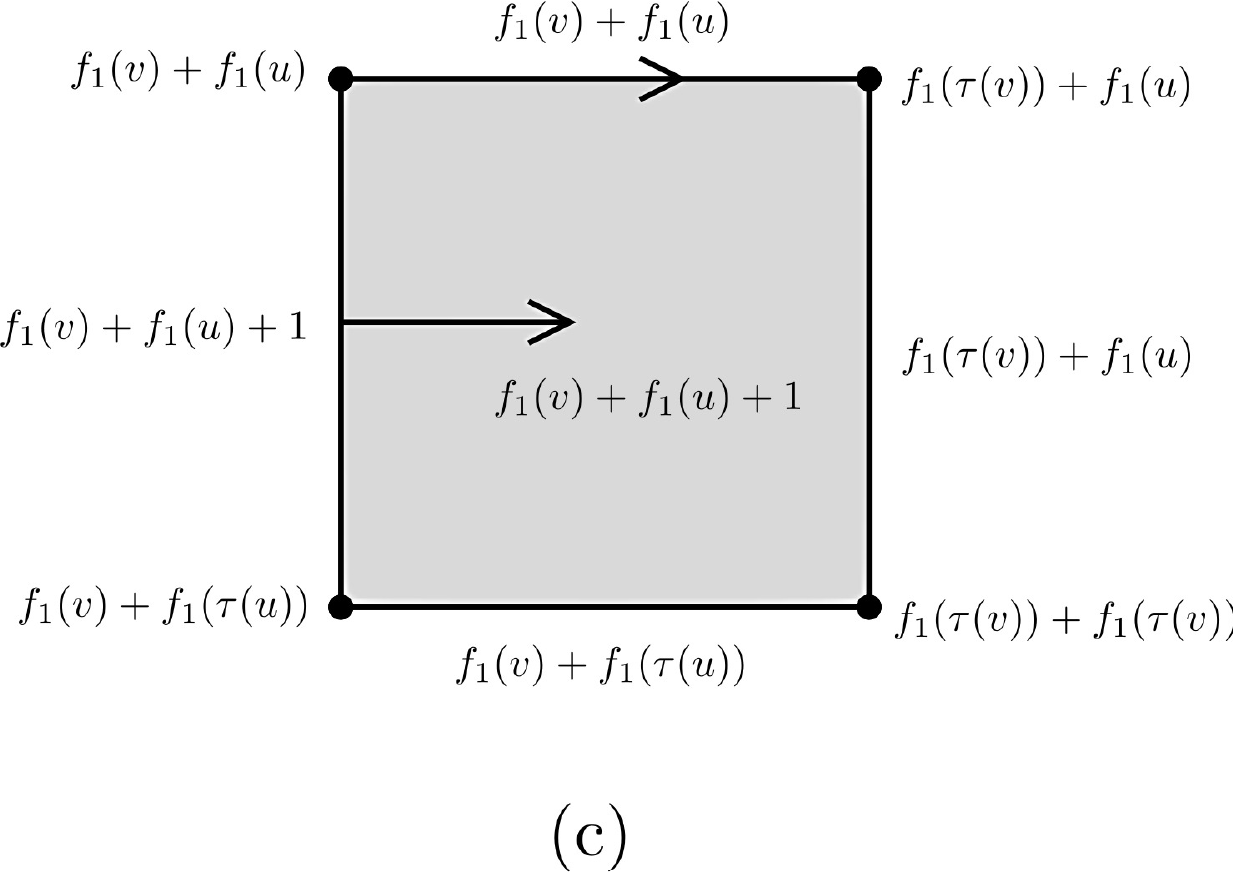}~~\includegraphics[scale=0.5]{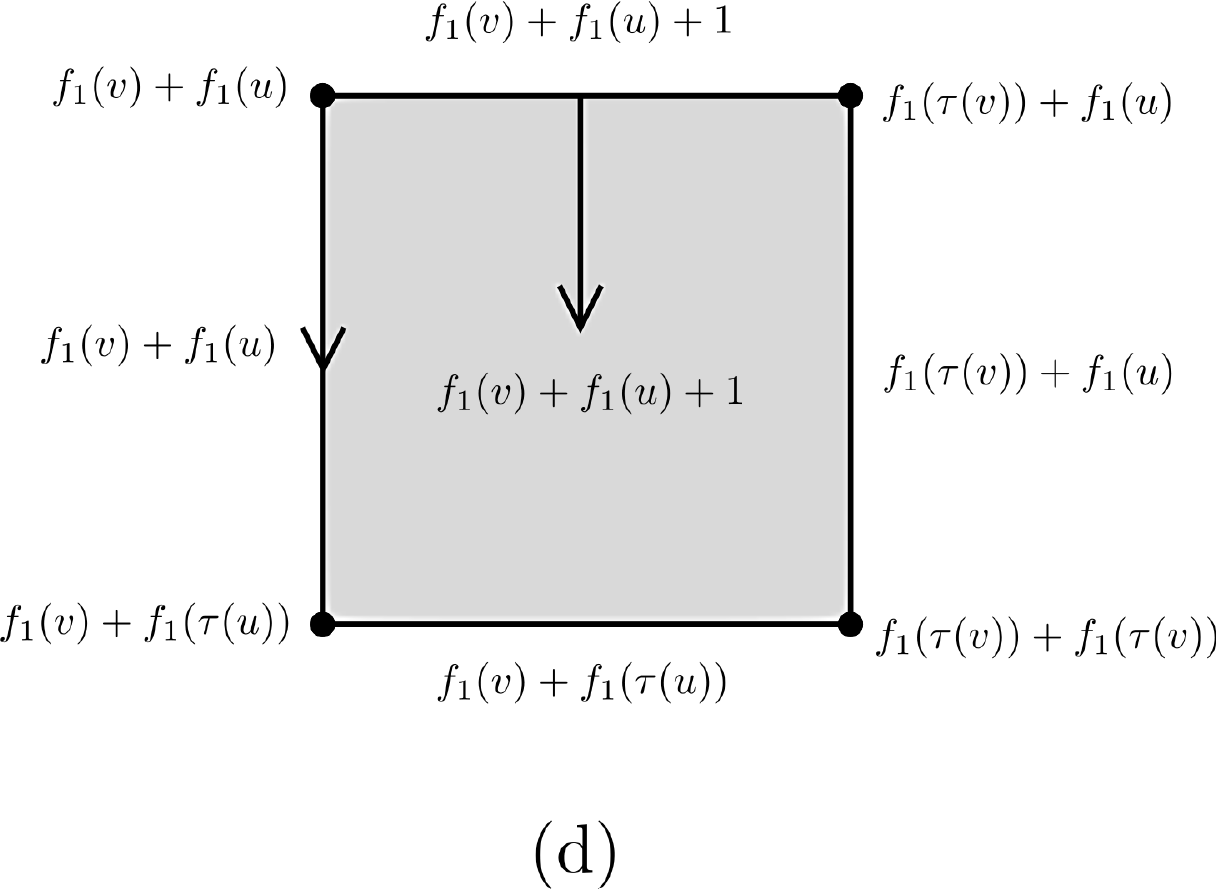}

\caption{(a) Two edges of $T$ with $e(v)\cap e(u)=\emptyset$, (b) The
problem of $2$-cell $e(v)\times e(u)$ (c),(d) two possible fixings
of $\tilde{f}_{2}$ }\label{figure12}
\end{figure}

\begin{Fact}\label{fact4}
For the $1$-cells $\beta=v\times e(u)$, where $e(u)\in T$ and $e(v)\cap e(u)\neq\emptyset$ the conditions (\ref{betacondition1}, \ref{betacondition2}) are satisfied.
\end{Fact}
\noindent Proof. Let us first calculate $\bar{f}_2(\beta)$. To this end we have to check if $\beta$ was modified in step 1. Notice that every $2$-cell which has $\beta$ in its boundary is one of the following forms:
\begin{enumerate}
    \item $e(v)\times e(u)$\label{1}
    \item $e\times e(u)$ with $e\in D_v$\label{2}
    \item $e\times e(u)$ with $e\in T_v$\label{3}

\end{enumerate}
Case (\ref{1}) is impossible since $e(v)\cap e(u)\neq\emptyset$. For any $2$-cell belonging to (\ref{2}) the value of $\tilde{f}_2$ was not modified on the boundary of $e\times e(u)$ (see fact \ref{fact2}). Finally, for $2$-cells belonging to (\ref{3}) the value of $\tilde{f}_2$ was modified on the boundary of $e\times e(u)$ but not on the cell $\beta$ (see fact 3). Hence $\bar{f}_2(v\times e(u))=\tilde{f}_2(v\times e(u))=f_1(v)+f_1(e(u))=f_1(v)+f_1(u)$. Let us now verify condition (\ref{betacondition2}). The $1$-cell $\beta$ is adjacent to exactly two $0$-cells, namely $v\times u$ and $v\times \tau(u)$. We have $\bar{f}_2(v\times u)=\tilde{f}_2(v\times u)=f_1(v)+f_1(u)$ and  $\bar{f}_2(v\times \tau(u))=\tilde{f}_2(v\times \tau(u))=f_1(v)+f_1(\tau(u))$. Now since $f_1(\tau(u))<f_1(u)$ condition (\ref{betacondition2}) is satisfied. For condition (\ref{betacondition1}) we have only to examine $2$-cells of forms (\ref{2}) and (\ref{3}) (listed above). For $2$-cells that belong to (\ref{2}) we have $f_2(e\times e(u))=f_1(e)+f_1(e(u))>f_1(v)+f_1(u)+2$ and for $2$-cells that belong to (\ref{3}) we have $f_2(e\times e(u))=f_1(e)+f_1(e(u))+1>f_1(v)+f_1(u)+1$. Hence in both cases $\bar{f}_2(e\times e(u))>\bar{f}_2(v\times e(u))$ and condition (\ref{betacondition1}) is satisfied.

\begin{Fact}\label{fact5}
For the $1$-cells $\beta=v\times e$, where $e\notin T$ and $e(v)\cap e\neq\emptyset$ conditions (\ref{betacondition1}, \ref{betacondition2}) are satisfied.
\end{Fact}
\noindent Proof. Let us first calculate $\bar{f}_2(\beta)$. To this end we have to check if $\beta$ was modified in step 1. Notice that every $2$-cell which has $\beta$ in its boundary is one of the following forms:
\begin{enumerate}
    \item $e(v)\times e$\label{c1}
    \item $e_i\times e$ with $e_i\in D_v$\label{c2}
    \item $e_i\times e$ with $e_i\in T_v$\label{c3}

\end{enumerate}
Case (\ref{c1}) is impossible since $e(v)\cap e\neq\emptyset$. For any $2$-cell belonging to (\ref{c2}) or (\ref{c3}) the value of $\tilde{f}_2$ was not modified on the boundary of $e_i\times e(u)$ (see fact \ref{fact1} and \ref{fact2}). Hence $\bar{f}_2(v\times e)=\tilde{f}_2(v\times e)=f_1(v)+f_1(e)$. Let us now verify condition (\ref{betacondition2}). To this end assume that $e=(j,k)$ with $j>k$. The $1$-cell $\beta$ is adjacent to exactly two $0$-cells, namely $v\times j$ and $v\times k$. We have $\bar{f}_2(v\times j)=\tilde{f}_2(v\times j)=f_1(v)+f_1(j)$ and  $\bar{f}_2(v\times k)=\tilde{f}_2(v\times k)=f_1(v)+f_1(k)$. Now since $f_1(e)=\mathrm{max}(f_1(j),f_1(k))+2$ condition (\ref{betacondition2}) is satisfied. For condition (\ref{betacondition1}) we have only to examine $2$-cells of forms (\ref{c2}) and (\ref{c3}) (listed above). It is easy to see that in both cases  $\bar{f}_2(e_i\times e)>\bar{f}_2(v\times e)$.

\begin{Fact}\label{fact6}
For the $1$-cells $\beta=v\times e(u)$, where $e(u)\in T$ and $e(v)\cap e(u)=\emptyset$ conditions (\ref{betacondition1}, \ref{betacondition2}) are satisfied.
\end{Fact}

\noindent Proof. Let us first calculate $\bar{f}_2(\beta)$. To this end we have to check if $\beta$ was modified in step 1. Notice that every $2$-cell which has $\beta$ in its boundary is one of the following forms:
\begin{enumerate}
    \item $e(v)\times e(u)$\label{ccc1}
    \item $e\times e(u)$ with $e\in D_v$\label{ccc2}
    \item $e\times e(u)$ with $e\in T_v$\label{ccc3}

\end{enumerate}
For any $2$-cell belonging to (\ref{2}) the value of $\tilde{f}_2$ was not modified on the boundary of $e\times e(u)$ (see fact \ref{fact2}). For the $2$-cells belonging to (\ref{3}) the value of $\tilde{f}_2$ was modified on the boundary of $e\times e(u)$ but not on the cell $\beta$ (see fact 3). Finally for the $2$-cell $e(v)\times e(u)$ the value of $\tilde{f}_2$ was modified on the boundary of $e(v)\times e(u)$ and by fact \ref{fact3} it might be the case that it was modified on $\beta$. Hence $\bar{f}_2(v\times e(u))=\tilde{f}_2(v\times e(u))=f_1(v)+f_1(e(u))=f_1(v)+f_1(u)$ or $\bar{f}_2(v\times e(u))=f_1(v)+f_1(u)+1$. Let us now verify condition (\ref{betacondition2}). The $1$-cell $\beta$ is adjacent to exactly two $0$-cells, namely $v\times u$ and $v\times \tau(u)$. We have $\bar{f}_2(v\times u)=\tilde{f}_2(v\times u)=f_1(v)+f_1(u)$ and  $\bar{f}_2(v\times \tau(u))=\tilde{f}_2(v\times \tau(u))=f_1(v)+f_1(\tau(u))$. Now since $f_1(\tau(u))<f_1(u)$ condition (\ref{betacondition2}) is satisfied. For condition (\ref{betacondition1}) we have to examine $2$-cells from (\ref{ccc1}), (\ref{2}) and (\ref{3}) (listed above). In case when $\bar{f}_2(v\times e(u))=f_1(v)+f_1(u)$ it is easy to see that $\bar{f}_2(e\times e(u))>\bar{f}_2(v\times e(u))$ for $e\in D_v,T_v$ and $\bar{f}_2(e(v)\times e(u))>\bar{f}_2(v\times e(u))$. For $\bar{f}_2(v\times e(u))=f_1(v)+f_1(u)+1$ we still have $\bar{f}_2(e\times e(u))>\bar{f}_2(v\times e(u))$ for $e\in D_v,T_v$ and $\bar{f}_2(e(v)\times e(u))=\bar{f}_2(v\times e(u))$. Hence condition (\ref{betacondition1}) is satisfied in both cases.

\begin{Fact}\label{fact7}
For the $1$-cells $\beta=v\times e$, where $e\notin T$ and $e(v)\cap e=\emptyset$ conditions (\ref{betacondition1}, \ref{betacondition2}) are satisfied.
\end{Fact}

\noindent Proof. Let us first calculate $\bar{f}_2(\beta)$. To this end we have to check if $\beta$ was modified in step 1. Notice that every $2$-cell which has $\beta$ in its boundary is one of the following forms:
\begin{enumerate}
    \item $e(v)\times e$\label{cc1}
    \item $e_i\times e$ with $e_i\in D_v$\label{cc2}
    \item $e_i\times e$ with $e_i\in T_v$\label{cc3}

\end{enumerate}
For any $2$-cell belonging to (\ref{cc1}), (\ref{cc2}) and (\ref{cc3}) the value of $\tilde{f}_2$ was not modified on the boundary of an appropriate $2$-cell (see fact \ref{fact2} and \ref{fact3}). Hence $\bar{f}_2(v\times e)=\tilde{f}_2(v\times e)=f_1(v)+f_1(e)$. Let us now verify condition (\ref{betacondition2}). To this end assume that $e=(j,k)$ with $j>k$. The $1$-cell $\beta$ is adjacent to exactly two $0$-cells, namely $v\times j$ and $v\times k$. We have $\bar{f}_2(v\times j)=\tilde{f}_2(v\times j)=f_1(v)+f_1(j)$ and  $\bar{f}_2(v\times k)=\tilde{f}_2(v\times k)=f_1(v)+f_1(k)$. Now since $f_1(e)=\mathrm{max}(f_1(j),f_1(k))+2$ condition (\ref{betacondition2}) is satisfied. For condition (\ref{betacondition1}) we have to examine $2$-cells form (\ref{cc1}), (\ref{cc2}) and (\ref{cc3}) (listed above). It is easy to see that  $\bar{f}_2(e_i\times e)>\bar{f}_2(v\times e)$ for $e_i\in D_v,\,T_v$ and $\bar{f}_2(e(v)\times e)=\bar{f}_2(v\times e)$.

\begin{Fact}\label{fact8}
For the $0$-cell $\kappa=u\times v$ such that $e(v)\cap e(u)\neq\emptyset$ with the terminal vertex $\tau(v)$ of $e(v)$ equal to $u$, condition (\ref{betacond}) is satisfied.
\end{Fact}
Proof.
The situation when $e(v)\cap e(u)\neq\emptyset$ and terminal vertex
$\tau(v)$ of $e(v)$ is equal to $u$ is presented in the figure
\ref{figure13}. For the $0$-cell $v\times u$ we have $\bar{f}_2=\tilde{f}_{2}(v\times u)=f_{1}(v)+f_{1}(u)$.
Notice that there is exactly one edge $v\times e(u)$ for which $\bar{f}_{2}\left(v\times e(u))\right)=\bar{f}_{2}(v\times u)$.
The function $\bar{f}_{2}$ on the other edges adjacent to $v\times u$
have a value greater than $\bar{f}_{2}(v\times u)$ and hence $v\times u$
and $v\times e(u)$ constitute a pair of noncritical cells.
\begin{figure}[H]
~~~~~~~~~~~~~~~~~~~~~~~~~~~~~~\includegraphics[scale=0.6]{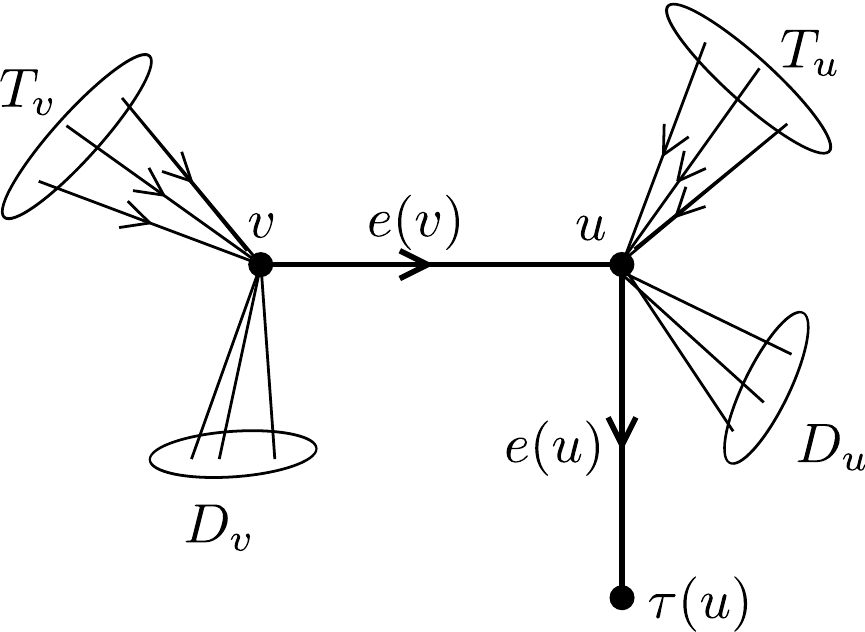}~~~~~~~

\caption{$e(v)\cap e(u)\neq\emptyset$ and $\tau(v)=u$}\label{figure13}

\end{figure}

\begin{Fact}\label{fact9}
For the $0$-cell $\kappa=u\times v$ such that $e(v)\cap e(u)\neq\emptyset$ with the terminal vertex $\tau(u)$ of $e(u)$ equal to the terminal vertex $\tau(v)$ of $e(v)$ condition (\ref{betacond}) is not satisfied. There are exactly two $1$-cells $\beta_1,\beta_2\supset\kappa$ such that $\bar{f}_2(\beta_1)=\bar{f}_2(\kappa)=\bar{f}_2(\beta_2)$. They are of the form $\beta_1=u\times e(v)$ and $\beta_2=v\times e(u)$. The function $\bar{f}_2$ can be fixed in two ways. We put $f_2(\beta_1):=\bar{f}_2(\beta_1)+1$ or $f_2(\beta_2):=\bar{f}_2(\beta_2)+1$. \end{Fact}
\noindent Proof.
The situation when $e(v)\cap e(u)\neq\emptyset$ and terminal vertex
$\tau(u)$ of $e(u)$ is equal to terminal vertex $\tau(v)$ of $e(v)$
is presented in the figure \ref{figure14}(a),(b). For the $0$-cell $v\times u$ we
have $\bar{f}_2(v\times u)=f_{1}(v)+f_{1}(u)$. There are two
edges $v\times e(u)$ and $u\times e(v)$ such that $\bar{f}_{2}(v\times e(u))=\bar{f}_{2}(v\times u)=\bar{f}_{2}(u\times e(v))$.
It is easy to see that the value of $\bar{f}_{2}$ on the other edges adjacent to $v\times u$
is greater than $\bar{f}_{2}(v\times u)$. So the function $\bar{f}_{2}$
does not satisfy condition (\ref{betacond}) and there are two possibilities \ref{figure14}(c),(d) to fix this problem. Either we put
$\bar{f}_{2}(v\times e(u))=\bar{f}_{2}(v\times u)+1$ or
$\bar{f}_{2}(u\times e(v))=\bar{f}_{2}(v\times u)+1$. They
both yield that the vertex $v\times u$ is non-critical. Notice finally that by the definitions of $f_1$ and $\tilde{f}_2$, increasing the value of $\bar{f}_2(\beta_i)$ by one does not influence $2$-cells containing $\beta_i$ in their boundary.
\begin{figure}[h]
~~~~\includegraphics[scale=0.6]{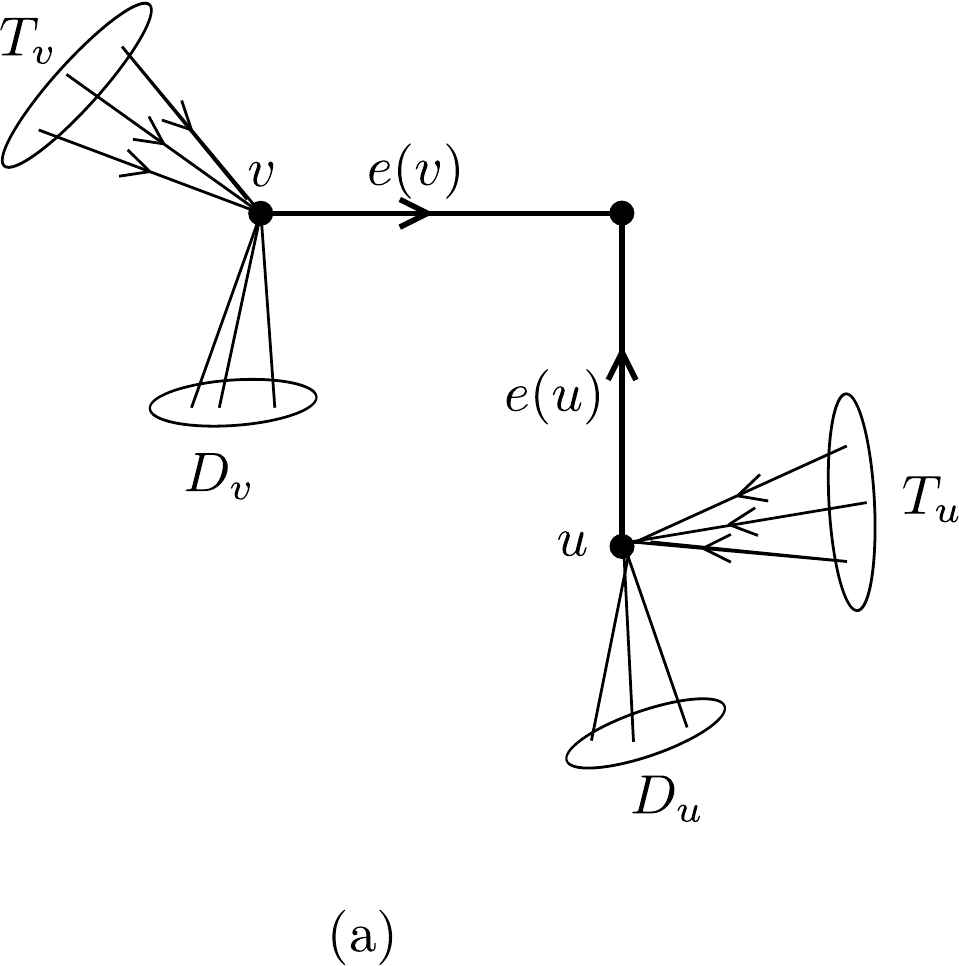}~~~~\includegraphics[scale=0.5]{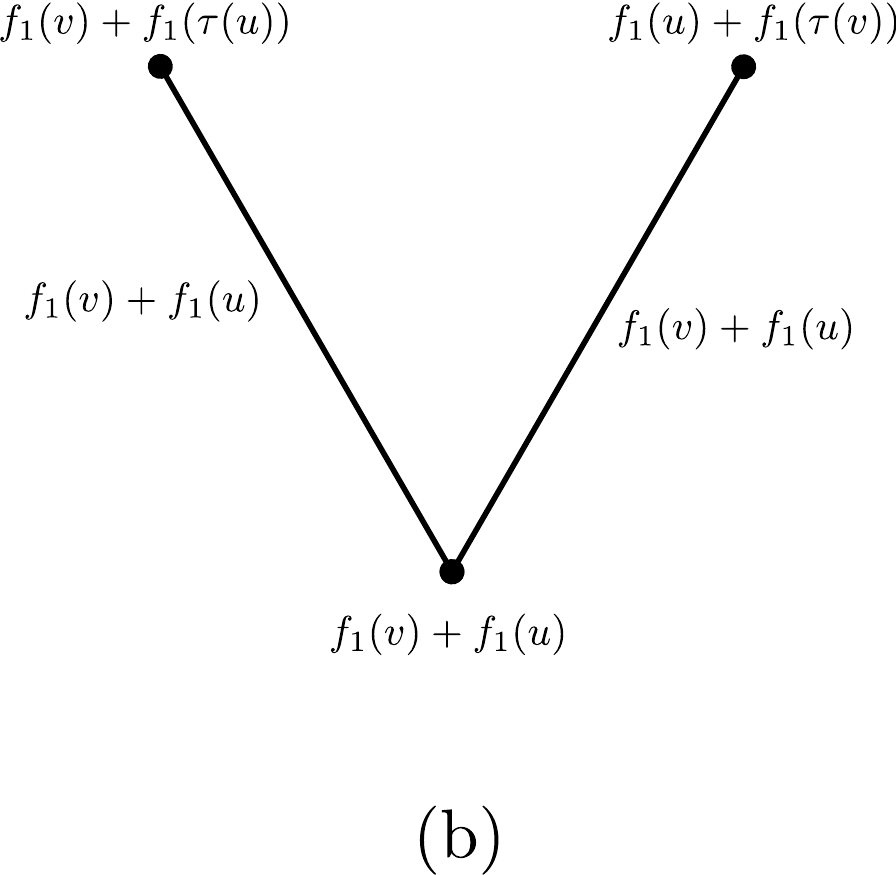}

\bigskip{}

~~~~\includegraphics[scale=0.5]{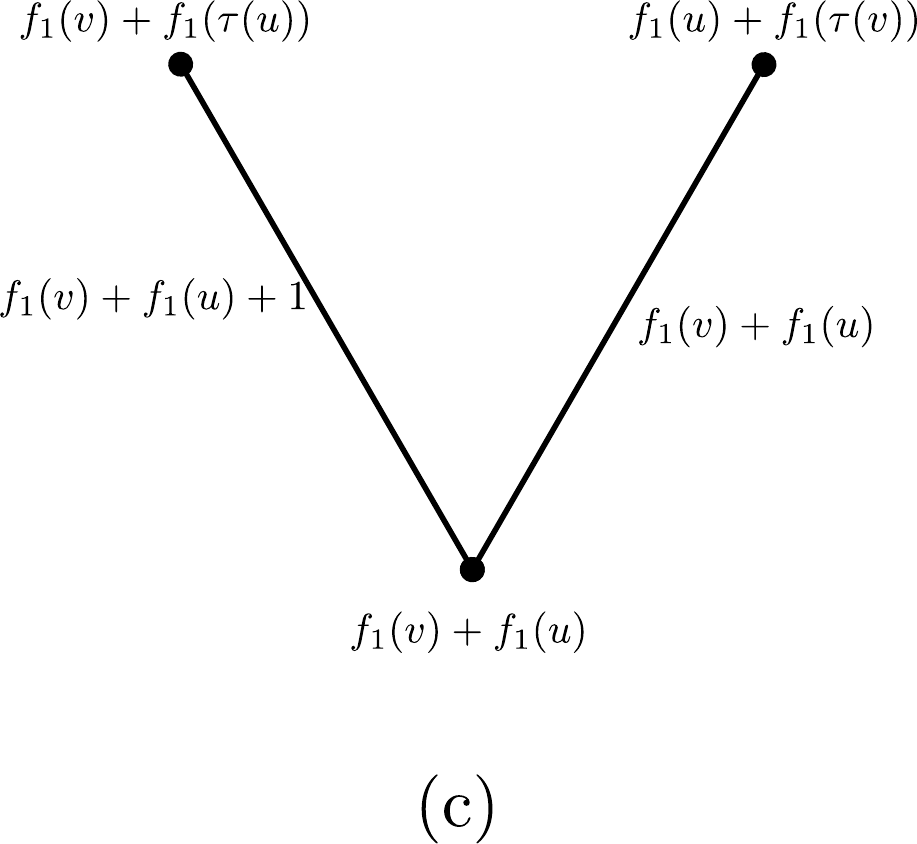}~~~~~\includegraphics[scale=0.5]{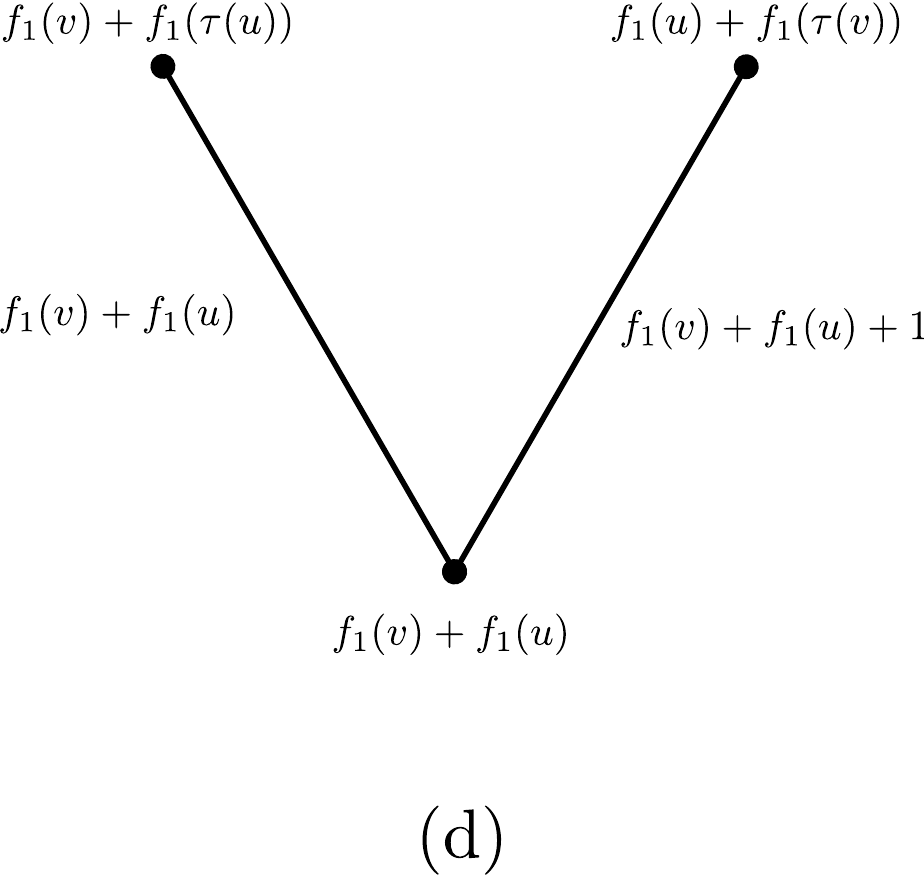}

\caption{(a) Two edges of $T$ with $e(v)\cap e(u)\neq\emptyset$, (b) The problem of $1$-cells $v\times (u,\tau(u))$ and $u\times (v,\tau(v))$ (c),(d) The two possible fixings of $\bar{f}_{2}$}\label{figure14}
\end{figure}

\begin{Fact}\label{fact10}
For the $0$-cell $\kappa=u\times v$ such that $e(v)\cap e(u)=\emptyset$ condition  (\ref{betacond}) is satisfied.
\end{Fact}
\noindent Proof. This is a direct consequence of the modification made for the $2$-cell $\alpha=e(v)\times e(u)$ in step 1. Moreover, $\kappa$ is noncritical.

\begin{Fact}\label{fact11}
For the $0$-cell $\kappa=1\times u$ condition  (\ref{betacond}) is satisfied.
\end{Fact}
\noindent Proof.  For the $0$-cell $1\times u$ we have $\bar{f}_2=\tilde{f}_{2}(v\times u)=f_{1}(u)$.
Notice that there is exactly one edge $1\times e(u)$ for which $\bar{f}_{2}\left(1\times e(u)\right)=\bar{f}_{2}(1\times u)$.
The function $\bar{f}_{2}$ on the other edges adjacent to $1\times u$
have a value greater than $\bar{f}_{2}(1\times u)$. Hence if $u\neq 2$ the $0$-cell  $1\times u$
and the $1$-cell $1\times e(u)$ constitute a pair of noncritical cells. Otherwise $\kappa$ is a critical $0$-cell.

\chapter{Summary and outlook}

In this thesis we developed a new set
of ideas and methods which gave a full characterization of all possible
abelian quantum statistics on graphs. Our approach enabled
identification of the key topological determinants of the quantum
statistics: 
\begin{enumerate}
\item  the connectivity of a graph,
\item the first homology
group $H_{1}(C_{n}(\Gamma))=\mathbb{Z}^{\beta_{1}}\oplus A$, where
$\beta_{1}$ is the number of independent cycles in $\Gamma$ and
$A$ determines quantum statistics 
\item for $1$-connected graphs number
of anyon phases depends on the number of particles,
\item for $2$-connected
graphs quantum statistics stabilizes with respect to the number of
particles $H_{1}(C_{n}(\Gamma))=H_{1}(C_{2}(\Gamma))$,
\item  for $3$-connected
non-planar graphs $A=\mathbb{Z}_{2}$, i.e. the usual bosonic/fermionic
statistics is the only possibility whereas planar $3$-connected graphs
support one anyon phase, $A=\mathbb{Z}$. Thus, from the quantum statistics
perspective, one can say that $3$-connected graphs mimic $\mathbb{R}^{2}$
when they are planar and $\mathbb{R}^{3}$ when not.
\end{enumerate}

It seems that the following problems can be approached using the methods developed in this thesis.

\noindent\textbf{ Problem 1}. It was noticed by V. I. Arnold in the
late 1960's \cite{Arnold1, Arnold2}, and then generalized to some classes of manifolds, that
the cohomology groups of the $C_{n}(\mathbb{R}^{2})$ possess three
basic properties:
\begin{enumerate}
\item finiteness: $H^{i}(C_{n}(\mathbb{R}^{2}))$
are finite except $H^{0}(C_{n}(\mathbb{R}^{2}))=\mathbb{Z}$, $H^{1}(C_{n}(\mathbb{R}^{2}))=\mathbb{Z}$
for $n\geq2$; also $H^{i}(C_{n}(\mathbb{R}^{2}))=0$ for $i\geq n$,
\item recurrence: $H^{i}(C_{2n+1}(\mathbb{R}^{2}))=H^{i}(C_{2n}(\mathbb{R}^{2}))$,
\item stabilization: $H^{i}(C_{n}(\mathbb{R}^{2}))=H^{i}(C_{2i-2}(\mathbb{R}^{2}))$
for $n\geq2i-2$. 
\end{enumerate}
These raises the following questions in graph's context: 
\begin{itemize}
\item what is
the minimal connectivity of $\Gamma$ that gives stabilization of
$H_{i}(C_{n}(\mathbb{R}^{2}))$ for planar and non-planar graphs?,
\item  Is 'quantum statistics' components of $H_{i}(C_{n}(\Gamma))$
given by the torsion part of $H_{i}(C_{n}(\Gamma))$, for $i>1$,
\item what is the minimal connectivity of $\Gamma$ for which 'quantum
statistics' components of $H_{i}(C_{n}(\Gamma))$ up to the given $i$
are the same as for $\mathbb{R}^{2}$ and $\mathbb{R}^{3}$, i.e.
when planar graphs mimic $\mathbb{R}^{2}$ up to $H_{i}(C_{n}(\mathbb{R}^{2})$
and non-planar graphs mimic $\mathbb{R}^{3}$ up $H_{i}(C_{n}(\mathbb{R}^{3})$
for given $i$.

\end{itemize}

\noindent\textbf{ Problem 2.} The aim is to lay the foundations for
the understanding of the influence of complex topology, which gives
rise to generalized anyon statistics, on many-particle transport properties
of complex networks. The principal attraction of quantum graphs is
that they provide mathematically tractable models of complex physical
systems. The fact that anyon statistics is present for many-particle
graph configuration spaces gives at least \textit{a priori} various
possible applications of this model. Using graph models one can investigate
topological signatures and effects of quantum statistics in many-particle
generalizations of a single-particle transport on networks. This should
provide models and variants of the quantum Hall effect extending to
many-particle quantum systems the transport theory for networks developed
by Avron (see for example  \cite{Avron1, Avron2, Avron3}).

\noindent\textbf{ Problem 3.} The importance of topology and geometry in quantum information theory is present on both foundational and application levels. Of course the Holy Grail in this area of research is still the construction of a quantum computer. One of the difficulties in building a many-qubit quantum computer is quantum decoherence. Physical systems typically remain in a coherent superposition of states for a very short time because generic interactions with the environment will decohere them, destroying the information encoded in quantum states. Recently, a new approach based on topology has been proposed to overcome some of the difficulties of this kind \cite{K103}. In simple words the idea is motivated by the fact that topological invariants are very robust. So if information is encoded in topology it is hard to destroy it as it is immune to a large class of perturbations. More precisely, topological quantum computing is based on the concept of anyons, and in particular, non-abelian anyons \cite{K103}. One of the most profound examples of these ideas is the celebrated Kitaev toric code, which is a realization of topological quantum error correcting code on a two-dimensional spin lattice \cite{K297}. The excitations for this model were proved to be of anyon type \cite{K297}. It is therefore natural to expect that the anyon statistics which are present on graph configuration spaces might be related to these ideas. One of the explicit tasks would be to construct a spin graph model for which excitations behave exactly like anyons corresponding to many-particle graph configuration spaces. It is also believed that the fractional Quantum Hall States are promising candidates for physical realization of topological computing \cite{N08}. So the study of transport properties described in the previous paragraphs is inevitably related to these concepts

\noindent\textbf{ Problem 4.} The entanglement of integer and fractional Quantum Hall States has recently been studied by several authors (see for example \cite{RS09,LBSH10}). An interesting problem would be to calculate the entanglement of eigenstates of many-particle graph configuration spaces with the topological gauge potential supporting anyon quantum statistics. When the topological gauge potential vanishes the Hamiltonian of the system is a non-interacting fermionic Hamiltonian. When this Hamiltonian has a non-degenerate spectrum, its eigenstates are given in terms of Slater determinants. Otherwise the topological gauge potential introduces an interaction to the system, so that eigenstates of Hamiltonian might be entangled. It seems interesting to understand how the degree of entanglement for these states is related to the topological invariants of the one-particle graph, e.g. its connectivity and planarity.


\bibliographystyle{plain}

\end{document}